\documentclass[article,12pt,oneside]{book}
\usepackage{bbm,framed,hyperref, bm,enumitem,subfigure,tcolorbox,nicefrac,mleftright,multirow}
\usepackage{amssymb,amsmath,amsfonts,amsthm,cases,pdfsync,color,graphicx,overpic,mathtools}
\usepackage[ruled,commentsnumbered]{algorithm2e} 
\usepackage[noend]{algorithmic}
\usepackage{mwe,float}
\newcommand{\beqn}{\begin{eqnarray*}}
\newcommand{\eeqn}{\end{eqnarray*}}
\newcommand{\beq}{\begin{equation}}
\newcommand{\eeq}{\end{equation}}

\newcommand{\R}{\mathbb{R}}

\newcommand{\argmin}{\arg\min}
\newcommand{\argmax}{\arg\max}

\newcommand{\mean}[1]{\mathbb{E}\!\left[#1\right]}

\newcommand{\prob}[1]{\Pr\left[#1\right]}

\newcommand{\indic}[1]{\mathbbm{1}_{\{#1\}}}

\newcommand{\Zipf}{\tau}
\newcommand{\youtube}{{YouTube}}
\newcommand{\sfrac}{\nicefrac}
\mleftright
\DeclareSymbolFont{symbolsC}{U}{pxsyc}{m}{n}
\DeclareMathSymbol{\multimapdot}{\mathrel}{symbolsC}{20}
\DeclareMathSymbol{\multimapbothvert}{\mathrel}{symbolsC}{149}
\newcommand{\downtruncated}{{ _\text{\rotatebox[origin=c]{270}{$\multimapdot$}}}}
\newcommand{\uptruncated}{{ _\text{\rotatebox[origin=c]{90}{$\multimapdot$}}}}
\newcommand{\nottruncated}{{ _{\text{\scalebox{.85}{$\multimapbothvert$}}}}}

\DeclareSymbolFont{symbolsC}{U}{pxsyc}{m}{n}
\DeclareMathSymbol{\multimapdot}{\mathrel}{symbolsC}{20}
\DeclareMathSymbol{\multimapbothvert}{\mathrel}{symbolsC}{149}

\newcommand{\ev}{{\bf e}} 

\newcommand{\Yc}{{\cal Y}}
\newcommand{\Xc}{{\cal X}}
\newcommand{\Sc}{{\cal S}}
\newcommand{\Nc}{{\cal N}}

\newcommand{\Mc}{{\cal M}}

\newcommand{\Sigmam}{\hbox{\boldmath$\Sigma$}}
\newcommand{\Mone}{{I1}}
\newcommand{\Mtwo}{{I2}}

\tcbuselibrary{theorems}
\tcbset{
	defstyle/.style={fonttitle=\bfseries\upshape,fontupper=\itshape,
		colframe=white,colback=white, colbacktitle=red!20!white, coltitle=black, sharp corners},
	theostyle/.style={fonttitle=\bfseries\upshape,fontupper=\itshape,
		colframe=white,colback=white!94!black,coltitle=blue!75!black, sharp corners},
}
\newtcbtheorem[number within=chapter]{dfn}{Definition}{defstyle}{dfn}
\newtcbtheorem[use counter from=dfn]{thm}{Theorem}{theostyle}{thm}
\newtcbtheorem[use counter from=dfn]{lem}{Lemma}{theostyle}{lem}
\newtcbtheorem[use counter from=dfn]{cor}{Corollary}{theostyle}{cor}

\newtcolorbox{box_example}[2][]{colbacktitle=red!10!white,colframe=white,
	colback=red!1!white,coltitle=red!70!black,
	title={#2}, fonttitle=\bfseries,#1}

\newtcolorbox{opt}[1]{title={#1},fonttitle=\sffamily\bfseries\large,colbacktitle=black!60!white,colframe=white!97!black,colback=white!94!black}

\begin{document}
\title{Cache Optimization Models and Algorithms}
\newcounter{one}
\setcounter{one}{1}
\newcounter{two}
\setcounter{two}{2}
\newcounter{three}
\setcounter{three}{3}
\author{
Georgios S. Paschos$^{\fnsymbol{one}}$, George Iosifidis$^{\fnsymbol{two}}$ and Giuseppe Caire$^{\fnsymbol{three}}$\vspace{0.1in}\\
\phantom{a}\\
$^\fnsymbol{one}$
Huawei Technologies Co.~Ltd., France \\
$^\fnsymbol{two}$  Trinity College, Ireland \\
$^\fnsymbol{three}$ TU-Berlin, Germany
}

\maketitle
\tableofcontents
%
\chapter{Introduction}\label{ch:intro}

Storage resources and caching techniques permeate almost every area of communication networks today. 
In the near future, caching is set to play an important role in storage-assisted Internet architectures,  information-centric networks, and wireless systems,  reducing operating and capital expenditures and improving the services offered to users.
 In light of the remarkable data traffic growth and the increasing number of rich-media applications, the impact of caching is expected to become even more profound than it is today. Therefore, it is crucial to design these systems in an optimal fashion, ensuring the maximum possible performance and economic benefits from their deployment. To this end, this article presents a collection of detailed models and algorithms, which are synthesized to  build a powerful analytical framework for caching optimization.

\section{Historical background and scope}

The term \emph{cache} was introduced in computer systems in 1970s to describe a memory with very fast access but typically small capacity. In computer applications, memory access often exhibits locality, i.e., the majority of requests are related to memory blocks in a specific area known as \emph{hot spot}. {By replicating these hot spots on a cache it is possible to accelerate the performance of the entire memory system.}
One of the most important first problems in this context was to select which memory blocks to replicate in order to maximize the expected benefits. The design of such \emph{caching policies} remains as  one of the main challenges {in caching systems} 
and several important results {were developed in the early era of computer systems}, for instance the optimality of the oracle MIN policy \cite{Belady66}.

The above caching idea was later applied to the Internet. As the population of users was growing fast in 1990s, the classical client-server connection model became impractical since all content requests (for web pages, in particular) were routed to few central servers. This was creating server congestion leading to an unsustainable network, and the idea of using \emph{Internet caches} was proposed to address this issue. Mimicking the computer cache, an Internet cache is deployed closer to end users and replicates certain web pages. Given the content popularity skewness, i.e., few pages attracting the majority of requests, even small caches can have impressive performance benefits by storing only subsets of the web pages. Indeed, it soon became clear that \emph{web caching} reduces the network bandwidth usage and server congestion, and improves the content access time for users. The caching policy design is a more intricate problem in these interconnected web caches, and decisions such as content routing and cache dimensioning must be {jointly} devised.

The last few years we witness a resurgence of interest in caching in the domain of wireless networks. The expansive growth of mobile video traffic in conjunction with exciting developments{ --- like the use of  \emph{coding techniques} --- }
have placed caching at the forefront of research in wireless communications. There is solid theoretical and practical evidence today that memory can be a game-changer in our efforts to increase the effective throughput and other key performance metrics, and there are suggestions for deploying caches at the network core, the base stations or even the mobile devices. At the same time, novel services which involve in-network computations, require pre-stored information (e.g., various machine learning services), or are bounded by low latency operation constraints, can greatly benefit from caching. In fact, many caching enthusiasts argue that such services can only be deployed if they are supported by intelligent caching techniques. 

Amidst these developments, it is more important than ever to model, analyze, and optimize the performance of caching systems. Quite surprisingly, many existing caching solutions, albeit practical, have not been designed using rigorous mathematical tools. Hence, the question of whether they perform optimally remains open. At the same time, the caching literature spans more than 40 years, different systems and even different research communities, and we are lacking a much-needed unified view on caching problems and solutions. This article aspires to fill these gaps, by presenting the theoretical foundations of caching and the latest conceptual and mathematical advances in the area. We provide detailed technical arguments and proofs, aiming to create a stable link between the past and future of caching analysis, and offer a useful starting point for new researchers. In the remainder of this chapter, we set the ground by discussing certain key caching systems and ideas, and explain how the contents of this monograph are organized.

\section{The content delivery network}

A central idea in caching systems is the \emph{Content Delivery Network}. CDNs consist of (i) the origin server; (ii) the caches; (iii) the backbone network; and (iv) the  points of ingress user traffic. The origin server is often deployed at a remote location with enormous storage capabilities (e.g., a datacenter), and stores all the content items a user of this service might request, i.e., the content \emph{catalog}. The caches are smaller local server installations, which are distributed geographically {near demand, and are connected with the backbone network}.
Before CDNs, the users would establish TCP connections with the origin server in order to obtain the content items. In CDNs however, these TCP connections are redirected to caches, which can serve only the items that are already locally replicated. 
 Apart from content caching, modern CDN systems {perform} traffic optimization, offer \emph{Quality of Service} (QoS), and protect from \emph{Denial of Service} (DoS) attacks. However, in this article we will focus on the aspect of content caching, and specifically on the intelligence involved in orchestrating the caching operations.

An iconic CDN system is the ``{Akamai Intelligent Platform}'', see \cite{akamai} for a detailed description. Akamai was one of the most prominent CDN providers in the booming Internet era of 2000s, and today is responsible for delivering 20$\%$ approximately of the overall Internet traffic. Its 216K caching servers are dispersed at network edges offering low-latency content access around the globe. The {Akamai model} was designed to intercept http traffic using a DNS redirect: when a user wants to open a website with the http protocol, it would first contact the local DNS server to recover the IP of the origin  server. The intelligent platform  replaces the DNS entry with the IP of an Akamai cache containing the requested content, and hence the http request is eventually served by that cache. The intelligent operations are handled by the \emph{mapping} system, which decides where to cache each content item, and accordingly maps DNS entries to {caches}.
Although the mapping system is effectively deciding the {placement of content}, the local caches are also operated with reactive policies such as the famous LRU and its variants.

\subsubsection{Benefits of caching}

The replication of few popular contents can significantly \textbf{reduce the traffic at the backbone network}. When a requested content  is available at a nearby cache (an event called \emph{hit}), the user request is redirected to the path that connects it with that cache instead of the origin server. Therefore, caches are often scattered around the network to minimize the geodesic distance, and/or network hops, from potential requesters. Previous research has investigated solutions for the {optimal placement of servers}, e.g., \cite{qiu_replicas01}, and the sizing of cache storage, called {dimensioning of caches} \cite{kelly00}. {Since more hits mean less network traffic, an important criterion for deploying caches is the increase of the \emph{cache hit ratio}}. Related mathematical optimization problems in this context include the choice of the eviction policy, i.e., the dynamic selection of the contents that are evicted from an overflowing cache \cite{Sleator85}, as well as the strategic content placement for cache collaboration \cite{BektasCOR07}. 
Typically, different such policies are combined  to optimize the transportation of Internet traffic and make CDNs profitable. 

Another benefit of caching is \textbf{latency reduction}, i.e., the decrease of elapsed time between the initiation of a request and the content delivery. Typically, the latency improvement is attributed to cutting down {propagation latency}. Packets traversing a transcontinental link, for example, may experience latency up to $250msec$ due to the speed-of-light limitation \cite{singla2014internet}. Given that each TCP connection involves the exchange of several messages, it might take seconds before a requested content is eventually delivered over such links. These large latencies are very harmful for e-commerce and other real-time applications, and their improvement has been one of the main market-entry advantages of CDNs. Indeed, when a user retrieves  contents 
from a nearby cache, the distance is greatly reduced, and so is the propagation time that delays the content delivery. Nevertheless, latency optimization in caching systems is an intricate goal, and there are some notable misconceptions.

Firstly, in most applications latency effects smaller than $30msec$ do not impact the user experience. Hence, one needs to be cautious in increasing the infrastructure costs in order to achieve faster delivery than this threshold. In other words, when it comes to latency criteria, a single local cache often suffices to serve a large metropolitan area. Secondly, regarding video content delivery, the latency requirements apply only to the first video chunks and not on the entire file. Delivering fast the first chunks and then exploiting the device's buffer is adequate for ensuring smooth reproduction even if the later video segments are delivered with higher latency. Finally, several low latency applications, such as reactive virtual reality, vehicular control, or industrial automation, cannot typically benefit from caching since their traffic is not reusable. Nevertheless, we note that there are scenarios where one can exploit caching (e.g., using proactive policies) in order to boost the Quality-of-Service performance of such demanding services.

Another important effect of web caching is that it \textbf{balances the load of servers}. 
For example, the Facebook Photo CDN leverages web browser caches on user devices, edge regional servers, and other caches in order to reduce the traffic reaching the origin servers. Notably, browser caches serve almost $60\%$ of traffic requests, due to the fact that users view the same content multiple times. Edge caches serve $20\%$ of the traffic (i.e., approximately $50\%$ of traffic not served by browser caches), and hence offer important off-network bandwidth savings by serving locally the user sessions. Finally, the remaining $20\%$ of content requests are served at the origin, using a combination of slow back-end storage and a fast origin-cache \cite{Huang_2013}. This CDN functionality shields the main servers from high load and increases the scalability of the architecture. Note that the server load is minimized when the cache hits are maximized, and hence the problem of server load minimization is equivalent to cache hit maximization. Therefore, in the remaining of this article  we will focus on hit maximization, as well as bandwidth and latency minimization.

\section{Wireless caching}

Caching has also been considered for improving content delivery in wireless networks \cite{paschos_16}. There is growing consensus that network capacity enhancement through the increase of physical layer access rate or dense deployment of base stations is a costly approach, and also outpaced by the fast-increasing mobile data traffic \cite{cisco}. Caching techniques promise to fill this gap, and several interesting ideas have been suggested to this end: {(i)} deep caching at the {evolved packet core} (EPC) in order to reduce content delivery delay \cite{caching-cellular}; {(ii)} caching at the base stations to alleviate congestion in their throughput-limited backhaul links \cite{golrezaei2012femtocaching}; {(iii)} caching at the mobile devices to leverage device-to-device communications \cite{femtocaching_d2d}; and {(iv)} coded caching for accelerating transmissions over a broadcast medium \cite{MaddahAli2014Fundamental}.

Recently developed techniques that combine caching with coding 
 demonstrate revolutionary \emph{goodput} scaling in bandwidth-limited cache-aided networks. This motivated researchers to revisit the fundamental question of how memory ``interacts'' with other types of resources. The topic of \emph{coded caching} started as a powerful tool for broadcast mediums, and led 
  towards establishing an information theory for memory. 
  Similarly, an interesting connection between memory and processing has been identified \cite{Li18}, creating novel opportunities for improving the performance of distributed and parallel computing systems. These lines of research have re-stirred the interest in joint consideration of bandwidth, processing {and memory resources, and promise high performance gains.}

Furthermore, the advent of technologies such as Software-Defined Networking (SDN) and Network Function Virtualization (NFV) create new opportunities for leveraging caching. Namely, they enable the fine-grained and unified control of storage capacity, computing power and network bandwidth, and allow a flexible deployment of in-network caching services in time and space. This gives rise to the new concept of content-centric network architectures that aspire to use storage and caching as a means to revolutionize the Internet,  as well as new business models are emerging today as new players are entering the content delivery market. Service providers like Facebook are acquiring their own CDNs, network operators deploy in-network cache servers to reduce their bandwidth expenditures, and content providers like Google, Netflix, and Amazon use caches to replicate their content world-wide. Interestingly, smaller content providers can buy caching resources on the cloud market to instantiate their service \emph{just in time and space}. 
These novel concepts create, unavoidably, new research questions for caching architectures and the caching economic ecosystem, and one of our goals in this document is to provide the fundamental underlying mathematical theories that can support research in these exciting directions.

\section{Structure}

In this section we provide a quick summary of the article, serving both as a warm-up for reading the entire article, as well as a map with directions to specific information.

We begin in \textbf{\emph{Chapter 2}} with a detailed treatment of content popularity, a key factor for the performance of caching policies. We first  explain the power-law popularity model, and how we can infer its parameters from a given dataset, and then use it to determine the optimal cache size. We define the \emph{Independence Reference Model} (IRM) for describing a request {generation} process. IRM is a widely used model for caching analysis, but it has limited accuracy since it fails to capture correlation effects between requests, namely \emph{temporal and spatial locality}. We discuss the state-of-the-art mathematical models which are more accurate than IRM in that respect, but also more difficult to use in practice. For the case of temporal correlations, we provide the optimal rule for popular/unpopular content classification that maximizes cache performance.

In \textbf{\emph{Chapter 3}} we explore the realm of \emph{online eviction policies}. A single cache receives content requests and we must design a  rule for evicting a content when the cache overflows. The design of eviction policies is an equally challenging and important problem, and we present the key theoretical results in this domain. We begin with the case where requests are arbitrary: an oracle policy --- known as ``Belady'' --- is shown to achieve the maximum number of hits under any request sequence, or sample path. This policy requires knowledge of the future requests, and therefore it is useful only as a benchmark. Using the Belady policy we prove that the ``Least Recently Used'' (LRU) rule provides the best competitive performance among all online policies, i.e.,  those that do not know the future. Then, for stationary IRM requests the ``Least Frequently Used'' (LFU) rule is the optimal, as it estimates (the considered static) content popularity using the observed frequencies. We also study the \emph{characteristic time} approximation, with which we obtain the performance of LRU for stationary requests, as well that of \emph{Time To Live} (TTL) caches which allow to optimally tune individual content hit probabilities. Last, we depart from the stationary assumption and take a model-free approach inspired by the Machine Learning framework of \emph{Online Convex Optimization}. We present an adaptation of the Zinkevich's online gradient policy to the  caching problem, and show that it achieves the optimal \emph{regret}, i.e., the smallest losses with respect to the best static cache configuration with complete future knowledge.

In \textbf{\emph{Chapter 4}} we study caching networks (CNs), i.e., systems where multiple caches are interconnected via a network graph. Here, we focus  on \emph{proactive caching} policies which populate the caches based on expected, i.e., estimated, demand. In CNs, the designer needs to decide where to cache each content item (caching policy), how to route the content from caches to the requesters (routing policy), and often decisions such as the dimensioning of the different caches and network links. Therefore, we start the chapter by explaining this general \emph{CN design and management problem}. This is a notoriously hard problem, that cannot be solved optimally for large CNs and content libraries. To gain a better understanding into the available solution methodologies, we survey a number of key subproblems: (i) the \emph{cache dimensioning} problem where we  decide where to place storage in the network; (ii) the \emph{content caching} in bipartite and tree graphs; and  (iii) the \emph{joint content caching and routing} problem in general graphs. Although these are all special cases of the general CN problem, they are governed by significantly different mathematical theories. Therefore, our exposition in this chapter serves to clarify where each mathematical theory applies best, and how to get a good approximate guarantee for each scenario.

In the following \textbf{\emph{Chapter 5}} we take an approach that combines the two previous chapters. Here we study a CN where the content popularity is unknown, and therefore the objective of the chapter is to design an CN eviction \& routing policy that at the same time learns the content popularity, decides in an online manner what to cache, and in a reactive manner how to route the content to the requestor. Our approach is also a generalization of the Online Gradient Ascent (OGA) explained in Chapter 3. We show that a policy that takes a step in the direction of a subgradient of the previous slot utlility function can provide ``no regret'' in the CN case as well. However, finding a subgradient direction is a much more complicated than the gradient of Chapter 3, and hence we provide directions as to cast  this problem as a Linear Program.

In the last \textbf{\emph{Chapter 6}} we examine a scaling network of caches, arranged in a square grid. We relax the original combinatorial problem and obtain a relaxed solution via convex optimization,  shown to be of the same order of performance with the actual integral. 
The chapter includes detailed results about the sustainability of networks aiming to deliver content in different regimes of (i) network size, (ii) catalogue size, and (iii) cache size.


\chapter{Content popularity}\label{ch:2}

Managing a cache effectively requires deep understanding of content popularity and the principles underlying the content request sequences. This chapter introduces basic concepts, presents the important IRM request model with power law popularities, and then proceeds with more detailed models that capture time and space correlations.

\section{Introduction to caching-related terms}

We begin with a few  terms that  will be useful across the entire article.  
A piece of reusable information is called a \emph{content}. Examples of contents include \youtube~videos, Netflix movies, music files, Facebook pictures, http web pages, documents, news feeds, etc. On the other hand, examples of non-reusable information include  VoIP calls, teleconference, sms, and control signaling {which are not relevant to caching}.
Fortunately, reusable contents account today for  80\% of the total Internet traffic~\cite{cisco2015}. This makes caching extremely relevant for communication networks and content delivery platforms.

\begin{dfn}[theorem style=plain]{Content catalog}{}
We consider a set of contents  $\mathcal{N}=\{1,2,\dots,N\}$,  called the catalog. The size of the catalog is  $N=|\mathcal{N}|$.
\end{dfn}

For simplicity we assume that all contents are of equal size, {unless otherwise stated}.\footnote{Capturing  varying sizes in the analysis is possible and interesting, see  \cite{gds,Fricker12}; {and we present some examples in Chapter 4.}} Sometimes we measure the catalog size in Bytes (B), with the understanding that $N$ can be recovered if we divide with the file size. The catalog, and its size,  depend on the application. For instance, the catalog of videos in a Video-on-Demand (VoD) application is in the order of TBs ($10^{12}$Bs), while it is estimated that the total catalog of the Internet has surpassed $10^{24}$Bs \cite{hilbert2012much,HowBigInternet16}.

Due to memory limitations, a cache can typically store only a subset of the catalog $\mathcal{M}\subset \mathcal{N}$, and we denote with $M=|\mathcal{M}|$ the total number of contents that fit in the cache. The ratio $\gamma=M/N \in [0,1]$ is the fraction of catalog that can be cached, called the \emph{relative cache size}. The value of $\gamma$ affects the cache performance. Ideally, we would like $\gamma$ to be as large as possible, since $\gamma=1$ means that we can cache the entire catalog, and therefore satisfy all possible requests.

However, in most applications $\gamma$ is very small. For example, \youtube's catalog grows by roughly 30PB every year, and the standing estimate today is 300PB. On the other hand, the Global Google Cache system --- used to cache \youtube~videos --- is based on  disk arrays with estimated size 200TB. 
 Related figures for Facebook catalog \cite{facebook_url}, Facebook photos \cite{Huang_2013}, and Netflix \cite{netflix_url}  are summarized in Table \ref{tab:gamma}. We see that the cache size is in most cases much smaller than the catalog. However, a small cache may  be useful due to the fact that some contents are very popular and may be requested very frequently. Hence 
caching may be effective even for small $\gamma$. To shed more light into this possibility, this chapter studies  the properties of content popularity, defined as follows.

\begin{table}[h!]
\begin{center}
\scriptsize
  \begin{tabular}{ | l | l | l  | l |}
    \hline
    \textbf{Application}           				& \textbf{Catalog Size $N$ (TB)}            & \textbf{Cache Size $M$ (TB)}      &  $\gamma\!=\!M/N$ \\ \hline \hline
    \text{Small CDN}   			& $10$	     				              & $1$						& $0.1$\\ \hline 
    \text{Netflix}              				& $314$   						      & $200+$		        		        & $0.63$\\ \hline 
    \text{Torrents file sharing}   			&  $1.5\times 10^{3}$	     	              & $40$					& $2.6\!\times\! 10^{-2}$\\ \hline 
        \text{Youtube}              				& $3\times 10^5$ 				      & $200+$					& $0.5\!\times\!10^{-3}$\\ \hline 
    \text{Facebook (photos)}              		& $1.5\times 10^{3}$ 			      & $300$  					& $0.2$\\  \hline
    \text{Facebook (total)}              		& $3\times 10^5$ 				      & $300$  					& $10^{-3}$\\  \hline
    \text{The Internet}					& $3\times 10^{8}$				      & --  					& -- \\  \hline
  \end{tabular}\vspace{-0.2in}
\end{center} 
\caption{Catalog estimates  and cache sizes of popular applications (source: web 2019).}
\label{tab:gamma}
\end{table}

\begin{dfn}[theorem style=plain]{Content popularity}{}
Consider the probability distribution $(p_n)$ over the set $\mathcal{N}$, where $p_n$ expresses how 
 likely it is that a request is issued for content $n$.
 The probability $p_n$ is called the popularity of content $n$. 
\end{dfn}

\section{Power law popularity}

In this section we ask the question: what is a ``good model'' for the distribution $(p_n)$? Where a good model should  capture reality and facilitate mathematical analysis and optimization.

\subsection{Zipf's law}
{The frequency of words appearing in text is} a well-studied topic  in the field of statistical linguistics \cite{manning1999foundations}. Zipf's law specifies that given a text of  natural language, the probability we encounter 
a word is inversely proportional to its rank in the frequency table.
We write,
 \[
p_n\propto n^{-1},\quad n=1,\dots,N,
\]
which means  that $p_1$ is the popularity of the most frequent word, $p_2$ that of the second most frequent, etc., and the Zipf  law governs the relation between all popularities, i.e.,  the distribution $(p_n)$.
 For example, the popularity ratio of the most frequent word  ``the''  over the 2nd most frequent ``be''  is exactly $p_1/p_2=2$. 
 
The reason why power laws appear in natural systems   has been a topic of extensive study. Links have been made to: mixtures of exponential distributions \cite{R_Farmer_08}, scale-invariance \cite{J_Chater_99} and invariance under aggregation \cite{R_Farmer_08}, random division of elements into groups \cite{R_Baek_11}, and  others mentioned in \cite{Piantadosi14}. Sometimes the connection between the rules of the underlying system and the formation  of power laws  remains unexplained  \cite{Mitz04,Piantadosi14}. Nevertheless, the frequency with which power laws appear in experimental data is too striking to be ignored. Hence, it is not surprising that Zipf models are used to emulate content requests and perform mathematical optimization. Besides, as statistician George Box once said, ``\emph{All models are wrong, but some models are useful}'' \cite{box1976science}.

\subsection{Evidence of Zipf's law in the Internet}
Extensive traffic measurements 
have revealed that the popularity of requests for Internet content (web pages, music, video, etc.) exhibits  a Zipf behavior, see \cite{breslau}. 
 A significant difference from word frequencies in texts is that the content popularity $p_n$ is not always inversely proportional to rank, but a different exponent may be observed: the distribution is written as  $p_n\propto n^{-\Zipf}$, where  $\Zipf$ sets the rate of popularity decline.
For $\Zipf\to 0$, we recover the discrete uniform distribution, where all contents are equally popular. For $\Zipf=1$, we recover the Zipf's law.
Taking the logarithm in both sides produces $\log p_n \propto -\Zipf \log n$, hence a log-log popularity/rank plot is a line with negative slope equal to $\Zipf$.  Larger values of $\Zipf$ result in steeper slopes, and hence  more skewed popularity distributions. One can observe evidence of power law in a dataset by looking at the log-log plot of the observed empirical frequencies, as in Fig.~\ref{fig:evidZipf}.

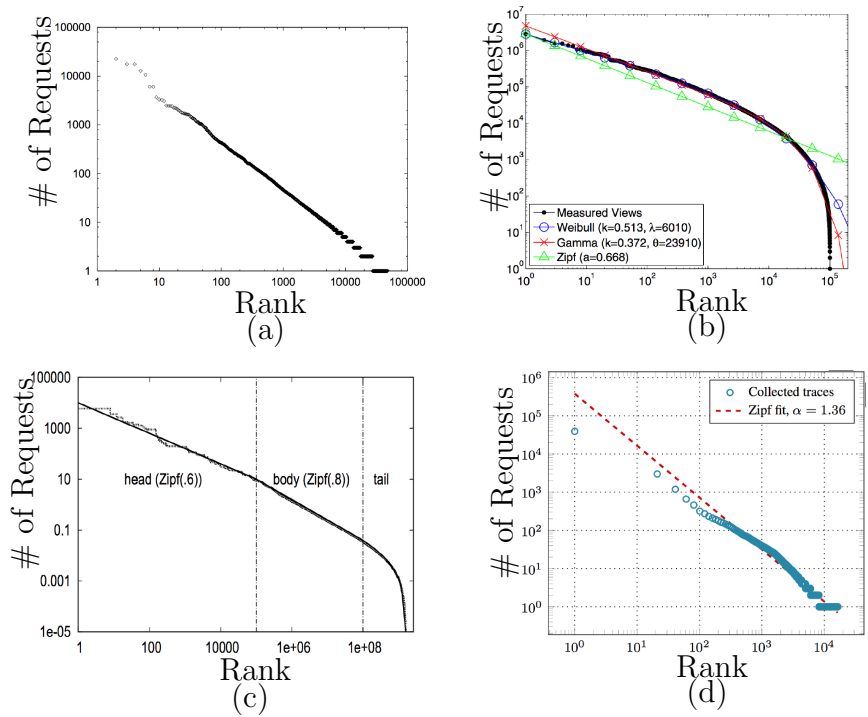
\begin{figure}[h!]
	\begin{center}
		\begin{overpic}[scale=.32]{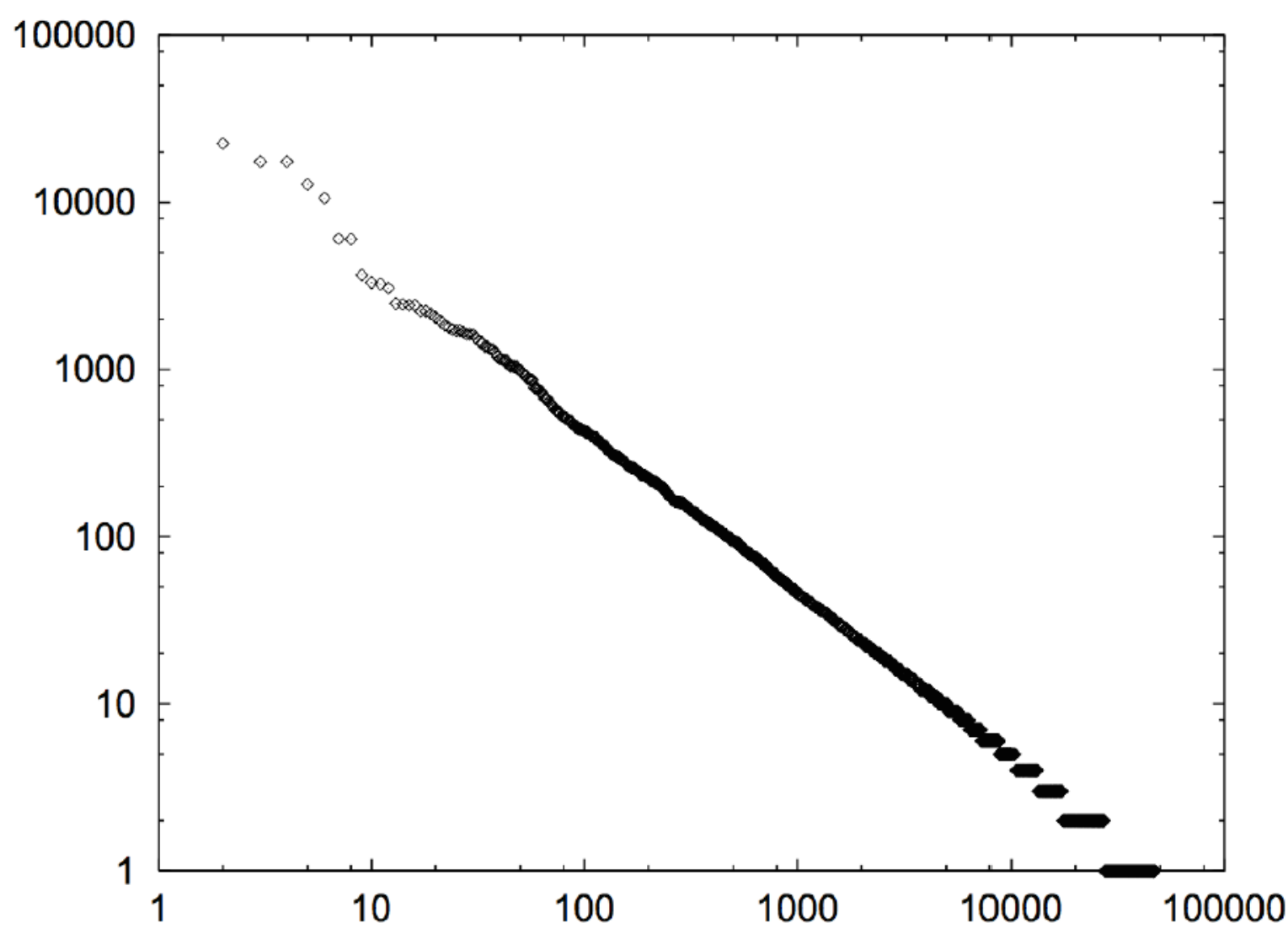}
			\put(49,-6){Rank}
			\put(-6,22){\rotatebox{90}{\# of Requests}}			
			\put(52,-13){(a)}			
		\end{overpic}\hspace{0.09in}\quad\quad
		\begin{overpic}[scale=.35]{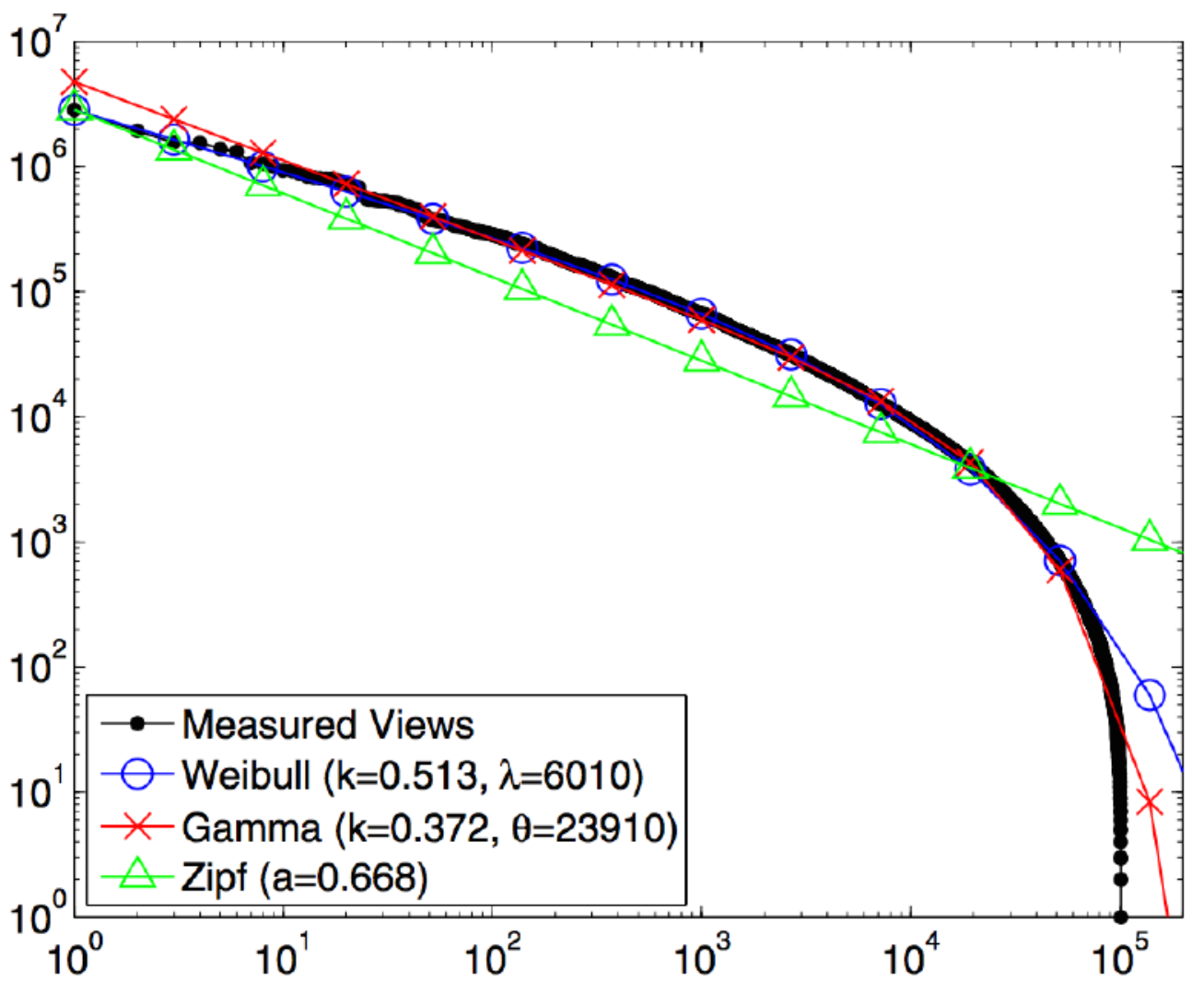}
			\put(49,-6){Rank}
			\put(-7,22){\rotatebox{90}{\# of Requests}}	
			\put(52,-13){(b)}
		\end{overpic}\hspace{0.03in}\\\vspace{0.4in}
		\begin{overpic}[scale=.35]{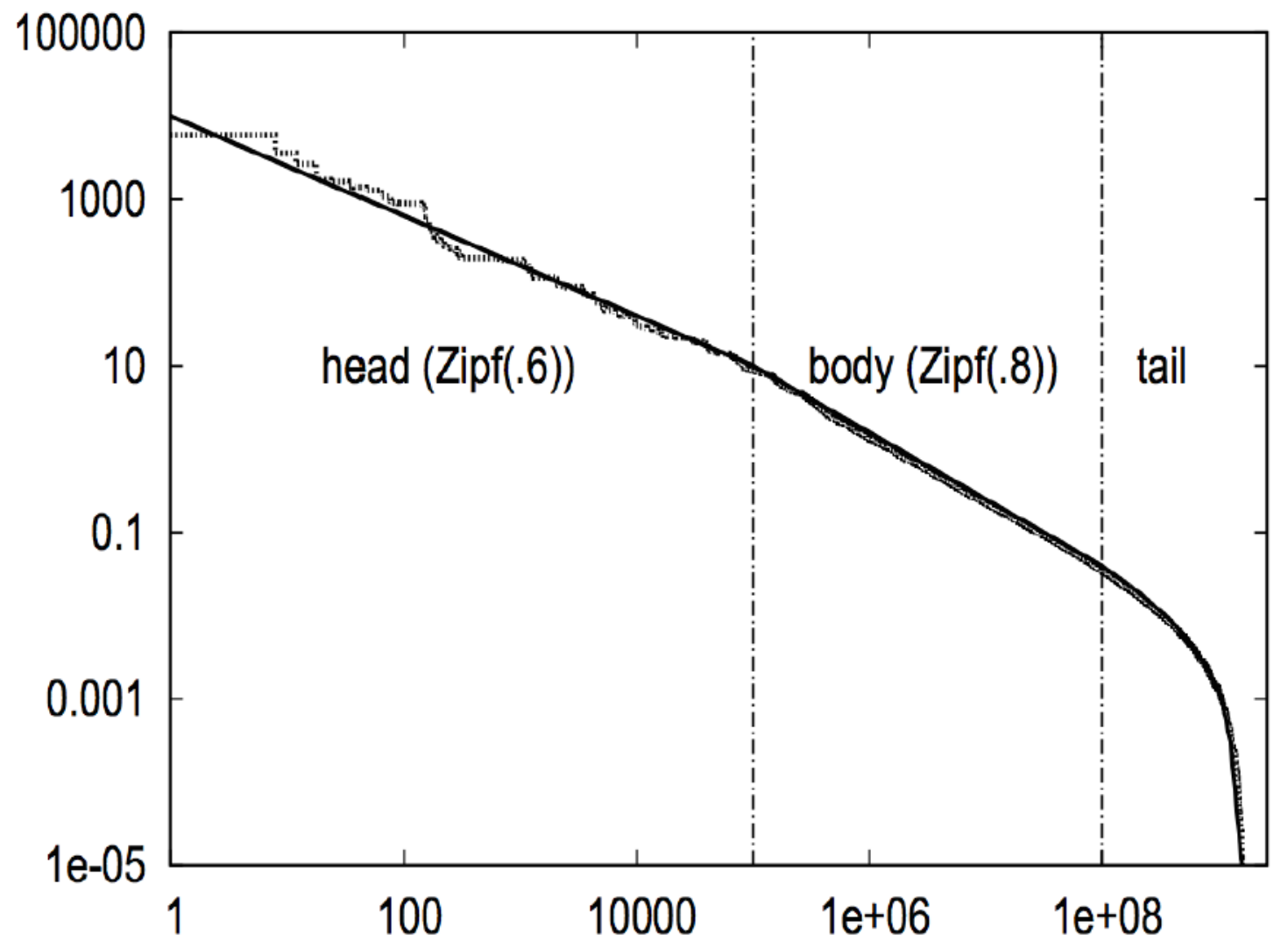}
			\put(49,-6){Rank}
			\put(-4,22){\rotatebox{90}{\# of Requests}}	
			\put(52,-13){(c)}
		\end{overpic}\hspace{0.03in}\quad\quad\quad
		\begin{overpic}[scale=.38]{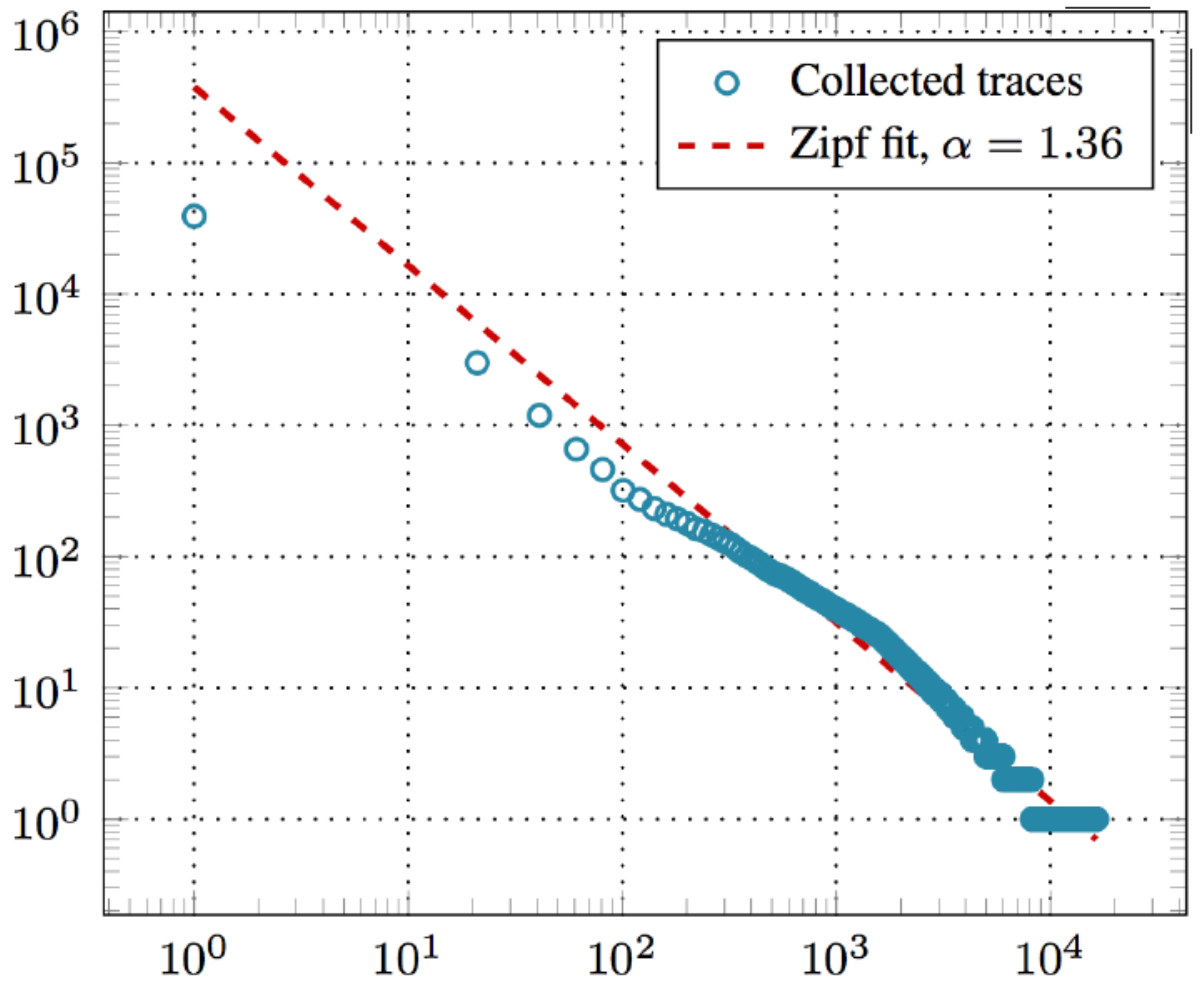}
			\put(45,-6){Rank}
			\put(-8,22){\rotatebox{90}{\# of Requests}}	
			\put(48,-13){(d)}
		\end{overpic}\vspace{0.2in}
		\caption{Log-log plots of number of content requests versus rank, showing evidence of power law popularity in  different applications. (a) www documents \cite{breslau} (1999),  (b) YouTube videos \cite{C_Cheng08} (2008) (c) p2p torrents \cite{Roberts} (2013), and (d) http requests in mobile operator \cite{Ejder15} (2015). }\vspace{-0.3in}
		\label{fig:evidZipf}
	\end{center}
\end{figure}

Since $(p_n)$ is a distribution, and hence $\sum_{n=1}^N p_n=1$, the power law distribution is given by
\begin{equation}\label{eq:powerlaw}
p_n=\frac{n^{-\Zipf}}{\sum_{j=1}^Nj^{-\Zipf}}, ~~n=1,\dots,N,
\end{equation}
where the denominator is the truncated zeta function    $H_{\Zipf}(N)\triangleq \sum_{j=1}^N j^{-\Zipf}$ evaluated at $\Zipf$, also called the $N^{\text{th}}$ \emph{$\Zipf$-order generalized harmonic number}. Although this sum does not have a closed form, we provide below a useful approximation.

\begin{box_example}[detach title,colback=blue!5!white, before upper={\tcbtitle\quad}]{Approximation of generalized harmonic number.}
{
\scriptsize
First for $n\ge m\ge 0$ we bound the sum  $\sum_{j=m}^n j^{-\Zipf}=H_{\Zipf}(n)-H_\Zipf(m)$ using integrals:\vspace{-0.14in} 
\begin{align*}\nonumber & \!
\int_m^n (x+1)^{-\Zipf}dx \leq H_{\Zipf}(n)-H_\Zipf(m)\leq 1+\int_{m+1}^n x^{-\Zipf}dx,
\Rightarrow \displaybreak[0]\\
&\left\{\begin{array}{rcll}	\frac{(n+1)^{1-\Zipf}-(m+1)^{1-\Zipf}}{1-\Zipf} & \le H_{\Zipf}(n)-H_\Zipf(m) &\le \frac{n^{1-\Zipf}-(m+1)^{1-\Zipf}}{1-\Zipf}+1, & \text{if } \Zipf\ne 1,\\
	\phantom{\big|^A}\ln \frac{n+1}{m+1} & \le H_{\Zipf}(n)-H_\Zipf(m)& \le \ln \frac{n+1}{m+2}, & \text{if } \Zipf=1.
\end{array}\right.
\end{align*}
Also, setting $m=0, n=N$ we obtain:
\begin{equation}\label{eq:H_bound}
\frac{(N+1)^{1-\Zipf}-1}{1-\Zipf}   \leq H_{\Zipf}(N)\leq  \frac{N^{1-\Zipf}-1}{1-\Zipf}+1, \text{if } \Zipf \neq 1.
\end{equation}
If $N$ is large, we then have the following approximation:
\begin{equation}\label{eq:H_approx}
H_{\Zipf}(N)\approx\left\{\begin{array}{ll}
 \frac{N^{1-\Zipf}}{1-\Zipf}, & \text{if } \Zipf < 1, \\
  \log N, &\text{if } \Zipf = 1, \\
    \frac{1}{\Zipf-1}, &\text{if } \Zipf > 1.
\end{array}
\right.
\end{equation}
As Fig.~\ref{fig:approxH} shows, the approximation error is very small for all values of $\Zipf$ as $N$ becomes large, while it is typically larger near $\Zipf=1$. Also, as $N$ increases, the approximation error diminishes.}\end{box_example}

\begin{figure*}[h!]
	\centering
	\subfigure[]{
		\includegraphics[width=0.45\linewidth ]{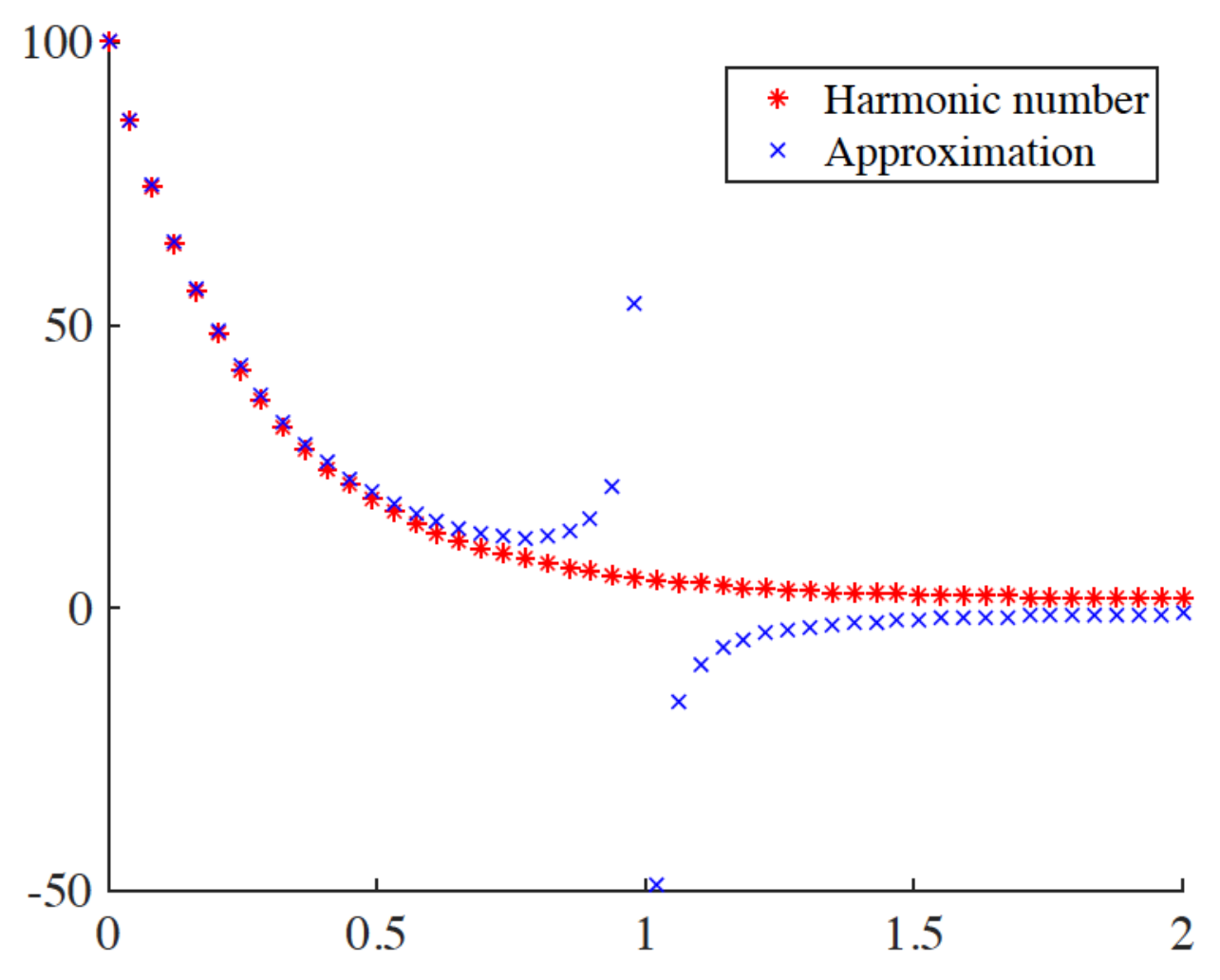}
		\put(-8,14){\small $\Zipf$}		}
	\subfigure[]{
		\centering
		\includegraphics[width=0.45\linewidth]{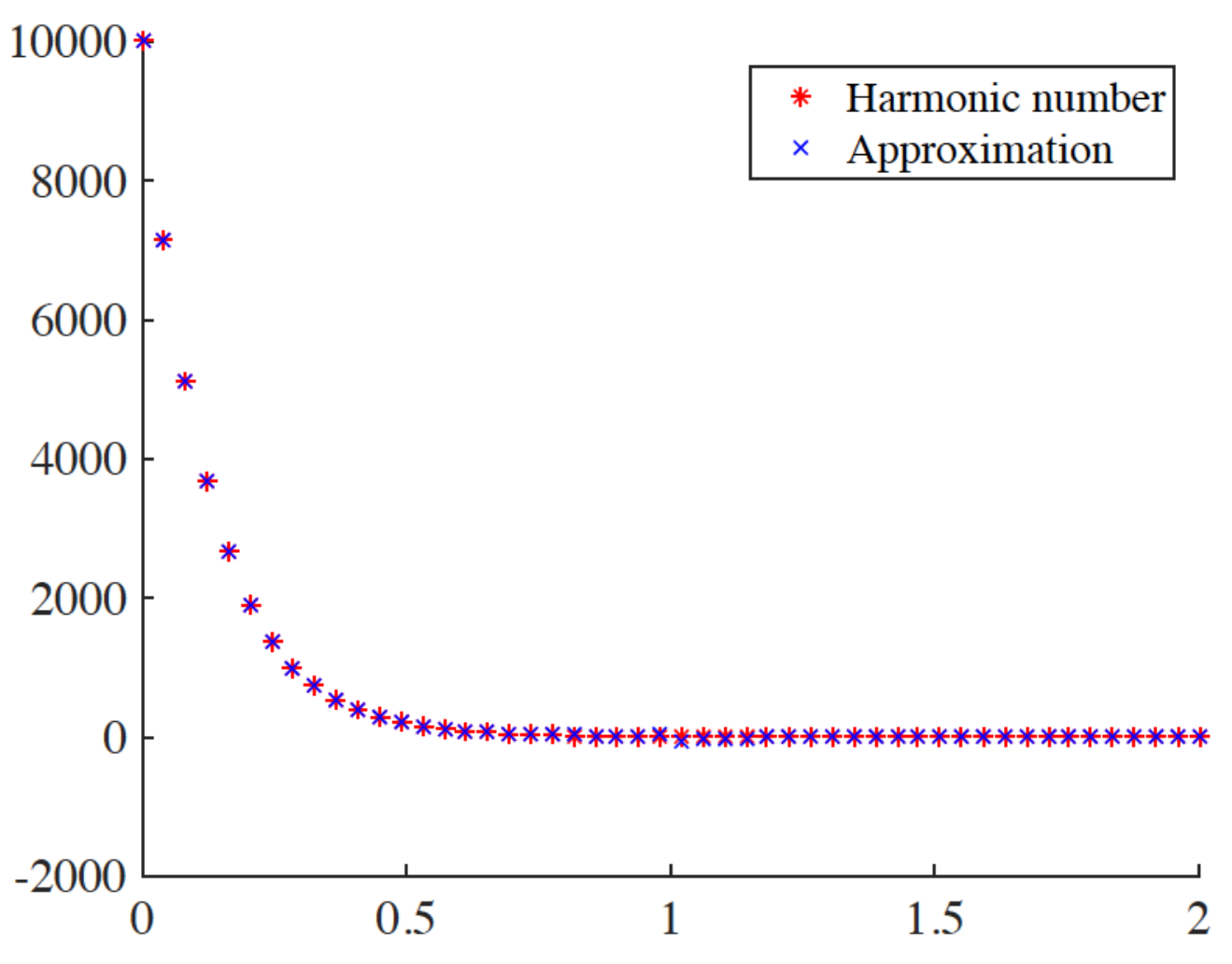}
		\put(-8,14){\small $\Zipf$}}
		\caption{Approximation \eqref{eq:H_approx} versus actual value of the generalized harmonic number. Parameters: 		$\gamma=0.1$, (left) $M=10^2, N=10^3$, (right) $M=10^4, N=10^5$.}
\label{fig:approxH}  
\end{figure*}

\subsection{Hit probability maximization}\label{sec:mostpop}

We present here a toy one-cache optimization problem and its solution, in order to understand  how the  popularity model affects the analysis of caching systems. 
 We draw one request for a content randomly from the catalog, and define the hit probability as follows.  

\begin{dfn}[theorem style=plain]{Hit probability}{}
The random event $A=$``the requested content is  cached'' is called a hit. The probability of  $A$ is called {hit probability}.
\end{dfn}

The hit probability is indicative of the fraction of  requests that is served by the cache. Since we would like to serve as many requests from the cache as possible, one way to optimize system performance is to decide which contents should be cached in order to  maximize hit probability. This procedure is described next.

We introduce decision variables $y=(y_n\in\{0,1\}: n=1,\ldots,N)$, where $y_n=1$ if we decide to cache content $n$, and $y_n=0$ otherwise. 
 Due to the limited cache size  $M<N$, a \emph{feasible} decision must satisfy the  constraint $\sum_{n=1}^N y_n\leq M$. Conditioning on the request of content $n$, the hit probability is computed by  $h( y,M)=\sum_{n=1}^N y_n p_n$, when the  decision $ y$  is taken.  Overall, we are interested in the question ``which feasible decision  $ y$ will maximize hit probability?'', which is answered by the following optimization problem.

\begin{opt}{Hit Probability Maximization}\vspace{-0.25in}
\begin{align}
h( y^*,M)=&\max_{\boldsymbol y\in \{0,1\}^N} \sum_{n=1}^N y_n p_n \label{eq:hit_prob_opt}\\
& \sum_{n=1}^N y_n\leq M.\quad \text{ (cache size constraint)} \notag
\end{align}
\end{opt}

We may solve the above optimization problem by inspection: set $y_n=1$ for the $M$  contents with greatest probability, and zero otherwise. 
This solution is also known as ``cache the most popular'' contents.
Hence, assuming w.l.o.g. that the contents are ordered in non-increasing popularity (i.e. $p_1\geq \dots \geq p_N$), the maximum hit probability is given by $h( y^*,M)=\sum_{n=1}^M  p_n $. When popularities are  power-law with exponent $\Zipf$,
\[
h( y^*,M)=\sum_{n=1}^M  p_n = \sum_{n=1}^M \frac{n^{-\Zipf}}{\sum_{j=1}^Nj^{-\Zipf}}= \frac{\sum_{n=1}^Mn^{-\Zipf}}{\sum_{n=1}^Nn^{-\Zipf}}=\frac{H_{\Zipf}(M)}{H_{\Zipf}(N)},
\]
and we may employ the approximation of $H_{\Zipf}(i)$ in \eqref{eq:H_approx} to obtain an  approximation of the maximum hit probability
\begin{equation}\label{eq:approx_hit}
h( y^*,M)\approx \left\{\begin{array}{ll}
\left(\frac{M}N\right)^{1-\Zipf}=\gamma^{1-\Zipf} & \text{ if } 0\leq \Zipf<1 \\
\frac{ \ln M}{ \ln N}=1-\frac{\ln \gamma^{-1}}{\ln N} & \text{ if } \Zipf=1 \\
1 & \text{ if } \Zipf>1.
\end{array}\right.
\end{equation}
Fig.~\ref{fig:approxhitrate} shows the maximum hit  probability comparing the actual value (found numerically by evaluating the sums) with the closed-form approximation of \eqref{eq:approx_hit}. We see that the  approximation is very accurate, while an  error is {introduced when $\Zipf\approx1$ 
 which  decreases as $M$ or $N$ increase}.

\begin{figure*}[t!]
	\centering
	\subfigure[ ]{
		\includegraphics[width=0.475\linewidth ]{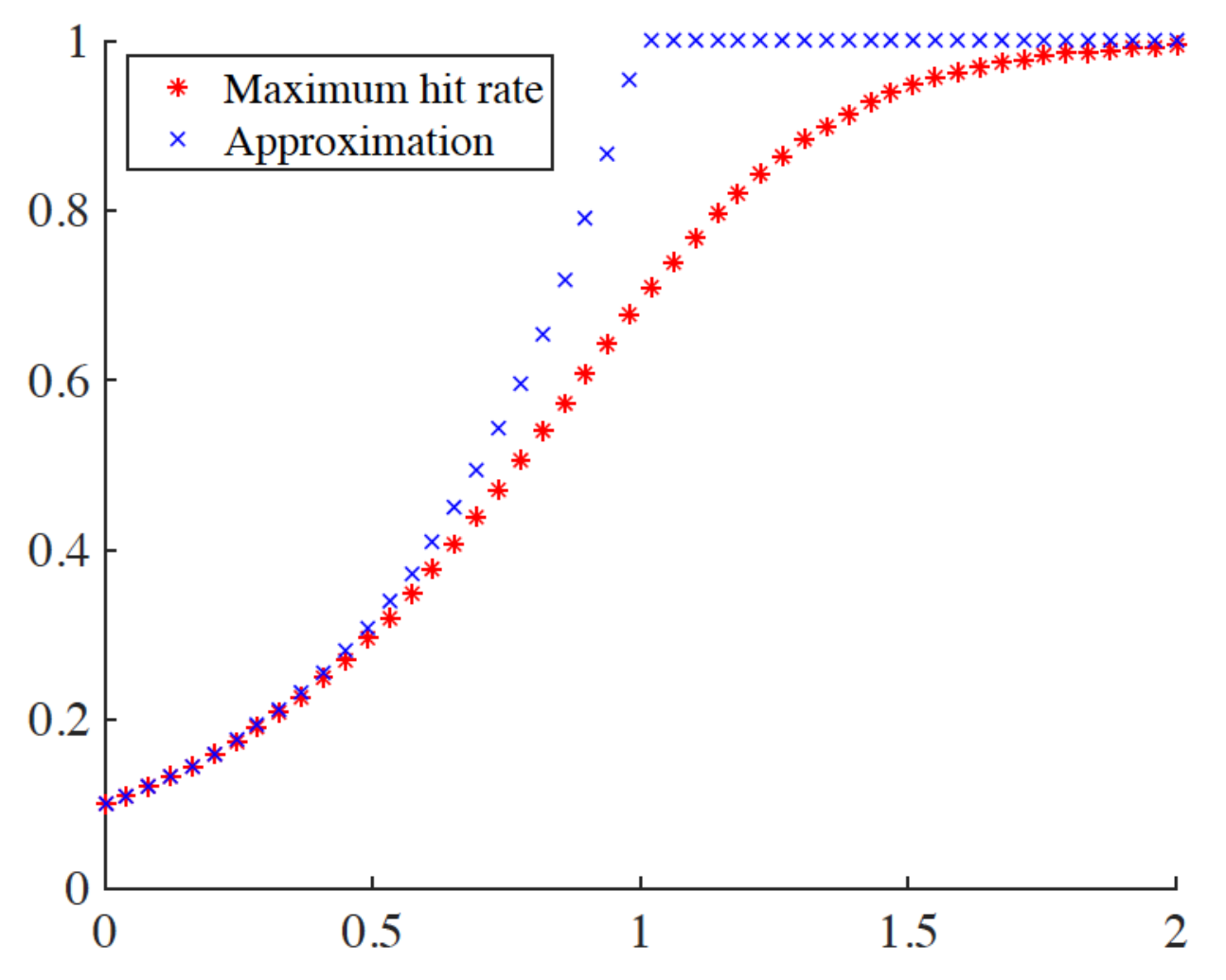}
		\put(-8,16){\small $\Zipf$}		}
	\subfigure[ ]{
		\centering
		\includegraphics[width=0.475\linewidth]{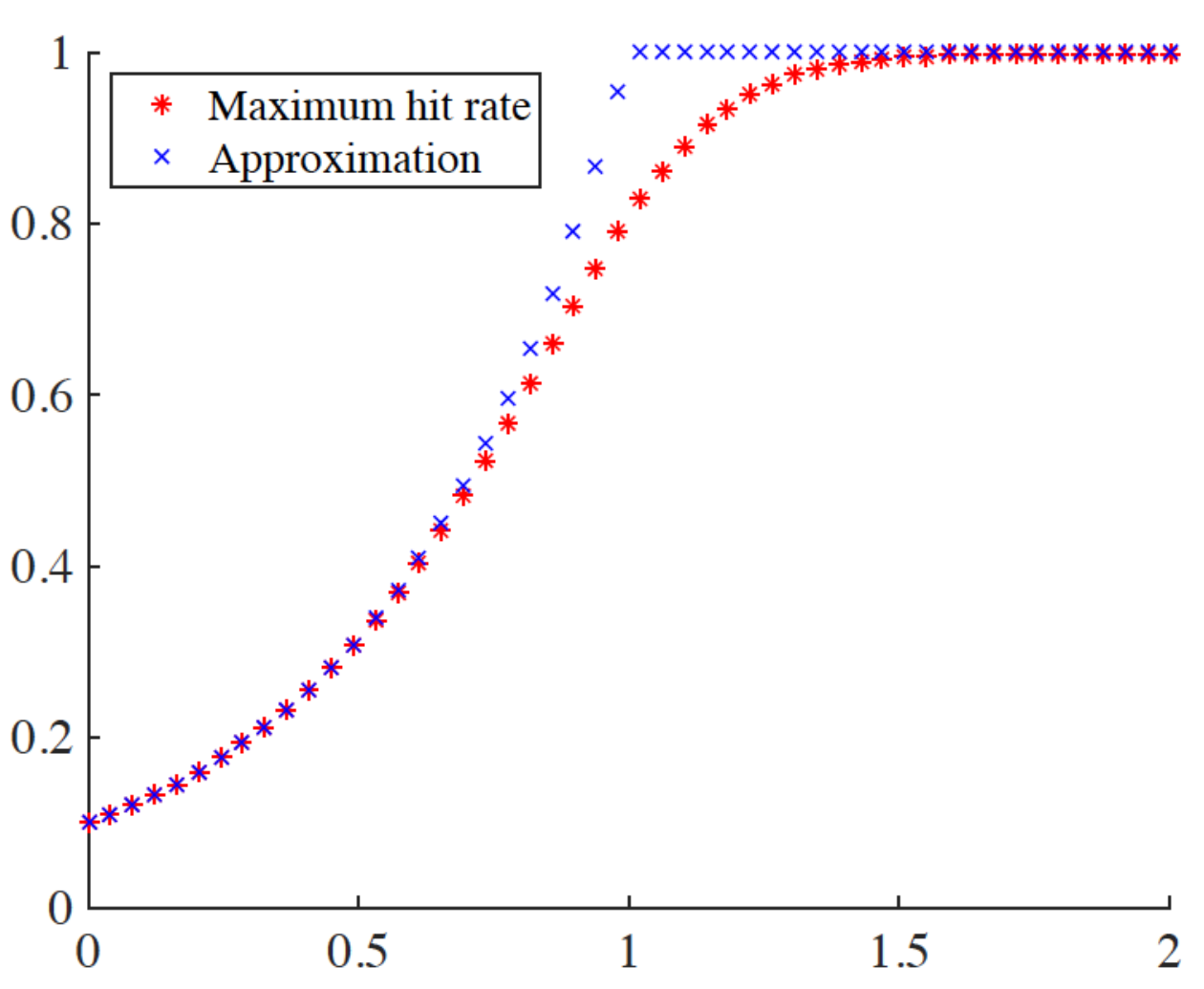}
		\put(-8,16){\small $\Zipf$}}
	\caption{Maximum hit probability: approximation versus actual value. Parameters: $\gamma=0.1$, (a) $M=100$, (b) $M=10000$.}
	\label{fig:approxhitrate}
\end{figure*}

Observe that in the common  case where $\Zipf<1$, the maximum hit probability can be approximated by the expression $\gamma^{1-\Zipf}$. The cache  performance depends on  the relative cache size  $\gamma$ and the power law exponent $\Zipf$. In the next section, we will discuss in detail how these parameters can be estimated in real systems.

\subsection{Fitting the power law model to data}\label{sec:inference}

One of the most important design problems in caching  is the problem of \emph{cache dimensioning}, where we must  decide the cache size $M$ that is appropriate for a specific system. Given storage costs, data transfer costs, and traffic estimation, the cache dimensioning problem aims to tune the cache size for cost minimization. In the example of one cache, this can be done via  the function $h( y^*,M)$ which provides an estimate of the cache efficiency {for a certain value of $M$}. 
However, as we saw above, this estimation requires knowledge about the Zipf law exponent $\Zipf$. Namely, we would like to determine an exponent $\Zipf^*$  that would accurately reproduce the observed system  behavior in terms of content requests. Inferring the parameters of a distribution from data is known as \emph{statistical inference}, and in this section we focus on highlighting some subtle points that appear in the inference of the power law exponent.

In practice, engineers sample a dataset with requests from the said system during peak-demand and then analyze it statistically. In this subsection, we discuss how $\Zipf^*$  can be extracted from such a dataset. Throughout the section we assume that the samples are drawn independently from a distribution $(p_n),~p_n\propto n^{-\Zipf}$ with unknown $\Zipf$, and we explain how to  compute the parameter  $\Zipf^*_{\text{\tiny MLE}}$ that best fits the samples in a maximum likelihood sense.

A few remarks are in order. First, we note that in general the popularity distribution changes over time. Our assumption in this section of fixed distribution  is meaningful only when the dataset is collected in a reasonably small time interval. 
If we knew there are time correlations, we could use this knowledge to improve our estimate further; this approach is explained in section \ref{sec:temporal} where we study temporal locality. Second, in practice it is common to fit data to a power law using the graphical method, where one simply plots the log-log requests versus rank and then decides a value for $\Zipf$ that  approximates the curve. Unfortunately, such graphical analysis can be grossly erroneous \cite{Goldstein2004,J_Clauset_09}. Wrong estimates for $\Zipf$ may lead us to wrong decisions about cache dimensioning. Therefore, we present a rigorous approach of fitting the data to a power law  by means of maximizing the log-likelihood function \cite{Goldstein2004}. We mention that  similar conclusions can be reached using the linear least squares method, albeit at  a higher error \cite{Bauke2007}.

We begin with a dataset $X=(x_i\in\mathcal{N}: i=1,\ldots,K)$, where the $i^{\text{th}}$ datum $x_i$ is the id of one of the contents in the catalog; for instance  we have in mind the \youtube~dataset from \cite{Kurose08}. 
For a single datum $x_i$ 
%
 we denote with $\ell(x_i;\Zipf)$ the likelihood that $x_i$ is generated from a power law model with exponent $\Zipf$, and
\[
\ell(x_i;\Zipf)=n^{-\Zipf}/H_{\Zipf}(N),~~\text{ iff}~~ x_i=n.
\]
The likelihood function is multiplicative over the data in the dataset. Hence, for a dataset $X=(x_1,\dots,x_K)$ with $K$ samples, the likelihood function is
\[
\ell(X;\Zipf)=\prod_{i=1}^K\ell(x_i;\Zipf).
\]
\begin{dfn}[theorem style=plain]{Empirical frequency}{freq}
For a sequence of requests, the empirical frequency of content $n$ is the ratio:
\[
f_n=\frac{\text{Number of requests for content } n}{\text{Total number of requests}}.
\] 
\end{dfn}
Disregarding the constant of total number of requests, we can simply write:
\[
\ell(X;\Zipf)=
p_1^{f_1}p_2^{f_2}\dots p_N^{f_N},\quad \text{where }~ p_n=n^{-\Zipf}/H_{\Zipf}(N).
\]
The log-likehood function $\Lambda(\Zipf)=\log \ell(X;\Zipf)$ is therefore:
\begin{align}
\Lambda(\Zipf)&=\sum_{n=1}^N f_n \log\left(n^{-\Zipf}/H_{\Zipf}(N)\right)=\sum_{n=1}^N f_n  \left(\log n^{-\Zipf}- \log H_{\Zipf}(N) \right)= \nonumber \\
&=-\Zipf \sum_{n=1}^N f_n  \log n - \log H_{\Zipf}(N)
\end{align}
The log-likelihood function is preferred from likelihood because it leads to simpler calculations. Since the logarithm is a strictly increasing function, the maximization of both functions will lead to the same solution. The MLE we are seeking can be recovered by solving the optimization problem:
\begin{equation}\label{eq:maxMLE}
\Zipf^*_{\text{\tiny MLE}}={\argmax}_{\Zipf\geq 0} \Lambda(\Zipf).
\end{equation}
The maximum likelihood estimator  is \emph{efficient} meaning that as the sample size increases,
$K\to\infty$, it achieves the Cram\'er-Rao bound \cite{rao1992information}, and therefore it has the smallest  mean squared error among all unbiased estimators.
Solving \eqref{eq:maxMLE} can be done in closed-form using  the approximation \eqref{eq:H_approx} and computing the stationary point of $\Lambda(\Zipf)$. 
Alternatively, we may solve the maximization numerically.
Fig.~\ref{fig:mlle} presents the results of the described method applied to the Youtube dataset of \cite{Kurose08}.

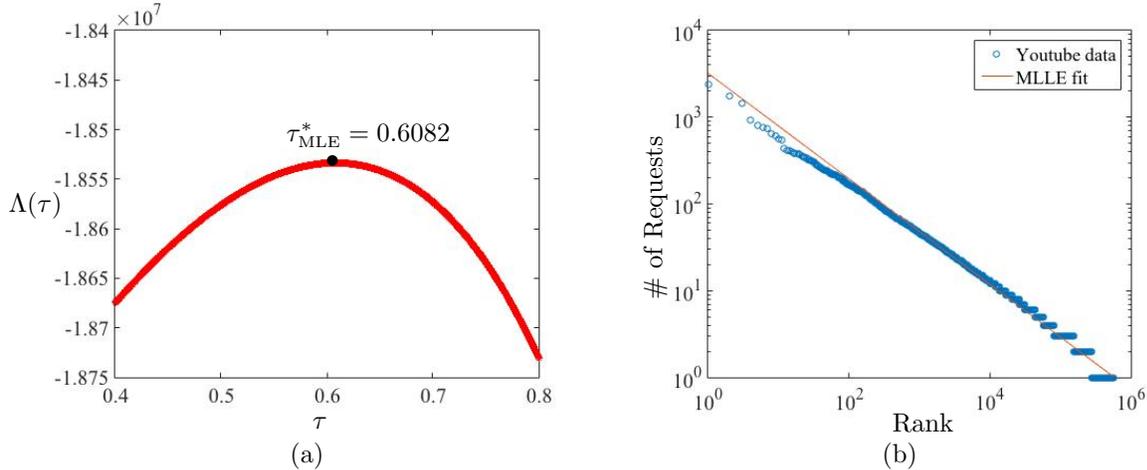
\begin{figure*}[t!]
	\centering
	\subfigure[ ]{
		\includegraphics[width=0.4\linewidth ]{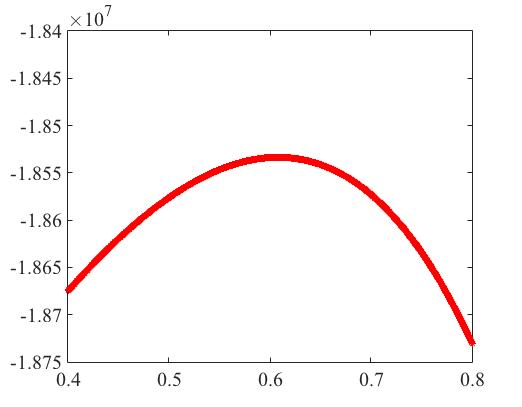}			
		\put(-105,-2){\footnotesize{$\Zipf$}}
		\put(-220,80){\footnotesize{$\Lambda(\Zipf)$}}
		\put(-115,107){\footnotesize{$\Zipf^*_{\text{\tiny MLE}}=0.6082$}}
		\put(-100,97){\footnotesize{${\color{black}\bullet}\mathllap{\circ}$}	}}\,\,\,\,\,
	\subfigure[ ]{
		\centering
		\includegraphics[width=0.4\linewidth]{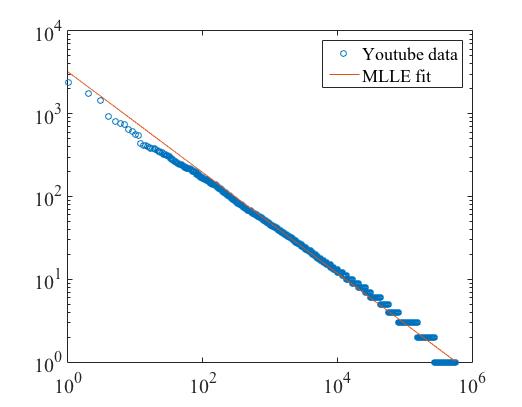}
		\put(-110,-3){\footnotesize{Rank}}
		\put(-204,47){\rotatebox{90}{\footnotesize{\# of Requests}} }}	
	\vspace{-0.2in}
		\caption{Experiments on the 2008 YouTube dataset \cite{Kurose08}: (left) log-likelihood function $\Lambda(\Zipf)$, (right) log-log plot of ranked requests and  MLE power-law vs rank.}
	\label{fig:mlle}
\end{figure*}

\subsubsection{Unknown labels}

A standard complication that arises in caching is the unavailability of content labels.
Here, we use the term {label} to refer to  \emph{the actual index of a content in the underlying popularity distribution $(p_n)$}. For instance, the content with popularity $p_n$ has label $n$. When working with real traces we observe how many times a specific content is requested, but its actual label remains unknown. The content with label 100 and  popularity $p_{100}$ might appear less times than the content with label 99 and popularity $p_{99}$, due to limited size of our sample. Therefore, although $p_{99}>p_{100}$, it may be that $f_{99}<f_{100}$. When the dataset is unlabelled with respect to the underlying popularity distribution one possibility is to adapt the maximum likelihood estimation as follows: (i) we rank the empirical  frequencies in decreasing order such that $f_{\sigma(1)}\geq \dots \geq f_{\sigma(N)}$, (ii) we \emph{assume} that the ranked frequency $f_{\sigma(n)}$ corresponds to label $n$, and (iii) we infer the power law exponent with the above MLE method. However, the implicit assumption $\sigma(n)=n$ {does not always hold in practice}, and therefore the known guarantees for maximum likelihood do not apply. We provide an example below. 

\begin{box_example}[detach title,colback=blue!5!white, before upper={\tcbtitle\quad}]{Fitting unlabeled data with MLE.} \footnotesize
We create $1.46m$ samples from a  power law distribution with exponent $\Zipf=0.6082$, and catalog size $N=566k$  (same parameters as in the Youtube dataset but ficticiously generated data). Then we extract from the samples three different empirical frequency vectors:
\begin{itemize}
\item $(f_n^1)$ is ranked according to the actual labels of $p_n$,
\item $(f_n^2)$ is permuted so that $f_1^2\geq \dots \geq f_N^2$, and
\item $(f_n^3)$ is same as $(f_n^2)$ but limited to the first 1000 elements (head of the ranked frequencies).
\end{itemize}
The vectors are then used to produce the  maximum loglikelihood estimates of the power law exponent; we find $\Zipf^1_{\text{\tiny MLE}}=0.6078$ ($0.06\%$ relative error), $\Zipf^2_{\text{\tiny MLE}}=0.6406$ ($5.32\%$ relative error), $\Zipf^3_{\text{\tiny MLE}}=0.6050$ ($0.52\%$ relative error). Fig.~\ref{fig:unlabelled} shows log-log requests vs rank plots in all cases. 
\end{box_example}

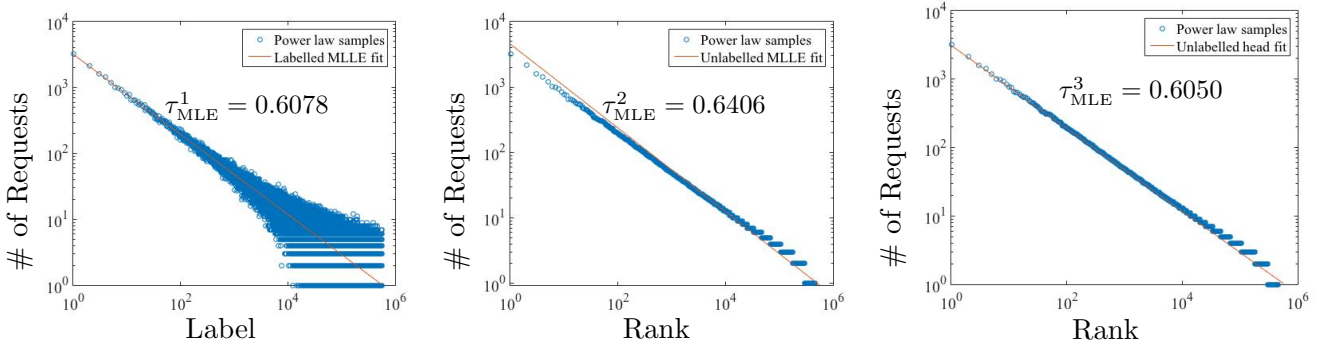
\begin{figure}[h!]
	\begin{center}
		\begin{overpic}[scale=.3]{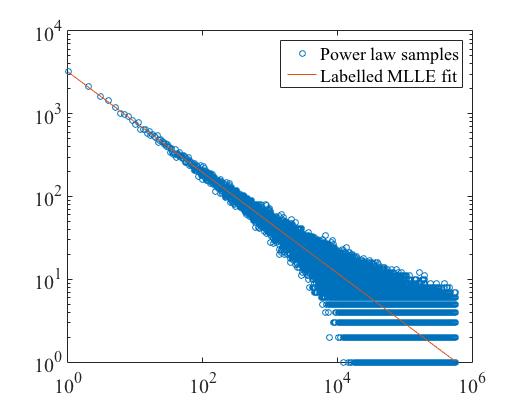}
			\put(40,-4){\small Label}
			\put(-3,12){\small \rotatebox{90}{\# of Requests}}	
			\put(35,50){\footnotesize $\Zipf^1_{\text{\tiny MLE}}=0.6078$}
		\end{overpic}
		\hspace{0.005in}
		\begin{overpic}[scale=.3]{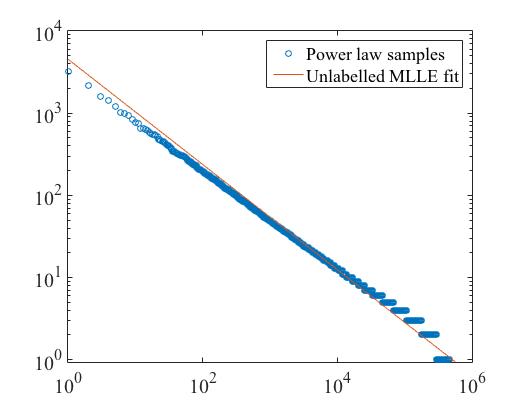}
			\put(40,-4){\small Rank}
			\put(-3,12){\small\rotatebox{90}{\# of Requests}}				
			\put(35,50){\footnotesize $\Zipf^2_{\text{\tiny MLE}}=0.6406$}
		\end{overpic}
		\hspace{0.005in}
		\begin{overpic}[scale=.3]{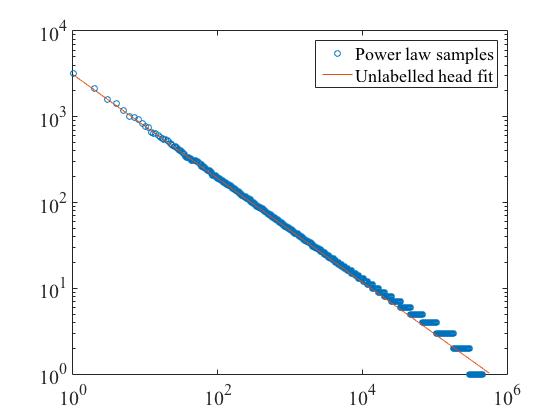}
			\put(40,-4){\small Rank}
			\put(-3,12){\small\rotatebox{90}{\# of Requests}}				
			\put(37,50){\footnotesize $\Zipf^3_{\text{\tiny MLE}}=0.6050$}
		\end{overpic}
		\caption{Generated power law data with $\Zipf=0.6082$, catalog $N=566k$ and $1.46m$ samples.
 (left) MLE fit  with the correct labels,  (middle) MLE fit with ranked frequencies, (right) MLE fit of the head with ranked frequencies.}
		\label{fig:unlabelled}
	\end{center}
\end{figure}

While the labelled MLE has a very small relative error ($0.06\%$) as expected, the example shows that the ranking of frequencies leads us to erroneous estimation ($5.32\%$ relative error). The reason 
is the  permutation of the labels at the tail of the distribution. Indeed, we can obtain better accuracy if we restrict the fitting to the head of the ranked empirical frequencies ($0.52\%$ relative error). While the tail is very noisy, the head \& body of the labelled frequencies and the ranked frequencies are similar, leading to better estimates; this approach is called ``trust your body'' in \cite{nair2013fundamentals}. A different approach to inference of unlabeled data is given in \cite{J_olmos_15}, where the authors assume that the popularity of each content is itself a random variable with the same power-law distribution.

\subsubsection{Observed values of $\Zipf$} 

We summarize literature measurements of $\Zipf$ in Table \ref{tab:zipfexponent}. Evidently, the power law model is  fitted in many different scenarios over the years, but each application --- or even each dataset --- fits a different value of $\Zipf$. Therefore, our mathematical analysis should be parametric to  the power law exponent $\Zipf$, and adjustable to the application dataset. What if data are not available--what should be our choice  of $\Zipf$ then? For visualizing results, it is customary to take $\Zipf\in \{0.6,0.8,1\}$, which are the most depictive values. Table \ref{tab:datasets} lists  publicly available datasets for experimentation with request sequences and power-law popularity.

\begin{table}[h!]
\begin{center}
\scriptsize
  \begin{tabular}{ | l || l || l  || l |}
    \hline
    \textbf{Reference}         		& \textbf{Value of $\Zipf$}      & \textbf{Year }   & \textbf{Application} \\ \hline \hline
    \text{Cuhna et al. \cite{bestavros}}   			& 1			  	& 1995		& Web traffic 	\\ \hline     
    \text{Breslau et al. \cite{breslau}}   			& 0.64 - 0.83	  	& 1999		& Web traffic	\\ \hline 
    \text{Roadknight et al. \cite{roadknight2000file}}   	& 0.64 - 0.91	   	& 1999		& Web traffic (university proxy)	\\ \hline 
    \text{Padmanabhan et al. \cite{padmanabhan2000content}}   & 1.39 - 1.81	& 2000	& Web traffic (MSNBC news) \\ \hline 
    \text{Artlitt and Jin \cite{Arlitt}}			   	& 1.16	          	& 2000		& Web traffic (1998 world cup)\\ \hline 
    \text{Mahanti et al. \cite{J_Mahanti_00}}   		& 0.74 - 0.84	   	& 2000		& Web traffic \\ \hline 
    \text{Adamic et al. \cite{adamic}}   				& 1			   	& 2002		& Web traffic (AOL)\\ \hline     
    \text{Chu et al. \cite{chu}}	   				& 0.58 - 0.64	   	& 2003		& P2P (Gnutella $\&$ Napster)	\\ \hline     
    \text{Challenger et al. \cite{challenger2004efficiently}}  	& 1		   	& 2004		& Web pages (Japanese)\\ \hline 
    \text{Wierzbicki \cite{wierzbicki2004cache}}	  	& 1		   		& 2004		& P2P (FastTrack)\\ \hline     
    \text{Krashakov et al. \cite{krashakov2006universality}} & 0.92 - 1.09	& 2006		& Web traffic (Russian)	\\ \hline 
    \text{Yamakami \cite{yamakami}} 				& 1 - 1.5			& 2006		& Mobile web 	\\ \hline 
    \text{Yu \cite{yu2006understanding}}			 & 0.6			& 2006		& VoD (China Telecom)	\\ \hline 
    \text{Gill et al. \cite{gill2007youtube}}   			& 0.56    		        & 2007		& Video (Youtube)	\\ \hline 
    \text{Cheng et al. \cite{cheng2013understanding}}  & 0.54    		        & 2007		& Video (Youtube) 	\\ \hline 
    \text{Hefeeda et al. \cite{hefeeda2008traffic}}  	& 0.6 - 0.78    		& 2008		& P2P 	\\ \hline 
    \text{Urdaneta et al. \cite{urdaneta2009wikipedia}} & 0.53    		        & 2009		& Wikipedia	\\ \hline 
    \text{Kang et al. \cite{kang2010understanding}} & 0.85    		        & 2009		& Video (Yahoo)	\\ \hline 
    \text{Dan et al. \cite{dan2010power}} 		& 0.6 - 0.86    		        & 2010		& P2P (BitTorrent)	\\ \hline 
    \text{Traverso et al. \cite{snm}} 			  	& 0.7 - 0.85	        & 2012		& Video (Youtube)	\\ \hline 
    \text{Huang et al. \cite{Huang_2013}}   			& 0.7 - 1		        & 2013		& Photos (Facebook)	\\ \hline 
    \text{Imbreda et al. \cite{imbrenda2014analyzing}}	& 0.83		        & 2014		& Web traffic (French)	\\ \hline 
    \text{Zotano et al. \cite{zotano2015analysis}}   	& 1.09 - 1.87	        & 2015		& Web traffic (Spanish)		\\ \hline     
    \text{Bastug et al.  \cite{Ejder15}}				& 1.36		        & 2015		& Mobile web (Turkish)		\\ \hline     
    \text{Hasslinger et al. \cite{hasslinger2017performance}}   	& 0.5 - 0.75	        & 2015		& Wikipedia		\\ \hline         
  \end{tabular}\vspace{-0.2in}
\end{center}  
\caption{List of $\Zipf$ values from datasets.}\label{tab:zipfexponent}
\end{table}

\begin{table}[h!]
\begin{center}
\scriptsize
  \begin{tabular}{ | l || l || l  || l |}
    \hline
    \textbf{Dataset}         				& \textbf{Ref.}            & \textbf{Number of requests}      &  \textbf{Time span} \\ \hline \hline
    \text{Web proxy accesses}			   	& \cite{gds}		   & $24M$ accesses		   	   & Aug '96 - Sep '96\\ \hline 
    \text{Netflix movie ratings} 				& \cite{Netflix:2009} 	   & $100M$ movie ratings		   & Oct '98 - Dec '05\\ \hline  
    \text{Anonymized URL requests}			& \cite{Meiss08WSDM}& $25B$ requests  		   & Sep '06 - Mar '08\\ \hline
    \text{YouTube traces, UMass campus}   	& \cite{Kurose08}	   & $1.5M$ requests			   & Jun '07 - Mar '08\\ \hline 
    \text{top-1000 Wikipedia pages}			& \cite{wikipedia15}	   & $34M$ requests			   & Oct '15 - Nov '15\\ \hline 
  \end{tabular}\vspace{-0.2in}
\end{center}  
\caption{Publicly available datasets for caching experimentation.}\label{tab:datasets}
\end{table}

\subsection{Catalog estimation}

The problem of cache dimensioning requires also the knowledge of catalog size $N$. While estimation of $\Zipf$ is  commonly encountered in the literature, the catalog size estimation is often overlooked.\footnote{We would like to acknowledge our discussions with Dr. J. Roberts for this section.} For a dataset $X$ with $K$ samples, the number of different observed contents is called  the \emph{visible catalog}, and has  size: 
\[
N_K=|\{n\in \mathcal{N} | f_{n}(X)>0 \}|
\]
which clearly depends on the dataset $X$, and its length. Using this notation, the true catalog size is $N_{\infty}$, and we have $N_{\infty}\geq N_K$.
 
Let $E_j$ be the number of contents observed exactly $j$ times in   $X$, and let $e_j$ denote its mean. We have,
\begin{equation}
E_j =  \sum_{n=1}^{N_{K}} \indic{\text{content }n \text{ is requested } j\text{ times}}, 
\label{eq:number}
\end{equation}
and taking expectations yields,
\begin{eqnarray*}
e_j(X)&  = & \sum_{n=1}^{N_{K}} \text{Pr}\: [\text{content }n \text{ is requested } j\text{ times}] \\
& \stackrel{(a)}{=} &  \sum_{n=1}^{N_{K}} \binom{K}{j} p_n^j (1-p_n)^{K-j},
\end{eqnarray*}
where in $(a)$ we have assumed that the samples are i.i.d. 
A mean estimator  $\widehat{N}_{\infty}$ of the actual catalog $N_{\infty}$ can therefore be obtained by
\[
\widehat{N}_{\infty}=N_K + e_0(X),
\]
where recall that $N_K$ is the size of the visible catalog and therefore easy to find by enumerating the unique ids in the dataset, and $e_0$ is the mean number of contents that are never requested, calculated as
\[
e_0(X)=\sum_{n=1}^{N_{K}} \binom{K}{0} p_n^0 (1-p_n)^{K-0}=\sum_{n=1}^{N_{K}}(1-p_n)^K.
\]
In summary, combining the calculation of $e_0(X)$ with the observable $N_K$, one can obtain an unbiased estimate of the true catalog size $\widehat{N}_{\infty}$, which subsequently  determines our estimate of $\gamma$ and allows for correct cache dimensioning  via $h( y^*,M)$.


\section{Request sequences}

So far we have considered a single isolated request, drawn randomly from  the distribution $(p_n)$. This  allowed us to discuss about content popularity, and introduce the optimization of hit probability. In this section, we will introduce the notion of the \emph{request sequence}:  a succession of content requests. Using the request sequence, we will investigate  more intricate aspects of content popularity such as  temporal and spatial correlations. 
\begin{dfn}[theorem style=plain]{Request sequence}{}
The request sequence $\boldsymbol R=(R_1,\dots,R_{T})$ is a sequence of integers, where $R_t\in \mathcal{N}$  denotes the content id of the $t^{\text{th}}$  request. 
\end{dfn}\vspace{-0.1in}
\noindent We present next different models of request sequences.

\vspace{-0.1in}
\subsection{Independent reference model}\vspace{-0.1in}
\begin{dfn}[theorem style=plain]{IRM model}{}
A request sequence $\boldsymbol R$ is called  IRM if $(R_t)_{t=1,2,\dots}$ are i.i.d. random variables drawn from $(p_n)_{n=1,\dots,N}$.
\end{dfn}
The term ``independent'' reflects the property of an IRM sequence that each random variable is drawn with distribution $(p_n)$ independently of all  others. The term ``reference'' is identical  to the term ``request'' in our context.
As a sequence of i.i.d.   random variables, IRM represents a special case where each element of the sequence repeats the same random experiment, in which case the  notion of hit probability continues to apply. 
 In the next chapter we provide the hit probability analysis  for several online caching policies under IRM.

Sometimes we define the request sequence in continuous time. The simplest extension in this direction is to consider a Poisson process whose  points denote the requests. Each point of the Poisson process is   \emph{marked} with a value from $\{1,\dots,N\}$ denoting which content is requested. The marks are chosen randomly and  i.i.d.  from distribution $(p_n)$.

\begin{dfn}[theorem style=plain]{Poisson IRM model}{poissIRM}
A request sequence is called Poisson IRM if it is determined by the marks of an associated homogeneous Poisson process with intensity $\lambda$ and random i.i.d. marks drawn from $(p_n)_{n=1,\dots,N}$.
\end{dfn}
The intensity parameter $\lambda$ defines how densely the different requests are spaced apart in time. We remark that the standard IRM model can be obtained from a Poisson IRM by discarding the inter-request  time information and keeping only the (ordered) marks at the points of the process. When restricting attention to  the marks, we may still use the notion of hit probability. For the associated continuous time process, we will later define the notion of \emph{hit rate}, see  discussion in Sec.~\ref{sec:stationary}.

IRM sequences are  the simplest and most established models. They are often combined with a power law distribution $(p_n)$, thus  yielding a one-parameter model (or two-parameter in case of a Poisson IRM), which is easy to fit. Past work  validated IRM models, showing that they fit well real request traces  obtained at the peak hour \cite{breslau}. On the other hand, IRM sequences by design have no  temporal or spatial correlations. Hence, if the actual data do exhibit such correlations, an IRM model will not capture them. We discuss next alternative models that capture temporal and spatial correlations.

\subsection{Temporal locality}\label{sec:temporal}

Past measurements have shown that content popularity exhibits temporal correlations. Specifically, the measurements show that the popularity of recent requests is higher than the popularity of past requests when these have the same empirical frequency. This phenomenon is called ``temporal locality'', and it is particularly important for caching as it suggests that \emph{we should not cache contents that are  popular on average, but rather those that are popular locally in time}. Since IRM is i.i.d., it does not capture temporal correlations. In this section  we focus on alternative request  models.

Evidence of temporal locality in content requests dates back to web page requests in the 90s. For example, \cite{Almeida96} attributed to temporal locality the hit probability discrepancy between dataset and generated data from  a distribution with the same power-law exponent; \cite{breslau} observed that the popularity of  highly popular contents changes less often than that of the unpopular ones; and
\cite{snm} explained how decaying popularity in \youtube~videos affects cache hits. \emph{In sum, content popularity is often time-varying, and this has a profound effect on hit probability and cache dimensioning.}

Regarding the span of temporal locality over time, a \youtube~study from  \cite{Tortelli16} confirms that correlations on time scales of a few hours can be ignored, as their impact on cache performance is negligible. Instead, the time scale of a few days up to weeks is deemed most significant. 
For example, Fig.~\ref{fig:temporal_locality} reproduced from \cite{gds}, shows how the  web proxy accesses evolve over time,  and it is indicative of strong ephemerality. Recent accesses are 100 times more probable than those requested one day ago, and 1000 times that those last requested one week ago. Another work
\cite{hasslinger2017performance} demonstrated that requests of wikipedia articles have a rapid day-to-day change in popularity rank. In particular, measuring the fraction of contents in  top-$x$ popularities of  a day  and the previous, they find $76\%$ for $x=1000$ and only $54\%$ for $x=25$, meaning that half of the top 25 contents change from day to day. In conclusion, a caching analysis that extends beyond a few hours must  take into account temporal correlations. This finding is particularly important for systems with few requests per unit time, where learning the evolving popularity might be very challenging.

\begin{figure}[h!]
  \centering
      \includegraphics[width=0.5\textwidth]{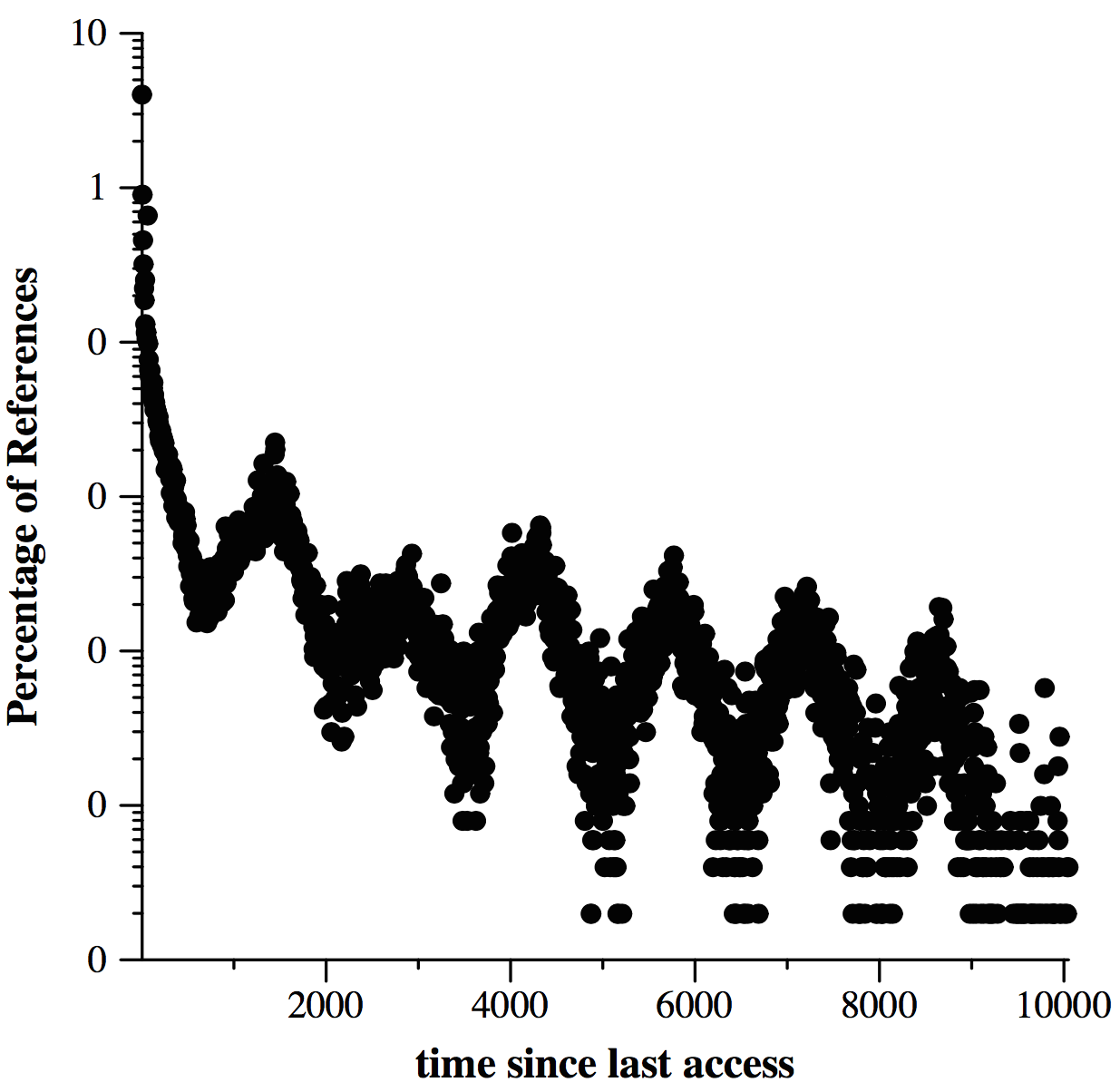}
\caption{Percentage of references as a function of the time since last access to the same document by the same user, from \cite{gds}. The time is in minutes and y-axis is in log scale.}
\label{fig:temporal_locality}
\end{figure}

In order to model temporal locality, we consider content request models with time-varying popularity distribution $(p_n(t))$. When designing such models we aim to fit well the observed correlations and maintain a simple model structure for easy fitting to the actual dataset, while  ensuring mathematical tractability. In the literature, there exist several such models. We mention a few: \cite{Gitz_14} uses the theory of variations, \cite{Elayoubi2015} makes random replacements of contents in the catalog, and \cite{kauffman} uses a dependent Poisson model. Below, we describe in detail the Shot Noise Model (SNM) proposed in \cite{snm}, which is also similar to the one independently proposed by \cite{kauffman}.

\subsubsection{The shot noise model}\label{sec:snm}

{The main idea in SNM} is to define the request sequence as a superposition of many independent \emph{inhomogeneous} Poisson processes called shots, where each shot is associated to an individual content and describes the temporal profile of its popularity \cite{snm}. This is a natural generalization of the Poisson IRM model in the following sense. Recall that the Poisson IRM model is represented by a  homogeneous Poisson process with intensity $\lambda$, where a Poisson point corresponds to a request for content $n$ with probability $p_n$. Using properties of the Poisson  process, this can be shown to be equivalent  to superposing  $N$ thinned homogeneous Poisson processes, one for each content, where the process for content $n$ has  intensity $\lambda p_n$. In SNM, we replace each thinned homogeneous Poisson process with an inhomogeneous Poisson process of a time-varying intensity $\lambda_n(t)$, like the one in  Fig.~\ref{fig:SNM_shot}.

If the shape of all shots is   horizontal $\lambda_n(t)=\lambda p_n,~\forall n$  then we recover the Poisson IRM. Typically however, we  use a time-varying shape to reflect the  changing nature of popularity. 
The temporal locality phenomenon can be modelled using a rapid growth towards maximum popularity (highest intensity) followed by a phase of slow decrease \cite{Zhao06}, similar to Fig.~\ref{fig:SNM_shot}. In general, the shot is characterized by its {(i)} shape, {(ii)} duration, {(iii)}  arrival instance, and {(iv)} height.

\begin{figure}[h!]
  \centering
      \includegraphics[width=0.7\textwidth]{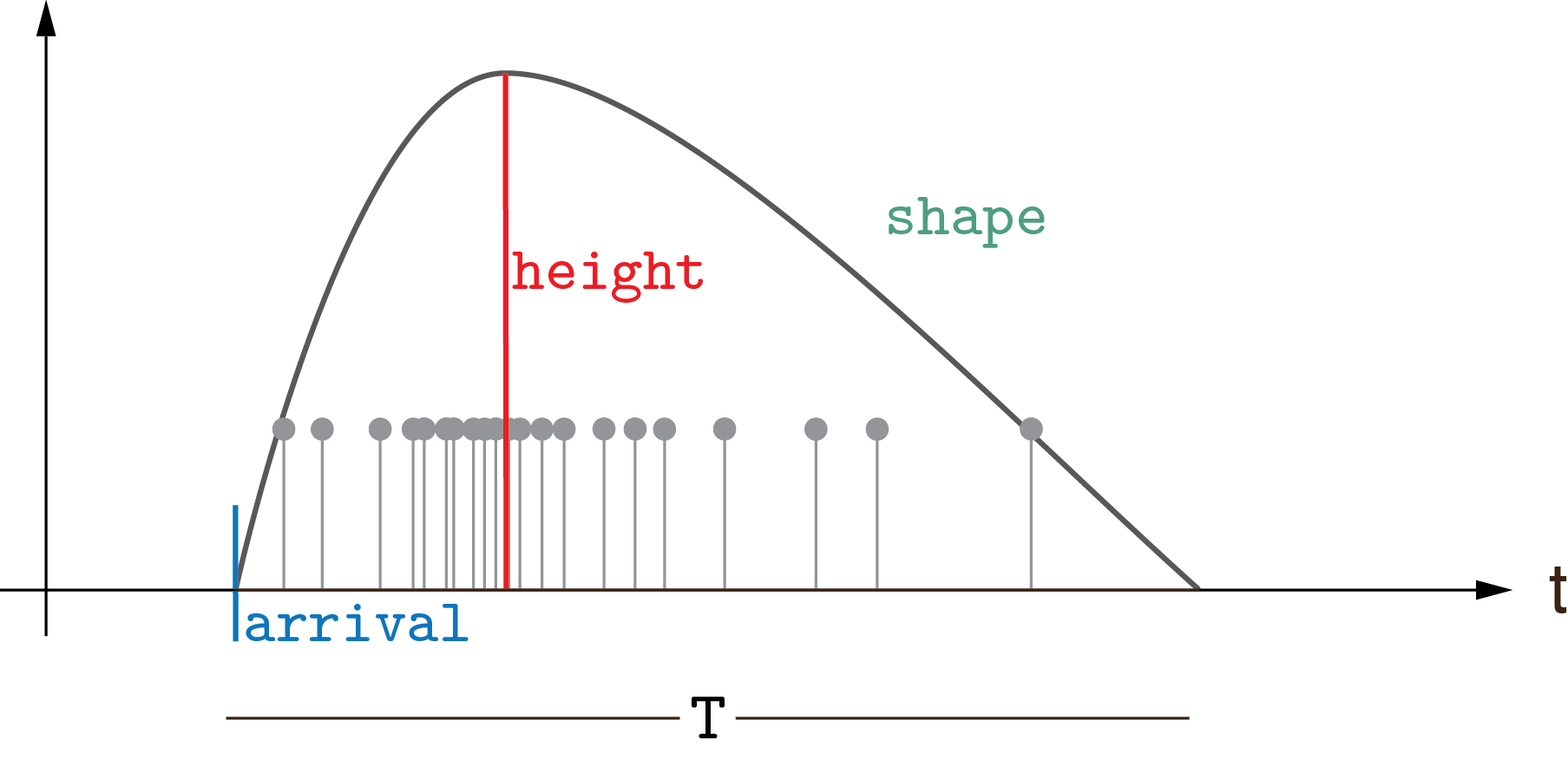}
\caption{Characteristics of a shot in SNM.}
\label{fig:SNM_shot}
\end{figure}

It is reported that different choices of (i)-(ii) result in similar  hit probability performance \cite{snm}, therefore as a general recipe we may fix the shape to   rectangular pulses and the duration to a fixed $T$ in order to maintain tractability. The shot arrival instances (feature (iii)) are points of another Poisson process with constant rate $\nu$. Denote with $t_n$ the arrival time of shot $n$. At time $t$ the active content catalog is given by the set:
\[
\mathcal{N}(t) = \{n : t-T\leq  t_n\leq t\},
\]
and we have that the expected size of the catalog at any time instance $t$ will be $\mathbbm{E}\left[|\mathcal{N}(t)|\right]=\nu T$. Hence, the parameter $\nu$ is used to tune the model to the estimated catalog. The shot height (iv) is often taken random, and power-law distributed, which ensures a good fit to the power law model when we inspect the instantaneous content popularity. Specifically, we set the height of content $n$--while it is alive--to the random variable $p_n$ constructed in the following way. First, for any $n$, let $U_n$ be an i.i.d. random variable drawn uniformly at random in $[0, 1]$. Then
\[
p_n=\frac{U^{-\Zipf}_n}{\int_0^1z^{-\Zipf}dz}\overline{p}=U_n^{-\Zipf}(1-\Zipf)\overline{p}, ~\text{for all}~n,
\]
where $\overline p$ is the mean popularity (over all contents), and $\Zipf$ is the power law exponent. We let $f(x)$ denote the density of $p_n$  at $x\in[\overline{p}(1-\Zipf),+\infty)$;\footnote{Note that density is not defined for smaller $x$.} $f(x)$ is in fact a Pareto distribution with parameters $\alpha_{\text{\tiny{Par}}}=1/\Zipf$ and $x_{\text{\tiny{Par}}}=\overline{p}(1-\Zipf)$ \cite{newman}; Pareto is the limit of Zipf distributions for large catalogs.

\begin{dfn}[theorem style=plain]{SNM model}{}
The SNM is the superposition of $N$ Poisson inhomogeneous processes, where process $n=1,\dots,N$ corresponds to  content $n$ and has a time-varying intensity $\lambda_n(t)$, characterized by (i) shape, (ii) duration, (iii) arrival instance determined by a Poisson point process with intensity $\nu$, and (iv) height $p_n$ Pareto distributed.  

Furthermore, a rectangular-SNM has (i) rectangular shape, and (ii) fixed duration $T$. 
\end{dfn}

\cite{snm} proposed to use an SNM  with four different shot shapes (all with rectangular shape but with different shot durations) to accurately fit a dataset with \youtube~video requests. Parallel work from Orange (French telecom operator) \cite{kauffman} observed that the duration of the shots is correlated with the height: highly popular contents tend to stay alive longer. The latter implies that  the SNM model can become more accurate if we correlate the duration $T_n$ of content $n$ to its shot height.

\subsubsection{Classifying contents according to instantaneous popularity}\label{sec:class}

When content popularity changes over time, we must follow a different approach from the analysis in \eqref{eq:hit_prob_opt}. The focus is now on the instantaneous hit probability, which  is random, and hence we wish to maximize its expectation. An equivalent --- and instructive --- approach is to  formulate the maximization of instantaneous hit probability  as a binary classification task. Given a cache with relative size  $\gamma$, we would like to \emph{observe  a request sequence up to $t$ and classify  each content into  the $\gamma N(t)$ most popular contents, or not,} where recall that $N(t)$ is the number of ``alive'' contents. The classification is performed by means of solving the following optimization problem.
\begin{opt}{Popularity classification}\vspace{-0.2in}
\begin{align}\label{eq:class}
&\min_{\boldsymbol y\in \{0,1\}^N} \sum_{n=1}^N y_n\mathbbm{E}\left[p_n(t)-\hat{p}_n(t)\right] ~~\quad\text{s.t. }~~ \sum_{n=1}^N y_n=\gamma N(t),
\end{align}
\end{opt}
where $p_n(t)$ is the instantaneous popularity of content $n$ at time instance $t$, {and $\hat{p}_n(t)$ is our respective estimation. Hence the objective function expresses} the difference between the maximum possible hit probability, and the one we would obtain by caching the most popular contents according to our estimations. The constraint ensures that we only classify $\gamma N(t)$ contents as cacheable. Note that $\sum_{n=1}^N y_n\mathbbm{E}\left[p_n(t)\right]$ is in fact always equal to $\overline{p}\gamma N(t)$, hence fixed and independent of the optimization variables $y_n$. Furthermore, the solution $\bm y^*$ of the classification problem \eqref{eq:class} is obtained as before by ordering the contents in decreasing estimated  popularity $\mathbbm{E}\left[\hat{p}_n(t)\right]$ and then setting $y_n^*=1$ for the top $\gamma N(t)$ of them. Hence, below we will focus on finding a good popularity estimate $\hat{p}_n(t)$ in order to decrease the hit rate loss $\sum_{n=1}^N y_n^*\mathbbm{E}\left[p_n(t)-\hat{p}_n(t)\right]$.

Hereinafter, we focus on the example of the rectangular-SNM model.  
To arrive at a popularity estimate, we assume that at time instance $t$  we know (i) the total number of requests for content $n$ up to $t$,  $K_n(t)=\sum_{i=1}^{t}\mathbbm{1}\{R_i=n\}$,  (ii) the  time $t_n$ of the first request for content $n$, from which we can infer the \emph{age} of the content $a_n(t)=t-t_n$, (the age of a content is defined as the time duration for which the content has been active), and (iii) the model of $p_n$ which is power law if $n$ active, and zero if inactive.  Accordingly,  $\hat{p}_n(t)$ is the popularity estimate based on measurements $(K_n(t),a_n(t))$. It is tempting  to use as estimates the frequencies $f_n(t)=K_n(t)/a_n(t)$. However, such a frequentist approach  would disregard the prior model. Instead, we follow the Bayesian approach, and use the \emph{a posteriori} estimate of the instantaneous popularity of content $n$ at time $t$, given by:
\begin{equation}\label{eq:apost}
\hat{p}_n \triangleq \mathbbm{E}[p_n|K_n,a_n]= \frac{\int_{p_n}p_n\prob{K_n|p_n,a_n}f(p_n)dp_n}{\int_{p_n}\prob{K_n|p_n,a_n}f(p_n)dp_n},
\end{equation}
where for notation simplicity we omitted the dependence on $t$ from all terms.  The term $\prob{K_n|p_n,a_n}$ is a Poisson probability, i.e., 
\[
\prob{K_n|p_n,t_n}=(p_na_n)^{K_n}\frac{e^{-p_na_n}}{K_n !},
\]
and $f(x)$ is the probability density function of Pareto.

Now let $F(x)$ be the distribution of $\hat{p}_n$ and $\theta$ be a threshold defined as follows: $\theta=F^{-1}(1-\gamma)$, hence $\theta$ is in fact the $\gamma$--upper quantile of $F(x)$, or more simply $\theta$ is the smallest value for which $\prob{p_n>\theta}=\gamma$ is true. We may also think of $\theta$ as a popularity threshold after which contents are classified as ``popular''. Then we define the corresponding threshold on the number of observed requests as a function of age $a$: 
\begin{equation}\label{eq:Nthres}
 \tilde{K}(a)=\min\{k=0,1,\dots | E[p_n| K_n=k,a_n=a]\geq \theta\},
\end{equation}
From \cite{mathieu} we have the following result.

\begin{thm}[theorem style=plain]{Popularity classification for rectangular SNM}{snm}
Consider a cache with relative  size $\gamma$, and let $\tilde{K}(a)$ be the threshold defined in \eqref{eq:Nthres}. Then if content $n$ has $K_n$ requests and age $a_n$ classify it according to a posteriori estimate rule:
\[
y_n^*=\left\{\begin{array}{ll}
1 & \text{if }  K_n\geq  \tilde{K}(a_n) \\ 
0 & \text{otherwise}.
\end{array}
\right.
\]
As $\nu\to\infty$,  $y_n^*$ is a solution to \eqref{eq:class} with probability 1.
\end{thm}
\begin{proof}
We focus on time instance $t$ and drop all mention of $t$, and let $\mathcal{N}$ denote the set of alive shots. As $\nu\to\infty$, the mean number of alive shots $\mathbbm{E}\left[|\mathcal{N}|\right]=\nu T$ scales to infinity as well, and the random variable $|\mathcal{N}|\to \infty$ w.p.1.  Each alive shot's popularity is Pareto distributed, and therefore we are observing a scaling number of independent Pareto experiments. In the following, we use this fact to obtain the proof.

Recall that $\hat{p}_n \triangleq E[p_n|K_n,a_n]$ is the expected popularity of a shot conditioned on  observing $K_n$ requests in $a_n$ time (and factoring the prior model of $p_n$ which is Pareto). Then, consider the empirical distribution of $\hat{p}_n$ of all alive shots:
\[
F^{\nu}(x)=\frac{1}{|\mathcal{N}|}\sum_{n=1}^{|\mathcal{N}|}\indic{\hat{p}_n<x}.
\]
Since $|\mathcal{N}|\to\infty$ w.p.1, we also have that $\lim_{\nu\to\infty}F^{\nu}(x)=F(x)$, where $F$ is the distribution of $\hat{p}_n$. Let now $\theta^*$ denote the $\gamma$-upper quantile of $F$, such that $\prob{\hat{p}_n\leq \theta^*}=1-\gamma$. This quantile is very convenient because it represents the threshold on $\hat{p}_n$ above which we should cache a shot in order to meet exactly the cache constraint of \eqref{eq:class}. Indeed, we note that this is the rule with which we decide $y_n^*$. By the convergence of distributions, it follows that the optimal empirical threshold also converges to $\theta^*$. As a result, classifying with $y_n^*$ ensures that (i) the cache constraint is exactly met (by the property of quantile), and (ii) the shots with $y_n^*=1$ have higher  $\hat{p}_n$ than those with $y_n^*=0$, hence the objective of \eqref{eq:class} is minimized.
\end{proof}

In Fig.~\ref{fig:class} we compare the above (optimal) classifier for the rectangular-SNM to the frequency estimates, and we see that they are indeed different. The differences are greater in the general SNM, without the assumption of equal shot duration $T$.

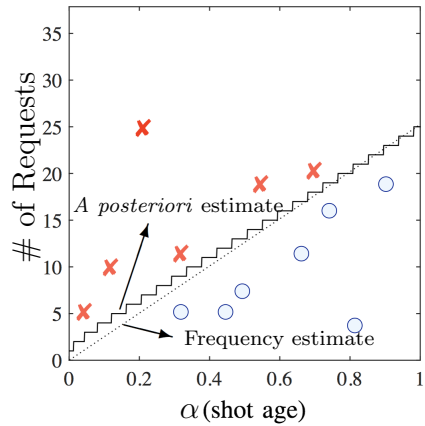
\begin{figure}[h!]
	\begin{center}
		\begin{overpic}[scale=.3]{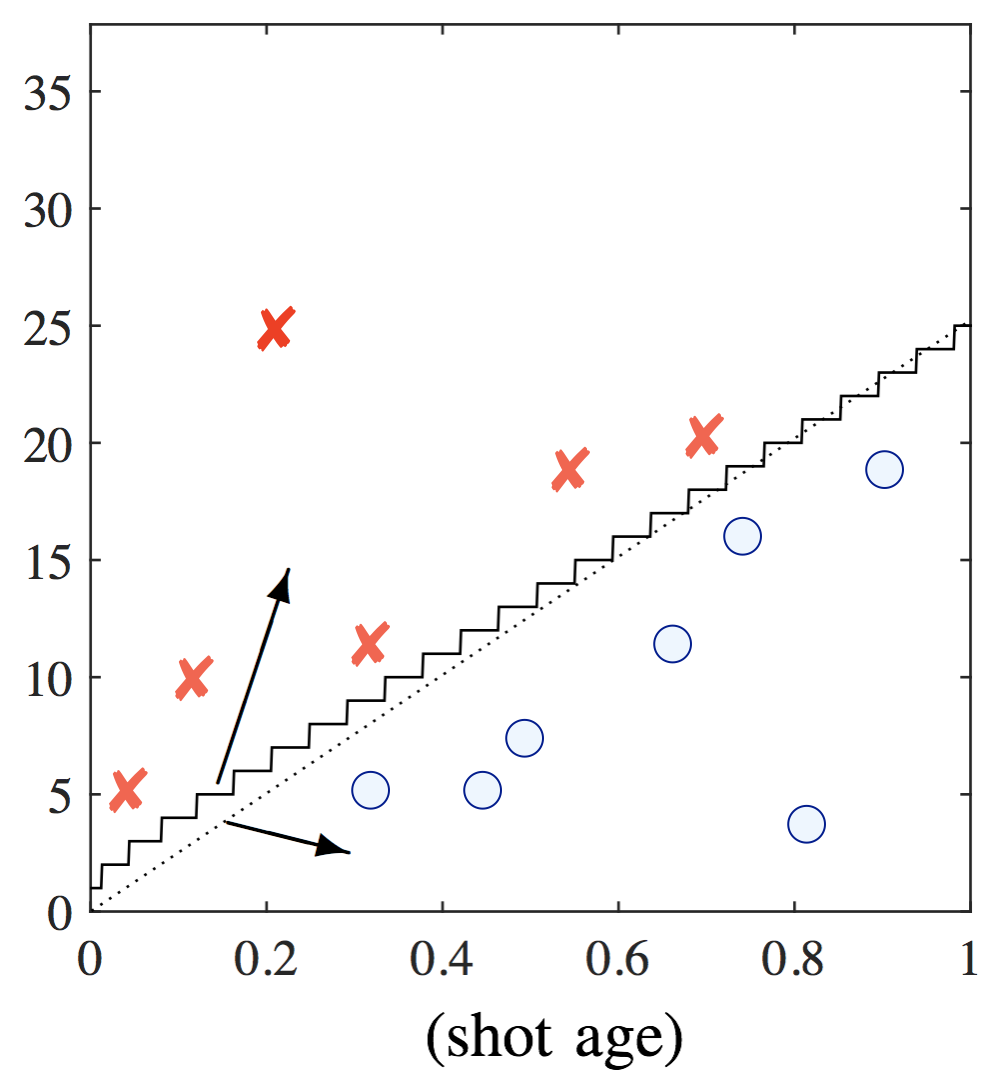}
			\put(34,2){$\alpha$}
			\put(-5,38){\rotatebox{90}{\# of Requests}}	
			\put(9,50){\scriptsize \emph{A posteriori} estimate}
			\put(35,19){\scriptsize Frequency estimate}
		\end{overpic}
		\caption{Two classifiers: the \emph{a posteriori} estimate (solid line) and the frequency (dotted line).  
		}
		\label{fig:class}
	\end{center}
\end{figure}

To summarize the above, we have discussed the common phenomenon of temporal locality, which amounts to time correlations in the request sequence. Such correlations can be captured accurately by non-stationary models, such as the SNM. The latter, also allows for a simple derivation of the optimal classifier (to popular/unpopular) by  exploiting the knowledge of the prior model to  outperform the standard frequency estimator. Future research is needed along these lines to solidify our understanding of what is a good non-stationary model, and derive Bayesian estimators for time-varying  popularity distributions.


\subsection{Spatial locality}

The requests also exhibit correlations across space. For example, crossing the boarders between two countries with different language has a severe impact on the observed content popularity. Recent work \cite{Kurose08,Scellato_2011} has also shown that the popularity distribution changes rapidly within a city as well.
Another work \cite{Fourno_17} focused on profiling traffic volumes at base stations in several European cities and showed that the traffic fluctuates based on a location-dependent pattern, called ``signature''. 
Metro station areas, residential, commercial, touristic, academic, and other type of areas have their own distinctive signature. Similarly, it is believed that the content preference may vary depending on the culture and activity  at each geographical area. 
In a multi-cache  installation spanning a large geographical area, such spatial correlations result in certain contents being very popular in a subset of caches, and unpopular in the rest. This also suggests that \emph{the popularity model has location-dependent  ``features'', and clustering together locations with similar features can improve popularity learning}.

We mention, however, that identifying 
spatial correlations in popularity is challenging since it requires collection of data at many different points in the network, simultaneous availability of content id (application information) and location information (user-specific information), and a large number of data in a short time interval from a sizeable geographical area. To the best of our knowledge, no such data are published to date. Our discussion is therefore restrained here to the proposition of mathematical models for generating spatially correlated requests, without the ability to validate the models with real data. In the next subsection we propose the stochastic block model that  captures the  geographical correlations of content requests, and allows  easy inference and optimization.   

\subsubsection{The stochastic block model}

In this subsection we present  the \emph{Stochastic Block Model} (SBM).  We consider $L+1$ caches, where the edge caches $\{1,\dots,L\}$ receive local requests from $L$ different geographical locations, while the parent cache $0$ receives the aggregate  requests. Cache $l\in \{0,1,\dots,L\}$ ``sees'' a popularity distribution $(p_{nl})_{n=1,\dots,N}$, which in general might vary across caches, see Fig.~\ref{fig:tree_learning}.

\begin{figure}[h!]
  \centering
      \includegraphics[width=0.75\textwidth]{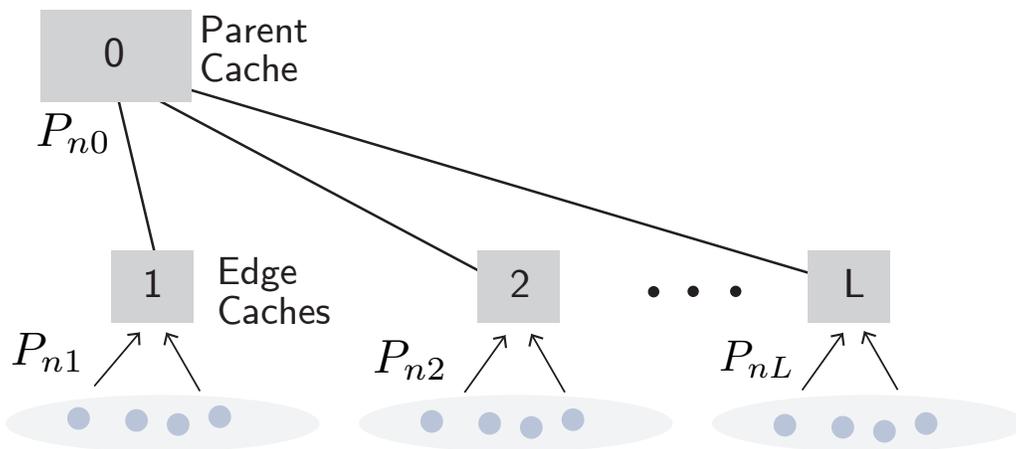}
\caption{Modeling correlated popularities in a system with $L+1$ caches.}
\label{fig:tree_learning}
\end{figure}

We propose here a model for correlated local popularities $(p_{nl})_{n}, l=1,\dots,L$. The global popularities $(p_{n0})_n$ can follow any law we desire, for example we may assume that they are described by an SNM with mean popularity $\overline{p}$. To model correlations between the local cache popularities, we draw inspiration from the field of community detection \cite{Lelarge13}:
\begin{itemize}
\item Each content $n$ is associated with a feature vector $X_n$. To simplify the model we let $(X_n)$ be independent uniform random variables, each taking values in $[0, 1]$.
\item Each location $l$ is associated with feature vector $Y_l$, which
are again chosen independently and uniformly in $[0, 1]$.
\item We define a kernel $K(x, y) = g(|x-y|)$, where $g$ is a continuous function, strictly decreasing on [0, 1/2], symmetric and 1-periodic, with $\int_{[0,1]} g(|x-y|)dy=1$ for all $x\in[0, 1]$. Since the kernel $K$  is periodic,  we can think of the feature vectors $(X_n),(Y_l)$ as lying on the torus $[0, 1]$ rather than the interval.
\end{itemize}

The local popularity of content $n$ at cache $l$ is defined as
\[
p_{nl}=p_{n0}\frac{K(X_n,Y_l)}{\sum_{l'\in L}K(X_n,Y_{l'})},~~\forall n,~l=1,\dots,L.
\]
One can trivially verify that by  averaging the popularity of locations $1,\dots,L$, we obtain $p_{n0}$ (the popularity of location 0). As such,  SBM is completely determined by the model of $p_{n0}$ and the chosen kernel $K(x, y)$, where the latter is responsible for inducing the correlations. We can further simplify the model by choosing the kernel to have a specific form, for example \cite{mathieu} uses:
\[
K(x,y)=5(1-2|x-y|)^4.
\]
The SBM is an analytical model where  popularity  exhibits spatial correlations. In this setting, we note  an inherent tradeoff of learning. The global location 0 aggregates the requests from all locations, and therefore  learns quickly the average popularity distribution $(p_{n0})$, owing to the large number of request samples per unit time. On the other hand, each individual location learns slower (due to smaller sample), but it  learns the local model $(p_{nl})$, which is more accurate than the average of all models $(p_{n0})$. An interim approach is also possible, where the distribution is learned from aggregating requests from a subset of locations with similar popularity. Such a cluster-based learning would be faster than local learning and  more accurate than global learning. As an example, suppose half of the locations share the same popularity distribution, and they are negatively correlated with the other half such that $p_1\geq p_2\geq \dots,\geq p_N$ in the first, and   $p_1\leq p_2\leq \dots,\leq p_N$ in the second. Then learning in two clusters would greatly outperform both  global and  local learning. This can be achieved with standard spectral clustering techniques from the literature, for example \cite{von2007tutorial}.

\section{Conclusions}

We have surveyed the key concepts of content popularity, explained spatio-temporal correlations in the request process, and presented the most widely used content popularity models. Predicting a time-varying popularity is a key challenge in caching, and we provide more information in Sec.~\ref{sec:prediction}. In the next chapter we study policies that decide which contents to cache in reaction to a unknown (and potentially changing) popularity.



\chapter{Cache eviction policies}\label{ch:3}

More often than not,
the cache is full, and hence in order to cache a new content, another content must be evicted to make room in the cache. 
This chapter focuses on cache eviction policies that decide  which content should be evicted. 
  Broadly speaking we seek to find dynamic eviction policies that maximize hit performance. 
In light of this, the chapter presents offline and online cache eviction policies with associated performance guarantees under various settings. We examine different assumptions for the underlying request sequence, emphasizing fundamental properties of caching policies, and providing intuition for designing new policies in order to match the needs of new systems. 

\begin{figure}[!h]
\centering	
\includegraphics[width=3.2in]{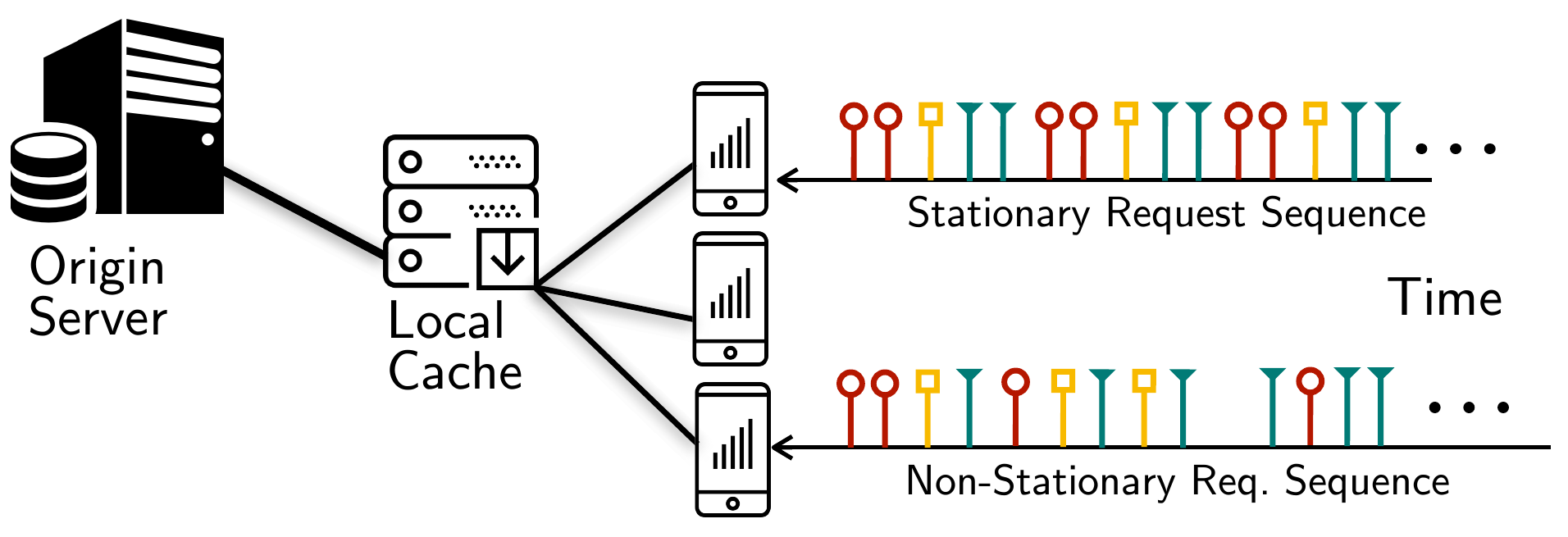}
\caption{
Content requests must  be served by  a local cache (hit) or the origin server (miss). An eviction policy observes past requests and decides which contents to cache in order to increase the hit probability (or hit rate).
}
\label{Fig:motivation}
\end{figure}

\section{Performance under arbitrary requests}

In this section we make no assumptions about requests, and analyze the performance of caching policies under an arbitrary request sequence.

\subsection{The paging problem}

Consider a catalog of $\mathcal{N}=\{1,2,\dots,N\}$ contents,  and a cache that fits at most $M<N$ of them, where all contents have equal size. Time is slotted, $t=1,2,\dots,T$, and there is  an arbitrary request sequence $\boldsymbol R=(R_1,R_2,\dots,R_T)$, {with $R_t\in \mathcal{N}$ for every $t$}. The cached contents at time $t$ are given by the set $\mathcal{M}_t\subseteq \mathcal{N}$, and there is a {hit} if $R_t\in \mathcal{M}_t$ and a {miss} if $R_t\notin \mathcal{M}_t$. In case of a miss, we add the newly requested content to the cache and, when $|\mathcal{M}_t\cup R_t|>M$ we \emph{evict} one of the cached items, denoted $E_t$. The evolution of the  \emph{cache state} is given by:
\begin{align*}
&\mathcal{M}_{t+1}=\left\{ \begin{array}{ll}
 \mathcal{M}_t\cup R_t - E_t & \text{if we have a miss}\\
\mathcal{M}_t & \text{if we have a hit}
\end{array}\right.
\quad t=1,\dots,T-1.
\end{align*}

\begin{dfn}[theorem style=plain]{Eviction policy}{}\label{eq:eviction}
The online  eviction policy $\pi$ is a rule that at time $t$ maps the current cache state $\mathcal{M}_t$ and  the observed request sequence $(R_1,R_2,\dots,R_t)$ into an evicted content with the constraint that $E^{\pi}_t\in \mathcal{M}_t \cup R_t$ when $R_t\notin \mathcal{M}_t$ and $E^{\pi}_t=\emptyset$ otherwise. 
\end{dfn}
We allow the possibility for a policy $\pi$ to be aware of the entire sequence $(R_1,R_2,\dots,R_T)$ at each slot $t$ (\emph{offline policy}) instead of just $(R_1,R_2,\dots,R_t)$ (\emph{online policy}). 
We will denote the set of online and offline policies with $\Pi$, which includes all the $T$-vectors of feasible eviction decisions. 
We will consider two models for eviction decisions:
\begin{enumerate}
\item[(\Mone)]  $E^{\pi}_t\in  \mathcal{M}_t\cup R_t$, where the requested item at slot $t$ is also eligible for eviction.
\item[(\Mtwo)]  $E^{\pi}_t\in  \mathcal{M}_t$,  where the requested item at slot $t$ is \emph{not} eligible for eviction.
\end{enumerate}
I2 is often encountered in the vast literature of databases and computing, while \Mone~is a reasonable option for Internet caching applications. For clarity, we will make explicit reference to the employed model, but we mention that the differences between the two models are minor.

We denote with  $h^{\pi}(\boldsymbol R)$ the \emph{number of hits} under policy $\pi\in\Pi$ and request sequence $\boldsymbol R$. We can now formally state the \emph{paging problem}: find an eviction policy that maximizes the number of hits for the worst-case request sequence.

\begin{opt}{Paging problem}\vspace{-0.1in}
\begin{equation}\label{eq:paging}
\max_{\pi\in\Pi}\min_{\boldsymbol R\in \mathcal{N}^{T}} h^{\pi}(\boldsymbol R)
\end{equation}
\end{opt}

The problem is called paging since it was originally  formulated in the scope of paging  memory of computing systems \cite{Sleator85}. In computer science it is used as  a didactic example of competitive analysis for deterministic algorithms.

\subsection{Belady's offline paging}

We first describe an offline policy that maximizes $h^{\pi}(\boldsymbol R)$  for any $\boldsymbol R$, and therefore solves \eqref{eq:paging}. The algorithm was first proposed in \cite{Belady66} by Belady, and it is also known as the \emph{max stack distance}, or the \emph{MIN} replacement policy (it minimizes the misses), or caching oracle.

We denote with $D(n,t)$ the \emph{stack distance} of content $n$, which at time $t$ is defined  as the number of slots until the next request for $n$ arises. Formally:
\[
D(n,t)=\inf \{ \tau> t, R_{\tau} = n\}-t. 
\]
If there is no other occurrence of $n$ in the remaining sequence, then we trivially set $D(n,t)=\infty$. {An example is presented in Figure  \ref{fig:belady}.}

\begin{figure}[h!]
	\centering
	\includegraphics[width=0.85\textwidth]{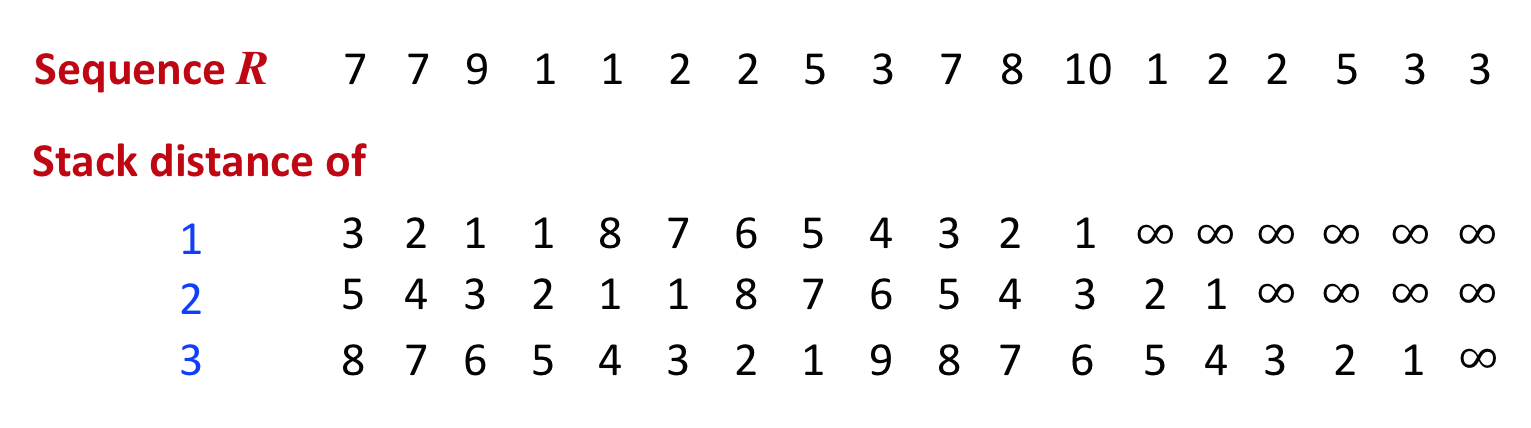}
	\caption{Showcased stack distances $D(1,t),D(2,t),D(3,t)$.}
	\label{fig:belady}
\end{figure}

\begin{box_example}[detach title,colback=blue!5!white, before upper={\tcbtitle\quad}]{Belady eviction policy.}
Belady evicts the content with the largest stack distance:
\[
E^{\textit{MIN}}_t={\arg\max}_{n\in \mathcal{M}_t\cup R_t} D(n,t),
\]
where we use ``\emph{MIN}'' for Belady. Hence the policy always evicts the content that will be requested further in the future. 
\end{box_example}

We remark that Belady is an impractical eviction policy because it requires the calculation of the stack distance, {and this in turn involves information about the future request arrivals}. Nevertheless, Belady has important practical uses. First, since it achieves the maximum possible hits, it provides  a performance upper bound for online eviction policies on datasets. {Therefore it can be used to test whether the hits achieved by another policy 
are satisfactory.} 
Second, as we will show in the following, it can be used to 
characterize the performance of practical online  eviction policies. Next, we rigorously state the optimality of Belady.

\begin{thm}[theorem style=plain]{Belady achieves optimal paging}{belady}
	Fix $T$, and choose any sequence $\boldsymbol R\in\mathcal{N}^T$. Belady achieves the highest number of hits  in class $\Pi$ under both \Mone~and \Mtwo. 
	\[
	\textit{MIN}\in {\arg\max}_{\pi\in\Pi} h^{\pi}(\boldsymbol R). 
	\]
\end{thm}

The classical proofs \cite{Mattson70,Roy10,Vogle08,Lee15} prove the statement under model \Mtwo. Here we give an elegant proof for \Mone.  
\begin{proof}[Proof of Theorem \ref{thm:belady} (case \Mone)]
Assume any eviction policy $A\in\Pi$. We will show that it yields less or equal  hits to Belady, i.e., $h^A(\boldsymbol R)\leq h^{\textit{MIN}}(\boldsymbol R)$. The proof is based on showing that  starting from policy $\pi_1$ we can produce a new policy $\pi_2$ that takes one more decision like Belady and performs no worse than $\pi_1$. Then we use induction on the number of steps. Below we drop the notation $(\boldsymbol R)$ with the understanding that $\boldsymbol R$ is fixed.

\begin{lem}[theorem style=plain]{One more Belady step}{dom}
Consider a policy $A_k\in \Pi$ with $h^{A_k}$ hits, which takes the same eviction decisions with {MIN} up to slot $k$, where $0\leq k\leq T-1$, that is
\[
E^{A_k}_t=E^{\textit{MIN}}_t,\quad \forall t=0,\dots, k.
\]
Then, there exists a policy $A_{k+1}\in \Pi$ satisfying:
\begin{align}
	& E^{A_{k+1}}_t=E^{MIN}_t,\quad \forall t=0,\dots, k+1.\label{eq:part1}\\
	& h^{A_{k+1}}\geq h^{A_{k}}.\label{eq:part2}
\end{align}
\end{lem}
Lemma~\ref{lem:dom} says that there exists a modification of any policy $A_k$  to $A_{k+1}$ which makes one more decision like Belady, and  the number of hits of $A_{k+1}$ are no less than those of $A_k$. Assuming w.l.o.g. that policy $A$ performs the same steps with Belady up to slot $a$, where $0\leq a \leq T$, by Lemma \ref{lem:dom} and induction on $k$, we get $h^A= h^{A_a} \leq h^{A_{a+1}} \leq \dots \leq h^{MIN}$, which proves the theorem.
\end{proof}

\begin{proof}[Proof of Lemma \ref{lem:dom}]
The proof is by construction. {We construct $A_{k+1}$ as follows:} (i) for slots $0,\dots, k$ it mimics $A_k$ and \emph{MIN} (which are the same by the premise of the Lemma), (ii) at slot $k+1$ it takes the decision that \emph{MIN} would have taken, (iii) for slots $t=k+2,\dots,T$ it takes feasible decisions that minimize $|\mathcal{M}^{k}_t-\mathcal{M}^{k+1}_t|$,\footnote{The set difference $\mathcal{M}^{k}_t-\mathcal{M}^{k+1}_t$ corresponds to the contents policy $A_k$ maintains in the cache but policy $A_{k+1}$ does not.} where in case of a tie the content to be evicted is the same as in $A_k$. We immediately obtain that $A_{k+1}\in \Pi$ and \eqref{eq:part1} hold. {It remains to show \eqref{eq:part2}.}
	
Up to time $k$, both policies take the same decisions and thus yield the same hits. Hence we will compare the two policies in the interval $\{k+1,\dots,T\}$. Additionally, we note that if at some slot $t>k$ it is $\mathcal{M}^{k+1}_t=\mathcal{M}^{k}_t$, then by construction of $A_{k+1}$ the caching states of the two policies $A_k,A_{k+1}$ become identical until $T$. 
Consider two useful slot events in the interval $\{k+1,\dots,T\}$:
\begin{align*}
	\Omega:~~\text{in a slot $A_k$ gets a hit while $A_{k+1}$ not}\\
	\hat\Omega:~~\text{in a slot $A_{k+1}$ gets a hit while $A_{k}$ not}
\end{align*}
We will prove that $\Omega$ occurs at most once and only after $\hat\Omega$ has previously occurred, which proves \eqref{eq:part2} and hence the lemma. Note that at time $k$ policy $A_{k+1}$ has evicted the content with the maximum stack distance (as Belady does) and hence $\Omega$ cannot occur at slot $k+1$. Therefore, it suffices to prove the following two statements:  
\begin{itemize}
	\item[(a)] If at slot $t\in \{k+2,\dots,T\}$ $\Omega$ occurs, then $|\mathcal{M}^{k}_t-\mathcal{M}^{k+1}_t|$ becomes zero, and as mentioned above the two policies keep the same states until the end. Hence,  $\Omega$ can not occur again.
	
	\item[(b)] Suppose that $\Omega$ occurs at $\hat{t} \in \{k+2,\dots,T\}$, then there exists another slot in the interval $\{k+1,\dots,\hat{t}-1\}$ such that $\hat\Omega$ occurs.
\end{itemize}
	
\textbf{Preliminary:} Let $a\in \mathcal{N}$ denote the max-distance content  evicted by $A_{k+1}$ in slot $k+1$ (mimicking Belady)  and $b\in \mathcal{N}$ denote another content which  is evicted by $A_{k}$, $b\neq a$. Also, denote by $\mathcal{I}$ the intersection of the caching states of the two policies. For clarity, at time $k+1$:
\begin{align*}
	&\mathcal{M}^k_{k+1}=\mathcal{I}\cup \{a\}\\
	&\mathcal{M}^{k+1}_{k+1}=\mathcal{I}\cup \{b\}\\
	&\mathcal{M}^k_{k+1}-\mathcal{M}^{k+1}_{k+1}=\{a\}\\
	& \text{Uncached contents: } \mathcal{N}- \mathcal{I}- \{a,b\},
\end{align*}
where, the notation $\mathcal{M}^m_n$ denotes the set of contents that are cached in slot $n$ by policy $A_m$.

\textbf{Proof of (a):} Observe that for any $t$ there are only two possible set differences $\mathcal{M}^k_t-\mathcal{M}^{k+1}_t\in \{\emptyset, \{a\}\}$. This is ensured by the construction of $A_{k+1}$ which takes decisions to minimize $|\mathcal{M}^k_t-\mathcal{M}^{k+1}_t|$. Any new content cached by $A_k$ must then also be cached by $A_{k+1}$. Therefore,  for  event $\Omega$ to occur, content $a$ must be requested. In this case, $A_{k+1}$ will take a decision that makes the set difference an empty set (evicting its own extra content in order to cache $a$). After this slot, we always have $|\mathcal{M}^k_t-\mathcal{M}^{k+1}_t|=0$.

\textbf{Proof of (b):} Since  $\Omega$ occurs, content $a$ was requested at slot $\hat t$, and since $a$ is the max stack distance content in slot $k+1$, it follows that  the interval $\{k+1,\dots,\hat{t}-1\}$ must contain a slot where $b$ is requested. Let $t_b$ denote the slot where $b$ is requested for the first time in the interval $\{k+1,\dots,\hat{t}-1\}$. For $\hat\Omega$ to occur at $t_b$ it suffices to show that  $b\in \mathcal{M}^{k+1}_{t_b}$, which we do next. Observe that no requests of $a$ or $b$ occur in $\{k+1,\dots,{t}_b\}$. Hence, requested contents belong either to $\mathcal{I}$ (which  preserve $b$), or to $\mathcal{N}-\mathcal{I}-\{a,b\}$ in which case a new content $x$ might be cached by both policies while another content is evicted. Here note that $b$ is not evicted by $A_{k+1}$ by the way the ties are resolved.
\end{proof}

\subsection{LRU and competitive online paging}
%

{Having studied the benchmark Belady offline policy, we now move to the study of online paging policies.} Our goal here is to study which online eviction policy maximizes the hits in the worst-case.  For presentation purposes, we consider the equivalent problem of minimizing the misses, i.e., $m^*(\boldsymbol R,M)=\min_{\pi\in \Pi}m^{\pi}(\boldsymbol R,M)$, where $m^{\pi}(\boldsymbol R,M)$ is the number of misses employing $\pi$ on a cache of size $M$ versus sequence $\boldsymbol R$. This is equivalent to hit maximization, and the minimum misses are  also  achieved by the Belady offline policy. 

We emphasize that the paging problem is an  ``adversarial" setting, where the cache selects its eviction policy first, and then an adversary chooses the request sequence that induces the most misses. Analyzing 
the worst-case misses of an online policy is not informative because for any online policy $\pi$ we can find a request sequence that yields $m^{\pi}(\boldsymbol R)=T$ (or zero hits), which implies that all online paging policies have the same worst-case performance. Hence, in order to have a meaningful competitive analysis 
 we will  disadvantage the  offline paging  with a smaller cache and then  compare the number of misses between online and  size-restricted offline. This methodology was first proposed in \cite{Sleator85}.

\begin{dfn}[theorem style=plain]{Miss competitive ratio}{}
A policy $\pi\in\Pi$ is $(\xi,\rho)$--competitive if 
\[
m^{\pi}(\boldsymbol R,M)\leq \xi\cdot m^{*}(\boldsymbol R,\rho M), ~~\forall \boldsymbol R.
\] 
\end{dfn}

In words, a $(\xi,\rho)$--competitive policy is a policy that for  any sequence $\boldsymbol R$  yields at most $\xi\geq 1$ times more misses than Belady would do  with $\rho\leq 1$ times smaller cache. 

\begin{thm}[theorem style=plain]{Lower bound on miss competitive ratio}{lbound}
Run an online policy $\pi$ with cache size $M$, and Belady with cache size $x\leq M$. There exists sequence $\boldsymbol R$ such that: 
\begin{align*}
&\text{\Mone:}\quad m^{\pi}(\boldsymbol R,M)\geq  \left(\frac{M+1}{M-x+1}\right)m^*(\boldsymbol R,x),\\
&\text{\Mtwo:}\quad m^{\pi}(\boldsymbol R,M)\geq  \left(\frac{M}{M-x+1}\right)m^*(\boldsymbol R,x).
\end{align*}
For both \Mone, \Mtwo, $(\frac{M}{M-x+1},\frac{x}M)$ is a lower bound on the  miss competitive ratio.
\end{thm}
\begin{proof} We prove the claim for \Mone, and we mention that  the same argument works for \Mtwo~with minor differences, see \cite{Sleator85}.
We  construct a sequence $\boldsymbol R$ of length $M+1$ on which ${\pi}$
has $M+1$ misses and Belady only $M-x+1$.  

 Specifically $\boldsymbol R$ is constructed to be the concatenation  of two subsequences, $\boldsymbol R_1$ with the first $M-x+1$ requests and $\boldsymbol R_2$ with the last $x$ requests. The $M-x+1$ requests of subsequence $\boldsymbol R_1$ are all for different contents, that are not cached  neither in ${\pi}$'s nor in Belady's cache, hence both policies incur $M-x+1$ misses. Fix $\boldsymbol R_1$, and let $\mathcal{S}$ be the set denoting the union of  Belady's original $x$ cached contents and the $M-x+1$ newly requested contents of $\boldsymbol R_1$. Note that $|\mathcal{S}|=M+1>M$, and it follows that irrespective of the choices of $\pi$, the set $\mathcal{S}\setminus  \mathcal{M}^{\pi}_{t}$ is non-empty for all $t=M-x+2,\dots,M+1$. The remaining $x$ requests (of subsequence $\boldsymbol R_2$) are selected  from these sets such that on the combined sequence of $M+1$ requests, ${\pi}$ misses all requests. 

Since all last $x$ requests are selected from either Belady's original cache state or $\boldsymbol R_1$, Belady achieves a hit for each one of them, hence Belady misses only $M-x+1$ times in total. Although this construction has a fixed sequence length equal to $M+1$, we may repeate it as many times as desired, giving the theorem. This concludes the proof for \Mone.

We mention that the proof for \Mtwo~is obtained by examining a sequence of length $M$, having a $\boldsymbol R_2$ of length $x-1$. The reason lies in the fact that Belady can only guarantee  hits if $\boldsymbol R_2$ is made of $x-1$ contents due to the restriction in \Mone~of always caching the arriving request even if it is the max stack distance one.
\end{proof}

The theorem shows that the best competitive ratio we can hope for is $(\frac{M}{M-x+1},\frac{x}M), x=1,\dots,M$. Next, we discuss an online policy that achieves this bound.
\begin{box_example}[detach title,colback=blue!5!white, before upper={\tcbtitle\quad}]{LRU eviction policy.}
\small Upon a miss, evict the content that is least recently requested (used). Formally, at time $t$ let $r_n(t)$ be the time elapsed since last request of content $n$,  called \emph{recency of  $n$ at time $t$}. Upon a miss, evict the content with maximum recency $\arg\max_{n} r_n(t)$.

Since recency $r_n(t)$ depends only on the past sequence $R_1,\dots, R_{t}$,  LRU is an online policy.
\end{box_example}

\begin{thm}[theorem style=plain]{LRU is $(\frac{M}{M-x},\frac{x}M)$-competitive}{compLRU}
Fix sequence $\boldsymbol R$ and some $x\leq M-1$ (or $x\leq M$ in \Mtwo). LRU satisfies:
\begin{align*}
\text{\Mone:}\quad &m^{LRU}(\boldsymbol R, M) \leq \frac{M}{M-x} m^{*}(\boldsymbol R, x),\\
\text{\Mtwo:}\quad &m^{LRU}(\boldsymbol R, M) \leq \frac{M}{M-x+1} m^{*}(\boldsymbol R, x).
\end{align*}
\end{thm}
Note that the LRU achievable performance matches that of Theorem \ref{thm:lbound} exactly for \Mtwo, and almost matches less a small term for \Mone. Specifically for \Mone, the LRU misses are in the interval ${\scriptstyle \left[ \left(\frac{M+1}{M-x+1}\right)m^*(\boldsymbol R,x), \frac{M}{M-x} m^{*}(\boldsymbol R, x)\right]}$ and the interval  vanishes as $M$ increases (independent of the value of $x$). 
\emph{We may conclude that  LRU provides the best  miss competitive ratio in the adversarial setting.} If we set $x=M-1$, we can see that LRU makes  at most $M$ times more misses than the offline optimal, and this also  holds for very long sequences where $T>>M$. More interestingly, set $x=\lfloor M/2 \rfloor$ and check that even an optimal algorithm that knows the future would have caused at least half misses using roughly half the cache. 

\begin{proof}[Proof of Theorem~\ref{thm:compLRU}]

Decompose $\boldsymbol R$ into phases that contain exactly $M$ \emph{different} requests (and possibly more than $M$ requests in total) and a possible remaining sub-phase. In particular, each phase ends just before a request for a content that is different than the $M$ previous. Let $p$ denote the number of phases containing $M$ different requests. Due to recency, multiple requests for the same contents within a phase result  always in hits, hence we have
\[
m^{LRU}(\boldsymbol R, M)\leq M\cdot p +M,
\]
where the last term represents misses in a possible remaining sub-phase. Considering Belady with a cache $x$, there must be at least $M-x$ misses at each phase irrespective of what is cached at the beginning of each phase, i.e., $m^{*}(\boldsymbol R, x)\geq (M-x)\cdot p$. Combining, we conclude
\[
m^{LRU}(\boldsymbol R, M) \leq M\cdot p +M \leq \frac{M}{M-x}m^{*}(\boldsymbol R, x)+M.
\]
The case of \Mone follows by amortizing the term $M$ over the much longer sequence $\boldsymbol R$. See  \cite{Sleator85} for \Mtwo.
\end{proof}


\begin{figure}[h!]
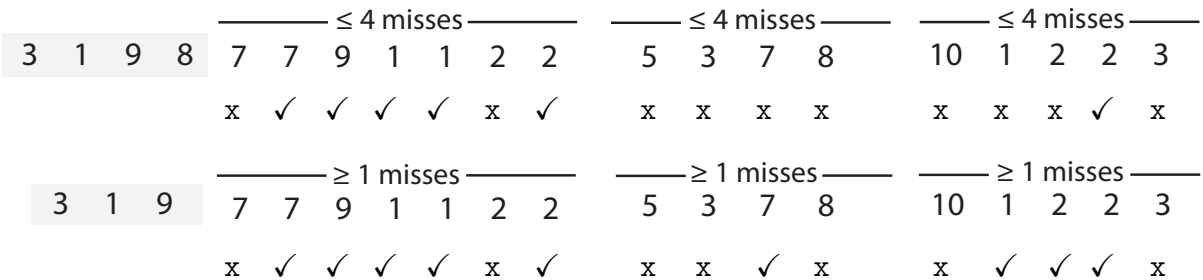

\begin{center}
	\begin{overpic}[scale=.62]{LRU_competitive_new3}
		\put(18.7,9){\texttt{x}}
		\put(22.7,9){\checkmark}
		\put(27,9){\checkmark}
		\put(31,9){\checkmark}
		\put(35.4,9){\checkmark}
		\put(40.3,9){\texttt{x}}
		\put(44.4,9){\checkmark}
		\put(53.2,9){\texttt{x}}
		\put(57.8,9){\texttt{x}}
		\put(62.8,9){\texttt{x}}
		\put(67.6,9){\texttt{x}}
		\put(77.5,9){\texttt{x}}
		\put(82.5,9){\texttt{x}}
			\put(87,9){\texttt{x}}
			\put(90.5,9){\checkmark}
			\put(95.5,9){\texttt{x}}
			\put(18.7,-4){\texttt{x}}
			\put(22.7,-4){\checkmark}
			\put(27,-4){\checkmark}
			\put(31,-4){\checkmark}
			\put(35.4,-4){\checkmark}
			\put(40.3,-4){\texttt{x}}
			\put(44.4,-4){\checkmark}
			\put(53.2,-4){\texttt{x}}
			\put(57.8,-4){\texttt{x}}
			\put(62.8,-4){\checkmark}
			\put(67.6,-4){\texttt{x}}
			\put(77.5,-4){\texttt{x}}
			\put(82.5,-4){\checkmark}
			\put(87,-4){\checkmark}
			\put(90.5,-4){\checkmark}
			\put(95.5,-4){\texttt{x}}
	\end{overpic}
\vspace{0.1in}\caption{Example of competitive paging with $M=4, x=3, p=3$ (Theorem \ref{thm:compLRU} ensures that $m^{LRU}(\boldsymbol R,4)\leq 4 m^{*}(\boldsymbol R, 3)$ for any sequence and on the given sequence we observe that indeed $m^{LRU}(\boldsymbol R,4)=10$, $m^*(\boldsymbol R,3)=7$). 
}\label{fig:compLRU}
\end{center}
\end{figure}

From the above matching bounds we conclude that  LRU is the \emph{optimal} online eviction policy under  arbitrary request sequences with respect to the miss competitive ratio metric. Additionally, LRU is a very practical policy. Its doubly linked list and hashmap implementation runs in $O(1)$ \cite{tanenbaum2009modern} (hence independent of cache size and catalog size) and it is often used in CDNs. However, as we shall see, if more information is known about the request sequence, then LRU may be outperformed by other practical online policies.

Furthermore, although no online policy can guarantee a better worst-case performance than LRU, consider the following: an online paging policy  is called \emph{conservative}, if on any  request sequence with $M$ or fewer  distinct contents will incur $M$ or fewer misses. It turns out that Theorem \ref{thm:compLRU} applies to every conservative policy, examples of which are LRU, \emph{First In First Out} (FIFO) and \emph{CLOCK} \cite{wsclock}. 
From practical experience,  LRU often performs well but FIFO has typically poor performance. This certainly reduces the value of LRU optimality result, and serves as an indication that perhaps the metric of competitive ratio 
does not allow us to successfully distinct good caching policies from bad ones when it comes to the operation of a cache with real requests.

\subsection{Discussion of Related Work}

{Apart from the typical}
FIFO and RANDOM policy (which evicts a content at random), perhaps the most well-known eviction policy is \emph{Least Frequently Used} (LFU), which maintains a list of content  empirical frequencies, and evicts the least frequent content (assumed to be the least popular). In the next section we show that this policy achieves the optimal hit rate under stationary requests.

The simplicity and appealing properties of the LRU policy has inspired a great deal of work on  online paging policies. Perhaps the simplest generalization is the \emph{CLOCK} policy, which provides low implementation complexity. Specifically, while LRU needs to lock the content list in order to move the accessed content to the most recently used position (in order to ensure consistency and correctness under concurrency), CLOCK's clever implementation eliminates  lock contention, and supports high concurrency and high throughput. It was later improved to \emph{WSCLOCK} \cite{wsclock}. 
  
In an interesting short survey  \cite[pp. 601--604]{enc_algo},  a number of generalizations of the  paging problem are summarized. \emph{Weighted caching} is the generalization  where the eviction cost  is  content-dependent. 
For weighted caching, \cite{Young02}  studied the linear programming structure of the more general $k$-cache problem  and  obtained the surprising insight that the LRU  and the \emph{Balance} policies are primal-dual algorithms. By  generalizing them both, it obtained \emph{GreedyDual}, a  competitive primal-dual strategy for weighted caching. 

In a similar manner, \emph{GreedyDualSize} is an extension for  paging with differently sized  contents \cite{gds}, where each content is assigned a size related credit, and the credits are adjusted according to requests. By making small adjustments in credits, the algorithm is able to promote popular and size-efficient contents. The \emph{LANDLORD} is a generalization of GreedyDualSize, where both  content-specific costs and sizes are considered; the authors in \cite{young2002line} showed that LANDLORD encapsulates  GreedyDualSize,  GreedyDual, and LRU, and achieves the standard competitive ratio $(\frac{M}{M-x},\frac{x}M)$.

Another line of work aims to improve LRU's hit  performance by making it more ``frequency-aware''. For example,  the \emph{$q$--LRU} policy works as follows: upon a miss  the content is cached with probability $q$, and then an LRU update is performed; a stationary performance analysis is provided  in \cite{garetto2016unified}, where it is shown that as $q\to 0$, the policy approaches LFU performance. In a different approach, the \emph{LRU-$k$} maintains $k$ connected layers of LRU caches. The content enters the cache at the first layer, and each subsequent hit moves the content to a deeper layer. The layers have the effect of ``filters'', such that only highly popular contents reach the deepest layers, this is shown to improve the ability of LRU to keep highly popular content \cite{Elayoubi2015}, but predictably makes the cache less reactive to popularity changes \cite{li2018accurate}. Martina et al. showed that  LRU-$k$ achieves the same performance as LFU for $k\to\infty$ \cite{garetto2016unified}, while similar results were previously established with a different approach \cite{o1999optimality}. Inspired by the above two, the \emph{ARC}  combines LRU and LFU by adaptively allocating parts of the cache to either recent or frequent content {\cite{megiddo2003arc}}. As a well-performing solution, it was later implemented in its ``CLOCK'' version, called \emph{CAR} \cite{bansal2004car}. Another related approach is to use a score for each content, and apply LRU only to high-score contents (the rest are never cached), cf.~\cite{hasslinger2017performance}. This general approach allows to combine LRU with any mechanism that discovers content popularity. 

Last, a seperate category of eviction policies is called \emph{Time To Live} (TTL), where the contents are evicted when a timer expires. These policies provide the extra capability to tune the hit rate of each individual content, allowing in this way to differentiate service between contents. An analysis of TTL  policies is presented  in section \ref{sec:ttl}. In summary, the properties of different online paging policies are collected in table \ref{tab:onlinepaging}.

\begin{table}[h!]
\begin{center}
\tiny
  \begin{tabular}{ | l | l | l | c  | l |}
    \hline
    \small\textbf{Name}         		& \small\textbf{Eviction rule}      						& \small\textbf{Guarantee}		& \small\textbf{Cplx}    		& \small\textbf{Ref.} \\ \hline \hline 
    Belady   			& max stack distance 					& max hits pathwise 			& N/A				& \cite{Belady66,Mattson70} 	\\ \hline\hline        
     LRU	   			& least recently requested 					& competitive 				& $\bigstar$			& \cite{Sleator85} 	\\ \hline        
    FIFO	   			& first-in				  					& competitive 				& $\bigstar$			& \cite{garetto2016unified} 	\\ \hline        
    WSCLOCK   			& linked list LRU	  						& competitive				& $\bigstar$			& \cite{wsclock} 	\\ \hline        
     {GreedyDualSize}		& least size-related credit	& competitive 			& $\bigstar\bigstar\bigstar$ & \cite{gds} 	\\ \hline        
    LANDLORD			& least size+utility credit		& competitive 		& $\bigstar\bigstar\bigstar$ & \cite{young2002line} 	\\ \hline\hline        
      LFU   			& least frequently requested	  				& opt. stationary				& $\bigstar\bigstar$		& --  	\\ \hline        
     $q$--LRU			& cache with prob $q$, evict LRU   				& opt. stationary $q\to 0$ 		& $\bigstar$			& \cite{garetto2016unified} 	\\ \hline        
        LRU-$k$			& $k$ LRU stages	& opt. stationary $k\to\infty$	& $\bigstar\bigstar$		& \cite{o1999optimality,garetto2016unified} 	\\ \hline        
         TTL	   			& upon timer expiration		  				& opt CUM stationary 					& $\bigstar$			& \cite{fofack2012analysis,C_Dehghan_16} 	\\ \hline\hline               
    RANDOM   			& random			  						& N/A 					& $\bigstar$			&\cite{garetto2016unified} 	\\ \hline         
     CLIMB				& hit moves up one position, evict tail  		& N/A					& $\bigstar$			& \cite{aven1987stochastic,gast2015transient, Starobinski01}	\\ \hline        
    partial LRU   			& least recently requested chunk	  				& N/A 					& $\bigstar\bigstar$		& \cite{wang2015optimal,maggi2018adapting} 	\\ \hline        

    { score-gated}		& threshold-based  	& N/A					& $\bigstar\bigstar$		& \cite{hasslinger2017performance} 	\\ \hline                  
        ARC				& split cache to LRU and LFU 	& 4--competitive					& $\bigstar\bigstar$		& \cite{megiddo2003arc,consuegra17} 	\\ \hline        
	CAR				& combine ARC and CLOCK			 	 	& 18--competitive					& $\bigstar\bigstar$		& \cite{bansal2004car,consuegra17} 	\\ \hline        
  \end{tabular}
\end{center}  
\caption{Online paging policies.}\label{tab:onlinepaging}
\end{table}

\section{Performance under stationary requests}\label{sec:stationary}

In this section, the request sequence is assumed to be stationary and independent. In particular, we will study the Poisson IRM model (see Definition  \ref{dfn:poissIRM}). {When the sequences have finite length we are interested in maximizing the number of hits, yet when $T\to\infty$ the criterion is the limit of the empirical hit probability, known also as the \emph{hit rate}.}
\begin{dfn}[theorem style=plain]{Hit rate}{}
Let $R_t$ denote the $t^{\text{th}}$ request, $\mathcal{M}^{\pi}_t$  the  cached contents under policy $\pi$, and $R_t\in \mathcal{M}^{\pi}_t$ the hit event. 
The hit rate  $\chi^{\pi}$ of policy $\pi$  is:
\[
\chi^{\pi}={\lim\inf}_{T\to\infty} \frac{\sum_{t=1}^{T}\indic{R_t\in \mathcal{M}^{\pi}_t}}{T}.
\]
\end{dfn}

Therefore, the standard caching problem when $T\to\infty$ is to  solve $\sup_{\pi\in\Pi}\chi^{\pi}$. We will see that the optimal online policy in this case  is LFU, and we will also analyze the performance of LRU, as well as the TTL policies that are tunable per each different content item.

\subsection{Least Frequently Used (LFU)}

As we explained in section \ref{sec:mostpop}, if distribution $(p_n)$ is static and known, then the optimal caching policy is the one that  caches the most popular contents. Here we will assume that the distribution is static but  unknown.

\begin{box_example}[detach title,colback=blue!5!white, before upper={\tcbtitle\quad}]{LFU eviction policy.}
	\small
Upon a miss, evict the least  frequently requested (used) content. Formally, let $f_n(t)$ denote the \emph{frequency} of content $n$ at time $t$; see definition \ref{dfn:freq}. Then upon a miss at $t$, evict item $\arg\min_{n} f_n(t)$ (ties solved randomly).

Since frequency $f_n(t)$ depends only on the past sequence $R_1,\dots, R_{t}$,  LFU is an online policy.

\end{box_example}

\begin{thm}[theorem style=plain]{LFU is optimal under IRM}{optLFU}
Suppose  $\boldsymbol R$ is Poisson IRM. LFU satisfies:
\[
\chi^{LFU} \stackrel{\text{w.p.1}}{=}\sup_{\pi\in\Pi}\chi^{\pi}\stackrel{\text{Pois}}{=}\max_{\boldsymbol y}h(\boldsymbol y).
\]
\end{thm}
\begin{proof}
Fix distribution $(p_n)$ unknown to the policy, and Poisson IRM requests drawn from $(p_n)$. Recall that the  empirical frequencies are given by $f_n(t)\triangleq{\sum_{i=1}^t\indic{R_i=n}}/{t}$. We first note that by the strong law of large numbers, the frequencies converge to the popularities $f_n(t)\to p_n$ w.p.1. It follows that the hit rate is maximized by the policy that caches the most popular contents, which implies
\[
\sup_{\pi\in\Pi}\chi^{\pi}\stackrel{\text{Pois}}{=}\max_{\boldsymbol y}h(\boldsymbol y),
\]
where $\max_{\boldsymbol y}h(\boldsymbol y)$ expresses the max hit probability of the corresponding static distribution.
In the remaining, we show that due to the convergence of frequencies, LFU eventually caches the most popular contents, maximizing the hit rate. 

Pick a  positive $\epsilon$ such that for any $n,m$ with $p_{n}\neq p_{m}$, we have $|p_{n}- p_{m}|>\epsilon$. From the convergence $f_n(t)\to p_n$ there is a finite time $T_{\epsilon}$, such that  we have $\max_n |f_n(T_{\epsilon}) - p_n|\leq \epsilon/2$ (w.p.1) and the following holds: 
\[
\text{if } n\in \mathcal M^{LFU}_t ~~\text{ then } p_n\geq p_m ~~\text{ for all } m\notin \mathcal M^{LFU}_t, t\geq T_{\epsilon}.
\]
The above implies that after convergence of sup norm of frequencies within $\epsilon/2$, all popular contents have been requested at least once.
Now, for all $t\geq T_{\epsilon}$, $\mathcal M^{LFU}_t$ contains the most popular items.  
Then we have:
\begin{align*}
\chi^{LFU} &= {\lim\inf}_{T\to\infty} \frac{\sum_{t=1}^{T}\indic{R_t\in \mathcal{M}^{\pi}_t}}{T} \geq  {\lim\inf}_{T\to\infty} \frac{\sum_{t=T_{\epsilon}}^{T}\indic{R_t\in \mathcal{M}^{\pi}_t}}{T}\\
& \stackrel{\text{SLLN}}{=}\mathbbm{E}_{\infty}\left[\indic{R_t\in \mathcal{M}^{\pi}_t}\right]=\max_{\boldsymbol y}h(\boldsymbol y), ~~\text{w.p.1},
\end{align*}
where in the  inequality we dropped  ${\lim\inf}_{T\to\infty} {\sum_{t=1}^{T_{\epsilon}}\indic{R_t\in \mathcal{M}^{\pi}_t}}/T=0$, and the expectation is taken with respect to the invariant measure (of the converged LFU cache). 
\end{proof}

From the proof we understand that LFU's good performance relies on the convergence of the empirical frequencies  to popularities, which requires the popularities to be static. When popularities are time-varying, LFU can perform quite poorly; for instance a content that becomes suddenly unpopular, may stay in an LFU cache for a long time. Another disadvantage of LFU is that we need to maintain the frequency of all contents in the catalog, which creates a large memory overhead.

\subsection{Least Recently Used (LRU)}

We next study the performance of LRU under Poisson IRM. 
 We will discuss an exact Markovian model, which unfortunately results in cumbersome computations. To facilitate {the analysis} 
we will subsequently  present the seminal Che's approximation.

\subsubsection{Exact analysis via Markovian model}

Under IRM, the  LRU cache can be described by  a Discrete Time Markov Chain (DTMC) such that every new request is associated to a transition in this chain. The resulting chain is irreducible and aperiodic, and thus if we can 
compute its stationary distribution we can then use it to calculate the hit probability.

Let $\sigma_t(.)$ be a mapping of cache positions $\{1,\dots,M\}$ to content indices, such that $\sigma_t(i)\in \{1,\dots,N\}$ denotes the content that is cached at location $i$ of the cache at slot $t$ (for example $\sigma_t(M)$ is the next content to be evicted under LRU), and let $\Sigma$ be the set of all such mappings. The number of these mappings equals the different ways we can craft permutations of length $M$ out of $N$ elements, i.e., $2^M \binom NM$. Note that every mapping corresponds to a precise order of recency. Then, let $J(\sigma_t)$ be the set of cached contents under mapping $\sigma_t$ and also the state of the Markov chain--further note that two different mappings $\sigma^1\neq \sigma^2$ may satisfy $J(\sigma^1)=J(\sigma^2)$ since the same contents might occupy different spots in the cache. King \cite{C_King_71} derived the stationary probability of set $J$ as:
\begin{equation}\label{eq:statLRU}
\pi_J=\sum_{\sigma:J(\sigma)=J}\frac{p_{\sigma(1)}p_{\sigma(2)} \dots p_{\sigma(M)}}{(1-p_{\sigma(1)})\times\dots
\times (1-p_{\sigma(1)}-p_{\sigma(2)}-\dots-p_{\sigma(M)})}, ~~\forall J,
\end{equation}
where $p_n$ denotes  the  popularity of content $n$ as before. After $\pi_J$ is computed, then the stationary hit probability of LRU cache is simply determined by:
\[
h^{LRU}(\infty)=\sum_J \pi_J\sum_{n\in J} p_n .
\] 

These analytical formulas are computationally expensive because they involve enumeration of all  ordered subsets of content with length $M$, which are exponential to $N$. Flajolet characterized the stationary probabilities $\pi_J$ with  an alternative  integral formula \cite{Flajolet92}, which was recently shown  to be exactly equivalent to \eqref{eq:statLRU} in terms of computations \cite{R_Berthet_16}. Therefore, computing exactly the hit probability of LRU under IRM for large caches is prohibitively complex. Several approximations of this probability exist in the literature, cf. \cite{C_Dan_90}, but the most commonly used  is the one attributed to Che et al \cite{J_Che_02}, which we present next.

\subsubsection{Che's approximation}

{We present a simple approximation for}
the characterization of the LRU cache performance under IRM. Although the  approximation  takes its name from \cite{J_Che_02}, we mention here that the idea  partially appears {in a prior work} in the context of computer cache  \cite{FAGIN1977}. Che's approximation not only is the de facto approach to analyzing LRU, and it has  been extended to many other setups. For example, \cite{garetto2016unified} extends it to non-stationary popularity and to the study of other policies (RANDOM, FIFO, LRU-$k$, $q$--LRU); \cite{maggi2018adapting} extends it to caching chunks of contents; and \cite{C_Jung_TTL} to the study of TTL caches.
\vspace{-0.1in}
\begin{dfn}[theorem style=plain]{Characteristic time of content $n$}{}
Consider a cache of size $M$. The characteristic time of content $n$, denoted by $t_n(M)$, is the elapsed time until  $M$ different contents other than $n$ are requested. 
Formally, letting $\delta_j$ denote the (random)  time until the next arrival of content $j$:
\[
t_n(M)=\inf\left\{t:\sum_{j\neq n}\indic{\delta_j < t}>M\right\}.
\]
\end{dfn}\vspace{-0.1in}
\noindent If no request for $n$ occurs in the interval $[t,t+t_n(M)]$, then content $n$ will be evicted by LRU. In this case, the cache is replenished; the state of the cache $\mathcal{M}_{t+t_n(M)}$ is independent of the state of the cache $\mathcal{M}_t$, and depends only on the arrivals that occured within the interval $[t,t+t_n(M)]$.  

When the request sequence is random, the characteristic time is a random variable, whose distribution is complicated to derive. However, its analysis can be greatly simplified if we assume that $t_n(M)$ is strongly concentrated around its mean:  
\begin{lem}[theorem style=plain]{Che's approximation}{che}
Consider an LRU cache of size $M$, fed with a Poisson IRM sequence with rate $\lambda$ and popularity $(p_n)$. The characteristic time of content $n$ can be approximated by $t_n(M)\approx \hat{t}(M)$, where $\hat{t}(M)$ is the unique solution to the following equation:
\[
M=\sum_{n=1}^N (1-e^{-\lambda p_n \hat{t}(M)}).
\]
\end{lem}
\begin{proof}[Proof of Lemma \ref{lem:che}]
Let $\delta_n$ denote the inter-request time for content $n$, which by the premise of Poisson arrivals is assumed to be exponentially distributed with mean $1/\lambda_n=1/(\lambda p_n)$. 
By definition, there are exactly $M$ contents other than $n$ that will arrive before the characteristic time has elapsed. This leads to the following identity:
\[
M=\sum_{j\neq n} \mathbbm{1}_{\{\delta_j < t_n(M)\}}.
\]
Let us assume that $t_n(M)$ is \emph{constant} (first part of the approximation), and take expectation with respect to $\delta_j$:
\[
M=\sum_{j\neq n}\left(1-e^{-\lambda p_jt_n(M)}\right).
\]
Then, let us further assume that $t_1(M)=\dots=t_N(M)=\hat{t}(M)$ and drop the $j\neq n$ in the sum (second part of the approximation), it follows $M=\sum_{n=1}^N\left(1-e^{-\lambda p_n\hat{t}(M)}\right)$. Note that the RHS is strictly concave, hence the equation has a unique solution.
\end{proof}

Very recently \cite{Fricker12} showed that under a Zipf-like popularity, the coefficient of variation of $t_n(M)$ vanishes as the cache size grows, hence as $M\to\infty$ the approximation becomes exact. 
Although $t_n(M)$ is content-dependent and $\hat{t}(M)$ is not,  \cite{Fricker12} also showed that the dependence of $t_n(M)$ on $n$ becomes negligible when the catalog size $N$ is sufficiently large.

Next, we will use Che's approximation to derive the hit probability of LRU. 
We begin with the hit probability for content $n$ conditioned on the characteristic time, denoted with $h_n^{LRU}(M | t_n(M))$.
Using the definition of the characteristic time, we write:
\begin{equation}\label{eq:hitprobn}
h_n^{LRU}(M | t_n(M))=\prob{\delta_n<t_n(M)}=1-e^{-\lambda p_nt_n(M)},
\end{equation}
where the first equality follows because LRU scores a hit if the next request arrives before the characteristic time has elapsed, and the second equality follows from the fact that $\delta_n$ is exponentially distributed in the Poisson IRM. Although $h_n^{LRU}(M | t_n(M))$ is in principle a random variable, using Che's approximation for ${t}_n(M)$, we can approximate the hit probability of LRU cache as:
\begin{align}\label{eq:hitLRU}
h^{LRU}(M)&=\sum_{n=1}^N p_n h_n^{LRU}(M | t_n(M))\stackrel{\eqref{eq:hitprobn}}{=}\sum_{n=1}^N p_n(1-e^{-\lambda p_nt_n(M)})\notag\\&\approx \sum_{n=1}^N p_n(1-e^{-\lambda p_n\hat{t}(M)}).
\end{align}
Fig.~\ref{fig:che} compares the hit probability of LRU cache as computed by \eqref{eq:hitLRU} versus simulations under an IRM model with $\Zipf=0.8$. The plot provides evidence of the approximation quality for different values of $\Zipf$, and $\gamma=0.1$. 

\begin{figure*}[t!]
	\centering
		\includegraphics[width=0.475\linewidth]{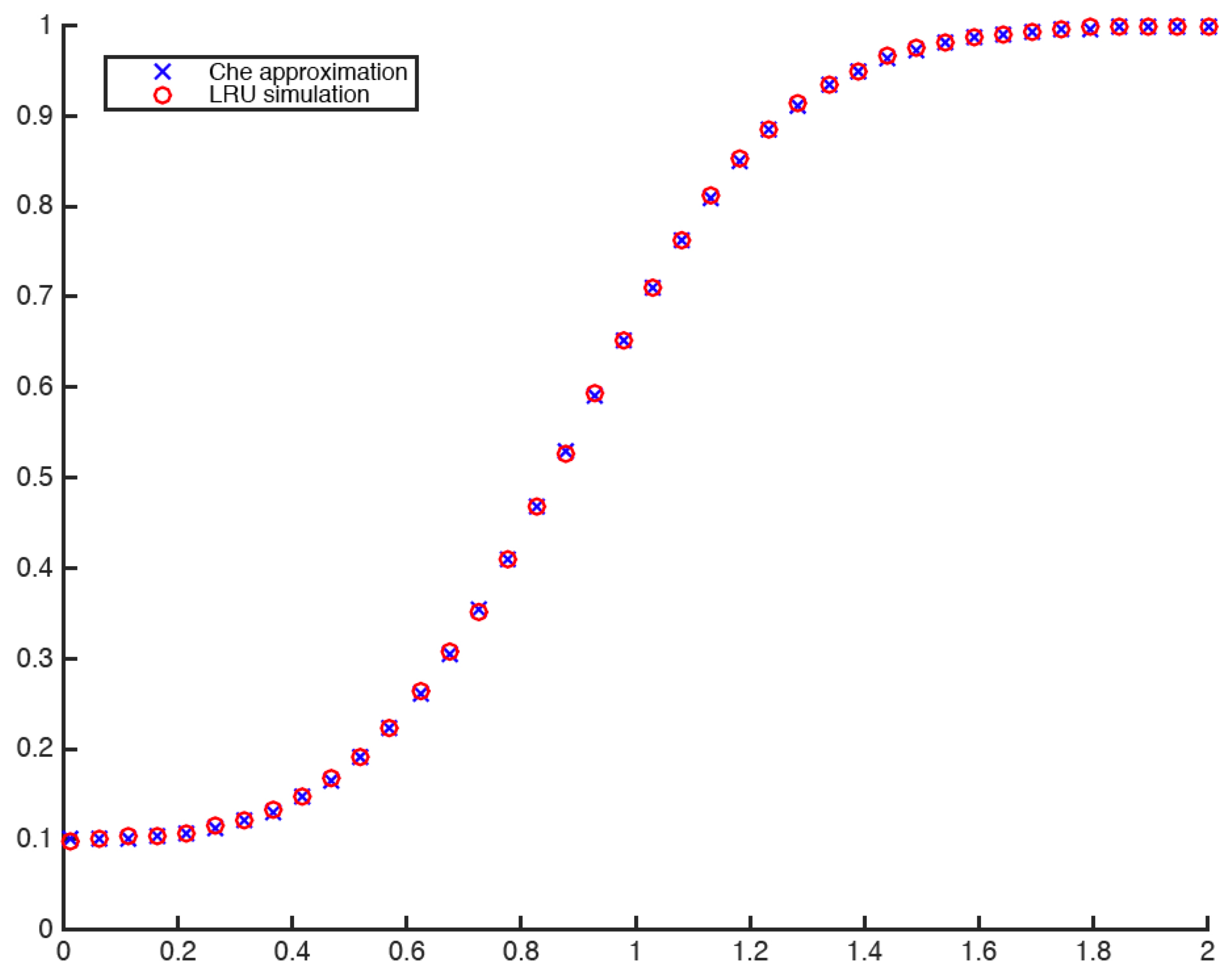}
	\caption{Simulated LRU hit probability under IRM versus Che's approximation for different $\Zipf$. Parameters: $\gamma=0.1$, $N=10000$.}
	\label{fig:che}
\end{figure*}

\subsection{TTL caches}\label{sec:ttl}


We move now to the analysis of caches with the ``Time To Live'' (TTL) attribute. When a content is cached in a TTL cache, a timer is set. Upon expiration of the timer, the content is evicted from the cache. TTL caches are extensively used in the Internet because of their usefulness in caching content that becomes obsolete. For instance, the content of the \url{www.cnn.com} web page may change rapidly within the day, and one can use the TTL attribute to keep information in the cache fresh. However, the TTL caches also allow us to surgically configure the hit probability of each individual content and balance hit probabilities of different contents in any desirable way. This is in contrast to  LRU and LFU where the relative hit probability of each individual content is automatically determined  by the popularity of the content. 

\subsubsection{How TTL caches work}

Let us denote with $T_n$ the timer associated with  content $n$. In all TTL caches  timer $T_n$ starts counting after a miss for content $n$, but there are two different TTL cache models for how we react to  a hit:
\begin{enumerate} 
\item In \emph{reset} TTL, the timer is reset after a hit for content $n$.
\item In \emph{non-reset} TTL, the timer is \emph{not} reset after a hit.
\end{enumerate}
See a graphical example in figure~\ref{fig:ttl}.

\begin{figure}[h!]
	\centering
	\includegraphics[width=0.95\textwidth]{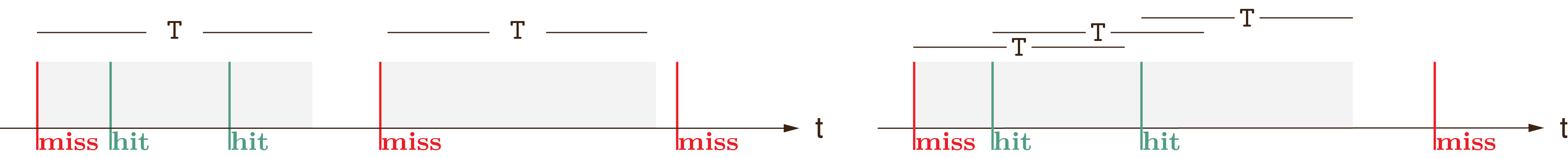}
	\caption{Non-reset (left) and Reset (right) TTL cache.}\label{fig:ttl}
\end{figure}

Assuming a Poisson IRM request model, we would like to derive the hit probability of content $n$ when we choose a timer $T_n$. At time $t$ the stationary hit probability is denoted with $h_n(\infty,T_n)$. Next, we derive  $h_n(\infty,T_n)$ for both TTL models. For the non-reset model, we use the elementary renewal theorem \cite{C_Jung_TTL} to obtain the limiting hit probability of a TTL cache as
\[
h^{non}_n(\infty, T_n)=\frac{\mathbbm{E}[N_n(T_n)]}{1+\mathbbm{E}[N_n(T_n)]},
\]
where $N_n(T_n)$ counts the number of requests for content $n$ from 0 up to the timer expiration, and  $\mathbbm{E}[N_n(T_n)]$ is its expectation. 
Since $N_n(t)$ is taken to be a Poisson process with rate $\lambda_n=\lambda p_n$, we have $\mathbbm{E}[N_n(T_n)]=\lambda_n T_n =\lambda p_n T_n$, and finally
\[
h^{non}_n(\infty, T_n)=\frac{\lambda p_n T_n}{1+\lambda p_n T_n}.
\]

\noindent For the reset model, we observe that $T_n$ plays the role of a (deterministic) characteristic time, and using Che's approximation we  obtain:
\[
h^{res}_n(\infty, T_n)= 1- e^{-\lambda p_n T_n}.
\]
In the next section, we ask the question how to select $(T_n)$  in order to induce a desirable per-content hit probability. We will use the notation $h_n^{\text{\tiny TTL}}(\infty, T_n)$ to simultaneously refer to both models.

\subsubsection{Cache Utility Maximization}

The task of designing a TTL caching policy consists of setting the timers of all  {contents ${\bm T}=(T_1,\dots,T_N)$}
so that (i) the number of contents with non-expired timers are equal to the cache size to avoid cache overflow and (ii) the hit probabilities of the contents are balanced in a desirable way. For this purpose we introduce the following optimization problem, called Cache Utility Maximization (CUM). The optimization variables are the content timers ${\bm T}$. To capture the cache overflow constraint we should ensure at every instance that: 
\[
\sum_n \indic{n \text{ is cached}} \leq M.
\]
However, to satisfy this demanding constraint we must fine-tune the placement of contents in the cache. Alternatively, an easier approach is to  satisfy the constraint on average  by  taking expectations in both sides and obtaining the relaxed probabilistic constraint:
\[
\sum_{n=1}^{N} h_n^{\text{\tiny TTL}}(\infty, T_n)=M.
\]
{Since this constraint is ensured on average,  the cache will  overflow when  many contents with large timers appear together. However, as \cite{C_Dehghan_16} shows, this happens rarely and can be avoided either by evicting a random content, or by under-subscribing the cache by a small fraction.}

Our objective is to optimize the individual  hit probabilities $h_n^{\text{\tiny TTL}}(\infty, T_n)$ for all $n$. This can be achieved using a concave, differentiable, and separable utility function $U(\boldsymbol T)=\sum_n g_n(h_n^{\text{\tiny TTL}}(\infty, T_n))$, where $g_n(.)$ 
{maps the limiting} hit probability of content $n$ to the utility obtained by it.

\begin{opt}{Cache Utility Maximization (CUM)}\vspace{-0.15in}
\begin{align}
& \max_{\boldsymbol T} U(\boldsymbol T) \label{eq:cum}\\
\text{s.t. } & \sum_{n=1}^{N} h_n^{\text{\tiny TTL}}(\infty, T_n)=M. \notag
\end{align}\vspace{-0.08in}
\end{opt}

If we select the objective $U(\boldsymbol T)$ to be a concave function of $\boldsymbol h=(h_n^{\text{\tiny TTL}}(\infty, T_n))$, then the CUM problem can be expressed in terms of $\boldsymbol h$ as a constrained convex optimization problem, known as \emph{resource allocation  with simplex constraints}, a relatively simple problem to solve. Before presenting a scalable distributed solution, we first discuss how to choose different utility functions. A powerful model for convex resource allocation problems is the family of $\alpha$-fair functions, defined by choosing $g(\cdot)$ to be a  polynomial parametrized with $\alpha$:
\[
g_n(x)=\frac{x^{1-\alpha}}{1-\alpha},~~\alpha \in (0,1)\cup(1,\infty).
\] 
Therefore, our $\alpha$-fair objective becomes
\[
U(\boldsymbol T)= \sum_n \frac{(h_n^{\text{\tiny{TTL}}}(\infty, T_n))^{1-\alpha}}{1-\alpha}.
\]
Depending on the choice of $\alpha$, the optimal solution vector of \eqref{eq:cum}, denoted with $\boldsymbol h^*$, has  special fairness properties in the set of feasible hit probabilities $\mathcal{H}=\{\boldsymbol h | \sum_{n=1}^{N} h_n^{\text{\tiny{TTL}}}(\infty, T_n)=M, \boldsymbol T\in \R_+^N\}$. Past work \cite{J_Mo_00} shows that:
(i) for $\alpha\to 0$, $\boldsymbol h^*$ maximizes the sum of hit probabilities,
(ii) for $\alpha\to 1$, $\boldsymbol h^*$ is proportionally fair,
(iii) for $\alpha=2$, $\boldsymbol h^*$ is potential delay fair, and
(iv) for $\alpha\to \infty$, $\boldsymbol h^*$ is max-min fair. 

Hence, the solution of \eqref{eq:cum} allows us to balance the individual hit probabilities in different ways. It remains to explain how to tune $\boldsymbol{T_n}$ to achieve a targeted solution $\boldsymbol h^*$. We remark that the convex program \eqref{eq:cum} can be solved in many different ways, aiming for distributization, runtime, or robustness criteria \cite{Ber99book}. Since in practice we encounter \eqref{eq:cum} in the form of a large constrained convex program, \cite{C_Dehghan_16} proposed a dual subgradient algorithm that adjusts the individual hit probabilities according to subgradients of $U$. Due to $U$ being separable, its subgradient at $\boldsymbol h$ can be denoted as $(g_1'(h_1), \dots, g_N'(h_N))$, where $g_n'(h_n)$ is the local directional derivative with respect to content $n$.
\begin{box_example}[detach title,colback=blue!5!white, before upper={\tcbtitle\quad}]{Dual subgradient TTL caching.}
	\footnotesize
Set the timer of content $n$ to:
\begin{align*}
& T_{n,t}=-\frac{1}{\lambda_n}\log\left(1-g_n'^{-1}(\mu_t)\right)~~~\text{(for reset)},\\
& T_{n,t}=-\frac{1}{\lambda_n}\log\left(1-1/[1-g_n'^{-1}(\mu_t)]\right)~~~\text{(for non-reset)},
\end{align*}
where $\mu_t$ is the Lagrangian multiplier at time $t$, and update the multiplier as:
\[
\mu_{t+1}=\left[\mu_t+\eta_t\left(\sum_{n=1}^{N} h_n^{\text{\tiny TTL}}(\infty, T_{n,t})-M\right)\right]^+,
\]
where $\eta_t$ is a stepsize parameter, that may affect the convergence rate of the algorithm. Often it is set to a fixed value, or set to a diminishing value $\propto 1/\sqrt{t}$. 
\end{box_example}

\section{Online popularity learning}\label{sec:online}

Comparing the optimal classification rule under SNM of Theorem~\ref{thm:snm}, the maximum hit rate of LFU under IRM, as well as other optimality results in the literature, we see that  to decide which policy  to use, one must know the underlying request model. However,  in practice the request model is \emph{a priori} unknown and time-varying. \emph{This renders imperative the design of a {universal} caching policy that will work provably well for all request models.} In fact, this was the inspiration behind  the online paging problem. However, we saw that the competitive ratio metric admits many optimizers, such as the  LRU policy, but also the commonly avoided FIFO policy. Although an eviction policy is shown to have optimal worst-case hit rate, the actual hit rate performance on a dataset  might be poor.  

In this section we return to the arbitrary requests, but we change our  perspective. We study a model-free caching model, along the lines of the Online Linear Optimization framework used extensively in machine learning literature to design robust algorithms \cite{Shalev12}. Specifically, we assume that content requests are drawn from a general distribution, which is equivalent to caching  versus an adversary who chooses the requests to harm our hit performance. At each slot,  {(i)} the caching policy decides which content parts to cache; {(ii)} the adversary introduces the new  request; and {(iii)} a content-specific utility is obtained proportional to the fraction of the request that was served by the cache. This generalizes the criterion of cache hit ratio and allows one to build policies that, for instance, optimize delay performance or provide preferential treatment to contents.

\begin{figure}[!h]
	\centering
	\includegraphics[width=5in]{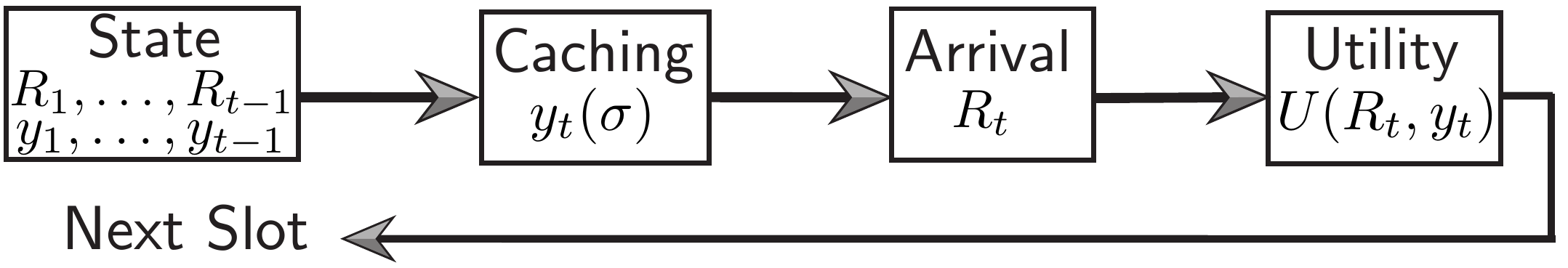} 
	\caption{Online caching model (slot $t$).}
	\label{fig:model}
\end{figure}  
In this setting, we seek to find a  caching policy with minimum \emph{regret}, {i.e., minimum utility loss over an horizon $T$, when compared to the best cache configuration on a given request sample path}.

\subsection{Regret in caching}\label{sec:model}
As usual, we consider a library $\mathcal{N}=\{1,2,\dots, N\}$ of equal size contents and a cache that fits $M<N$ of them. The system evolves in time slots;  in slot $t$ a single request is made for content $n\in\mathcal{N}$, denoted with the event $R_{t}^{n}\!=\!1$.\footnote{Notice the change in notation. While before $R_t\in\mathcal{N}$, now we use $R_t^n\in\{0,1\}$ to denote if the content $n$ was requested in slot $t$. Hence, while before $R_t$ would be a scalar id, now $R_t$ is a ``hot vector'', with the only non-zero element indicating the id.} The  vector $R_t=(R_{t}^n, n\in\mathcal{N})$ represents the $t$-slot request, chosen from the set of feasible request vectors:
\[
\mathcal{X}=\left\{R\in \{0,1\}^N ~\bigg |~ \sum_{n=1}^N R^n=1\right\}.
\]
The instantaneous content popularity is determined by the probability distribution $\prob{R_t}$ (with support $\Xc$), which is allowed to be unknown and arbitrary, and the same holds for the joint distribution $\prob{R_1,\dots,R_T}$ that describes the content popularity evolution. This generic model captures all possible request sequences, including stationary (i.i.d. or otherwise), non-stationary, and adversarial models. 

The cache is managed with the caching configuration variable $y_t\in [0,1]^N$, where $y_t^n$ denotes the fraction of content $n$ that is cached in slot $t$.\footnote{Why caching of content fractions $y_t^n\in[0,1]$ makes sense?
Large video contents are typically composed of thousands of chunks,  stored independently, as it  is well explained in the literature of \emph{partial caching}, cf.~\cite{maggi2018adapting}.
The fractional variables may also represent probabilities of caching a content (or a chunk) \cite{Shalev12,Blaszczyszyn2014Geographic}, or coded equations of chunks \cite{golrezaei2012femtocaching}. 
For practical systems, the fractional output $y_t^n$ of our schemes should be rounded to the closest element in this finer granularity, which will induce a small application-specific error.}
 Taking into account the size of the cache $M$, 
the set $\mathcal{Y}$ of admissible caching  configurations is:
\begin{equation}\label{eq:capped}
\mathcal{Y}=\left\{y\in [0,1]^N ~\bigg |~ \sum_{n=1}^Ny^n\leq M\right\}.
\end{equation}
  
\begin{dfn}[theorem style=plain]{Caching Policy}{}
A caching policy $\sigma$ is a (possibly randomized)  
rule that at slot $t=1,\dots,T$ maps past observations $R_1,\dots,R_{t-1}$ and configurations $y_1,\dots,y_{t-1}$ to the configuration $y_t^{\sigma}\in \mathcal{Y}$ of slot $t$.
\end{dfn}

We denote with $w^n$ the  utility obtained  when  content $n$ is requested and found in the cache (when we have a hit). This content-dependent utility can be used to model bandwidth economization from cache hits \cite{maggi2018adapting}, QoS improvement \cite{golrezaei2012femtocaching}, or any other cache-related benefit. We will also be concerned with the special case $w^n=w, n\in\mathcal{N}$, i.e., the \emph{hit probability} maximization. A cache configuration $y_t$ accrues in slot $t$ a utility $U\big(R_t,y_t\big)$  determined as follows:
\[
U\big(R_t,y_t\big)=\sum_{n=1}^N w^n R_t^ny^n_t.
\]

Consider now that an adversary selects the utility function of the system $U_t(y)$, by means of choosing $R_t$, i.e., $U_t(y)\equiv U(R_t,y)$. Differently from the competitive ratio approach of \cite{Sleator85} however, we introduce a new metric that compares how our caching policy fares against the \emph{best static policy in hindsight}. This metric is often used in the literature of machine learning \cite{Shalev12,Mert18} with the name {worst-case static regret}. In particular, we define the \emph{regret of policy $\sigma$} as:
\begin{align*}
\texttt{Reg}_T(\sigma)&=\max_{\prob{R_1,\dots,R_T}}
\mean{\sum_{t=1}^T U_t\big(y^*\big)-\sum_{t=1}^T U_t\big(y_t(\sigma)\big)} 
\end{align*}
where $T$ is the horizon, the maximization is over the admissible adversary distributions, the expectation is taken with respect to the possibly randomized $R_t$ and $y_t(\sigma)$, and $y^*\in\arg\max_{y\in \Yc}\sum_{t=1}^T U_t(y)$ is the best configuration in hindsight, i.e., a benchmark policy that knows the sample path $R_1,\dots,R_T$ but is limited to a static configuration $y^*$. Intuitively, measuring the utility loss of $\sigma$ w.r.t a static $y^*$  constrains the power of the adversary; for example  a rapidly changing adversary will challenge $\sigma$ but also incur huge losses in $y^*$. This benchmark comparison makes regret different from  the standard competitive hit ratio of \cite{Sleator85}. 
 
 We seek to find a caching policy that minimizes the regret by solving the problem $\inf_{\sigma}\texttt{Reg}_T(\sigma)$. 
\begin{opt}{ Caching for Regret Minimization}
Find an online caching policy $\sigma^*$ that attains:
\begin{equation}\label{eq:minregret}
\inf_{\sigma}\texttt{Reg}_T(\sigma)
\end{equation}
where $\texttt{Reg}_T(\sigma)=\max_{\prob{R_1,\dots,R_T}}\mean{\sum_{t=1}^T U_t\big(y^*\big)-\sum_{t=1}^T U_t\big(y_t(\sigma)\big)}$
\end{opt}

This problem is known as \emph{Online Linear Optimization} (OLO) \cite{Shalev12}. 
The analysis in OLO aims to discover how the regret scales with the horizon $T$. A caching policy with sublinear regret $o(T)$ produces   average losses $\texttt{Reg}_T(\sigma)/T\!\to\! 0$ with respect to the benchmark, hence it  learns to adapt to the best cache configuration without any knowledge about the request distribution (smaller regret  means faster adaptation). We emphasize that our problem is further compounded due to the cache size dimension. Namely, apart from optimizing the regret with respect to $T$, it is of  huge practical importance to consider the dependence on $N$ (or $M$).  
\begin{lem}[theorem style=plain]{}{}
The regret of \textup{LRU}, \textup{LFU}  satisfies:
\begin{align*}
&\texttt{Reg}_T(\textup{LRU})\geq wM\left(\frac{T}{M+1}-1\right)\,,\\
&\texttt{Reg}_T(\textup{LFU})\geq wM\left(\frac{T}{M+1}-1\right)\,.
\end{align*}
\end{lem}

\begin{proof}  
Suppose the adversary chooses the following periodic request sequence: $\{1,2,\dots,M+1,1,2,\dots\}$. For any $t>M$, since  the requested content is the $M+1$ least recent content, it is not in the LRU cache, and no utility is received. Hence overall, LRU can achieve at most $wM$ utility--from the first $M$ slots. However, a static policy with hindsight  achieves at least $wTM/(M+1)$ by caching the first $M$ contents. The same proof holds for LFU by noticing that due to the form of the periodic arrivals, the least frequent content is also the least recent content.
\end{proof}
The poor $\Omega(T)$ performance of standard caching policies is  not surprising. These policies are designed to perform well only under specific request models, i.e., LRU for requests with  temporal locality and LFU for stationary requests. On the contrary,  they are known to under-perform in other models: LRU in stationary, and LFU in decreasing popularity request patterns. Poor regret performance means that there exist request distributions for which the  policy fails completely to ``learn'' what is the best learnable configuration. \emph{Remarkably, we will show below that there exist universal caching policies that can provide low regret under any request model.}

\subsection{Regret lower bound}\label{sec:regret-bound}
A regret lower bound is a powerful theoretical result that provides the fundamental limit of how fast any algorithm can learn to cache, much like the upper bound of the channel capacity. Regret lower bounds in OLO have been previously derived for different action sets: for  $N$-dimensional unit ball centered at the origin in \cite{Abernethy08} and  $N$-dimensional hypercube in \cite{Hazan06}, both resulting in $\Omega(\sqrt{\log N\cdot T})$. In our case, however, the above results do not apply since  $\Yc$  in \eqref{eq:capped} is a capped simplex, i.e., the intersection of a box and a simplex inequality. 
Due to the specific form of our constraint set, we expect that a tighter lower bound may exist.
Indeed, below we provide such a lower bound  tailored to the online caching problem.

\begin{thm}[theorem style=plain]{General Caching Regret Lower Bound}{sumz}
The regret of any  online caching policy $\sigma$   satisfies: 
	\[
\texttt{Reg}_T(\sigma) > \sum_{i=1}^M\mean{Z_{(i)}}\sqrt{T},\quad \text{ as } T\to\infty,
	\]
where $Z_{(i)}$ is the $i-$th maximum element of a gaussian random vector with zero mean and covariance matrix $\Sigmam(w)$ given by \eqref{eq:covarianceMatrix}.  
	Furthermore, assume $M<N/2$ and define $\phi$ any permutation of $\mathcal{N}$ and $\Phi$ the set of all such permutations:
	\[
\texttt{Reg}_T(\sigma)>\frac{\max_{\phi\in\Phi}\sum_{k=1}^M\sqrt{w^{\phi(2(k-1) +1)}+ w^{\phi(2k)}}}{\sqrt{2\pi\sum_{n=1}^N1/w^n}}\sqrt{T}.	
	\] 
\end{thm}

\begin{cor}[theorem style=plain]{Hit Regret Lower Bound }{lowerBoundRegret}
Fix a $\gamma\triangleq M/N$, $\gamma\in(0,1/2)$ and  $w^n=w,~\forall n$ (which corresponds to the case of adversarial hit maximization). Then, the regret of any caching policy $\sigma$  satisfies: 
	\[
		\texttt{Reg}_T(\sigma) > w\sqrt{\frac{\gamma}{\sigma}}\sqrt{MT}, \quad \text{ as } T\to\infty.
	\]
\end{cor}

Before providing the proof, we remark that our bound is tighter than  the classical $\Omega(\sqrt{\log N\cdot T})$ of OLO in the literature  \cite{Hazan06,Abernethy08}, which is attributed to the difference of  sets $\Xc,\Yc$. In the next section we will provide a caching policy that achieves regret  $O(\sqrt{MT})$, which will establish that the regret of online caching is in fact $\Theta(\sqrt{MT})$.

\begin{proof}[Proof of Theorem \ref{thm:sumz} and Corollary \ref{cor:lowerBoundRegret}]
To find a lower bound, we will analyze a specific adversary $R_t$. In particular, we will  consider an i.i.d.  $R_t$ such that content $n$ is requested with probability 
\[
\prob{R_t=\ev_n}=\frac{1/w^n}{\sum_{i=1}^N1/w^i}, ~~\forall n,t,
\] 
where $\ev_n$ is a vector with element $n$ equal to one and the rest zero. With such a choice of $R_t$,	any causal caching policy 	yields an expected utility  at most $MT/\sum_{n=1}^N(1/w^n)$, since 
\begin{align}
\mean{\sum_{t=1}^T U_t(y_t(\sigma))}
&=\sum_{t=1}^T\sum_{n=1}^Nw^n\prob{R_t=\ev_n}y_t^n(\sigma)\notag\\
&\hspace{-0.2in}=\sum_{t=1}^T\frac{1}{\sum_n 1/w^n}\sum_{n=1}^Ny_t^n(\sigma) 
\leq\frac{MT}{\sum_n 1/w^n}, \label{eq:util}
\end{align} 
independently of $\sigma$. 
To obtain a regret lower bound
we show that a static policy with hindsight can exploit the knowledge of the sample path $R_1,\dots,R_T$ to  achieve a higher utility than \eqref{eq:util}. Specifically, defining $\nu^n_t$ the number of times content $n$ is requested in slots $1,\dots,t$, the best static policy will cache the $M$ contents with highest products $w^n\nu^n_T$. In the following, we characterize how this compares against the average utility of \eqref{eq:util} by analyzing the order statistics of a Gaussian vector. We have the following technical lemma:


{
For i.i.d. $R_t$ we may rewrite regret as the expected difference between the best static policy in hindsight and  \eqref{eq:util}: 
\begin{equation}\label{eq:regretIID}
\texttt{Reg}_T =\mean{ \max_{y\in\Yc}y^T\sum_{t=1}^T w \odot R_t }- \frac{MT}{\sum_n 1/w^n}
,\end{equation}
where $w \odot R_t=[w^1R^1_t, w^2R^2_t,...,w^NR^N_t]^T$ is the Hadamard product between the weights and request vector. Further, \eqref{eq:regretIID} can be rewritten as a function:
\[
\texttt{Reg}_T=\mean{g_{N,M}(\overline{z}_T)} = \mean{\max_{b\in \stackrel{\circ}{\Yc}}\left[b^T\overline{z}_T\right]},
\]
where, \emph{(i)} $\stackrel{\circ}{\Yc}$ is  the set of all possible integer caching configurations (and therefore $g_{N,M}(.)$ is the sum of the maximum $M$ elements of its argument); and \emph{(ii)} the process $\overline{z}_T$ is  the vector of utility obtained by each content after the first $T$ rounds, centered around its mean:
\begin{align} \nonumber
\overline{z}_T & = \sum_{t=1}^{T}w \odot R_t - w \odot\frac{T}{\sum_{n=1}^N1/w^n}w^{-1} \\ \label{eq:centeredDemandVector}
& = \sum_{t=1}^{T}\left(z_t - \frac{1}{\sum_{n=1}^N1/w^n}\mathbf{1}_N\right) \end{align}
where $z_t = w \odot R_t$ are i.i.d. random vectors,  with distribution 
\[ \mathbb{P}\left(z_t = w^i \ev_i\right) = \frac{1/w^i}{\sum_{n=1}^N1/w^n}, \forall t, \forall i,
\]
and, therefore, mean $\mean{z_t}=\frac{1}{\sum_{n=1}^N1/w^n}\mathbf{1}_N$.\footnote{Above we have used the notation $w^{-1} = \left[1/w^1, 1/w^2,...,1/w^N\right]^T.$} 
}
The A main ingredient of our proof is the following limiting behavior of $g_{N,M}(\overline{z}_T)$:
\begin{lem}[theorem style=plain]{}{regretConvDistribution}
	Let ${Z}$ be a Gaussian vector  $\Nc\left(\mathbf{0}, {\Sigmam}(w)\right)$, where ${\Sigmam}(w)$ is given in \eqref{eq:covarianceMatrix}, and ${Z}_{(i)}$ its $i-$th largest element. Then
	\[
\frac{g_{N,M}(\overline{z}_T)}{\sqrt{T}} \xrightarrow[T\rightarrow\infty]{\text{distr.}} \sum_{i=1}^M{Z}_{(i)}
	.\]
\end{lem}  
\begin{proof}[Proof of Lemma \ref{lem:regretConvDistribution}]
Observe that $\overline{z}_T$ is the sum of $T$ uniform i.i.d. zero-mean random vectors, where the covariance matrix  can be calculated using \eqref{eq:centeredDemandVector}: ${\Sigmam}(w) =$
\begin{align}\nonumber
&=\mean{\left(z_1 - \frac{1}{\sum_{n=1}^N1/w^n}\mathbf{1}_N\right)\left(z_1 - \frac{1}{\sum_{n=1}^N1/w^n}\mathbf{1}_N\right)^T}\\ \label{eq:covarianceMatrix}
&= \frac{1}{\sum_{n=1}^N1/w^n}\begin{cases}
w_i -\frac{1}{\sum_{n=1}^N1/w^n}, i=j\\
-\frac{1}{\sum_{n=1}^N1/w^n}, i\neq j
\end{cases}
,\end{align}
where the second equality follows from the distribution of  $z_t$ and some calculations.\footnote{For the benefit of the reader, we note that $Z$ has no well-defined density (since ${\Sigmam}(w)$ is singular). For the proof, we only use its distribution.}
Due to the Central Limit Theorem, we then have 
\begin{equation}\label{eq:convCenteredX}
\frac{\overline{z}_T}{\sqrt{T}} \xrightarrow[T\rightarrow\infty]{\text{distr.}} {Z}.\end{equation}
Since $g_{N,M}(x)$ is continuous, \eqref{eq:convCenteredX}  and the Continuous Mapping Theorem \cite{Billingsley99} imply
\[
\frac{g_{N,M}\left({\overline{z}_T}\right)}{\sqrt{T}} = g_{N,M}\left(\frac{\overline{z}_T}{\sqrt{T}}\right) \xrightarrow[T\rightarrow\infty]{\text{distr.}} g_{N,M}\left({Z}\right)
,\] 
and the proof is completed by noticing that $g_{N,M}(x)$ is the sum of the maximum $M$ elements of its argument. 
\end{proof}
An immediate consequence of Lemma \ref{lem:regretConvDistribution},  is that 
\begin{equation}\label{eq:convExpectations}
\frac{\texttt{Reg}_T}{\sqrt{T}}=\frac{\mean{g_{N,M}(\overline{z}_T)}}{\sqrt{T}}\xrightarrow{T\rightarrow\infty}\mean{\sum_{i=1}^M{Z}_{(i)}} = \sum_{i=1}^{M}\mean{{Z}_{(i)}}
\end{equation}
which proves the first part of Theorem \ref{thm:sumz}. 

To prove the second part, we remark that the RHS of \eqref{eq:convExpectations} is the expected sum of $M$ maximal elements of vector $Z$, and hence larger than the expected sum of any $M$ elements of $Z$. In particular, we will compare with the following: Fix a permutation $\bar{\phi}$ over all $N$ elements, partition the first $2M$ elements in pairs by combining 1-st with 2-nd, ..., $i$-th with $i$+1-th, $2M$-1-th with $2M$-th, and then from each pair choose the maximum element and return the sum. Using this, we obtain: 
\begin{align}\nonumber
\mean{\sum_{i=1}^M{Z}_{(i)}} &\geq \mean{\sum_{i=1}^M\max\left[Z^{\bar{\phi}(2(i-1)+1)}, Z^{\bar{\phi}(2i)}\right]} \\ \nonumber
& = \sum_{i=1}^M\mean{\max\left[Z^{\bar{\phi}((2(i-1)+1)}, Z^{\bar{\phi}(2i)}\right]}
,\end{align}
where the second step follows from the linearity of the expectation, and the expectation is taken over the marginal distribution of a vector with two elements of $Z$. We now focus on  $\max\left[Z^{k}, Z^{\ell}\right]$ for (any) two fixed $k,\ell$. We have that 
\[
(Z^{k}, Z^{\ell})^T \sim \Nc\left(\mathbf{0}, \Sigmam(w_k, w_{\ell})\right)
\]  
where
\begin{align*}
&\hspace{0.3in}\Sigmam(w^k, w^{\ell})= \\&= \frac{1}{\sum_{n=1}^N1/w^n}\begin{bmatrix}
w^k-\frac{1}{\sum_{n=1}^N1/w^n} & -\frac{1}{\sum_{n=1}^N1/w^n}\\
-\frac{1}{\sum_{n=1}^N1/w^n} & w^{\ell}-\frac{1}{\sum_{n=1}^N1/w^n}
\end{bmatrix}
.\end{align*}
From \cite{Clark1961} we then have
\[
\mean{\max\left[Z^{k}, Z^{\ell}\right]} = \sqrt{\frac{1}{\sum_{n=1}^N1/w^n}}\frac{1}{\sqrt{2\pi}}\sqrt{w^{k}+w^{\ell}}
,\]
therefore 
\begin{equation}\label{eq:orderStatTransformed}
\mean{\sum_{i=1}^M{Z}_{(i)}} \geq \frac{1}{\sqrt{2\pi}}\frac{\sum_{i=1}^M\sqrt{w^{\bar{\phi}((2(i-1)+1)}+w^{\bar{\phi}(2i)}}}{\sqrt{\sum_{n=1}^N1/w^n}},
\end{equation}
for all $\bar{\phi}$. This completes the proof of Theorem \ref{thm:sumz}.
Corollary \ref{cor:lowerBoundRegret} follows noticing that the  tightest bound is obtained by maximizing   \eqref{eq:orderStatTransformed}  over all permutations. 
\end{proof}

\subsection{Online gradient ascent}\label{sec:gradient}

We show that the online variant of the standard gradient ascent algorithm achieves the best possible regret.
Recall that the utility in slot $t$ is described by the linear function $U_t(y_t)=\sum_{n=1}^Nw^nR_t^ny_t^n$, hence  the gradient at $y_t$ is an $N$-dimensional vector $\nabla U_t=
(\frac{\partial U_t}{\partial y_t^1},\dots,\frac{\partial U_t}{\partial y_t^N})=(w^1 R^1_t,\dots,w^N R^N_t)$.

\begin{box_example}[detach title,colback=blue!5!white, before upper={\tcbtitle\quad}]{Online Gradient Ascent Policy.}
	\small
Upon a request $R_t=(R_t^1,\dots,R_t^N)$, adjust the caching decisions ascending in the direction of the gradient:
\[
y_{t+1}=\Pi_{\mathcal{Y}}\left(y_t+\eta_t \nabla U_t\right),
\]
where $\eta_t$ is the stepsize, and $\Pi_{\mathcal{Y}}\left(z\right)\triangleq \argmin_{y\in\Yc}\|z-y \|$ 
is the Euclidean projection of the argument vector $z$ onto $\mathcal{Y}$.
\end{box_example}

The projection step is discussed in detail in the  following subsection. 
From a practical perspective, OGA works as follows. Upon the request of file $n$, the cache configuration is updated by caching more chunks from content $n$ (the number of extra chunks is decided according to the stepsize selection) and dropping chunks equally from all files in order to satisfy the cache constraints. 
Therefore, OGA bases the decision $y_{t+1}$ on the caching configuration $y_t$ and the most recent request $R_t$ and  it is a very simple online policy that does not require memory for storing the entire state (full history of $R$ and $y$). 

Let us now discuss the regret performance of OGA. We define first the set diameter $\textit{diam}(\mathcal{S})$  to be the largest $L_2$ distance between any two elements of set $\mathcal{S}$. To determine the diameter, we inspect two vectors $y_1,y_2\in \Yc$ which cache entire and totally different contents and obtain
\[
\textit{diam}(\mathcal{Y})=\left\{
\begin{array}{ll}
\sqrt{2M} & \text{if } 0<M\leq N/2,\\
\sqrt{2(N-M)} & \text{if } N/2<M\leq N. 
\end{array}
\right.
\]
Also, let $L$ be an upper bound of $\|\nabla U_t\|$, we have $L\leq\max_n(\sum_n w^nR_t^n)\leq \max_n(w^n)\equiv w^{(1)}$.
\begin{thm}[theorem style=plain]{Regret of OGA}{2}
Fix stepsize $\eta_t=\frac{\textit{diam}(\mathcal{Y})}{L\sqrt{T}}$,  the regret of \textup{OGA} satisfies:
\[
\texttt{Reg}_T(\textup{OGA})\leq \textit{diam}(\mathcal{Y})L{\sqrt{T}}.
\]
\end{thm}

\begin{proof}
Using the non-expansiveness property of the Euclidean projection \cite{Ber99book} we can bound the distance of the algorithm iteration from the best static policy in hindsight:
\begin{align*}
\|y^{t+1}\!-\!y^*\|^2&\triangleq\|\Pi_{\mathcal{Y}}\left(y_t\!+\!\eta_t \nabla U_t\right)\!-\!y^*\|^2\!\leq\! \|y_{t}\!+\!\eta_t\nabla U_t\!-\!y^*\|^2\notag\\
&= \|y_t-y^*\|^2+2\eta_t{\nabla U_t}^T(y_t-y^*)+\eta_t^2\|\nabla U_t\|^2,
\end{align*}
where we expanded the norm.
If we fix $\eta^t=\eta$ and sum telescopically over horizon $T$, we obtain:
\begin{equation*}
\|y_{T}-y^*\|^2\!\!\leq\! \|y_1-y^*\|^2+2\eta\sum_{t=1}^T{\nabla U_t}^T(y_t-y^*)+\eta^2\sum_{t=1}^T\!\|\nabla U_t\|^2.
\end{equation*}
Since $\|y_{T}-y^*\|^2\geq 0$, rearranging terms and using $\|y_1-y^*\|\leq \textit{diam}(\Yc)$ and $\|\nabla U_t\|\leq L$: 
\begin{equation}\label{eq:teleonl}
\sum_{t=1}^T{\nabla U_t}^T(y^*-y_t)\leq \frac{\textit{diam}(\Yc)^2}{2\eta} +\frac{\eta TL^2}2.
\end{equation}

For  $U_t$ convex it holds $U_t(y_t)\geq U_t(y)+ {\nabla U_t}^T(y_t-y)$, $\forall y\in\mathcal{Y}$, and with equality if $U_t$ is linear (as here). Plugging these in the OGA regret expression ($\max$ operator is removed) we get:
\begin{align*} 
\texttt{Reg}_T(OGA)&\leq \sum_{t=1}^T(U_t(y^*)-U_t(y_t)) =\sum_{t=1}^T{\nabla U_t}^T(y_t-y^*)\\
&\stackrel{\eqref{eq:teleonl}}{\leq} \frac{\textit{diam}(\Yc)^2}{2\eta} +\frac{\eta TL^2}2,
\end{align*}
and for $\eta=\textit{diam}(\Yc)/L\sqrt{T}$ we obtain the result. 
\end{proof}
Using the above values of $L$ and $\textit{diam}(\mathcal{Y})$ we obtain:
\[
\texttt{Reg}_T(OGA)\leq {w^{(1)}\sqrt{2MT}},\,\,\,\,\text{for}\,\, M<N/2\,.
\]

\begin{cor}[theorem style=plain]{Minimum Hit Regret}{2}
Fix a $\gamma\triangleq M/N$, $\gamma\in(0,1/2)$ and  $w^n=w,~\forall n$. Then, the regret of any caching policy $\sigma$  satisfies: 
\[
w\sqrt{\frac{\gamma}{\pi}} \sqrt{MT}\leq \inf_{\sigma}\texttt{Reg}_T(\sigma) \leq w\sqrt{2}\sqrt{MT}~~~\text{as}~T\to\infty.
\]
\end{cor}
\noindent Corollary \ref{cor:2} follows  immediately from Corollaries \ref{cor:lowerBoundRegret}-\ref{cor:2}. We conclude that disregarding a constant $\sqrt{2\pi/\gamma}$ (which is amortized over $T$), OGA achieves the best possible regret and hence fastest possible model-free learning rate for caching.

\subsubsection{Projection algorithm}\label{sec:projection}

We explain next the Euclidean projection $\Pi_{\mathcal{Y}}$ used in  OGA, which can be written as a constrained quadratic program:
\begin{align}
\Pi_{\mathcal{Y}}\left(z\right)\triangleq \argmin_{y\geq \bm{0}} & \sum_{n=1}^N (z^n-y^n)^2 \label{eq:projection}\\
\text{s.t. }& \sum_{n=1}^N y^n \leq M~\text{ and}~y^n\leq 1,~\forall n\in\mathcal{N}.\notag 
\end{align}
In practice $N$ is expected to be large, and hence we require a fast algorithm. 
Let us introduce the Lagrangian:
\begin{align}
	L(y, \rho, \mu, \kappa )&=\sum_{n=1}^N(z^n-y^n)^2+\rho(\sum_{n=1}^Ny^n-M)\nonumber \\ &+\sum_{n=1}^N\mu_n(y^n-1) - \sum_{n=1}^N\kappa_ny_n,\nonumber
\end{align}
where $\rho,\mu,\kappa$ are the  Lagrangian multipliers. The KKT conditions of \eqref{eq:projection} ensure that  the values of $y^n$  at optimality will be partitioned into three sets $\Mc_1, \Mc_2, \Mc_3$:
\begin{align}
&\Mc_1=\{n\in\mathcal{N}: y^n\!=\!1  \},\,\,\, \Mc_2\!=\!\{n\in\mathcal{N}:\! y^n=z^n\!-\!\rho/2\},\nonumber\\
 &\Mc_3\!=\!\{n\in\mathcal{N}:\! y^n\!=\!0 \},\label{eq:proj2} 
\end{align}
where $\rho=2\big(|\Mc_1|-M+\sum_{n\in \Mc_2} z^n\big)/|\Mc_2|$ follows from the tightness of  the simplex constraint. In order to solve the projection problem, it suffices to determine a partition of contents into these sets. Note that given a candidate partition, we can check in linear time whether it satisfies all KKT conditions (and only the unique optimal partition does). Additionally, one can show that the ordering of contents in $z$ is preserved at optimal $y$, hence a known approach is to search exhaustively over all possible ordered partitions, which takes $O(N^2)$ steps \cite{wang2015projection}. For our problem, we propose Algorithm 1, which exploits the property that all elements of $z$ satisfy $z^n\leq 1$ except at most one (hence also $|\mathcal{M}_1|\in \{0,1\}$), and computes the projection in $O(N\log N)$ steps (where the term $\log N$ comes from sorting $z$). In our simulations each loop is visited at most two times, and the OGA simulation  takes comparable time with LRU.

\begin{algorithm}[h!]  
\caption{Fast Cache Projection on Capped Simplex}  
\label{alg1}  
\begin{algorithmic}[1]  
    \REQUIRE $M$; sorted $z^{(1)}\geq\dots\geq z^{(N)}$
    \ENSURE $y = \Pi_{\mathcal{Y}}\left(z\right)$
    \STATE $\Mc_1\leftarrow \emptyset,\Mc_2\leftarrow \{1,\dots,N\},\Mc_3 \leftarrow \emptyset$
\REPEAT 
    \STATE $\rho\leftarrow 2({|\Mc_1|-M+\sum_{n\in \Mc_2} z^n})/{|\Mc_2|}$
    \STATE $y^n\leftarrow\left\{\begin{array}{ll}
1 &  n\in \Mc_1, \\
z^n-\rho/2 & n\in \Mc_2, \\
0 & n\in \Mc_3 \end{array}\right.$
    \STATE $\Sc \leftarrow \left\{n\in\mathcal{N}: y^n<0\right\}$
    \STATE $\Mc_2 \leftarrow \Mc_2\setminus \Sc$, $\Mc_3 \leftarrow \Mc_3\cup \Sc$
\UNTIL{$\Sc=\emptyset$}
        \IF{$y^1 >1$}
        \STATE $\Mc_1\leftarrow \{1\},\Mc_2\leftarrow \{2,\dots,N\},\Mc_3 \leftarrow \emptyset$
        \STATE Repeat 2-7
    \ENDIF \,\,\,\,\,\,\% KKT conditions are satisfied. 
\end{algorithmic} 
\end{algorithm}

\subsubsection{Performance comparison}\label{sec:FoLreg}

The online gradient descent (similar to OGA) is identical to the {well-known} \emph{Follow-the-Leader} (FtL) policy with a Euclidean regularizer $\frac{1}{2\eta_t}\|y\|$, cf.~\cite{Shalev12}, where FtL chooses in slot $t$ the configuration that maximizes the average utility in  slots $1,2,\dots,t-1$; it is the hypothesis that best describes the existing observations. We may observe that the unregularized FtL {applied here} would cache the contents with the highest frequencies, hence it is identical to the LFU. Therefore, OGA can be thought of as a regularized version of a utility-LFU policy, where additionally to largest frequencies, we smoothen the decisions (by means of a Euclidean regularizer).

\begin{figure*}[!t]
\centering
\subfigure[CDN aggregation (IRM model)] 
{\includegraphics[width=1.6in]
{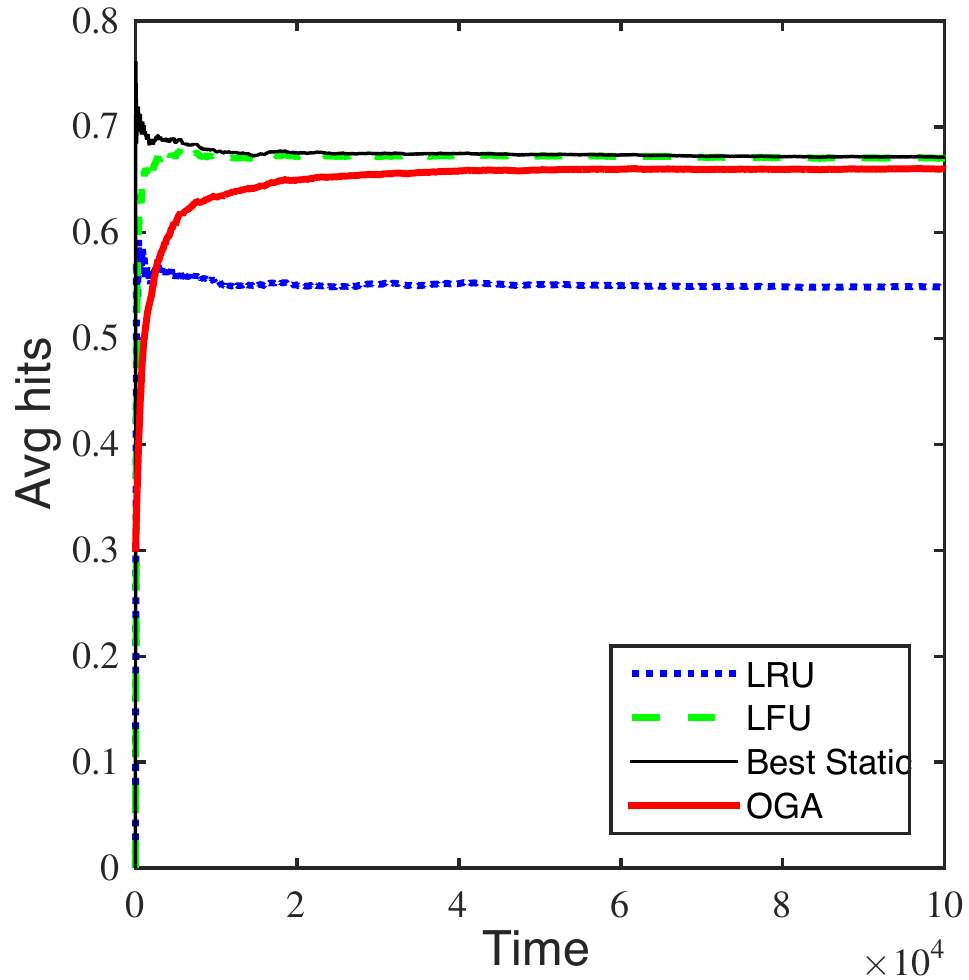}
\label{Fig:Fig1}}
\subfigure[Youtube videos (model \cite{snm})]{\includegraphics[width=1.6in]
{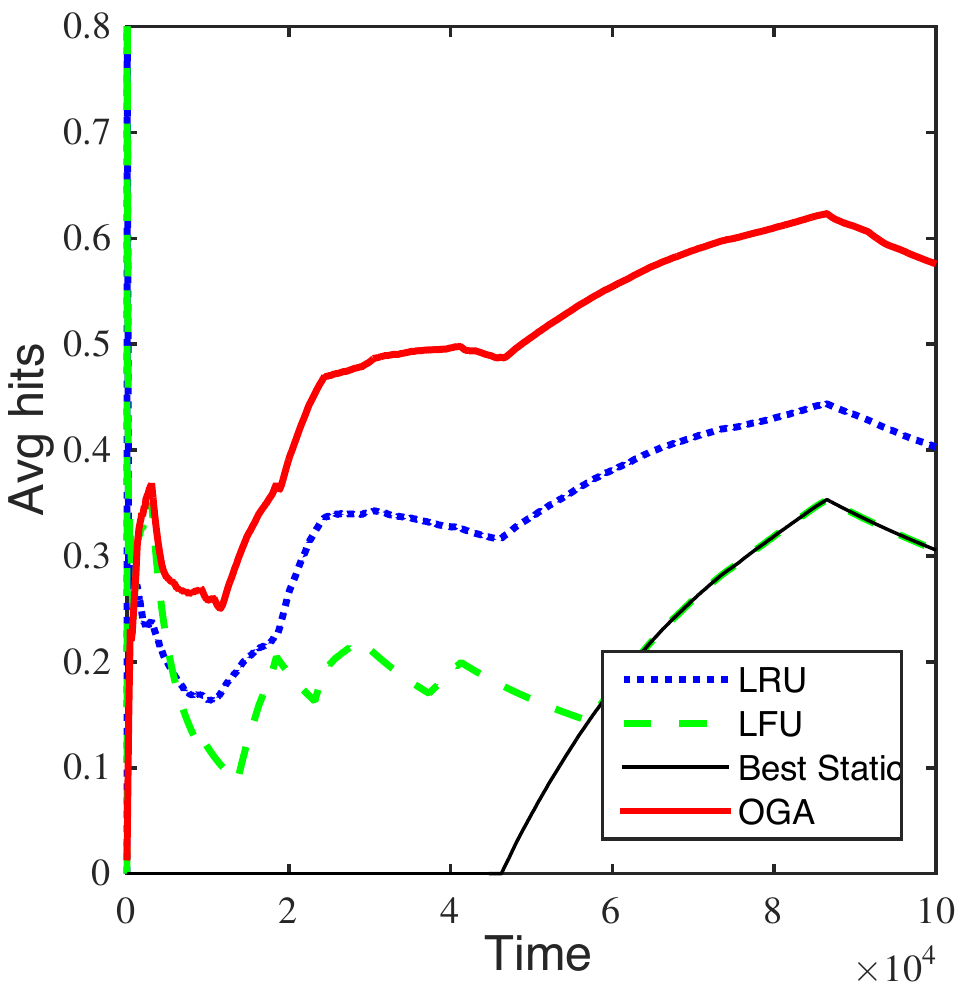}
\label{Fig:Fig2}}
\subfigure[Web browsing (trace \cite{Kurose08})]{\includegraphics[width=1.6in]
{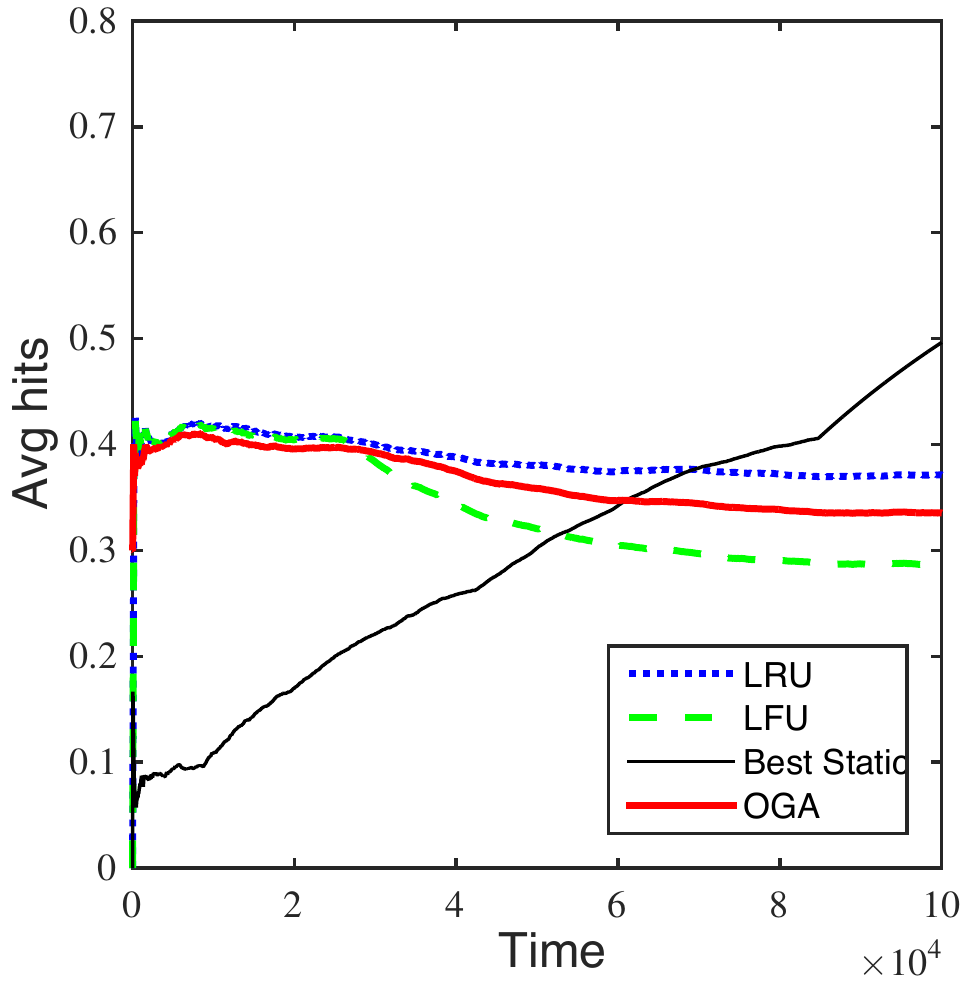}
\label{Fig:Fig3}}
\subfigure[Ephemeral Torrents (model \cite{Elayoubi2015})]{\includegraphics[width=1.6in] 
{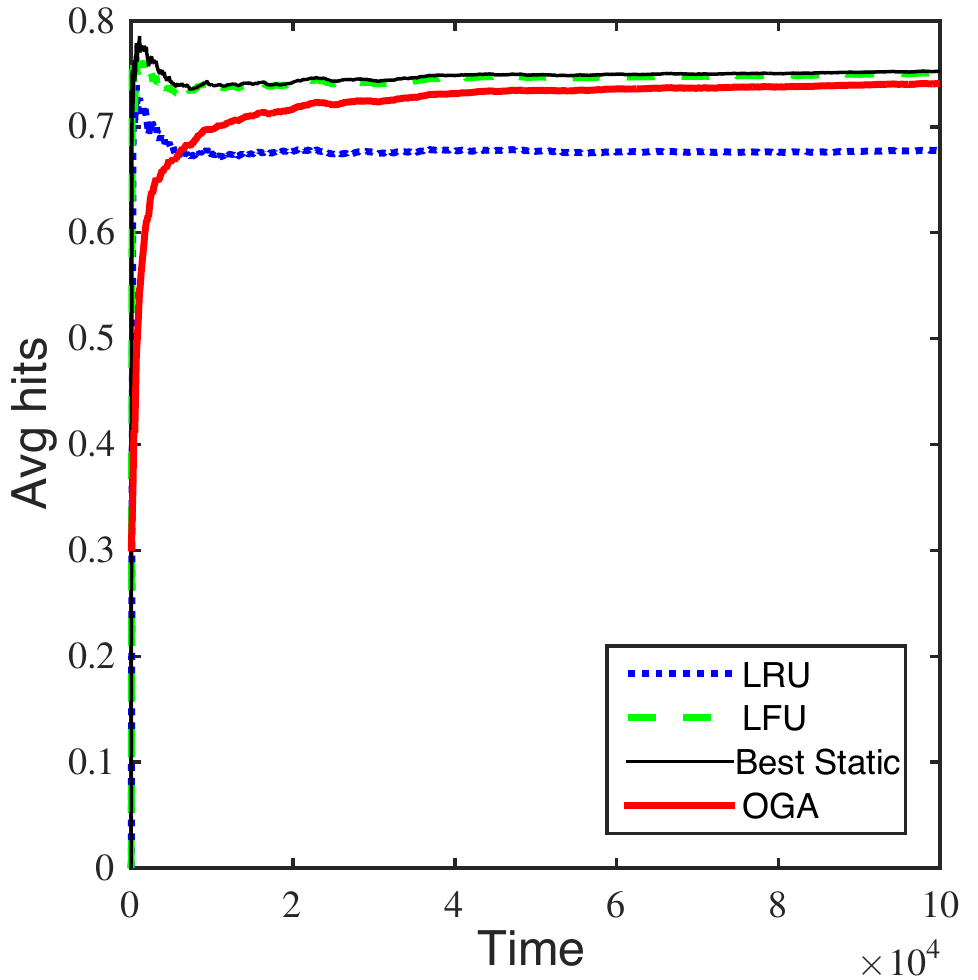}
\label{Fig:Fig4}}
\caption{Average hits under different request models \cite{fricker2012impact}; (a) i.i.d. Zipf, (b) Poisson Shot Noise \cite{snm}, (c) web browsing dataset \cite{Kurose08}, (d) random replacement \cite{Elayoubi2015}; Parameters: $\gamma=0.3, T=2\times 10^5, \eta=0.1$.}
\label{Fig:perf2}
\end{figure*}

Furthermore, OGA for $\eta=1$, $w^n=w$ also  bears similarities to the LRU policy, since recent requests enter the cache at the expense of older requests. Since the Euclidean projection drops some chunks from each  content  (see projection algorithm), we expect  least recent requests to drop first in OGA. The difference  is that OGA evicts  contents  gradually, chunk by chunk, and not in one shot. As a consequence, the two policies take strongly correlated decisions, but  OGA additionally ``remembers'' which of the {recent requests  are  also infrequent, and decreases corresponding $y^n$ values (as LFU would have done)}. 

Finally, in Fig.~\ref{Fig:perf2} we compare the performance of OGA to LRU, LFU, and the best in hindsight static configuration. We perform the comparison for catalogs of $10$K contents, with a cache that fits $3$K contents, and we use four different request models: {(a)} an i.i.d. Zipf model that represents requests in a CDN aggregation point \cite{fricker2012impact}; {(b)} a Poisson shot noise model that represents ephemeral YouTube video requests \cite{snm}; {(c)} a dataset from \cite{Kurose08} {with actual web browsing requests} at a university campus; and {(d)} a random replacement model from \cite{Elayoubi2015} that represents ephemeral torrent requests. We observe that OGA performance is always close to the best among LFU and LRU due to its universality. The benefits from the second best policy here is as high as 16$\%$ over LRU and 20$\%$ over LFU.

\section{Further notes on popularity prediction}\label{sec:prediction}

It should be now apparent that learning content popularity is central to caching optimization. Standard paging and caching policies such as LRU, LFU, and others, learn the content popularity in a certain heuristic manner, which is often embedded in the way they take caching decisions. Rather than following such reactive techniques, a recent research trend aims to predict content popularity and then optimize the content placement accordingly. In this concluding subsection, we provide a quick survey of \emph{popularity prediction} techniques.

Perhaps the most standard way of predicting popularity is to   assume its distribution is time-invariant, and use historical data for prediction. In this case,  an unbiased predictor of popularity is the content frequency, i.e., LFU  predicts  a static popularity. As we saw in chapter 2, however, the assumption of static popularity is rarely true since real datasets exhibit temporal locality. Indeed,  time-invariant models such as the IRM are considered  accurate for short time intervals \cite{breslau}, but applying them to a large time-scale analysis is problematic. If we assume time-variant popularity, we may take a Bayesian approach in predicting popularity, similar to the  rectangular-SNM classifier of \cite{mathieu} explained in section \ref{sec:class}. This approach had the advantage of correcting the prediction according to the information provided by the prior model. The  Bayesian methodology can be applied to other non-stationary models as well. However,  non-stationary models such as SNM have their own limitations when it comes to popularity prediction: (i) they have many degrees of freedom, hence making data fitting quite cumbersome, e.g., in \cite{snm} the authors advocate to restrict to 4 different  popularity profiles, and (ii) they subscribe to a specific ``non-stationarity'', which leads to overfitting, i.e., producing predictions that are tightly related to the specific non-stationary model. When reality deviates from our selected non-stationary model, prediction errors occur. The quest for the right non-stationary model of content popularity is still open.

Another methodology for prediction is to take a model-free  approach, which is agnostic to the prior stochastic model. For instance \cite{gunduz-reinforcement,giannakis-q-learning} employ \emph{Q-learning}, and \cite{mihaela-video-caching} leverages a scalable prediction approach; these methods assume that the evolution of popularity is  stationary across time but its distribution is unknown, see also stochastic bandit models which assume partial feedback \cite{blasco2014multi}, and transfer learning \cite{bacstuug2015transfer}. In the case of non-stationary unknown distributions, the online caching framework explained in section \ref{sec:online} can be used to derive the optimal learning technique if caching decisions are fractionals.
If we restrict to caching entire contents, then popularity prediction is studied in \cite{geulen2010regret,englert2013economical,lykouris-ML}.  

It is important for caching policies to  assimilate popularity changes as fast as possible. \cite{li2018accurate} suggests that (most) eviction policies can be modeled as Markov chains, whose mixing times give us a figure of how ``reactive'' the  policy is, or else how true to its stationary performance. A practical lesson learned is that although multi-stage $q$-LRU offers LFU-like hit rate performance, it  adapts slowly to popularity changes. 
From the analysis presented in this chapter, we saw that the regret in online caching is $\Theta(\sqrt{MT})$. This also provides an insight on how quickly we should expect to learn a good prediction. If $M$ is similar to $T$, the regret becomes $\Theta(T)$, which means that there is not enough time to accurately learn the optimal configuration. In other words, we expect to learn a good prediction at an horizon which is at least as big as the number of contents we can cache.

Other related works  look at how a trending file will evolve \cite{figueiredo2013prediction}, or how social networks can be used to predict the file popularity \cite{asur2010predicting}. Notably, some works even propose the use of recommendations as a means to actually engineer content popularity \cite{jordan_nudging,giannakas2018show}, and even jointly optimize recommendations and caching decisions \cite{kastanakis2018cabaret}.

The topic of learning popularity is an exciting one. In wireless networks, request patterns have low spatial intensity, which makes accurate popularity prediction essential but challenging. We mention here some practical challenges: {(i)} the learning rate depends on the rate of requests (number of requests per unit time), and hence on the  location (and aggregation layer) of the studied cache, {(ii)} the content popularity is affected by user community characteristics, and thus  learning relates to the  geographical characteristics of caches, and (iii) while a popularity model can enhance the accuracy of prediction, it also creates an inherent weakness, that on relying on a pre-specified model which might not be the true depiction of reality; the designer of a caching system must decide which of the two approaches (model-based or model-free) is the most beneficial at each occasion.

\chapter{Caching networks}\label{ch:4}


In this chapter we shift our attention to caching networks (CNs), i.e., systems of interconnected caches. The management of these complex systems requires the joint design of caching and routing policies, but involves also decisions about the deployment and dimensioning of the cache servers, Fig. \ref{fig:hierarchical}. These decisions were historically realized in different time scales, but the increasingly flexible network control tools tend to blur these boundaries. We focus on \emph{proactive caching} which refers to techniques that populate the caches based on the expected demand \cite{Bastug2014LivingOnTheEdge}, \cite{tadrous-proactive}. In the sequel we focus on the most important CN models and present key algorithms for optimizing their performance. 



\section{CN Deployment and Management}

\subsection{Caching Network Model}

The basic element of a caching network is a graph $G=(\mathcal{V},\mathcal{E})$, where $\cal V$ is the set of $V\!=\!|\cal V|$ nodes and $\cal E$ is the set of $E\!=\!|\cal E|$ directional links, defined as ordered pairs of nodes. Graph $G$ summarizes the available locations for the deployment of caches, the points of end-user demand, and the links connecting them. Each link $e\in \cal E$ has a maximum capacity of $b_e$ bytes/sec and each transferred byte induces $d_e$ units of monetary, delay or other type of cost. We assume that there is a total storage budget of $B$ bytes that can be allocated to nodes in $\mathcal{V}$, which in some cases is node-specific. The deployment of a cache at node $v$ induces $c_v$ cost per storage unit. Each node $v$ generates requests for files in the catalog $\mathcal{N}$ at a rate $\lambda_{vn}\geq 0$, $n\in\mathcal{N}$. This may correspond to the aggregate demand of several users which are accessing the network via node $u$.

We consider the \emph{anycast} model, where a file can be retrieved from any node that has a replica of it. We denote with $\mathcal{P}_{v,u}$ the set of paths\footnote{The paths can be given exogenously, i.e., being precomputed and used as parameters in the CN problem; or they can be calculated as part of the problem. } that can be used to transfer files from node $v$ to $u$. The set of all paths in the network is $\mathcal{P}$, and to facilitate notation we also define the aggregate cost $d_p=\sum_{e\in p}d_e$ of each path $p\in \cal P$. In summary:

\begin{opt}{Caching Network (CN)}
A caching network is a network graph $G=({\cal V}, {\cal E}, \bm{c}, B, \bm{\lambda}, \bm{b}, \bm{d})$ where:
\vspace{-1mm}
\begin{itemize}
\item ${\cal V}$ is a set of nodes, 
\item ${\cal E}$ is a set of directional links,
\item $\bm{c}=(c_v, v\in\mathcal{V})$ is the storage cost vector (\$/byte),
\item $\bm{B}=(B_v,v\in\mathcal{V})$ is the vector of cache capacities (bytes), 
\item  $\bm{\lambda}=(\lambda_{u,n}, v\in\mathcal{V}, n\in\mathcal{N} )$  is the demand matrix (requests/sec),
\item $\bm{b}=(b_e, e\in\mathcal{E})$ is the link capacity vector (bytes/sec), and
\item $\bm{d}=(d_e, e\in\mathcal{E})$ is the routing cost vector (\$/byte/sec).
\end{itemize}
\end{opt}

We can obtain certain caching networks from the above CN model as follows. The set of links ${\cal E}$ can be used to determine a specific graph structure, such as a bipartite or tree network; setting $b_e=\infty, \forall e$ makes the CN  \emph{link-uncapacitated}; and if $c_v=0, \forall v\in\mathcal{V}$ and/or $d_e=0, \forall e\in\mathcal{E}$ then the storage and routing costs become irrelevant. Finally, in some CNs the caches can be deployed only in a subset $\mathcal{V}_{c}\subset \cal V$ of nodes, and the requests emanate from certain nodes $\mathcal{V}_r\subset \cal V$. These can be captured by setting $\lambda_{v,n}=0$ for $v\in {\cal V}\setminus \mathcal{V}_{r}$, and $B_v=0$ for $v\in {\cal V}\setminus \mathcal{V}_{c}$, respectively.

\subsection{CN Optimization}

The full scale version of the CN design and management problem, henceforth referred to as DCR, includes the cache deployment, content caching and routing decisions. In detail, the CN-related decisions can be classified as follows:
\begin{itemize}
	\item	Cache Deployment: decide at which nodes to deploy the caches. This is also referred to as \emph{storage deployment}.		
	\item	Cache Dimensioning: select the storage capacity of each cache.	
	\item	Content Caching: decide which files will be placed at each cache. This is also referred to as the (proactive) \emph{caching policy}.
	\item	Content Routing: determine which cache will satisfy each request, and which paths will be employed to route files to requesters. These decisions constitute the \emph{routing policy}.	
\end{itemize}
Clearly, one can consider additional decisions involved in the design of CNs as, for example, dimensioning the network links or the servicing capacity of each cache server. Due to the different time-scale of these decisions, and not least due to the complexity of DCR, it is not common to optimize them jointly; with few notable exceptions such as \cite{BektasCOR07}, \cite{laoutarisComNet05}.

\begin{figure}
	\centering
	\includegraphics[width=0.8\textwidth]{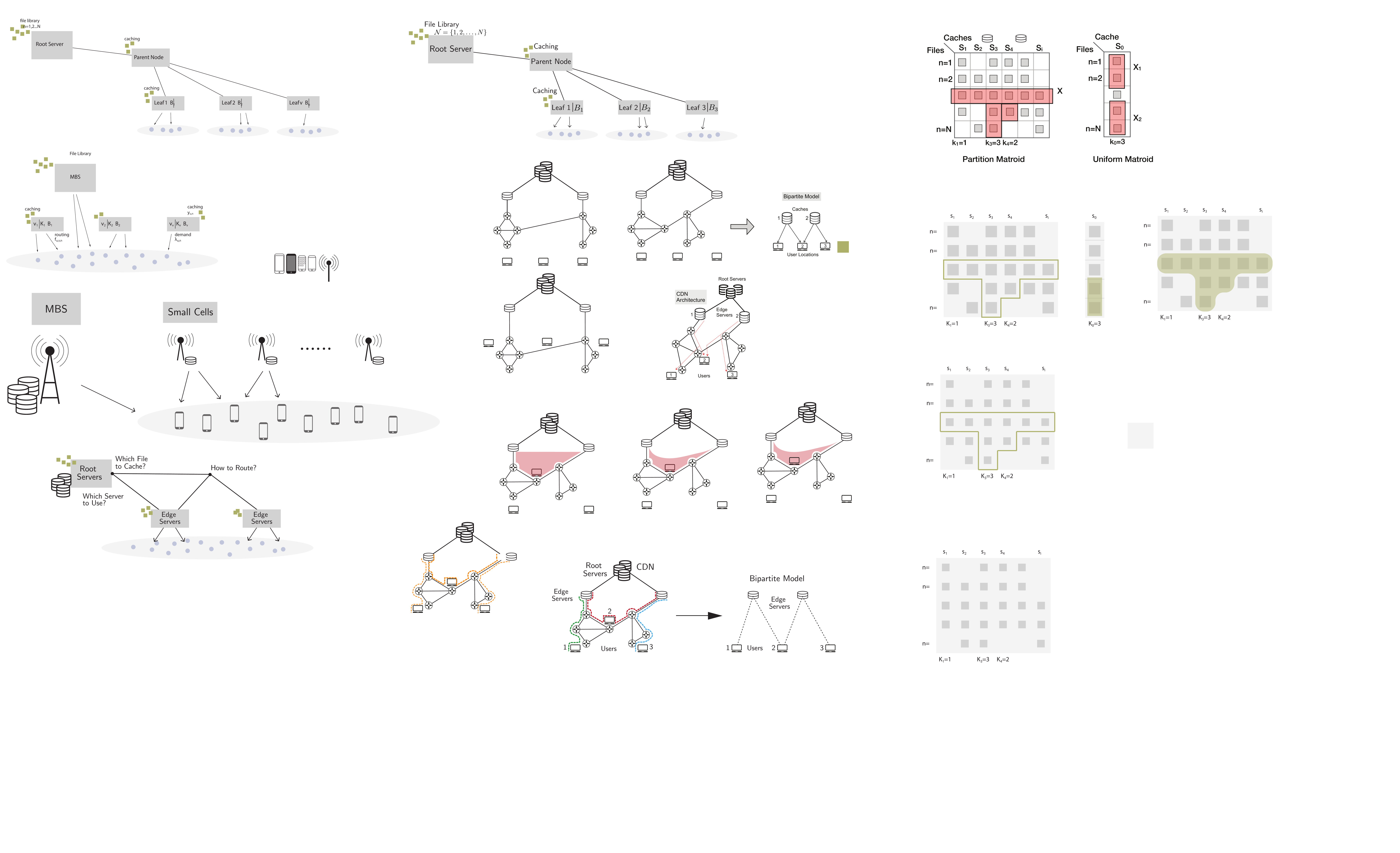}
	\caption{\small{Decisions in Caching Networks. \textbf{Small time scale}: at which server to cache? and from which server to fetch the file (routing)? The caching and routing decisions are inherently coupled, as a request can only be fetched from a server having the requested file. \textbf{Large time scale}: where to place servers, and how to dimension the server storage and the capacities of the  links connecting them?}}
	\label{fig:hierarchical}    
\end{figure}

\emph{Cache Dimensioning}. Each node $v\in \mathcal{V}$ can be equipped with a cache of size $x_v\geq 0$ storage units, and we define the vector $\bm{x}=(x_v,\,v\in \mathcal{V})$ for the network. Note that $x_v$ has been defined here in such a way that it simultaneously encompasses the deployment and dimensioning of cache $v$, but these decisions can be also considered separately. Vector $\bm{x}$ must satisfy the \emph{storage budget constraint}:
\begin{equation}
	\sum_{v\in \mathcal{V}} x_{v} \leq B, \label{eq:storage-budget-constraint}
\end{equation}
which might arise for economic or other reasons.\footnote{Dimensioning decisions are increasingly important because softwarization and cloud technologies allow us today to modify the size of a cache over time in a flexible manner.
} Then, the set of eligible cache deployment decisions can be defined as follows:
\begin{equation}
	\mathcal{X}=\{ \bm{x}\in \mathbb{R}_{+}^V \, \vert \, (\ref{eq:storage-budget-constraint})  \}.
\end{equation}

\emph{Content Caching}. We introduce the caching variables $y_{v,n}\in\{0,1\}$ which determine if file $n$ is placed at node $v$ ($y_{v,n}\!=\!1$), or not. The vector $\bm{y}=(y_{v,n}:v\in \mathcal{V}, n\in\mathcal{N})$ specifies the caching policy, which belongs to set:
\begin{equation}
	\mathcal{Y}=\{\bm{y}\in\{0,1\}^{V\times N}\}.
\end{equation} 
Clearly, the caching decisions must satisfy the \emph{cache size constraint}:
\begin{equation}
	\sum_{n\in\mathcal{N}} y_{v,n} \leq x_v,~~\forall v\in \mathcal{V},
\end{equation}
in order to ensure that files are cached only at nodes where storage has been deployed, and that their aggregate size does not exceed the cache capacity.

\emph{File Routing}. The routing is decided by variables $f_{p,n}\in\{0,1\},~p\in \mathcal{P}_{v,u}$ which determine whether path $p$ will be used to transfer file $n$ from cache $v$ to user $u$ ($f_{p,n}\!=\!1$). The routing policy is denoted $\bm{f}=( f_{p,n}:p\in\mathcal{P},n\in\mathcal{N})$. This definition corresponds to the \emph{unsplittable} traffic model. Another option is to allow \emph{splitting} the traffic in an arbitrary fashion among paths, i.e., $f_{p,n}\in [0,1]$. When demand is \emph{inelastic}, i.e., all requests must be satisfied, the following \emph{file delivery constraint} must be enforced:
\begin{equation}
\sum_{v\in \mathcal{V}}\sum_{p \in \mathcal{P}_{v,u} } f_{p,n} = 1,~~\forall u\in \mathcal{V},\,n\in\mathcal{N}. \label{eq:content-delivery}
\end{equation}
In case the link costs represent delay, one could consider a quality-of-service (QoS) requirement for delivering the files within a certain delay window $D_{th}$. This can be ensured by imposing the \emph{QoS constraint}:
\begin{equation}
f_{p,n}d_{p} \leq D_{th},~~\forall p\in \mathcal{P}, n\in\mathcal{N}. \label{eq:Qos-constraint}
\end{equation}
Note that if $f_{p,n}$ is continuous, then (\ref{eq:Qos-constraint}) needs to include an indicator function that does not allow a non-zero flow unless the path has an admissible delay. Finally, the total traffic over each link cannot exceed its capacity, hence we have the \emph{routing constraint}:
\begin{equation}
\sum_{v,u\in V}\sum_{p\in \mathcal{P}_{v,u}:e\in p}\sum_{n\in \mathcal{N}} f_{p,n}\lambda_{u,n} \leq b_e,~~\forall e\in \cal E. \label{eq:routing-constraint}
\end{equation}
Collecting the above constraints, we define the set of eligible routing decisions:
\begin{align}
	\mathcal{F}=\{\bm{f}\in\{0,1\}^{P\times N} \, \vert \, (\ref{eq:content-delivery}); (\ref{eq:Qos-constraint});(\ref{eq:routing-constraint}) \}.
\end{align}

One key property of DCR is the coupling among the caching and routing decisions. In particular, a request should not be routed to a node that does not have cached the requested file. This can be enforced with the \emph{file availability} constraint:
\begin{equation}
	f_{p,n}\leq y_{v,n}\,\,\,\,\forall p\in\mathcal{P},\,n\in\mathcal{N},\,u\in\mathcal{V},
\end{equation}
that allows $f_{p,n}$ to take non-zero values only if $y_{v,n}=1$. We will see in Sec. \ref{sec4:arbitrary} that this constraint can be omitted if the cost of a cache miss is incorporated in the objective. 

The objective function of DCR, $J(\bm{f}, \bm{x})$, captures the cost of deploying the caches and routing the files. A typical case is when these costs are linear:
\begin{equation}
J(\bm{f}, \bm{x})=\beta_1\underbrace{\sum_{u,v\in \mathcal{V}}\sum_{p\in \mathcal{P}_{v,u}}\sum_{n\in \mathcal{N}}f_{p,n}\lambda_{u,n}d_p}_{\text{routing cost}}  + \beta_2\underbrace{\sum_{v\in \mathcal{V} }x_vc_v}_{\text{deploym. cost}}\,,
\end{equation}
where parameters $\beta_1,\beta_2\geq 0$ tune the relative importance of the two cost components. Clearly, there are several other options for $J(\cdot)$. For example, one can include the caching operational expenditures which increase with the volume of stored files; and the cost $d_p$ might be a function of the amount of data traversing path $p$.

We can now formulate the joint cache deployment, file caching and routing problem:  
\begin{opt}{Deployment, Caching, and Routing (DCR)}
\vspace{-0.2in}
\begin{align}
	& \min_{\bm{f}\in \mathcal{F},\, \bm{x}\in\mathcal{X},\, \bm{y}\in\mathcal{Y} }\,\,\, J(\bm{f}, \bm{x}) \nonumber\\
	\text{s.t.}\,\,\,\,\,\,\,\,\,	&\,\,\, \sum_{n\in\mathcal{N}} y_{v,n} \leq x_v,\,\,\,\, \forall v\!\in\! \mathcal{V}; \label{eq:cache-storage-coupling} \\
	&\,\,\,\, f_{p,n}\leq y_{v,n},\qquad\forall\, u,v\in\mathcal{V},\,\,\, p\!\in\! P_{v,u},\,n\!\in\!\mathcal{N}. \label{eq:routing-caching-coupling}
\end{align}
\vspace{-3mm}
\end{opt}
\noindent In the general case and assuming linear cost functions, DCR is a Mixed Integer Linear Program (MILP) that is NP-hard to solve optimally, e.g. see \cite{bateni-uhcfl} for a reduction from the unsplittable hard-capacitated facility location problem. Furthermore, most DCR instances are solved for networks with hundreds or thousands of nodes, thousands or millions of files, and possibly many paths. Therefore, most often we focus on polynomial-time approximation algorithms or employ greedy algorithms or other heuristics that perform well in practice.

In the sequel we discuss certain important instances of DCR. First, we ignore the caching decisions and consider the storage deployment problem which can be studied separately or jointly with routing. This is a network design problem, solved at a coarse time scale. We next introduce the seminal femtocaching problem and study how to optimize caching policies on bipartite network graphs. We continue with capacitated bipartite networks, and then analyze joint routing and caching policies for hierarchical networks. We conclude with general CN problems involving arbitrary network graphs, congestible links, and non-linear objective functions.


\section{Design of Caching Networks} \label{sec:cn-design}

The design of CNs consists in  selecting at which nodes to deploy the caching servers, and how to assign the user demand to them. In some cases, it is also necessary to dimension the storage or servicing capacity\footnote{The servicing capacity constraint might arise for practical reasons, e.g., due to computing capacity or I/O limitations of the servers, or might be imposed by the system designer to achieve load balancing.} of the deployed caches. This problem is more challenging to solve in the presence of such servicing constraints; and when the network links are capacitated or induce load-dependent routing costs. In these cases the deployment decisions for the different caches are intertwined, and it does not suffice to consider shortest-path routing.

A meaningful design criterion is to minimize the cache deployment costs, including also the routing costs when relevant, for the expected demand $\bm{\lambda}$. In other cases we might need to consider QoS constraints that bound the delay for every delivered file, or we might decide to minimize the maximum delay experienced by any user. In sum, the CN design criteria can be broadly classified in two categories: \emph{(i)} those that optimize aggregate quantities (e.g., total cost); and \emph{(ii)} those that bound worst-case performance (e.g., maximum delay).


CN design problems can be solved by employing \emph{facility location theory} \cite{FL-book-1, FL-book-2} which studies the optimal deployment of facilities based on certain cost criteria. Although the problem of facility location is primarily encountered in the context of deploying warehouses for supply chains, the deployment of caching servers is not so different. In particular, caches for the Internet play the role of warehouses since they store content (product) before its delivery to the end users. Hence, their location is of extreme importance for optimizing the content delivery service (in terms of delay or routing costs), in a similar way that the location of warehouses shapes the performance of a supply chain. 

There is a vast literature studying the different FL problem variants. A basic criterion for their taxonomy is whether the number of  facilities is predetermined ($k$-center and $k$-median problems) or not (fixed charge problems). Also, the serving capacity of each facility can be bounded (hard-capacitated) or not (uncapacitated), and the demand of each point can be served concurrently by many facilities ({splittable} demand) or only by one facility ({unsplittable}). Any combination of these features creates a different FL problem, which are typically NP-hard to solve optimally as they generalize the set cover problem.


\begin{box_example}[detach title,colback=blue!5!white, before upper={\tcbtitle\quad}]{Facility Location Problems}
\small
\noindent FLP can either optimize maximum distance, e.g., k-center problem, or average distance criteria, e.g., k-median and fixed charge problems. We refer the reader to \cite{FL-book-3, tardosSTOC97} for a detailed taxonomy.

\vspace{1.5mm}

\noindent $\bullet$ \emph{k-center}. Given a set $\mathcal{V}$ of locations and distances $\mathcal{P}$, open $k$ facilities and select the demand-facility assignments $\mathcal{F}$, to \textbf{minimize the maximum distance} $P_{v,u}$ for any pair $(v,u)$ with $f_{v,u}=1$.

\vspace{1.5mm}

\noindent $\bullet$ \emph{k-median}. Given a set $\mathcal{V}$ of locations and distances $\mathcal{P}$, open $k$ facilities and select the demand-facility assignments $\mathcal{F}$, to \textbf{minimize the aggregate distance} $\sum_{v,u}P_{v,u}$ for all $(v,u)$ pairs with $f_{v,u}=1$.

\vspace{1.5mm}

\noindent $\bullet$ \emph{Uncapacitated fixed-charge}. Given a set $\mathcal{V}$ of locations, distances $\mathcal{P}$, and deployment costs $\bm{c}$, decide how many and which facilities to open, and select the assignments $\mathcal{F}$, to \textbf{minimize the aggregate distance} and \textbf{deployment costs}. 
\end{box_example}

There are several algorithms for solving the different FLP instances. Starting from the $k$-center problem, a greedy algorithm which iteratively opens a new facility at a location that is farther away from the existing ones, ensures cost no more than twice the optimal, i.e., guarantees a 2-approximation solution. It has been shown that there is no polynomial time algorithm that can achieve a better ratio \cite{shmoys-k-center}. The $k$-median problem is more intricate. It is known that it cannot be approximated better than $1+2/e \approx 1.73$ \cite{jain2002new}. The first constant approximation algorithm was based on LP-rounding and achieved  a $6\frac{2}{3}$ ratio \cite{shmoys-k-median-2}, subsequently reduced to 4 in \cite{vazirani-k-median}. An interesting approach was presented in \cite{arya2004local}, where a local search algorithm iteratively swaps $p$ of the $k$ facilities, and selects the minimum cost configuration. This was shown to provide a $3+2/p$ approximation, but the complexity increases fast with $p$. This algorithm has been extensively used in caching where typically $p\leq 3$ \cite{jamin2000placement, RossComCom02, Jsourlas11}. The approximation ratio has been improved to  $(1+\sqrt{3}+\epsilon)$ \cite{li-k-median} using a \emph{pseudo-approximation} technique, i.e., by disregarding the effect of the number of facilities in the approximation. 


For the fixed-charge problems, the complexity depends on whether there is an underlying metric space or not. The non-metric problems are  NP-hard to approximate. For metric settings, the \emph{Uncapacitated Facility Location} (UFL) instance attains a ratio no-less than 1.463 \cite{GuhaGreedy99}, while \cite{ShiLiFLP13} proposed a 1.488-approximation algorithm. A practical 1.861-approximation algorithm is given in \cite{jain2002new}, that uses the concept of budget for locations which are bidding to open facilities. The capacitated variant (CFL) is {NP-hard to approximate} when the demand is unsplittable \cite{bateni-uhcfl, tardosSTOC97}; admits a 5.83-approximation algorithm if demand is splittable \cite{zhang-capacitated} that can be further improved to 5 for uniform opening costs \cite{shmoys-cflp12}. See also \cite{vazirani-general-algo} for a technique that employs UFL algorithms for solving soft-capacitated FLPs, where multiple capacitated facilities can be deployed in each location. The above algorithms constitute a very useful toolbox for formulating and solving the various CN design problems. We study some representative instances below.









\subsubsection{Deployment Formulations}


We start with problem (D1). First, observe that the caching decisions $\bm{y}$ can be dropped, and this reduces substantially the dimension of the problem. Second, we are interested here in the aggregate demand $\lambda_u$ of each user $u\in \mathcal{V}_r$, instead of the per-file requests $\lambda_{u,n}$. Third, if the network links are uncapacitated and non-congestible, i.e., cost $d_e$ does not depend on the amount of data routed over link $e\in\mathcal{E}$, then the routing policy can be replaced by assignment decisions $f_{v,u}\in\{0,1\}, \forall v,u\in\mathcal{V}$, where $f_{v,u}\!=\!1$ selects the shortest path $p^\star\in\mathcal{P}_{u,v}$, which has delay $d_{v,u}=d_{p^\star}$. Finally, without loss of generality, we assume below that all caches have the same size $B$ and that only one cache can be deployed at each location.

Under these assumptions, the eligible routing and storage deployment policies for (D1) are:
\begin{align}\small
	\mathcal{F}_{D_1}\!=\!\left\{\!\bm{f}\in\{0,1\}^{V\times V} \Big| \! 
	\begin{array}{l}
	\sum_{v\in \mathcal{V}} f_{v,u}= 1,\,\forall\,u\in \mathcal{V}\\
	f_{v,u}=0\,\, \text{if}\,\,\mathcal{P}_{u,v}=\emptyset
	\end{array}\!\right\},	\mathcal{X}_{D_1}=\left\{ \bm{x} \in\{0,1\}^V \right\} \nonumber
\end{align}
and the respective problem is formulated as follows:   
\begin{opt}{Deployment (D1)}
	\vspace{-0.2in}
\begin{align}
	&\,\,\, \min_{\bm{f}\in\mathcal{F}_{D_1},\,\,\, \bm{x}\in\mathcal{X}_{D_1} }\,\,\, \sum_{v\in V}x_vc_v +\sum_{v,u \in \mathcal{V}}f_{v,u}\lambda_ud_{v,u}  \label{eq:D1-objective}\\
	\text{s.t. } &\,\,\,\,\,\,f_{v,u}\leq x_v, \,\,\,\,\forall\,v,u\in\mathcal{V}.
\end{align}
	\vspace{-0.25in}
\end{opt}
\noindent It is interesting to note that once $\bm{x}$ is decided, then the optimal $f$ for (D1) can be found by assigning each request to the cache offering the minimum routing cost. However, when there is a servicing capacity $F_v$ at each facility $v$,
\begin{equation}
	\sum_{u\in\mathcal{V}}f_{v,u}\lambda_{u}\leq F_v,\,\,v\in\mathcal{V},
\end{equation}
the routing policy is not trivially specified by the cache deployment. For example, the closest cache to a node might be reachable only via highly congested links. Problem (D1) can be mapped to the UFL problem, see Fig. \ref{fig:facility}, and hence we can employ the above approximation algorithms for its solution. 

\begin{figure}
	\centering
	\includegraphics[width=0.43\textwidth]{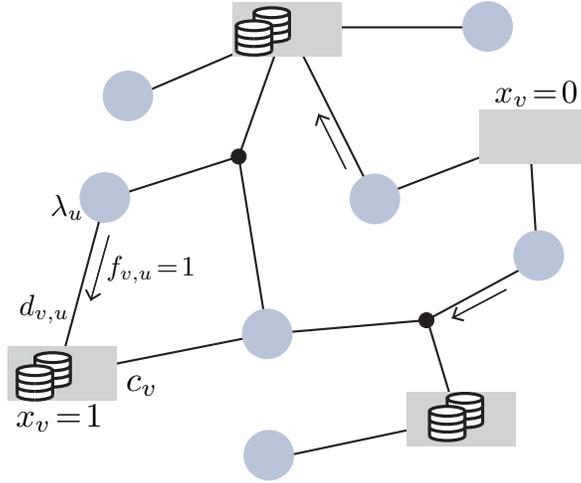}
	\caption{\small{The cache deployment as a facility location problem. Each node $v$ represents an eligible location for opening a facility, i.e., deploying a cache ($x_v=1$) with cost $c_v$. The files must be routed from open facilities to each requester $u$ ($f_{v,u}=1$), paying the respective routing cost $d_{v,u}$ for each unit of demand $\lambda_u$ at node $u$.
	}}
	\label{fig:facility}    
\end{figure}

A different version of this problem arises if there is no cache deployment cost but, instead, an upper bound $B$ on the number of caches we can deploy.\footnote{In other formulations, $B$ can capture the total amount of storage available for deployment in the different caches; or we might have a vector of such bounds $\bm{B}=(B_v, v\in\mathcal{V})$, one for each cache.} This is a classical $k$-median problem:
\begin{opt}{Deployment (D2)}
	\vspace{-0.2in}
	\begin{align}
		&\,\,\, \min_{\bm{f}\in\mathcal{F}_{D_2},\,\,\, \bm{x}\in\mathcal{X}_{D_2} }\,\,\, \sum_{u,v\in \mathcal{V}} f_{v,u}\lambda_ud_{v,u}  \label{eq:D2-objective}\\
		\text{s.t. }\,\,\,\, &f_{v,u}\leq x_v, \,\,\,\,\forall\,v,u\in\mathcal{V}.
	\end{align}
	\vspace{-0.25in}
\end{opt}
\noindent where $\mathcal{F}_{D_2}= \mathcal{F}_{D_1}$ and: 
\begin{equation}
	\mathcal{X}_{D_2}=\{ \bm{x}\in\{0,1\}^V\,\vert\, \sum_{v\in \mathcal{V} } x_{v}\leq B \}. \nonumber
\end{equation}
An interesting twist of (D2) appears when we wish to minimize the worst delivery cost. In this case, the objective includes (a function of) this cost variable, denoted $J_{th}$, and the set $\mathcal{F}_{D_2}$ must be amended to include the constraints: 
\begin{equation}
	\sum_{v \in\mathcal{V} }f_{v,u}\lambda_ud_{v,u}\leq J_{th},\,\,\forall u\in\mathcal{V}, \nonumber
\end{equation}
which ensures that no active route will exceed the selected value for $J_{th}$. 

Finally, in some CNs the caches need to sync with each other, periodically or on-demand, in order to ensure content consistency. This synchronization cost will be an additional component in the objective function and can be expressed as a Steiner tree connecting the caches. This is a suitable choice, for example, when we wish to minimize the bandwidth consumption or the delay of synchronization. This problem is a generalization of the connected facility location problem (CoFL), see \cite{guptaSTOC03}, \cite{swamyCFLP04}. CoFLP was introduced in \cite{guptaSTOC01}, and \cite{eisenbrand08} provided a randomized algorithm with an expected approximation ratio of $4$; \cite{BaevSIAM08} uses the term \emph{connected data placement} to describe this problem in the context of CNs. In another case the CN designer might wish to bound the synchronization delay $D_s$ for any pair of installed caches. This introduces the following constraint:
\begin{equation}
	x_{v_1}\cdot x_{v_2}\cdot d_{v_1, v_2}\leq D_{s},\,\,\forall\,v_1,v_2\in \mathcal{V}. \nonumber
\end{equation}
Such non-linear expressions in constraints or the objective compound further the optimization problem, but in certain cases they can be linearized, see \cite{Bektas} for an example.

\subsubsection{Discussion of Related Work}


The CN design problem has been extensively studied in the context of CDNs, see \cite{SahooSurvey17}. One of the first related studies is \cite{Li:Infoc99} which finds the delay-minimizing cache deployment when routing costs are fixed. The work \cite{qiu_replicas01} formulates an uncapacitated $k$-median CN problem to minimizes the delay. The uncapacitated variant of FLP is used in \cite{SmaragdInfocom07}, while \cite{bateni-uhcfl} analyzes the unsplittable hard-capacitated case. The work in \cite{DRazJSAC02} considers the cache synchronization cost, showing that increasing the number of caches beyond a certain threshold induces costs that surpass benefits. The same problem is formulated as a CoFLP in \cite{KazCNSM13}. Finally, \cite{GeorgiadisTPDS06} models the problem of server deployment as a splittable soft-capacitated FLP with path selection.
The CN design problem is simplified when the network is a tree graph (contains no cycles). For example, \cite{RobertTPDS08} considers the capacitated cache deployment problem in trees and formulates an ILP for the case of splittable and non-splittable requests; \cite{krishnan-ToN00} models the cache deployment as a {$k$-median} problem; and \cite{GuhaFOCS00} studies hierarchical caching with heterogeneous caches.


\section{Bipartite Caching Networks}\label{chapter4:bipartite}


\begin{figure}[t!]
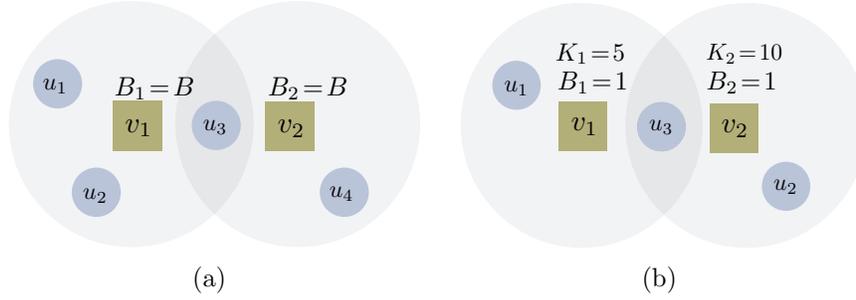

	\centering
	\subfigure[]{\includegraphics[width=5.5cm, height=3.3cm]{Femtocaching-example-new3}} \,\,\,
	\subfigure[]{\includegraphics[width=5.5cm, height=3.3cm]{TCOM-system-example-new3}}
	\caption{\small{\textbf{(a)}: \textbf{Impact of overlapping SBS areas}. If each SBS can cache $B$ files, users $u_1$ and $u_2$ prefer SBS $v_1$ to cache the $B$ most popular files; and similarly user $u_4$ for $v_2$. However, $u_3$ would prefer $v_1$ to cache the $B$ popular files and $v_2$ the second $B$ most popular. \textbf{(b)} \textbf{Link capacities impact caching policy}. Each SBS can cache at most one of the two files $n_1$ and $n_2$. The demands are $\lambda_{u_1, n_1}\!=\!1, \lambda_{u_2, n_1}\!=\!2, \lambda_{u_3, n_2}\!=\!10$. Assuming $v_1$ can serve $5$ requests and $v_2$ $10$ requests due to bandwidth constraints, the policy that maximizes the SBS hit ratio places $n_1$ to $v_1$ and $n_2$ to $v_2$. This leaves only 2/13 requests for the MBS. If however there were no bandwidth limitations, then the optimal caching policy places $n_2$ to $v_1$ and $n_1$ to $v_2$. Hence, adding the capacity limitation changes completely the optimal solution.}}
	\label{fig:bipartite-example}
\end{figure}


We turn now our focus to the design of caching and routing policies for bipartite network graphs where a set of caches serve a set of users. The bipartite model can capture a range of general caching networks, and has been recently used for the design of wireless edge caching systems. Namely, \cite{golrezaei2012femtocaching} proposed the \emph{femtocaching} architecture where edge caches are placed at small cell base stations (SBSs) that underlay a macrocellular base station (MBS), and proactively cache files which are then delivered to closely-located users. Femtocaching reduces network expenditures as it saves the expensive MBS capacity; alleviates the congestion at the SBS backhaul links; and improves the user experience since it employs low-latency energy-prudent links. 

 
 
The main goal in femtocaching design is to devise a caching policy that maximizes the hit rate of the SBS caches, or minimizes the file delivery delay by matching users to closest SBSs. These problems are equivalent when all wireless links have equal delay, and become NP-complete to solve optimally when the SBSs have overlapping coverage, Fig.~\ref{fig:bipartite-example}(a). An important version of this problem arises in massive demand (or, ultra-dense) scenarios where the SBS wireless capacity can be drained by the user requests. These limitations affect the effectiveness of caching, see Fig. \ref{fig:bipartite-example}(b), and require the joint derivation of user association and caching policies. We start below with the typical femtocaching model in Sec.~\ref{sec:uncap_femto} 
 and then we proceed to the analysis of its capacitated version in Sec.~\ref{sec:capacitated-femtocaching}. 



\begin{figure}[h!]
	\centering
	\includegraphics[width=0.9\textwidth]{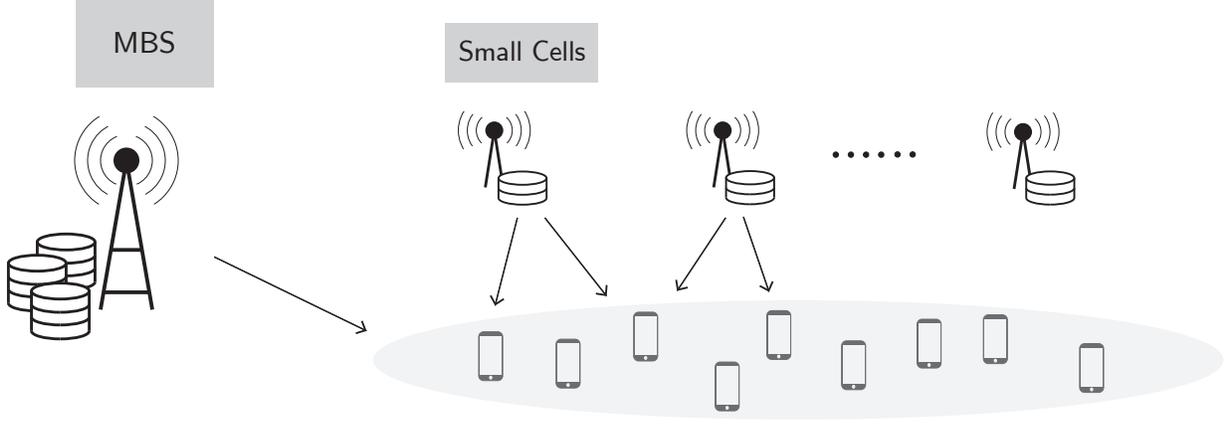}
	\caption{\small{Femtocaching system model: a set of users in range with cache-enabled small cell base stations that underlay a macrocellular base station where the root servers are located.}}
	\label{fig:femto_system_model}
\end{figure}

\subsection{The Femtocaching Problem}\label{sec:uncap_femto}

Consider an area where wireless users place random file requests to a set of SBSs which serve as dedicated content distribution nodes, \cite{femtocachingTransIT13}, \cite{golrezaei2012femtocaching}. The caches have limited storage capacity and transmission range, imposing constraints both on the file placement and the network connectivity. We assume there is an MBS which stores the entire file library and can serve all users, Fig. \ref{fig:femto_system_model}. In this context, the caching policy design can be formulated as follows: for a given file popularity distribution, cache capacity and network topology, \emph{how should the files be placed at the caches such that the average sum downloading delay of all users is minimized}? We next provide the detailed model and solution approach.



\subsubsection{Problem Definition}

Consider a caching network $G=({\cal V}, {\cal E}, \bm{c}, \bm{B}, \bm{\lambda}, \bm{b}, \bm{d}
)$, where nodes are partitioned to user locations $\mathcal{V}_r\subset \mathcal{V}$ and caching servers $\mathcal{V}_c=\mathcal{V}_{c}^0\cup \{0\}$, where element $\{0\}$ denotes the MBS and $\mathcal{V}_{c}^0$ the SBSs. Costs are irrelevant $c_v=0, v\in\mathcal{V}_c$, caches have equal capacity $B_v=B, v\in\mathcal{V}_{c}^0$ measured in number of files they fit, and link capacities are unbounded ($b_e=\infty, e\in\mathcal{E}$).


The existence of a link $e=(u,v)\in\mathcal{E}$ denotes that cache $v\in\mathcal{V}_c$ is reachable by user $u\in\mathcal{V}_r$, i.e., a wireless communication link connects them.\footnote{We assume the SBSs and the MBS operate in disjoint channels, hence there is no interference in these links.} The network graph is bipartite because no links exist between nodes of the same set. For convenience we also use $\mathcal{V}_c(u)$ and $\mathcal{V}_r(v)$ to denote the set of caches in range with user $u$, and the set of users in range with cache $v$, respectively. Also, we will be using $\mathcal{V}_{c}^{0}(u)=\mathcal{V}_{c}(u)\setminus\{0\}$ for each user $u$. The network performance is characterized by the average delay $d_{v,u}$ for delivering one byte from each cache $v$ to each connected user $u$, and we assume the MBS has the worst performance, i.e., $d_{0,u}\geq d_{v,u},\,\forall v\in\mathcal{V}_c$. Finally, the demand is considered homogeneous across users, i.e., $\lambda_{u,n}=\lambda_{n}, \forall u,n$. We revisit this assumption in Sec. \ref{sec:capacitated-femtocaching}.


The caching policy is described by the binary variables $y_{v,n}\in\{0,1\}$ which determine whether file $n$ is cached at SBS $v$. The set of eligible policies is: 
\begin{equation}
	\mathcal{Y}_{C_1}=\{\bm{y}\in \{0,1\}^{V_{c}^0\times N}\, \vert\, \sum_{n\in\mathcal{N}}y_{v,n}\leq B,\,\,v\in\mathcal{V}_{c}^0 \}.
\end{equation}
Let us denote with $\mathcal{M}_{v}$ the set of cached files at cache $v$. Taking the perspective of user location $u$, we denote with $\sigma_{u(j)}$ the index of the cache with the $j$th smallest delay for that user. That is, $\sigma_u$ is a permutation that ranks the $\mathcal{V}_c(u)$ caches in increasing delay: $d_{\sigma_{u(1)},u}\leq \ldots \leq d_{\sigma_{u(|\mathcal{V}_{c}(u)|)},u}\leq d_{0,u}$. Therefore, the average delay that $u$ experiences under policy $\bm{y}$ is: 
\begin{equation}
D_{u}(\bm{y})= d_{0,u}\sum_{n\in\mathcal{N} }\mathbbm{1}(n,u,0)\lambda_n + \sum_{j\in\mathcal{V}_{c}^{0}(u)}d_{\sigma_{u(j)},u}\sum_{n\in\mathcal{N}}\mathbbm{1}(n,u,j)\lambda_n\,, \label{eq:femto-delay1}
\end{equation}
where $\mathbbm{1}(n,u,j)$ is an indicator function taking value 1 if and only if under caching policy $\bm{y}$ file $n$ is placed at cache $\sigma_{u(j)}$, and not at any of the caches $i=1,2,\ldots, \sigma_{u(j-1)}$; and $\mathbbm{1}(n,u,0)$ denotes that the file is not cached in $\mathcal{V}_{c}^{0}(u)$. That is:
\begin{equation}
\mathbbm{1}(n,u,j)\!=\!\Big[\prod_{i=1}^{j-1}(1-y_{n,\sigma_{u(i)}} ) \Big] y_{n,\sigma_{u(j)}},~~
\mathbbm{1}(n,u,0)\!=\!\Big[ \prod_{i\in \mathcal{V}_{c}^{0}(u)}\big( 1-y_{n,\sigma_{u(i)}} \big) \Big]\,. \nonumber
\end{equation}
An equivalent formulation for (\ref{eq:femto-delay1}) that will be proved useful is:
\begin{align}  
D_{u}(\bm{y})=&\sum_{n\in \mathcal{N}_{ \sigma_{u(1)}}}\hspace{-0.15in}\lambda_nd_{\sigma_{u(1)},u} +\hspace{-0.052in} \sum_{n\in \mathcal{N}_{\sigma_{u(2)}}\setminus \mathcal{N}_{\sigma_{u(1)}}}\hspace{-0.2in}\lambda_nd_{\sigma_{u(2)},u} + \label{eq:femto-obj2} \\ &\sum_{n\in \mathcal{N}_{\sigma_{u(3)}}\setminus\{\mathcal{N}_{\sigma_{u(1)}}\bigcup \mathcal{N}_{\sigma_{u(2)}}\}}\hspace{-0.6in}\lambda_nd_{\sigma_{u(3)},u}+  \ldots
	+\hspace{-0.12in} \sum_{n\not\in \bigcup_{j=1}^{{V}_{c}^{0}(u)}\mathcal{N}_{\sigma_{u(j)}}}\hspace{-0.3in}\lambda_nd_{0,u}\,, \nonumber
\end{align}
where each term represents the aggregate delay experienced by $u$ when downloading files from the cache with the $j$th lowest delay, i.e., $v=\sigma_{u(j)}$, under the condition file $n$ is not cached in any other cache with smaller delay. 

Our goal is to minimize the average per-byte file delivery delay for all users, which can be equivalently expressed as maximizing the delay savings when using the SBSs instead of MBS: 
\begin{opt}{Uncapacitated Caching (C1)}
	\vspace{-0.2in}
	\begin{align}
		\max_{\bm{y}\in\mathcal{Y}_{C_1} }\qquad &J(\bm{y})=\sum_{u \in\mathcal{V}_r}\big( d_{0,u}- D_u(\bm{y}) \big) 	\nonumber 
	\end{align}
	\vspace{-0.2in}
\end{opt}
\noindent It is important to note that (C1) involves only caching decisions. This is possible because each user obtains the requested file from the cache that can deliver it with the smallest delay, and hence the routing decisions are directly determined by the caching policy. However, when the links are congestible or when there are hard capacity constraints, one needs to jointly (and explicitly) optimize the routing and caching decisions.


\subsubsection{Computational Intractability}

Given the theoretical and practical importance of the femtocaching problem, it is crucial to characterize its complexity. Unfortunately, (C1) is computationally intractable as the following theorem states.
\begin{thm}[theorem style=plain]{}{femto}
 Problem (C1) is NP-Complete.
\end{thm}

To prove this result, it suffices to consider a simplified version of the problem and use a reduction from the 2-disjoint set cover problem which is known to be NP-complete \cite{disjoint-proof}, and also prove that there is a polynomial-time verifier. In particular, consider an instance of (C1) where all delay parameters are equal, i.e., $d_{v,u}=d<d_{0,u}$, $\forall\,(v,u)\in\mathcal{E}$. In this case, we can define $d_u=d_{0,u}-d$ and rewrite (\ref{eq:femto-delay1}) as:
\begin{equation}
J(\bm{y})=\sum_{n=1}^N \lambda_n\sum_{u=1}^{V_r} d_u\mathbbm{1}(n,u),\, \text{where}\,\,\, \mathbbm{1}(n,\!u)\!=\!1-\prod_{v\in\mathcal{V}_{c}^{0}(u)}(1-y_{v,n})\,. \nonumber
\end{equation}
The indicator function $\mathbbm{1}(n,u)$ assumes value $1$ when at least one cache $v\in\mathcal{V}_{c}^{0}(u)$ has the file $n$ requested by user $u$. The last summation can be interpreted as the \emph{value} of each user, which is proportional to the probability of finding a file at the caches, multiplied by the respective delay savings. In order to prove our complexity claim, we consider the corresponding \emph{decision} problem. 

\textbf{Cache Decision Problem (CDP)}: Given a CN 
$G=({\mathcal{V}_c \cup \mathcal{V}_r},$ ${\cal E}, 0, \bm{B}, \bm{\lambda}, \infty, \bm{d})$, the catalog $\mathcal{N}$, the union $\mathcal{A}_u$ of the cached files at caches $\mathcal{V}_{c}^0(u)$ in range with $u$, and a number $Q\geq 0$, determine if there exists a policy $\bm{y}\in\mathcal{Y}_{C_1}$  such that:
\begin{align}
	\sum_{u \in\mathcal{V}_r} d_u\!\!\sum_{n\in\mathcal{A}_u}\!\lambda_n\geq Q\,. \label{eq:hlp-proof-object}
\end{align}
Let the above problem be denoted by $CDP({G}, \mathcal{N}, Q)$. Given a policy $\bm{y}$, we may verify in polynomial-time whether \eqref{eq:hlp-proof-object} is true, therefore $CDP$ is in NP. We will now use a reduction from the following NP-complete problem. 

\textbf{2-Disjoint Set Cover Problem (2DSC)}: Consider a bipartite graph $G_{sc}=(A,\Gamma,E)$ with edges $E$ between two disjoint vertex sets $A$ and $\Gamma$. For $b\in \Gamma$, define the neighborhood of $b$ as $\mathcal{N}(b)\subseteq A$. Clearly it is $A=\bigcup_{b\in \Gamma}\mathcal{N}(b)$. Do there exist two disjoint sets $\Gamma_1, \Gamma_2\subset \Gamma$ such that $|\Gamma_1|+|\Gamma_2|=|\Gamma|$ and $A=\bigcup_{b\in \Gamma_1}\mathcal{N}(b)=\bigcup_{b\in \Gamma_2}\mathcal{N}(b)$? Let this problem be denoted by $2DSC(G_{sc})$. 

2DSC is NP-complete \cite{disjoint-proof}, and we show in the following lemma that given a unit time oracle for the $CDP$ we can solve it in polynomial time.

\begin{lem}[theorem style=plain]{}{lem1}
2-Disjoint Set Cover $\leq_{P}$ Cache Decision Problem.
\end{lem}
\begin{proof}
Consider an instance of the $2DSC(G_{sc})$ problem and the CDP oracle algorithm. First, build an $CDP$ instance in the following manner. Set ${G}=G_{sc}$, $\mathcal{N}=\{1,2\}$, $\bm{\lambda}=\big(\frac{1}{1+\epsilon}$, $\frac{\epsilon}{1+\epsilon} \big)$ for some $\epsilon<1$, $\bm{d}=\{1,1,\ldots,1\}$, $B=1$, $Q=V_r$. For this $CDP$ instance, note that for any user $u$, $\mathcal{A}(u)$ can be $\{ \,\,\}, \{1\}, \{2\}, \{1,2\}$. Also, observe that the value of any user is $\sum_{n\in\mathcal{A}_u}\lambda_n\leq 1$ and it is exactly equal to $1$ only if $\mathcal{A}_u=\{1,2\}$. However, since $B=1$, any cache can cache at most one file $n=1$ or $n=2$. It follows that the objective is equal to $V_r$ only if $\mathcal{A}_u=\{1,2\}, \forall u$, implying that every user $u$ can reach at least one cache with file $1$ and another with file $2$. 

On the $CDP$ instance constructed above, call the oracle which returns a YES or NO in polytime. If it is a YES, we use the solution (content placement) and create a partition for $2DSC(G_{sc})$ by putting in set $\Gamma_1$ all caches with file $1$ and in set $\Gamma_2$ those with file $2$. These two sets then form a 2-disjoint set cover and we may conclude that the answer to $2DSC(G_{sc})$ is YES. If the oracle returns a NO for CDP, it follows that there is no content placement that will allow us to deliver both files to all users and achieve $Q=V_r$. It follows that there exist no partition of caches to sets $\Gamma_1$, $\Gamma_2$ such that every location is covered by both $\Gamma_1$ and $\Gamma_2$, therefore the answer to $2DSC(G_{sc})$ is NO. See Fig. \ref{fig:set-cover} for an example.
\end{proof}

\begin{figure}
	\centering
	\includegraphics[width=0.5\textwidth]{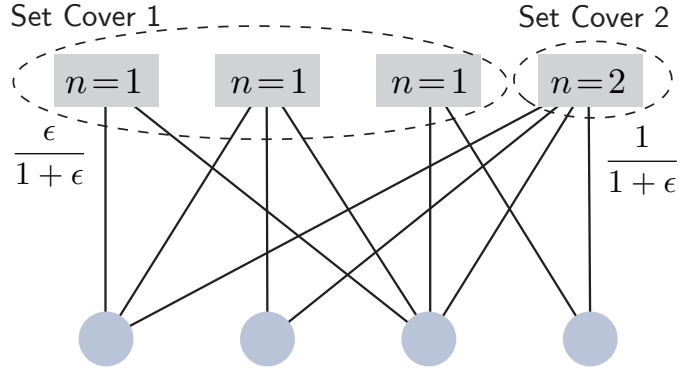}
	\caption{\small{Reduction of CDP to the 2-disjoint set cover problem. An example with 4 users and 4 caches.}}
	\label{fig:set-cover}    
\end{figure}

Given this result, we can only hope for a good approximation algorithm. It turns out that (C1) can be re-formulated as the maximization of a submodular function subject to matroid constraints. This favorable structure can be exploited to devise a computationally efficient algorithm with constant approximation ratio.

\subsubsection{A Greedy Algorithm with Constant Ratio}

First, we show that the caching constraints are independent sets of a matroid, see the box below for an introduction to matroids.

\begin{box_example}[detach title,colback=blue!5!white, before upper={\tcbtitle\quad}]{Matroids}
	\small
	extend the concept of independence to general sets. Informally, given a finite ground set $\mathcal{S}$, a matroid is a way to label subsets of $\mathcal{S}$ as ``independent''. In vector spaces, the ground set is a set of vectors, and subsets are called independent if their vectors are linearly independent in the usual linear algebraic sense. Formally, \cite{wolseybook88} a matroid $\mathcal{P}$ is a tuple $\mathcal{P}=(\mathcal{S}, \mathcal{I})$, where $\mathcal{S}$ is a finite ground set and $\mathcal{I}\subseteq 2^{\mathcal{S}}$ (a subset of the power set of $\mathcal{S}$) is a collection $\{X,Y,\ldots\}$ of subsets of $S$ that are \emph{independent}, where this property admits a problem-specific interpretation, such that:
	\begin{enumerate}[leftmargin=4.5mm]
		\item $\mathcal{I}$ is nonempty. 
		\item $\mathcal{I}$ is downward closed; i.e., if $Y\in \mathcal{I}$ and $X\subseteq Y$, then $X\in\mathcal{I}$. 
		\item If $X, Y\in\mathcal{I}$, and $|X|<|Y|$, then $\exists y\in Y\setminus X$ such that $X\cup \{y\}\in \mathcal{I}$.
	\end{enumerate}
	An important class of matroids for caching are \textbf{partition matroids}, where the ground set is partitioned into (disjoint) sets $\mathcal{S}_1, \mathcal{S}_2, \ldots, \mathcal{S}_l$ and each independent set $Y$ has a limited overlap with all partitions, i.e.,
	\begin{equation}
	\mathcal{I}=\{ Y\subseteq \mathcal{S}: |Y\cap\mathcal{S}_i|\leq k_i,\,\forall\,i=1,\ldots,l \}\,, \label{eq:partition-matroid}
	\end{equation}
	where $k_i>0$ is the overlap bound allowed for partition $\mathcal{S}_i$. Note that this definition can be used to express the constraint that each caching policy needs to respect the storage capacities, which are modeled with parameters $k_i:i=1,\ldots,l$. When the caches have equal capacity $k_i=k,\forall i$, then we can use as well the \textbf{uniform matroid}:
	\begin{equation}
	\mathcal{I}=\{ Y\subseteq \mathcal{S}: |Y|\leq k \}. \label{eq:uniform-matroid}
	\end{equation}
	The key benefit of matroids is that any maximal independent set is also maximum; and as a result, finding maximal independent sets can be achieved by greedy algorithms. 
\end{box_example}

For our femtocaching problem we define the following ground set:
\begin{align}
{S}=\{ s_{1}^1, s_{2}^1, \ldots, s_{N}^1; \ldots; s_{1}^{V_{c}^0}, s_{2}^{V_{c}^0}, \ldots, s_{N}^{V_{c}^0} \},\label{eq:femto-matroid1}
\end{align}
where $s_{n}^v$ is an abstract element denoting the placement of file $n$ into cache $v$. Set $S$ can be partitioned into $V_{c}^0$ disjoint subsets, $S_1, \ldots, S_{V_{c}^0}$, where $S_v=\{s_{1}^v, s_{2}^{v}, \ldots, s_{N}^v\}$ is the set of all files that might be placed at cache $v$.
\begin{lem}[theorem style=plain]{}{lem2}
The constraints of (C1) can be written as a partition matroid on the ground set $S$.
\end{lem}
\begin{proof}
In (C1) a content caching policy is expressed by matrix $\bm{y}$, and we define the respective caching set ${Y}\subseteq S$ such that $s_{n}^v\in Y$ if and only if $y_{v,n}=1$. Notice that the nonzero elements of the $v$th column of $\bm{y}$ correspond to the elements in $Y\cap S_v$. Hence, the eligible caching policies can be described as elements of the set $\mathcal{I}$, where:
\begin{equation}
	\mathcal{I}=\{Y\subseteq S: |Y\cap S_v |\leq B,\,\,\forall\, v=1,\ldots,V_c \} \label{eq:femto-matroid2}
\end{equation}
Comparing $\mathcal{I}$ and the partition matroid definition \eqref{eq:partition-matroid}, we see that our constraints form a partition matroid with $l=V_c$ and $k_i=B$, for $i=1,\ldots,V_{c}^0$. The partition matroid is denoted by $\mathcal{P}=(S,\mathcal{I})$. 
\end{proof}


Using now (\ref{eq:femto-obj2}), we can rewrite the average delay expression as a set function. This will allow us to show that (C1) has a submodular objective function. We define for each user $u$: 
\begin{align}  
D_{u}(Y)&=d_{\sigma_{u(1)},u}\sum_{n}\lambda_n Y_{\sigma_{u(1)}}^b(n) + d_{\sigma_{u(2)},u}\sum_{n}\lambda_n\big[Y_{\sigma_{u(2)}}^b(n)\wedge \bar{Y}_{\sigma_{u(1)}}^b(n)\big] \nonumber \\
&+ d_{\sigma_{u(3)},u}\sum_{n}\lambda_n\big[Y_{\sigma_{u(3)}}^b(n)\wedge \bar{Y}_{\sigma_{u(2)}}^b(n)\wedge \bar{Y}_{\sigma_{u(1)}}^b(n)\big]+  \ldots\nonumber \\
& + d_{0,u}\sum_{n}\lambda_n\big[\bar{Y}_{\sigma_{u(|\mathcal{V}_c(u)|-1)}}^b(n) \wedge \ldots\wedge \bar{Y}_{\sigma_{u(1)}}^b(n) \big]\,,\label{eq:femto-delay-2}
\end{align}
where $\wedge$ is the ``AND'' operation, $Y_{v}^{b}$ is the boolean representation of $Y_v = Y\cap\mathcal{S}_v$ and $\bar{Y}_{v}^b$ is the complement of $Y_{v}^{b}$. The $n$th element of $Y_{v}^b$ is denoted by $Y_{v}^b(n)$ which means that if $s_{n}^{v}$ is included in set $Y_v$, then $Y_{v}^b(n)=1$, and equal to $0$ otherwise. Using \eqref{eq:femto-delay-2} we can write:
\begin{equation}
	J(Y)=\sum_{u\in\mathcal{V}_r}\big(d_{0,u} -  D_{u}(Y) \big),
\end{equation}
that expresses function $J(\cdot)$ as a set function. The following lemma characterizes $J(y)$ \cite{femtocachingTransIT13}.
\begin{lem}[theorem style=plain]{}{lem3}
The objective of (C1) is a monotone submodular function.
\end{lem}
Given the above properties of the objective and constraints of (C1), we can use a \emph{greedy} algorithm that myopically selects in each round the action that maximizes the objective. The algorithm starts with an empty set and at each iteration adds the element from the ground set with the highest marginal value while maintaining the solution feasibility.

The detailed steps are presented in Algorithm \ref{algo:greedy-submodular}. First we find the element from all currently available ground elements that has the largest contribution in our objective function (line 3). Selecting $s_{n^*}^{v^*}$ means that we need to place file $n^*$ at cache $v^*$, and then update the set $Y_{v^*}$ of this cache (line 5) and the overall policy $Y$ (line 6). At the same time, we remove this element from the set of available elements $Z$ (line 7), and we keep track of the available elements at cache $v^*$ by updating $Z_{v^*}$ (line 8). In case the capacity of the selected cache $v^*$ is exhausted, we reduce the ground set by removing all elements of $Z_{v^*}$ (line 9) to save computation time in the next iteration. These steps are repeated $B\cdot V_{c}^{0}$ times, i.e., as many as the number of files that we can cache at caches. Since the objective function is submodular, the marginal value of elements decreases in each iteration. Thus, the algorithm terminates at the iteration where the largest marginal value is zero. Each iteration would involve evaluating marginal value of at most $N\cdot V_c$ elements and takes $O(V_r)$ time. Hence the running time would be $O(N^2V_c^2V_r)$.

\IncMargin{1.3em}
\begin{algorithm}[t]
	\small
	\nl \textbf{Initialize}: $Z_v\leftarrow S_v,\, Y_v\leftarrow \emptyset,\,\,\forall\,v=1,\ldots,V_c$; $Z\leftarrow S$; $Y\leftarrow\emptyset$.\\%
	\nl \For{ $i=1,2,\ldots, B\!\cdot\!V_c$  }{
		\nl Find the placement maximizing the marginal utility:\\
		$s_{n^*}^{v^*}=\arg\max_{z\in Z}J_{Y}(z)=J(Y\cup\{z\})-J(Y)$;\\
		\nl \If{ $J_{Y}(s_{n^*}^{v^*})=0$}{Terminate;}
		\nl $Y_{v^*}\leftarrow Y_{v^*} \cup \{s_{n^*}^{v^*}\}$\,; \,\,\,\,\,\, \footnotesize{\%Update the cached files of cache $v$}. \\%
		\nl $Y \leftarrow Y \cup \{s_{n^*}^{v^*}\}$\,; \qquad\,\,\,\,\,\,\, \footnotesize{\%Update the caching policy}. \\%
		\nl $Z \leftarrow Z\setminus \{s_{n^*}^{v^*}\}$; \\%
		\nl $Z_{v^*} \leftarrow Z_{v^*}\setminus \{s_{n^*}^{v^*}\}$; \\%
		\nl \If{$|Y_{v^*}|=B$ }{ $Z\leftarrow Z\setminus Z_{v^*}$ ; \,\,\,\,\,\,\,\,\,\,\,\,\,\, \footnotesize{\%Remove the elements of caches with no cache space}. \\}}
	\nl Output the optimal content placement set $Y$.\\%
	\caption{Greedy Bipartite Caching Algorithm}\label{algo:greedy-submodular}
\end{algorithm}\DecMargin{1em}
\normalsize

Despite its simplicity, Algorithm 1 achieves a quite satisfactory result as the following Theorem states.
\begin{thm}[theorem style=plain]{Greedy Approximation}{greedy_approx}
	The Greedy Bipartite Caching algorithm achieves a result within a factor of 1/2 of the optimal solution for (C1).
\end{thm}
This is a classical approximation result for this type of maximization problem, see \cite{wolsey1978}. Furthermore, a randomized algorithm which gives a $(1-1/e)$-approximation has been proposed in \cite{chekuri2007}, that consists of two parts. First, the combinatorial integer programming problem is replaced with a continuous one that is solved optimally. Then, the obtained feasible point is rounded using \emph{pipage rounding} \cite{pipage}, a technique that achieves tight approximations through integer relaxation, and has been proved particularly useful in caching problems. However, the implementation of this algorithm is computationally demanding and increases fast with quantity $V_c\cdot N$.

\subsection{Capacitated Femtocaching} \label{sec:capacitated-femtocaching}


The femtocaching problem becomes more challenging when the users are connected to the SBSs via links with limited bandwidth. In this capacitated case, routing to a cache is not to be taken for granted, since apart from the placement of the requested file, sufficient bandwidth for its delivery must also be secured \cite{Poularakis2014Approximation}. The arising problem is similar to the capacitated facility location (CFL) problem discussed in Sec. \ref{sec:cn-design}, yet more challenging since a cache can serve a request only if it has enough capacity to store the file, and enough link bandwidth to transmit it. In other words, it is a two-dimensional hard-capacitated problem, and its mapping to CFL is not straightforward. As a matter of fact, if there were no bandwidth limitations, one could consider this caching problem as a collection of Uncapacitated Facility Location (UFL) instances (one per file) that are coupled due to the storage constraint of each cache \cite{BaevSIAM08}. In the following, we use here the Unsplittable Hard-Capacitated Facility Location (UHCFL) formulation in order to devise a joint caching and routing policy.

\subsubsection{Problem Definition}

We begin with (C1) and further add the consideration of finite transmission capacity. The set of eligible caching policies is:
\begin{equation}
	\mathcal{Y}_{CR1}\!=\!\{ \bm{y} \!\in\! \{0, 1\}^{V_{c}^0\times N}\,\vert\,\sum_{n \in \mathcal{N}} y_{v,n}\! \leq\! B_v,\,\, v\!\in\!\mathcal{V}_{c}^0 \}, \nonumber
\end{equation}
which differs slightly from $\mathcal{Y}_{C_1}$ as each cache $v\in\mathcal{V}_{c}^0$ may have different storage capacity of $B_v$ files. Every SBS $v$ has a bounded transmission capacity and can deliver $K_v\geq 0$ data bytes within a time period $T$. Unlike (C1), we consider here heterogeneous demand and denote with $\lambda_{u,n}$ the expected number of requests for each file $n$ in the catalog $\mathcal{N}$, generated by user $u\in\mathcal{V}_r$ during $T$.\footnote{We simplify the presentation by assuming that files have unit size and that $T\!=\!1$, but the model and results can be easily generalized for any file size $s$ and $T$; see \cite{Poularakis2014Approximation}. } A request can be satisfied by any cache in $\mathcal{V}_c(u)$, Fig. \ref{fig:capacitated_femtocaching}, but the SBSs are preferable since they introduce smaller delay $d$ than the MBS, i.e., $d\!<\!d_0$.

\begin{figure}[t]
	\centering
	\includegraphics[width=0.75\textwidth]{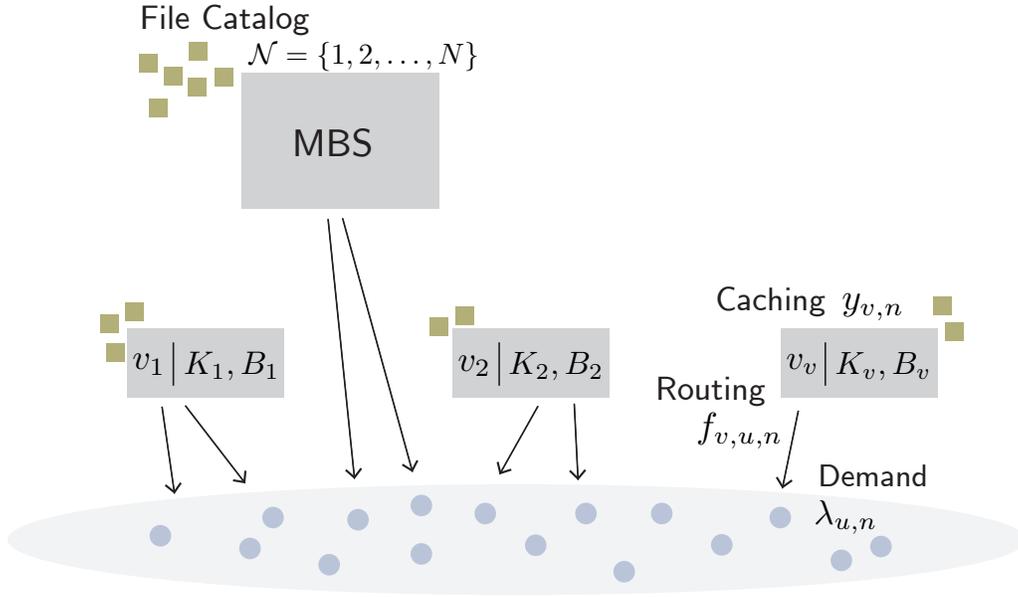}
	\caption{\small{System model for cache-enabled and bandwidth-constrained SBSs that underlay a macrocellular base station.}}
	\label{fig:capacitated_femtocaching}
\end{figure}

Due to the link capacities, we need to introduce routing decisions in order to explicitly determine how each request is satisfied. We consider the unsplittable routing model and denote with $f_{v,u,n}\in\mathcal{Z}^{+}$ the integer variable that decides the number of requests $(n,u)$ routed to SBS $v$, and with $f_{0,u,n}\in\mathcal{Z}^{+}$ those routed to MBS.\footnote{Since we have one-hop paths, we use directly pairs $(v,u)$ to index the routing decisions instead of the path indices $p\in\mathcal{P}_{v,u}$. } The {routing policy} is described by the vector $\bm{f}=\big(f_{v,u,n}: n\in\mathcal{N},v\in\mathcal{V}_c, u\in\mathcal{V}_r\big)$ which belongs to set:
\begin{align}
	&\mathcal{F}_{CR1}=\left\{ \bm{f}\in \mathcal{Z}_{+}^{V_c\times V_r\times N}\,\Bigg|\, \begin{array}{l}
	\sum_{v \in \mathcal{V}_c}  f_{v,u,n}\!=\!\lambda_{u,n}\\
	\sum_{u \in \mathcal{V}_r, n \in \mathcal{N}}\! f_{v,u,n} \leq K_v, v\!\in\!\mathcal{V}_{c}^{0}\\
	f_{v,u,n}\!=\!0,\, v\!\in\!\mathcal{V}_c \!\setminus\! \mathcal{V}_c(u), n\!\in\!\mathcal{N}, u\!\in\!\mathcal{V}_r
	\end{array}\right\},
\end{align}
where we have included constraints to ensure that all requests are satisfied (inelastic demand), each SBS cannot be assigned more demand than its bandwidth capacity,  and each SBS can serve only users within its range. 

Given the above, the problem of devising the joint routing and caching policy that minimizes the requests routed to MBS, is modeled as follows:
\begin{opt}{Caching and Routing (CR1)}
	\vspace{-0.2in}
 \begin{align}
 	\min _{\bm{y}\in \mathcal{Y}_{CR1},\,\,\bm{f}\in\mathcal{F}_{CR1} } \qquad  &\sum_{u \in \mathcal{V}_r} \sum_{n \in \mathcal{N}}  f_{0,u,n}   \nonumber \\
 	\text{s.t.}\qquad  & f_{v,u,n} \leq y_{v,n} \lambda_{u,n}, \text{ } n \in \mathcal{N}, u \in \mathcal{V}_r, v \in \mathcal{V}_{c}^{0}. \nonumber
 \end{align}
	\vspace{-0.3in}
\end{opt}
\noindent Note the coupling constraint between the routing and caching decisions, which ensures that requests are routed only to SBSs having cached the respective files, and do not exceed the total demand from each node $u$. Since the link delays are uniform, (CR1) is equivalent to a problem that minimizes the requests routed to MBS. Finally, it is easy to see that (CR1) is an NP-hard problem as it generalizes (C1) by incorporating bandwidth constraints.

\subsubsection{Reduction to FLP and Algorithms}

\begin{box_example}[detach title,colback=blue!5!white, before upper={\tcbtitle\quad}]{Unsplittable Hard-Capacitated Metric FLP (UHCMFL)}
	\small
	We are given a set $\mathcal{V}$ of locations, with a subset $\mathcal{V}_1\subseteq \mathcal{V}$ of locations at which we may open a facility, and a subset $\mathcal{V}_2\subseteq \mathcal{V}$ of client locations that must be assigned to an open facility. Let $\lambda_i\geq 0$ denote the demand of client $i\!\in\!\mathcal{V}_2$. Let $c_j$ denote the cost for opening a facility at location $j\in \mathcal{V}_1$, and $F_j$ its servicing capacity. Each client needs to assign its entire demand to a single facility. We denote $d_{ij}\geq 0$ the unit cost when facility $j$ serves one demand unit of client $i$. We assume that these costs form a \emph{metric}, i.e., they are non-negative, symmetric ($d_{ij}=d_{ji}$), and satisfy the triangle inequality: $d_{ij}+d_{jk} \geq d_{ik}$, $\forall i,j,k \in \mathcal{V}$.
	
	\vspace{2mm}
	
	Our goal is to determine the facility opening variables $x_j\in\{0,1\}$, $j\in\mathcal{V}_1$, and the client assignment variables $f_{ij}\in\{0,1\}$, $j\in\mathcal{V}_1$, $i\in\mathcal{V}_1$ so as to minimize the aggregate cost $Q$:
	\begin{equation}
	Q=\sum_{j \in \mathcal{V}_{1}} x_jc_j + \sum_{i \in \mathcal{V}_2}\sum_{j\in\mathcal{V}_1}\lambda_i d_{ij}f_{ij},
	\end{equation}
	while satisfying the demand, $\sum_{j\in\mathcal{V}_2}f_{ij}=1, \forall i$ and the facility capacity constraints, $\sum_{i\in\mathcal{V}_1}\lambda_if_{ij}\leq F_jx_j, \forall j$.
	
	\vspace{2mm}
	
	UHCMFL is NP-hard to approximate, and can be approximated within $O(\log V)$ distance if the capacities are violated by $(1+\epsilon)$, for some $\epsilon>0$; cf. \cite{bateni-uhcfl}. 
\end{box_example}

We will next establish the equivalence of (CR1) with a specific instance of the facility location problem. This will serve as the building block for our optimization framework, helping us to develop practical approximation algorithms based on existing FLP solutions. We start with the definition of the FLP instance that we will use here; see also \cite{bateni-uhcfl}.

In order to obtain a mapping between the two problems we will model caching a file as ``opening a facility''. However the connection between UHCMFL and (CR1) is non-trivial due to the co-existence of the bandwidth and storage constraints. Hence, we need to define demand in such a way so as to force that the opened facilities will respect the caching and link capacity constraints. In detail, we reduce any (CR1) problem 
to a specific UHCMFL instance, henceforth called $U_{CR1}$, by following the steps below:

\begin{itemize}[leftmargin=5mm]
\item	The set $\mathcal{V}_1$ consists of: (i) one facility named $a_{0}$ for the MBS, and (ii) a candidate facility location named $a_{vn}$ for every SBS $v\in\mathcal{V}_{c}^0$ and every file $n\in\mathcal{N}$. The capacity of $a_{0}$ is set to $+\infty$ and to $K_v$ for each $a_{vn}$, $\forall v\in\mathcal{V}_{c}^{0},\,n\in\mathcal{N}$. The opening cost is set to 0 for every facility location.
\vspace{-2mm}
\item	The set of clients $\mathcal{V}_2$ comprises subsets: (i) $\mathcal{V}_{2,1}$ that contains $\lambda_{u,n}$ clients $\forall u\in\mathcal{V}_r, n\in\mathcal{N}$, denoted as $b_{un1},\,b_{un2}\ldots,b_{un\lambda_{u,n}}$, (ii) $\mathcal{V}_{2,2}$, with $N-B_v$ clients, denoted $b'_{v1}$, $b'_{v2}$, etc., $\forall v\in\mathcal{V}_{c}^{0}$, and (iii) $\mathcal{V}_{2,3}$ which contains $(B_v -1)K_v$ clients, which are denoted $b''_{v1},\,b''_{v2}$ etc., $\forall v\in\mathcal{V}_{c}^{0}$. The demand of each client $b'_{vn}\in\mathcal{V}_{2,2}$ is equal to $K_v$. Each of the remaining clients $b_{unj}\in\mathcal{V}_{2,1}$ and $b''_{vn}\in\mathcal{V}_{2,3}$ has demand $1$. 
\vspace{-2mm}
\item Let $c$ be an arbitrarily small positive constant. Then, the unit serving cost for each facility-client pair is specified as follows: (i) each pair of type $(a_{0},\,b_{unj})$, $\forall u, n, j$, has cost $1+0.5+c$, (ii) each pair $(a_{vn},\,b_{unj})$, such that $v \in \mathcal{V}_c(u)$ and $j \in \{1,...,\lambda_{u,n}\}$, has cost $0.5+c$, (iii) each pair $(a_{vn},\,b'_{vj})$, $\forall v,n,j$, has cost $0.5+c$, (iv) each pair $(a_{vn},\,b''_{vj})$, $\forall v,n,j$ has cost $0.5+c$. The cost value for each one of the remaining pairs is equal to the cost of the shortest path that connect this pair. Thus, the costs form a metric. 
\vspace{-2mm}
\end{itemize}

Let us explain the intuition of this mapping. Facility $a_{0}$ represents the MBS, and facilities $a_{vn}$, $\forall n$ the SBS $v$. Hence, the facility capacity choices indicate that MBS can serve all requests, while each SBS $v$ can serve a certain number of them. Each client of type $b_{unj}\in\mathcal{V}_{2,1}$, $\forall v,n,j$ (whose demand is $1$) represents one user request, while $b'_{vn}\in\mathcal{V}_{2,2}$, and $b''_{vn}\in\mathcal{V}_{2,3}$, $\forall v,n$ denote \emph{virtual} requests that are designed to preserve the SBSs resource constraints. In particular, clients $b'_{vn}\in\mathcal{V}_{2,2}$ represent the files that cannot be cached in each SBS $v$, and there are $N-B_v$ such files. When a client $b'_{vn}$ with demand $K_v$ is associated to facility $a_{vn}$, it consumes all its capacity; which means that file $n$ will not be cached at $v$ and $a_{vn}$ is removed from the list of eligible facilities. On the other hand, $b''_{vn}$ clients ensure that the bandwidth constraint will be satisfied at SBS $v$. Without these clients the facilities violate $K_v$ since each file (if cached at $v$) is allowed to serve up to $K_v$ requests. This is achieved by the introduction of $(B_v -1){K_v}$ such virtual clients which will consume all the virtual bandwidth in the FLP instance. It is important to note that these virtual clients are indirectly connected to MBS through paths of high cost, and hence any FLP solution will opt to associate them with facilities of type $a_{vn}$. This choice of cost parameters ensures that the bandwidth and storage capacity constraints are satisfied.

After introducing the above, we can use any oracle for the $U_{CR1}$ to build a solution for the (CR1) problem, by following three simple rules:
\begin{tcolorbox}[boxrule=1.3pt,arc=0.6em, top=-0.4mm, bottom=-0.4mm] 
\begin{itemize}[leftmargin=0.51mm]
	\small
	\item Rule $1$: For each facility $a_{vn}$ \emph{not} serving any client of the form $b'_{vj}\in\mathcal{V}_{2,2}$, $\forall j$, place file $n$ to the cache of SBS $v$.
\vspace{-1.5mm}
	\item Rule $2$: For each facility $a_{vn}$ serving a client of the form $b_{unj}\in\mathcal{V}_{2,1}$, $\forall v,n,u,j$ route the $j^{th}$ request of user $u$ for file $n$ to SBS $v$.
\vspace{-1mm}	
	\item Rule $3$: The remaining requests are routed to MBS.
\end{itemize}
\end{tcolorbox}
\noindent Figure \ref{fig:FLP-mapping} presents an example. Squares represent the facilities and circles the clients. Users $u_1$ and $u_2$ request every file one time, while $u_3$ submits two requests for file $n=1$. Solid lines connect clients to facilities with cost $0.5\!+\!c$, and dashed lines model the links with cost $1\!+\!0.5\!+\!c$. The cost value of each of the remaining pairs is equal to the cost of the shortest path that unites this pair. For example, the cost between $b'_{11}$ and facility $a_{21}$ is $1.5\!+\!3c$ (3 hops over links with cost $0.5+c$). The demand of each virtual client is $1$, except for the clients $b'_{vn}\in\mathcal{V}_{2,2}, \forall v, n$, whose demand is $2$. The capacity of each facility is $2$, except for $a_{0}$ which is uncapacitated.

\begin{figure}[h!]
	\centering
	\includegraphics[width=0.99\textwidth]{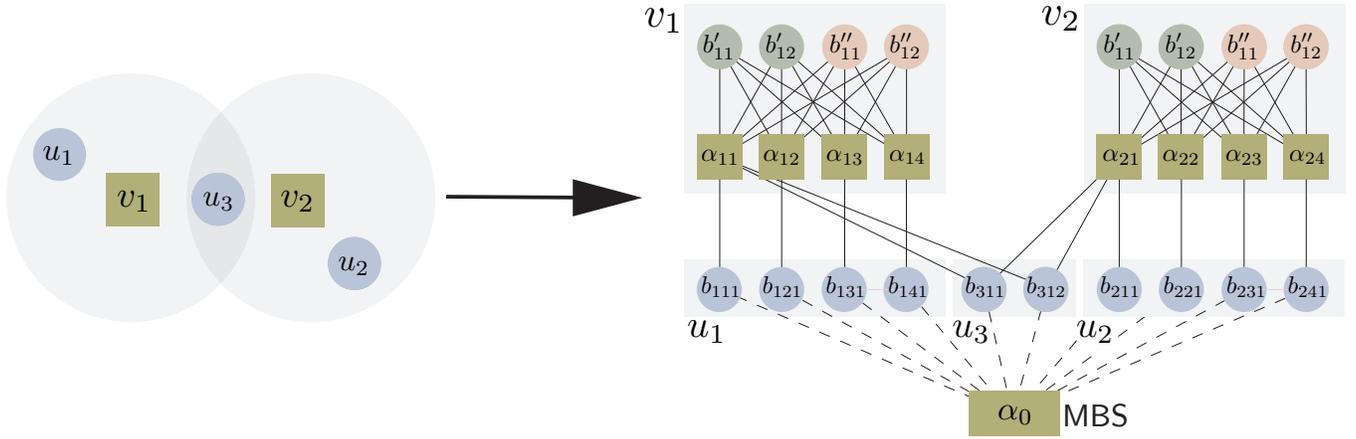}
	\caption{\small{An example of constructing a UHCMFL instance from a (CR1) problem with $N=4$ files, 2 SBSs with $B_1=B_2=2$ and $K_1=K_2=2$, and 3 users. Each SBS corresponds to the 8 components in the top of the UHCMFL instance ($N=4$, thus we need $4$ facilities and $4$ clients), each user $u$ to $\sum_{n=1}^4\lambda_{u,n}$ of the bottom clients, and the MBS is the facility $a_0$}.}
	\label{fig:FLP-mapping}
\end{figure}

We now prove that the preceding reduction holds, using two lemmas. Let $D$ denote the total demand of the clients in $U_{CR1}$. Then, we have:
\begin{lem}[theorem style=plain]{}{lem4}
	For every feasible solution of (CR1) with value $C$, there is a feasible solution to $U_{CR1}$ with total cost $C + D(0.5+c)$.
\end{lem}
\begin{proof}
We construct the solution to $U_{CR1}$ as follows:
\begin{enumerate}[leftmargin=5mm]
		\item We open all the facilities at zero cost.
		\item For each file $n$ \emph{not} cached at SBS $v$, we assign the (entire) demand of one client of type $b'_{vj}\in\mathcal{V}_{2,2}$, $j \in \{1,...,N-B_v\}$ to the facility $a_{vn}$.
		\item For each request generated by a user $u$ for a file $n$ served by an SBS $v$, we assign the demand of one client of type $b_{unj}\in\mathcal{V}_{2,1}$, $ j \in \{1,...,\lambda_{u,n}\}$ to the facility $a_{vn}$.
		\item The demand of a client of type $b''_{vj}\in\mathcal{V}_{2,3}$, $j \in \{1,...,(B_v-1)K_v\}$, is randomly assigned to one of the facilities of the form $a_{un}$, $\forall n\in\mathcal{N}$, without violating their capacity constraints.
		\item For each client $b_{unj}\in\mathcal{V}_{2,1}$ that has not been covered, we assign its demand to facility $a_{0}$. Thus, every demand unit of the clients was assigned to a facility. An assignment to facility $a_{0}$ induces a per unit cost $1+0.5+c$, while all other assignments induce a per unit cost $0.5+c$. By construction, the total demand assigned to $a_{0}$ is equal to the requests that are routed to MBS (C). Thus, the solution has cost equal to $C+D(0.5+c)$.
	\end{enumerate}\vspace{-0.3in}
\end{proof}
\vspace{-0.1in}
\begin{lem}[theorem style=plain]{}{lem5}
	For every minimum cost solution of the $U_{CR1}$ instance with total cost $C$, there is a feasible solution to the (CR1) problem with value $C-D(0.5+c)$.
\end{lem}\vspace{-0.2in}
\begin{proof}
	We construct the solution to (CR1) as follows:
	
	(i) \emph{For each facility $a_{vn}$ \emph{not} serving any client of the form $b'_{vj}\in\mathcal{V}_{2,2}$, $\forall j$, place file $n$ to the cache of SBS $v$} (Rule $1$). Observe that each client $b'_{vj}\in\mathcal{V}_{2,2}$, $\forall j$, must be assigned to a facility of the form $a_{vn}$, $\forall n$, at per unit cost $0.5+c$. This is because each of the other choices incurs at least $1+0.5+3c$ per unit cost. Thus, the extra cost paid is at least $1+2c$. On the other hand, each client $b_{unj}\in\mathcal{V}_{2,1}$, $\forall u, n, j$, can always be assigned to the facility $a_{0}$ at per unit cost $1+0.5+c$. This means that the potential gain for assigning it to a facility of the form $a_{vn}$, $\forall n\in\mathcal{N}$, at cost $0.5+c$, is equal to $1$, which is strictly lower than the extra cost paid above. We also observe that, the demand of each of these clients is equal to the capacity of each of the facilities of the form $a_{vn}$, $\forall n\in\mathcal{N}$. There are $N-B_v$ such clients. Thus, these clients fully occupy the capacity of $N-B_v$ of these facilities. Consequently, exactly $B_v$ of the above facilities will remain \emph{uncovered} corresponding to the files placed at the cache of SBS $v$.
	
	(ii) \emph{For each facility of the form $a_{vn}$ serving a client of the form $b_{unj}\in\mathcal{V}_{2,1}$, $\forall v,n,u,j$ route the $j$th request of user $u$ for file $n$ to SBS $v$} (Rule $2$). Observe that each of the clients of type $b''_{vj}\in\mathcal{V}_{2,3}$, $\forall j$, must be assigned to one of the $B_v$ \emph{uncovered} facilities of the form $a_{vn}$, $\forall n$, similarly to the above case. The capacity of each of these facilities is equal to $B_n$. There exist $(B_v -1)K_v$ such clients, each of them with demand equal to $1$. Thus, the remaining capacity suffices for serving at most $K_v$ units of demand of the clients $b_{unj}\in\mathcal{V}_{2,1}$, $\forall u,n,j$. By construction, a client $b_{unj}\in\mathcal{V}_{2,1}$ can be served by a facility $a_{vn}$ with cost equal to $0.5+c$ iff $v \in \mathcal{V}_c(u)$. The cost for serving $b_{unj}\in\mathcal{V}_{2,1}$ by $a_{vn}$, $\forall v \notin \mathcal{V}_c(0)$ is more than the serving cost by $a_{0}$. Thus, at most $K_v$ requests generated by users in the coverage area of an SBS $v$ will be routed to $v, \forall v\in\mathcal{V}_c$. The remaining $C-D(0.5+c)$ requests will be routed to the MBS (Rule $3$).
\end{proof}

We now present an approximation framework for (CR1) based on the above two-way reduction. Given that it is NP-hard to approximate the optimal solution of UHCMFL, previous work focused on obtaining \emph{bi-criteria} approximation algorithms \cite{bateni-uhcfl}, \cite{tardosSTOC97}. Formally, an $(\alpha,\beta)-$bi-criteria approximation algorithm finds an infeasible solution with cost at most $\alpha \geq 1$ times higher than the optimal cost, and aggregate demand assigned to each facility at most $\beta\geq 1$ times larger than its capacity. Similarly, we can define an $(\alpha,\beta)-$bi-criteria approximation algorithm for (CR1) such that its solution violates the SBS link capacities by a factor of $\beta$. Clearly, when $\beta=1$, a feasible solution is attained. In case the facilities have equal capacities (\emph{uniform} case) we can improve the approximation ratios \cite{Poularakis2014Approximation}.


Although (CR1) and UHCMFL problems are equivalent in terms of their optimal solution, the extension of approximation algorithms from one to the other is not straightforward. The following Theorem describes how the bi-criteria bound changes when translating the solution to handle CR1.
Let us define:
\begin{equation}
	c'=\frac{D(0.5+c)}{ \sum_{u \in \mathcal{V}_r} \sum_{n \in \mathcal{N}} (\lambda_{u,n})- \sum_{v\in \mathcal{V}_c}K_v }
	\label{eqn:c'}
\end{equation}
Then, we have the following theorem proved in \cite{Poularakis2014Approximation}:
\begin{thm}[theorem style=plain]{Capacitated Femtocaching with Violations \cite{Poularakis2014Approximation}}{poular-approxim1}
For any $(\alpha,\beta)-$bi-criteria approximation algorithm for the UHCMFL problem there is an $\big(\alpha + (\alpha-1)c',(\beta-1) N+1\big)-$bi-criteria approximation algorithm for the (CR1) problem, requiring the same computational complexity. 
\end{thm}


We can use this result to design approximation algorithms. For example, Theorem $\ref{thm:poular-approxim1}$ combined with the Algorithms presented in \cite{tardosSTOC97} provides bi-criteria approximation algorithms for (CR1), see also Fig. \ref{fig:FLP-mapping}. However, as the bandwidth capacities of the SBSs may be violated with this approach, the designer may need to endow the base stations with additional bandwidth in order to ensure the described approximation ratio. In many cases, it is not possible to perform additional investments and hence a fraction of the requests that reach an SBS will be rerouted to the MBS, further increasing its load. This additional cost can be quantified using the following Theorem that characterizes the worst case scenario:
\begin{thm}[theorem style=plain]{Capacitated Femtocaching without Violations \cite{Poularakis2014Approximation}}{TCOM-3}
	For any $(\alpha,\beta)-$bi-criteria approximation algorithm for the UHCMFL problem there is an $(\alpha + (\alpha-1)c') c''(\beta)-$approximation algorithm for the (CR1) problem, requiring the same computational complexity, where:
	\begin{equation}
			c''(\beta)=\frac{\sum_{u \in \mathcal{V}_r} \sum_{n\in\mathcal{N}}(\lambda_{u,n})-\sum_{v \in \mathcal{V}_c} K_v} {\sum_{u \in \mathcal{V}_r}\sum_{n \in \mathcal{N}}(\lambda_{u,n})-((\beta -1)N+1)\sum_{v\in \mathcal{V}_c}K_v}
	\end{equation}
\end{thm}
\noindent The detailed proof is provided in \cite{Poularakis2014Approximation} that also discusses the special case of SBSs with identical bandwidth capacities.

\subsection{Discussion of Related Work}




There are several works on the (un)capacitated femtocaching problem, including \cite{shakkottai2014} that considers a dynamic model and stabilizes the request queues, and \cite{massoulieATC15} which minimizes costs through load-balancing and content replication. The idea of creating a femtocaching network through leased caches was proposed in \cite{poularakisTNSM16}, and further extended in \cite{tasos-jsac18} which investigated the interplay between the user association rule and the caching policy. In \cite{poularakisTWC16}, the authors proposed the joint design of MBS multicast transmissions and SBSs caching policies, noting that both techniques aim to exploit the recurrence of user requests either in space (caching) or in time (multicast). Finally, the femtocaching model has been  extended to deployment of services (instead of content) at edge servers, see \cite{hong-hou-jsac} and references therein.

It is interesting to note the existence of different methods for using facility location theory to tackle caching problems. Here, we modeled every cached item as a facility \cite{Poularakis2014Approximation}, but it is possible to use the FLP variant with service installation costs \cite{shmoysSODA04}, where facilities model the caches and services represent the cached files. For an additional discussion about the application of FL theory to CN problems we refer the reader to \cite{BaevSIAM08}.

The bipartite CN model can be used for networks that do not have, at-a-first-glance, this structure. In fact, any arbitrary network where the routing paths do not have hard capacity constraints or load-dependent costs, can be transformed to an equivalent uncapacitated bipartite network, see Fig. \ref{fig:bipartite_model}. Moreover, one can use various criteria for selecting the paths when building the equivalent bipartite graph. For example, paths that exhibit longer delay than an agreed threshold can be excluded from the new graph, see e.g., \cite{RossComCom02}. Using this transformation, one can apply the bipartite greedy caching algorithm on the transformed graph, and find a caching policy with bounded optimality gap for the initial network when the objective satisfies the necessary properties. And, finally, there are also sophisticated approaches that extend this model to caching networks with link capacities, see, e.g., \cite{fleischerSoda06}. A question that remains is whether in case of specific graphs we can actually guarantee a better result than the $50\%$ of greedy and the $63\%$ of pipage routing. We explain in the next section that for tree network graphs, it is possible to improve the approximation results and/or obtain algorithms with lower complexity.

\begin{figure}[t]
	\centering
	\includegraphics[width=0.95\textwidth]{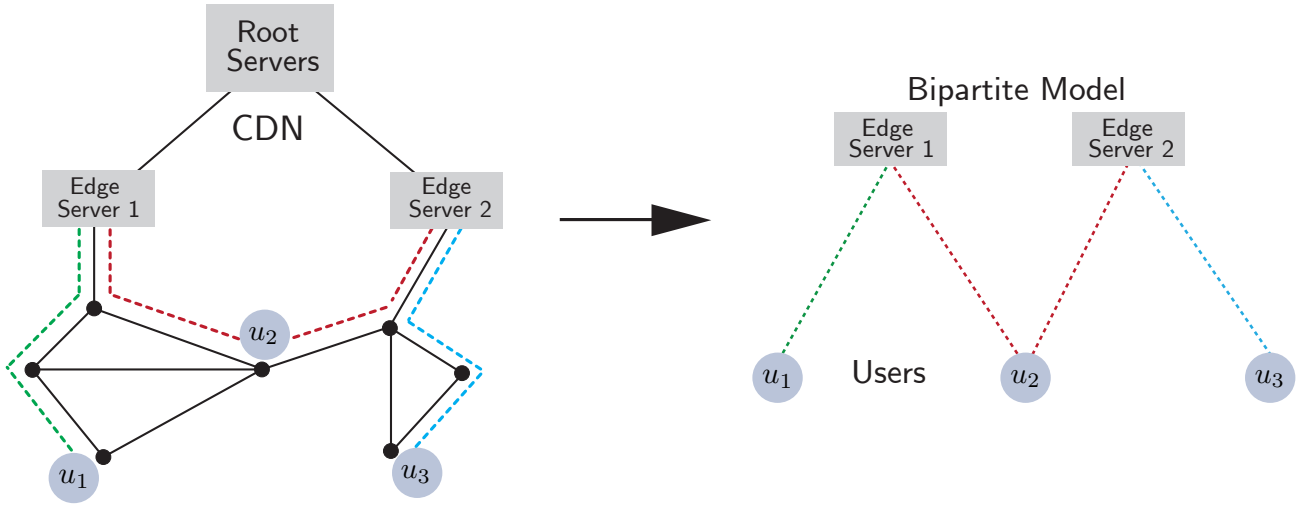}
	\caption{\small{The bipartite caching network can be used to model other CNs that do not have a bipartite structure but their operation is based on a threshold rule, such as a maximum delay for the content delivery. Here, we present an example of a content distribution network (CDN) that can be modeled with a bipartite graph. Users $1$ and $3$ can be only served by the respective nearby edge servers (caches), while user $2$ can be served by both servers.}}
	\label{fig:bipartite_model}
\end{figure}


\section{Hierarchical Caching Networks}



An important class of CNs arises when the network graph has a tree structure, i.e., when caches are organized in layers and each cache is only connected to caches in the layers immediately above and below. Examples of such caching networks are CDNs, IPTV, and cellular networks \cite{paschos-jsac}. These tree networks have the following interesting features: (i) the demand emanates from the leaf nodes; (ii) an origin server is always placed at the root; (iii) a request that is not served by a cache at a low layer is routed at a higher layer (i.e., towards the root) at the expense of additional routing cost; (iv) caching contents at the lower layers of the hierarchy reduces the delivery cost  but  serves less requests, while at the higher layers we have more costly delivery  but can potentially serve more users. The tree network structure facilitates the deployment of caches as we already discussed \cite{krishnan-ToN00, GuhaFOCS00, RobertTPDS08}, but also the design of caching and routing policies as we will explain in detail in this section.



We also consider here the general case where the leaf caches may also be connected to each other, hence having an hierarchical but not necessarily a tree structure. And we will consider both the case of cooperative networks where leaf nodes can exchange files directly (if connected) or indirectly through their parent nodes; and non-cooperative networks where the content can be fetched only from higher layer caches.

%

\subsection{Minimizing Routing Costs}

\begin{figure}[t]
	\centering
	\includegraphics[width=1.0\textwidth]{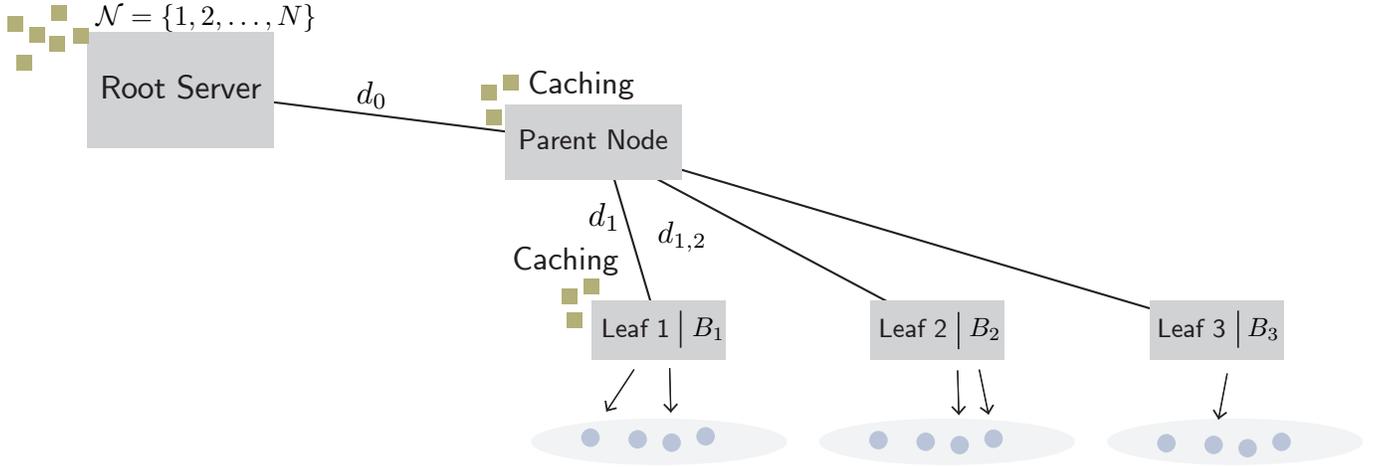}
	\caption{\small{An hierarchical caching network with three tiers. The demand is generated at the leaf nodes $\mathcal{V}^0$. Each node $v$, including the parent, has limited storage capacity of $B_v$ files. The root server stores the entire catalog $\mathcal{N}$. An item is delivered from the parent node to leaf $1$ with cost $d_1$, from root to parent node with cost $d_0$, and from leaf 1 to leaf 2 with cost $d_{1,2}$.}}
	\label{fig:hierarchical_model}
\end{figure}

We start with the simplest non-trivial hierarchical CN, with 3 layers comprising 1 root server, 1 parent node, and $V$ leaf nodes. The goal is to find the routing and caching policy minimizing the aggregate routing cost, see also \cite{BorstInfocom10}. We will discuss both the cooperative and non-cooperative model. In detail, this CN has $\mathcal{V}^{0}$ \emph{leaf} nodes\footnote{In this model all nodes, except the parent, create content requests; hence we simplify notation and use $\mathcal{V}$ (and $\mathcal{V}^0$) instead of $\mathcal{V}_c$, and also omit set $\mathcal{V}_r$.} connected to a single parent node ($v\!=\!0$) which in turn is linked to a distant root server ($v\!=\!-1$) that stores the entire catalog, same as in Fig. \ref{fig:hierarchical_model}. The files have different sizes $s_1,\ldots, s_N$, and $\lambda_{u,n}$ is the demand for the $n$th file at leaf $u$. Parameter $d_0$ is the routing cost for transferring a unit of file from the root to parent node, $d_u$ the cost for transferring a unit of file from the parent node to $u$, and $d_{u,v}$ the respective cost when it is delivered to $v$ through leaf $u$. We assume that it is cheaper to fetch a file from another leaf than from the root, i.e., $d_{u,v}\leq d_0+d_u$, $\forall u,v\in\mathcal{V}^0$.

The caching policy $\bm{y}$ is drawn from set:
\begin{equation}
	\mathcal{Y}_{CR2a}=\{ \bm{y}\in\{0,1\}^{V\times N}\,\vert\, \sum_{n\in\mathcal{N}}s_n y_{v,n} \leq B_v,\,   v\in\mathcal{V}\}\,, \nonumber
\end{equation}
where $\mathcal{V}=\mathcal{V}^0\cup\{0\}$. And the set of eligible routing policies $\bm{f}$ belongs to:
\begin{equation}
	\mathcal{F}_{CR2a}=\{\bm{f}\in\{0,1\}^{V^0\times V\times N} \}\,. \nonumber
\end{equation} 
Note that the assignment of a requester $u$ to cache $v$ specifies also the data transfer route since there exists a unique path connecting any two nodes. Hence, we have replaced again the subscript $p$ in the routing variables with $(v,u)$.

Our goal is to minimize the routing cost, or equivalently, maximize the routing savings. The latter are defined as the costs we avoid paying when a file requested at $u\in\mathcal{V}^0$ is delivered through a nearby node $v\in\mathcal{V}$ (including the case $v\!=\!u$) instead of fetching it from the root server, a formulation we used also in (C1). For example, when $f_{v,u,n}=1$, the $\lambda_{u,n}$ requests are satisfied by $v$ with cost $d_{u,v}$, and we save cost $d_u+d_0$, per unit of file, that we would have paid if it was fetched from the root. The objective function can be thus written as $J_{CR2a}(\bm{f}, \bm{y})=$
\begin{equation}
\sum_{u\in\mathcal{V}^0} \sum_{n\in\mathcal{N}}  s_n\lambda_{u,n}\Big[ (d_0+d_u)y_{u,n}+c_0f_{u,0,n}+\sum_{j\in\mathcal{V}^{0}} (d_0\!+\!d_u\!-\!d_{u,v})f_{v,u,n}  \Big]. \nonumber
\end{equation}

Combining the above, we define the following  problem:
\begin{opt}{Hierarchical Caching and Routing (CR2a)}
	\vspace{-0.2in}
\begin{align}
\underset{ \substack{\bm{y} \in\mathcal{Y}_{CR2}\\ \bm{f}\in\mathcal{F}_{CR2a}}}\max   &J_{CR2a}(\bm{f}, \bm{y})   \nonumber \\
	\text{s.t.}\,\,\,\, & f_{v,u,n}\leq y_{v,n}, \,\,\,u\in\mathcal{V}^0, v\in\mathcal{V}, n\in\mathcal{N} \label{eq:xyconstraint-borst-1a} \\
	& y_{u,n}\!+\!f_{0,u,n}\!+\!\sum_{v\in\mathcal{V}^0}f_{v,u,n}\leq 1,\,\, u\in\mathcal{V}^0,\,n\in\mathcal{N}. \label{eq:xyconstraint-borst-1b}
\end{align}
	\vspace{-3mm}
\end{opt}
\noindent Constraints (\ref{eq:xyconstraint-borst-1b}) ensure that each request is satisfied one time. The reader might have noticed that the demand satisfaction constraint is set as equality in some cases and as inequality in others. The typical formulation is the former, yet we can use inequalities when the objective is to maximize a function that is increasing in $y$ (as the case of caching gain in (CR2a)). Another case where inequalities are valid is for elastic demand, where the CN might select not to satisfy some requests in order to save costs.


To facilitate the presentation we discuss next the symmetric case of (CR2a) where the leaf nodes are similar in terms of bandwidth costs, user demand and cache sizes, i.e., $d_u=d, d_{u,v}\!=\!d_{u}^{\prime}\!=\!d^{\prime}, \lambda_{u,n}\!=\!\lambda_n$, and $B_u\!=\!B$ for all $u\in\mathcal{V}^0$. The asymmetric case is discussed in \cite{BorstInfocom10}. First, we relax the integrality constraints and obtain a linear program which gives an upper bound on the bandwidth savings (or, a lower bound for the achievable routing cost). Let us define the following parameters and variables:
\begin{align}
	d^{\prime\prime}=V^0(d_0+d)-(V^0-1)d^{\prime}\,;\,\,\,\,\, r_n=\min\{1, y_{0,n}+\sum_{i\in\mathcal{V}^0}y_{u,n} \},\nonumber \\ p_n=r_n-y_{0,n}\,;\,\,\,\,\,q_n=(y_{0,n} +\sum_{i\in\mathcal{V}^0}y_{u,n}-r_n)/(V^0-1)\,, \nonumber
\end{align}
where $r_n$ is the fraction of file $n$ that is collectively stored at the leaf nodes and the parent node; variable $p_n\in[0,1]$ decides whether some fraction of $n$ is stored at the leaf nodes (this enables cooperative exchanges), and variable $q_n$ indicates whether $n$ is fully replicated across leaf nodes. Then, (CR2a) can be transformed to an equivalent Knapsack-type problem with the auxiliary continuous variables $\bm{p}$, $\bm{q}$ and the discrete caching variables $\bm{y}_0=(y_{0,n}:n\in\mathcal{N})$. In detail, we can write: 
\begin{opt}{Symmetric HRC (CR2b)}
	\vspace{-0.2in}
 \begin{align}
 \underset{\bm{p}, \bm{q} \in [0,1]^N, \bm{y}_0\in\{0,1\}^N}	\max \,\, &\sum_{n\in\mathcal{N}} s_n\lambda_{u,n}\Big[ d^{\prime\prime}p_n+d^{\prime}(V^0-1)q_n+V^0d_0y_{0,n}\Big]   \nonumber \\
 	\text{s.t.}\qquad &   \sum_{n\in\mathcal{N}}s_n y_{0,n} \leq B_0, \label{eq:storeconstraint-borst-2} \\
 	&  \sum_{n\in\mathcal{N}}s_n\big(p_n+(V^0-1)q_n\big)\leq V^0 B, \label{eq:xyconstraint-borst-2a} \\
 	&  p_n+y_{0,n}\leq 1,\,\,\,n\in\mathcal{N}, \label{eq:xyconstraint-borst-2b} \\
 	&  q_n+y_{0,n}\leq 1,\,\,\,n\in\mathcal{N}. \label{eq:xyconstraint-borst-2c}
 \end{align}
	\vspace{-0.35in}
\end{opt}
There are two knapsacks of sizes $B_0$ (cache of parent node) and $V^0B$ (aggregate leaf cache storage) and $2N$ items of sizes $a_n=s_n$, $a_{N+n}=(V^0-1)s_n$, $n=1,\ldots,N$. Items $N+1,\ldots,2N$ cannot be included in the first knapsack, by definition. The value of item $n$ when included in the first knapsack is $d_0V^0$, which is equal to the aggregate bandwidth savings for all leaf nodes that will receive this item from the parent and not the root node. On the other hand, the values of items $n$ and $N+n$ when included in the second knapsack are $d^{\prime\prime}$ and $(V^0-1)d^{\prime}$, respectively. The latter observation may be interpreted as follows. Storing item $n$ in the first of $V^0$ leaf nodes yields bandwidth savings of $s_n\lambda_nd^{\prime\prime}$ since this node will save $d+d_0$ bandwidth and all other leaf nodes will save $d+d_0-d^{\prime}$. Also, caching $n$ in each additional leaf node yields further bandwidth savings of $s_n\lambda_nd^{\prime}$ because this node will have to pay the cost $d^{\prime}$. When $y_{0,n}=1$, $p_n$ and $q_n$ must be zero, since serving a content request both from the parent and leaf nodes cannot be optimal.


\IncMargin{0.2em}
\begin{algorithm}[t]
\small
	\nl \textbf{Initialize}: Set arbitrary values for $y_{u,n}\in\{0,1\},\, n=1,\ldots,N$.\\%
	\nl \Repeat{}{
		\nl Consider the new request for $n$;\\%
		\nl Case (a): If $y_{u,n}=0$, item is fetched from leaf node $v$ with $y_{v,n}=1$; \\%
		\nl Case (b): If $y_{u,n}=0$ and $y_{v,n}=0$ $\forall v\in\mathcal{V}^0$, then $n$ is fetched from root; \\%
		\nl Set the \emph{utility} $w_{u,n}$ for item $n$, $w_{u,n}=d^{\prime}\lambda_n$ in (a), and $w_{u,n}=d^{\prime\prime}\lambda_n$ in (b); \\%
		\nl Compare $w_{u,n}$ with $w_{u,m}$ where $m$ is a cached item at $u$ offering currently the smallest utility; \\%
		\nl If $w_{u,n}>w_{u,m}$, then replace $m$ at $u$ with $n$. \\%
	}
	\caption{\small{Local Search Caching Algorithm}}\label{algo:greedy-borst}
\end{algorithm}\DecMargin{1em}
\normalsize

There are two interesting versions of the above problem: (i) when there is only intra-layer cooperation and the parent node does not have storage; and (ii) when the leaf nodes cannot cooperate. We can obtain the formulation for (i) by removing variables $y_{0,n}$ from the objective and constraints in (CR2b). Similarly, we obtain the problem formulation for case (ii) by removing variables $f_{v,u,n}$, $u,v\in\mathcal{V}^0$ from (CR2a). These problems can be conveniently solved by \emph{local search algorithms}, which are distributed and achieve the best known performance by iteratively searching for an improved solution in the ``neighborhood'' of the current cache configuration. 





Algorithm \ref{algo:greedy-borst} is an example that executes local search in distributed fashion for the case of intra-layer cooperation. In each iteration, it selects randomly a node and a non-cached file from the library and examines if a swap with a cached file would improve the overall CN performance. There are two ways for evaluating this improvement, namely as an increase in the utility of the specific cache (strictly local; shown in line 8), and the version where the swap is assessed with respect to the total demand in the network. We refer to the algorithm using this latter metric as the \emph{generalized} local search. The following theorem summarizes the main result:
\begin{thm}[theorem style=plain]{Hierarchical Caching Approximation}{}
The local search caching algorithm achieves a $3/4$-approximation for (CR2b); and the generalized local search algorithm achieves a $1/2$-approximation for (CR2a).
\end{thm}
\noindent This result is due to \cite{BorstInfocom10} which also shows that a similar greedy algorithm achieves a tight approximation ratio for the model with inter-layer cooperation.  

From a different perspective, this algorithm can be seen as a dynamic policy which, upon the arrival of a request at a node (line 3), delivers the content (if cached) or fetches it from another leaf node (line 4), or the root if not available at any leaf node (line 5). Then, the algorithm decides if the content will be cached locally. In other words, this algorithm could encompass an eviction policy as well (see Chapter 3).

Finally, it is interesting to discuss some aspects of Algorithm \ref{algo:greedy-borst}. First, the property of locality arises due to the fact that the swapping decisions can be devised by each node independently. However, for the general asymmetric model this requires knowledge of the global CN utility. Second, given that this is a link-uncapacitated and non-congestible network, one could transform it to an equivalent bipartite caching graph, and use the respective algorithms. However, note that in the symmetric case, Algorithm \ref{algo:greedy-borst} achieves an improved approximation compared to the best result for the femtocaching problem, and involves low-complexity local search. This advantage disappears in the asymmetric case, for which one could use Algorithm \ref{algo:greedy-submodular}. Finally, it would be possible even to combine the two solutions and obtain an improved result. That is, first use the greedy bipartite caching algorithm to produce a feasible placement with 0.5 ratio, and then continue with local search to refine the approximation. Intuitively, this improvement will achieved due to the local search algorithm re-evaluating some of the arbitrarily-broken ties of the greedy iterations. 




\subsection{Minimizing Root-Server Traffic}


We next study the design of caching policies in a tree network \cite{Poularakis2016Complexity}, aiming to minimize the requests that are sent to the root server by satisfying as many of them as possible from in-network caches, see Fig. \ref{fig:hierarchical_model}. Reducing the load of root servers is very important as they are typically deployed at distant locations, and hence the off-network content transfer cost is typically very high. Moreover, in large networks these servers cannot satisfy concurrent requests from all leaf nodes, and hence reducing their load is crucial for sustaining load surges and avoiding flash crowd phenomena. 


A request can be served locally if a  leaf node $v$  has already cached the file; otherwise it is routed upwards following the path ${P}_{v}$ and eventually reaching the root server\footnote{We change slightly here the terminology of \cite{Poularakis2016Complexity}. That is, given that we will be studying 2-layer hierarchical networks, we use the term \emph{root} to describe the distant content server, and refer to the higher layer node in the tree as \emph{parent node}.} if none of the caches on ${P}_{v}$ has the file. The goal is to find the caching policy $\bm{y}$ that minimizes the requests sent to the root, i.e.: \vspace{-0.15in}
\begin{align}
J(\bm{y})=\sum_{v\in\mathcal{V}^0}\sum_{n\in\mathcal{N}}\lambda_{v,n}\mathbbm{1}\big(\sum_{v'\in P_{v}}y_{v^{\prime}n}<1 \big).
\end{align}
Clearly, there is no benefit caching a file at more than one locations on each path ${P}_{v}$ for every leaf $v\in\mathcal{V}^0$. Based on this observation, minimizing function $J(\bm{y})$ is equivalent to maximizing the number of requests served by the caches, which can be formulated as follows: 
\begin{opt}{Hierarchical Caching (HC2)}
	\vspace{-0.2in}
\begin{align}
	\max_{\bm{y}\in\mathcal{Y}_{HC2} } \qquad  &\sum_{v\in\mathcal{V}^0}\sum_{n\in\mathcal{N}}\lambda_{v,n}\sum_{v^{\prime}\in{P}_v} y_{v,n} \nonumber
\end{align}
	\vspace{-3mm}
\end{opt}
\noindent where the set of eligible policies is:
\begin{align}
	\mathcal{Y}_{HC2}=\left\{\bm{y}\in\{0,1\}^{P_v\times N}\,\Bigg|\begin{array}{l}
		\sum_{v^{\prime}\in{P}_v}y_{v,n}\leq 1 \\
		\sum_{n\in\mathcal{N}}y_{v,n} \leq B_v
		\end{array},~\forall v\in\mathcal{V}, n\in\mathcal{N} \right\}, \nonumber
\end{align}
and we have assumed that all files have unit size. The solution will dictate that either all requests $\lambda_{v,n}$ will be served by one cache along path ${P}_v$, or will be routed to the root server. This structure arises due to the objective here which, unlike (CR2a-b), does not optimize the routing cost savings and hence does not depend on which in-network node has cached the items. It is important to clarify that we have no routing variables, since the latter is completely determined by the caching variables: when we decide to place a file in a certain cache, then all requests generated at its connected leafs are routed there through the unique path.

Despite its simpler structure compared to femtocaching, (HC2) is NP-hard \cite{Poularakis2016Complexity}. However, for networks with two layers (leaf nodes connected to parent nodes, connected to a root server) it can be expressed as the maximization of a submodular function subject to uniform matroid constraints, as stated in the following lemma. 
\begin{lem}[theorem style=plain]{}{}
The constraints of (HC2) can be written as a uniform matroid on the ground set defined $\mathcal{S}=\{e_n:n\in\mathcal{N}\}$, where item $e_n$ denotes the placement of file $n$ at the parent node. Also, the objective function is a monotone, increasing and submodular function.
\end{lem}
\begin{proof}
Using as reference the parent node that has a capacity of $B_0$ files, every caching policy can be described by a subset $Y\subseteq \mathcal{S}$, where the elements of $Y$ correspond to files placed at the parent cache. Hence, it should be $Y\subseteq \mathcal{I}$ with:
\begin{align}
	\mathcal{I}=\{ Y\subset \mathcal{S}: |Y|\leq B_0 \}.
\end{align}
Comparing $\mathcal{I}$ and the definition of uniform matroid, we see that the constraints of this 2-layer caching problem constitute a uniform matroid $\mathcal{P}=(\mathcal{S}, \mathcal{I})$.

We now prove that the objective of (HC2) is a submodular function, see \cite{Poularakis2016Complexity}. A similar analysis can be found in \cite{femtocachingTransIT13, pacificiITC16}. First, recall that when a file $n$ is placed at the parent node, $e_n\in Y$, there is no need to place it at the leaf nodes; and let us denote with $\mathcal{A}_v(Y)$ the set of files cached at leaf $v$ for each policy $Y$. Then, for 2-layer networks, the objective is:
\begin{align}
	h(Y)=\sum_{v\in\mathcal{V}^0}\bigg( \sum_{n\in\mathcal{N}:e_n\in Y} \lambda_{v,n} + \sum_{n\in \mathcal{A}_v(Y)}\lambda_{v,n}\bigg),
\end{align}
which aggregates the requests served by the parent and leaf nodes. 

Since the summation preserves submodularity, it is sufficient if we show that every term of the external sum in $h(Y)$ is a submodular function. Focus on a single leaf cache $v\!\in\!\mathcal{V}^0$, and consider adding element $e_n$ in $Y$. We distinguish the following cases for a file $n$: 
\begin{itemize}
\item (i) $n\notin\mathcal{A}_v(Y)$. Then, if $n$ is placed at the parent node ($e_n$), the respective requests of leaf $v$ will not be sent to the root servers. This will yield additional value $h(Y\cup \{e_n\})-h(Y)=\lambda_{v,n}$. 
\item	(ii) $n\in\mathcal{A}_v(Y)$. Then, placing $n$ at the parent node will induce leaf $v$ to swap $n$ with the most popular file (based on $\lambda_v$), excluding the files already cached at the parent node and at $v$. We denote with $n^{\prime}$ that file, and the marginal value will be $h(Y\cup \{e_n\})-h(Y)=\lambda_{v,n^\prime}$. 
\end{itemize}

We now consider adding this element $e_n$ in a set $X\supseteq Y$. We distinguish the following two cases: 
\begin{itemize}
\item (i) $n\notin\mathcal{A}_v(X)$. Then $n$ is also not included in $\mathcal{A}_v(Y)$. To see this, recall that $\mathcal{A}_v(X)$ contains the $B_v$ most popular files after the $|X|$th popular file, while $\mathcal{A}_v(Y)$ contains the $B_v$ most popular files after the $|Y|$th popular file; and it holds $|Y|\!\leq\! |X|$. Hence, if demand $\lambda_{v,n}$ is not big enough to place $n$ among the $|X|$th and the $(|X|\!+\!B_v)$ most popular files, it would certainly not be enough to place it among the $|Y|$th and $(|Y|\!+\!B_v)$ most popular files. Hence, the marginal value will be $h(X\cup\{e_n\})-h(X)=h(Y\cup \{e_n\})-h(Y)=\lambda_{v,n}$. 

\item (ii) $n\in\mathcal{A}_v(X)$. Then, the marginal value is $h(X\cup \{e_n \} )-h(X)=\lambda_{v,n^{\prime\prime}}$, where $n^{\prime\prime}$ is the most popular file after those files already cached at the parent node or leaf $v$. We further distinguish the following two subcases here: (ii.a) $n\notin\mathcal{A}_v(Y)$. This is possible since it can be $Y\!\subseteq\! X$ and $v$ may not have the capacity to store the residual most popular files up to $n$. Then, it is $h(Y\cup\{e_n\})-h(Y)=\lambda_{v,n}\geq \lambda_{v,n^{\prime\prime}}$. (ii.b) $n\in\mathcal{A}_v(Y)$. Then, $h(Y\cup \{ e_n \})-h(Y)=\lambda_{v,n^{\prime}}\geq \lambda_{v,n^{\prime\prime}}$ where the last inequality holds because file $n^{\prime\prime}$ is picked among a subset of the files used for picking $n^{\prime}$. 
\end{itemize}
Therefore, the marginal value for adding an element in $X$ is always lower or equal to the one in $Y$, which results in $h(\cdot)$ being submodular. Similarly, we can show the submodularity of the component that is related to caching files at the parent node.
\end{proof}


\IncMargin{0.2em}
\begin{algorithm}[t]
	\small
	\nl \textbf{Initialize}: $\mathcal{Y}\leftarrow \emptyset$.\\%
	\nl \For{ $i=1,2,\ldots, B_0$ }{
		\nl Find the highest contribution element:\\
		$e_{n^*}=\arg\max_{ e_n\in \mathcal{S}\setminus Y }h_{Y}(e_{n})=H(Y\cup \{e_n \})-H(Y)$;\\
		\nl $Y \leftarrow Y \cup \{e_{n^*}\}$\,; \qquad\,\,\,\,\,\,\, \footnotesize{\%Update the caching policy}. \\%
	}
		\nl Place files in $Y$ at the parent node; \\%
		\nl Place files in $\mathcal{A}_v(Y)$ at each leaf node $v$;
	\caption{Hierarchical Greedy Caching Algorithm}\label{algo:greedy-poular}
\end{algorithm}\DecMargin{1em}
\normalsize


Given these properties of the problem, it can be solved with a greedy algorithm which, moreover, achieves an improved approximation ratio due to the special structure of the problem. 

\begin{thm}[theorem style=plain]{Greedy Hierarchical Caching}{}
	The greedy hierarchical algorithm achieves a $1/(1-1/e)$-approximation for the root request minimization problem (HC2).
\end{thm}

Note that unlike the femtocaching problem, here we obtain the $1/(1-1/e)$ bound without the need to use the computationally intensive pipage algorithm. On the other hand, this model is simpler as each leaf has only one parent node and the root-traffic minimization objective is insensitive in the different delivery delays the caching policies might yield, as long as the manage to satisfy the requests within the caching tree. The authors in \cite{Poularakis2016Complexity} provide also bounds -- and a methodology to construct them -- for trees with more than 2 layers. For such cases, and given that this is an uncapacitated and unweighted graph, one can consider also transforming the tree graph to a bipartite one and use Algorithm \ref{algo:greedy-submodular}.

\subsection{Discussion of Related Work}

Hierarchical or tree caching networks have been studied from the early days of Internet \cite{danzig-hierarchical} and distributed systems \cite{leff-hierarchical}. This problem is typically studied for two-layer networks, often with the addition of a distant root server. The seminal work \cite{korupolu99} studied a cooperative model with caches installed only at the leaf nodes, where the routing costs form an ultra-metric. The authors presented a polynomial-time exact algorithm with running time quadratic to $N\cdot V$. This work does not consider caches at higher layers nor hard link capacities. The idea of cooperative caching in trees arises naturally, and has been studied extensively in the context of web caching, including measurements for quantifying cooperation benefits \cite{wolman} and comparisons with flat architectures \cite{rodriguez-hierarchical}.

Caching policies for hierarchical networks have been proposed in \cite{li-hierarchical}, \cite{li-hierarchical2}, \cite{jia-hierarchical} that minimize routing costs using dynamic programming or heuristics, while \cite{poularakis-ciss} employed a matching problem for finding the optimal caching policy in 2-layer network. Unlike these works \cite{DaiInfocom12} proposed a 3-layer model, motivated by an actual IPTV network, with hard link capacity constraints. A routing and caching policy was designed, using a Lagrange decomposition technique (see Sec. \ref{sec:arbit-congest}), and the model was further extended for leaf nodes that can cooperate. Such hierarchical CN models arise also in wireless networks, see \cite{c-ran-caching}, \cite{pacificiITC16}.

In terms of solution algorithms for uncapacitated networks with non-congestible links, the hierarchical structure provides two important benefits. If the network is symmetric, the approximation ratio provided by the femtocaching Algorithm \ref{algo:greedy-submodular} can be improved with local search. Furthermore, a greedy algorithm offers the same guarantee as the computationally intensive pipage algorithm. In either case, clearly the exact performance of each algorithm depends on the selected optimization criterion and we presented here two examples, minimizing content delivery costs and reducing the root server load.

\section{Arbitrary Caching Networks} \label{sec4:arbitrary}


In this section we study more general CN models with hard link capacity constraints, or with routing costs that increase non-linearly with the volume of transferred content. Moreover, the objective criteria are non-linear functions of the cache hit ratio or of the routing costs, and the routing decisions can be multihop or multipath over arbitrary network graphs.

\subsection{Congestible Links and Non-linear Objectives} \label{sec:arbit-congest}

We start with a variation of the femtocaching model where the links are congestible and the objective is to minimize the aggregate delay \cite{Dehghan2015Complexity}. In detail, consider the network $G=( \mathcal{V}_c, \mathcal{V}_r, \mathcal{E}, \bm{\lambda}, D(\cdot))$, and the catalog $\mathcal{N}$ with files of unit size. Unlike the typical femtocaching model, the caches here are also connected with the root server and they can fetch a file in case of a cache miss, see Fig. \ref{fig:femto-towsley}. The SBS caching decisions should be drawn from set: 
\begin{equation}
	\mathcal{Y}_{CR3}=\{ \bm{y}\in\{0,1\}^{V_{c}^0\times N}\,\vert\, \sum_{n\in\mathcal{N}}y_{v,n}\leq B_v,\,v\in\mathcal{V}_{c}^{0} \},
\end{equation}
and routing is splittable, hence the routing policy $\bm{f}$ belongs to set:
\begin{equation}
	\mathcal{F}_{CR3}=\left\{\bm{f}\in[0,1]^{V_{c}^0\times V_r\times N }\,\Bigg|\begin{array}{l} \sum_{v\in\mathcal{V}_{c}^{0}}f_{v,u,n}\leq 1 \\
	 f_{v,u,n}=0, v\in\mathcal{V}_c\setminus \mathcal{V}_c(u)\end{array}\right\}. \nonumber
\end{equation}
We consider inelastic demand, where the fraction of traffic for each request $(u,n)$ routed to SBS caches $\mathcal{V}_{c}^0$ should not exceed $1$, and the remaining traffic is routed to MBS.\footnote{In earlier models we define $f_{v,u,n}$ as number of requests and not as a fraction. These are equivalent representations. Also, note that in $\mathcal{F}_{CR3}$ we do not include a hard constraint for satisfying all requests. This is not necessary because of the definition of $D(\cdot)$, where a request is routed to MBS unless satisfied by the SBSs.}

\begin{figure}[t]
	\centering
	\includegraphics[width=0.58\textwidth]{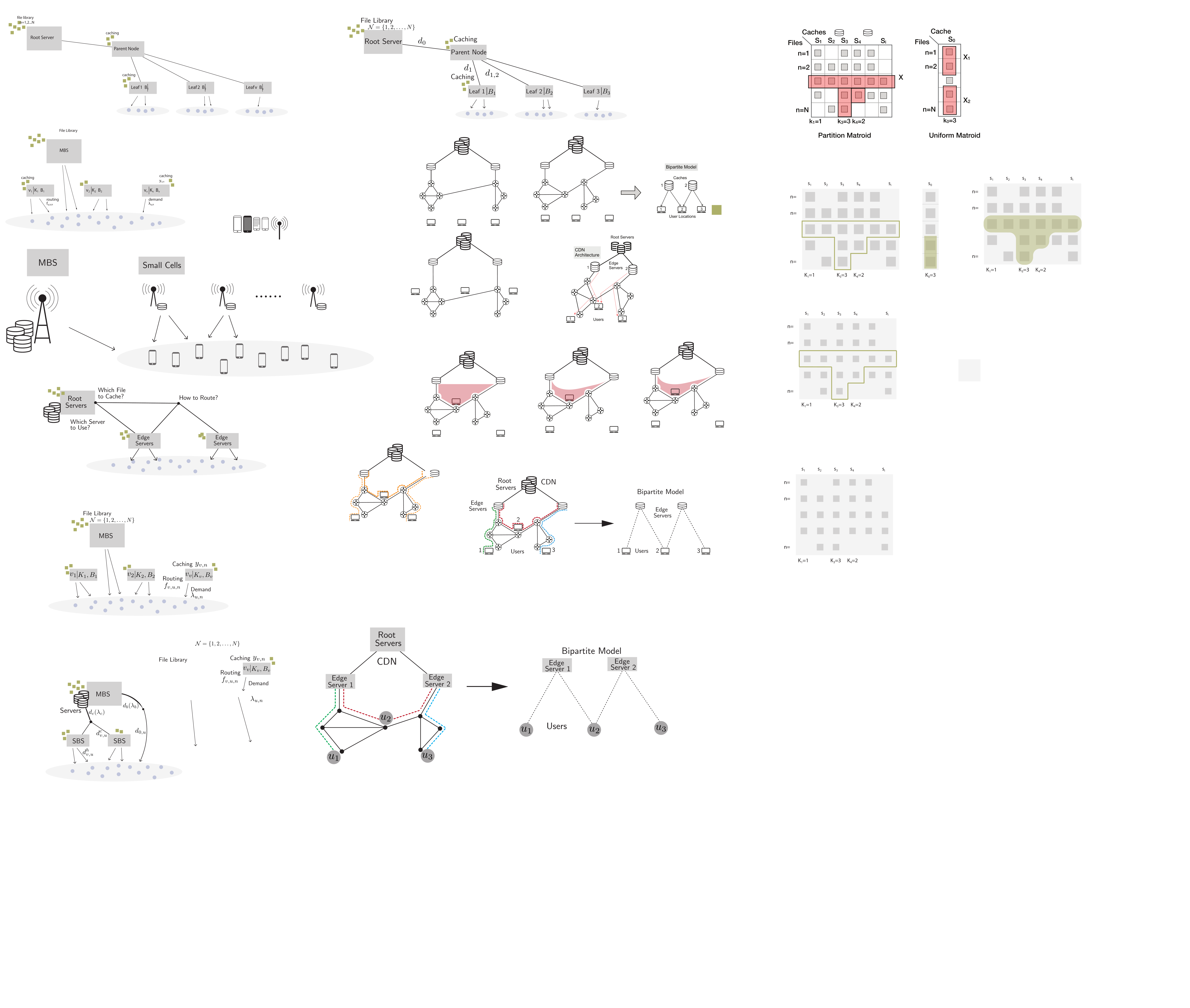}
	\caption{\small{Congestible femtocaching network with $\mathcal{V}_{c}^0$ SBSs, 1 MBS ($v=0$) and $\mathcal{V}_r$ users. The SBSs backhaul links are congestible, and the MBS transmissions induce load-dependent delay. Each user $u\in\mathcal{V}_r$ generates requests for file $n$ with rate $\lambda_{u,n}$, and all users are in range with the MBS that stores the catalog.}}
	\label{fig:femto-towsley}
\end{figure}

When a user $u$ receives a file from the MBS, it experiences delay equal to $d_{0,u}+d_b
\big( \lambda_{b}(\bm{f})\big)$, where $d_{0,u}$ is a constant delay parameter (due to distance) and $d_b(\cdot)$ a convex delay function that increases with the MBS link traffic:
\begin{equation}
\lambda_b(\bm{f})=\sum_{u \in\mathcal{V}_r}\sum_{n \in\mathcal{N}}\lambda_{u,n}\Big( 1-\sum_{v\in\mathcal{V}_{c}^{0}}f_{v,u,n} \Big)\,. \nonumber
\end{equation}
On the other hand, when the file is fetched from the servers through the backhaul links and transmitted by the SBS (cache miss), user $u$ experiences delay equal to $d_{v,u}^c+d_{c}\big(\lambda_c(\bm{f})\big)$, where $d_{v,u}^c$ is a constant and $d_c(\cdot)$ a convex function increasing with the backhaul load:
\begin{align}
\lambda_c(\bm{y},\bm{f})=\sum_{u\in\mathcal{V}_r}\sum_{n\in\mathcal{N}}\lambda_{u,n}\sum_{v\in\mathcal{V}_{c}^0}\Big( 1-y_{v,n}\Big)f_{v,u,n}\,. \nonumber
\end{align}
Finally, in case the requested file is available at the SBS (cache hit), it is delivered to the user with a fixed delay cost of $d_{v,u}^h$. Putting the above together, we can express the total delay as: 
\begin{align}
&D(\bm{y}, \bm{f})=\sum_{u\in\mathcal{V}_r}\sum_{n\in\mathcal{N} }\lambda_{u,n}\Big( \sum_{v\in\mathcal{V}_{c}^{0}}f_{v,u,n}y_{v,n}d_{v,u}^{h} +\big(1-\sum_{v\in\mathcal{V}_{c}^0}f_{v,u,n}\big)d_{0,u} + \nonumber\\ &\sum_{v\in\mathcal{V}_{c}^{0}}\big(1-y_{v,n}\big)f_{v,u,n}d_{v,u}^{c}  \Big)
+\lambda_b(\bm{f})d_b\big( \lambda_b(\bm{f}) \big) +\lambda_c(\bm{y}, \bm{f})d_c\big( \lambda_c(\bm{y}, \bm{f}) \big)\,.\nonumber
\end{align}
\noindent The goal here is to devise the policy $(\bm{y}, \bm{f})$ that minimizes $D(\bm{y}, \bm{f})$:
\begin{opt}{Caching and Routing with Congestible Links (CR3)}
	\vspace{-0.2in}
	\begin{align}
		\min _{\bm{y}, \bm{f} } \qquad  &  D( \bm{y}, \bm{f}) \nonumber \\
		\text{s.t.}\qquad 		&  \bm{y} \in\mathcal{Y}_{CR3}, \bm{f}\in\mathcal{F}_{CR3}
	\end{align}
	\vspace{-0.35in}
\end{opt}
\noindent It is interesting to note the absence of the typical constraint that prevents routing requests to caches that do not have the files. Indeed, this is not necessary here as this condition is incorporated in the objective function. Namely, we can only achieve delay savings for the requests that satisfy $f_{v,u,n}y_{v,n}>0$, i.e., the file is available at the employed cache. This creates a non-linear objective but simplifies the constraint set.

One of the simplest approaches for modeling the congestion-sensitive delay at the MBS is to assume that there is an $M/M/1$ queue for each link through which the files are served. The service delay  increases with the queue length and the total link traffic load. For this model, \cite{Dehghan2015Complexity} showed that (CR3) can be solved with a greedy algorithm achieving a $2$-approximation ratio. The proof shows that for any given caching policy $\bm{y}$ the delay is a convex function of the routing variables $\bm{f}$, and for fixed routing it is submodular on $\bm{y}$. This is then exploited by an algorithm that places greedily the items in caches and decides accordingly the routing policy.  More sophisticated delay models have been also considered in the literature, e.g., see \cite{carofiglioComNet16}.

The second CN model we consider in this section differs from (CR3) in two aspects. It includes hard constraints for the MBS and SBS link capacities, and the objective is to minimize a weighted cost function of delivery delay and network operating expenditures. This model was studied in \cite{poularVideoInfocom14} for the delivery of video files. The set of eligible caching policies $\mathcal{Y}_{CR4}$ is the same as in (CR3), $\mathcal{Y}_{CR4}= \mathcal{Y}_{CR3}$, but the routing decisions are more intricate due to the capacity constraints. Namely, they should belong to set:
\begin{align}
&\mathcal{F}_{CR4}\!=\!\left\{ \bm{f}\in[0,1]^{V_c\times V_r\times N}\Bigg|\!\! 
\begin{array}{l}
\sum_{u\in \mathcal{V}_r, n\in\mathcal{N}}\!\lambda_{u,n}f_{v,u,n}\!\leq \!K_v \\ 
\sum_{u\in \mathcal{V}_r, n\in\mathcal{N}}\!	\lambda_{u,n}f_{0,u,n}\!\leq\! K_0\,\\
f_{0,u,n}+\!\sum_{v\in\mathcal{V}_{c}^{0}(u)}f_{v,u,n}\!\leq\! 1
\end{array}
; \forall 
\begin{array}{l}
v\!\in\!\mathcal{V}_c\\
n\!\in\!\mathcal{N}\\
 u\!\in\!\mathcal{V}_r \end{array}
\right\}. \nonumber
\end{align}
where note that we also use explicit routing variables for the MBS, as an elastic model is assumed in this problem. Parameter $K_v$ is the wireless bandwidth of SBS $v$, and $K_0$ the MBS wireless link capacity. Additionally, we introduce the backhaul routing variables $z_{v,u,n}\in[0,1]$ denoting the fraction of requests that will be satisfied by fetching file $n$ via the backhaul link of SBS $v$ having capacity\footnote{The wireless links and backhaul link capacities are measured in bytes with reference to a certain time period $T$, which can be set equal to $1$.} $G_v$. The set of eligible backhaul routing policies is:
\begin{equation}
	\mathcal{Z}_{CR4}=\{ \bm{z}\in[0,1]^{V_{c}^0\times V_r \times N}\,\vert\, \sum_{u\in\mathcal{V}_r, n\in\mathcal{N}}\lambda_{u,n}z_{v,u,n}\leq G_v; v\!\in\!\mathcal{V}_c, n\!\in\!\mathcal{N}\}. \nonumber
\end{equation}

Assuming there are no expenditures for caching the content, the joint routing and caching policy can be formulated as follows: 
\begin{opt}{Caching and Routing with  Capacitated Links (CR4)}
	\vspace{-0.2in}
 \begin{align}
 	\min _{\bm{y}, \bm{f}, \bm{z} } \qquad  &  J(\bm{f}, \bm{z}) \nonumber \\
 	\text{s.t.}\qquad 	&  f_{v,u,n}\leq y_{v,n}\!+\!z_{v,u,n},\,\forall u\in\mathcal{V}_r,n\in\mathcal{N}, v\in\mathcal{V}_{c}^{0}(u) \label{eq:integer-video-poul-constraint4}\\
 	&  \bm{y} \in\mathcal{Y}_{CR4}, \bm{f}\in\mathcal{F}_{CR4}, \bm{z} \in\mathcal{Z}_{CR4}
 \end{align}
	\vspace{-0.25in}
\end{opt}
\noindent Constraint (\ref{eq:integer-video-poul-constraint4}) indicates that in order for an SBS to send a file to a user it needs either to have it cached or to fetch it from the root server through the backhaul link, and in the latter case, with an equal rate. The objective function may include a cost component that introduces a load-dependent delay as in (CR3); an increasing convex penalty for the fraction of unsatisfied requests; or an increasing cost function for the energy or other operating costs at the base stations. This detail modeling however increases the complexity of the problem, and reduces the chances to design efficient constant approximation algorithms. We discuss next a general method for solving such problems.

\subsubsection{Lagrange Decomposition Methods}

In particular, we will employ a Lagrange relaxation method where we dualize constraints (\ref{eq:integer-video-poul-constraint4}) that couple the different types of variables. This approach has several advantages. First, it decomposes the problem into a caching and a routing sub-problem. This is very important and has many practical implications. For example, in several cases a certain business entity is responsible for the routing policy (e.g., a network operator) and a different one for the caching policy (e.g., a CDN), see \cite{smaragdBook13, tasos-jsac18}. By decomposing the problem we enable them to coordinate and optimize jointly the overall content delivery performance. Second, this decomposition results in smaller, often also simpler, subproblems that can be solved faster and oftentimes in a parallel fashion.

The partial Lagrangian can be defined if we relax (\ref{eq:integer-video-poul-constraint4}) and introduce the dual variables $\mu_{v,u,n}\geq 0$:
\begin{align}
L(\bm{y}, \bm{f}, \bm{z}, \bm{\mu})=J(\bm{f}, \bm{z}) +\sum_{n}\sum_{u}\sum_{v}\mu_{v,u,n}\big( f_{v,u,n}-y_{v,n}-z_{v,u,n} \big)\,,\nonumber
\end{align}
and subsequently we can define the dual problem:
\begin{align}
	\max_{\bm{\mu}\geq \bm{0}} q(\bm{\mu})=\,\,\min_{\bm{y}\in\mathcal{Y}_{CR4}, \bm{f}\in\mathcal{F}_{CR4}, \bm{z}\in\mathcal{Z}_{CR4} }\,\,\,L(\bm{y}, \bm{f}, \bm{z}, \bm{\mu})\,, \nonumber
\end{align}
where $q(\cdot)$ is the dual function. Our goal is to solve the dual problem and obtain a lower bound for (CR4). The dual solution can be used either for recovering a feasible solution using a problem-specific heuristic, or for improving the performance of a branch-and-bound method \cite{fisherLagrange}.

The key observation here is that both the objective function and the constraints of the relaxed problem are separable. Therefore, for any given value of the dual variables $\hat{\bm{\mu}}$, we can find the optimal value of the Lagrangian subject to the remaining constraints. Namely, we can decompose $L(\cdot)$ and obtain the caching (CR4a) and routing (CR4b) subproblems: 
\begin{align}
	&\text{(CR4a)}\!:\min_{\bm{y} \in\mathcal{Y}_{CR4} }-\sum_{u\in\mathcal{V}_r }\sum_{v\in\mathcal{V}_c(u)}\sum_{n\in\mathcal{N} }\hat{\mu}_{v,u,n}y_{v,n}\,. \nonumber \\
	&\text{(CR4b)}\!:\min_{\bm{f}\in\mathcal{F}_{CR4}, \bm{z}\in\mathcal{Z}_{CR4} }J(\bm{f},\bm{z}) \!+\!\sum_{n\in\mathcal{N} }\sum_{u\in\mathcal{V}_r}\sum_{v\in\mathcal{V}_c(u)}\hat{\mu}_{v,u,n}(y_{v,u,n}\!-\!z_{v,u,n})\,.\nonumber
\end{align}
Under some typical assumptions about $J(\cdot)$, (CR4b) becomes a convex optimization problem. On the other hand (CR4a) can be further decomposed to $N$ Knapsack problems which, although NP-complete, can be solved in pseudo-polynomial time using dynamic programming methods\footnote{In some applications with equal  file sizes, the Knapsack problem simplifies and can be solved exactly.}. Clearly, this decoupling reduces the size of (CR4) and parallelizes its solution. However, note that in some problem formulations the objective function is not convex, even if the routing variables are continuous, e.g., see \cite{femtocachingTransIT13}, \cite{ioannidis_sigm16},  \cite{ioannidis_icn17}. This issue arises also when the objective depends directly on the caching policy, e.g., when there is cost for caching the files. 

More often than not, we cannot derive an analytical expression for $q(\bm{\mu})$ that would allow us to directly solve the dual problem. A typical method to overcome this obstacle is to use the subgradient method. The main idea is that we iteratively obtain improved dual values $\bm{\mu}^{(1)}, \bm{\mu}^{(2)},\ldots,\bm{\mu}^{(\tau)}$, by using the subgradient of $q(\cdot)$ which we calculate at each iteration $\tau$ using the relaxed problem. The details are provided in Algorithm \ref{algo:dual-algorithm} where note the bounds LB and UB are used to flag convergence. 


\IncMargin{1.9em}
\begin{algorithm}[t]
	\small
	\nl \textbf{Initialize}: Dual vars $\bm{\mu}^{(1)}=\bm{0}$;\,\, $LB=-\infty$;\,\, $UB=\infty$.\\%
	\nl Set $\tau\leftarrow 1$;\\
	\nl \Repeat{ $\frac{UB-LB}{UB}\leq \epsilon$ }{
		\nl Solve (CR4a) and find $\bm{y}^{(\tau)}$\,.\\%
		\nl Solve (CR4b) and find $\bm{f}^{(\tau)},\, \bm{z}^{(\tau)}$\,.\\%
		\nl \If{$q(\bm{\mu}^{(\tau)})>LB$}{$LB=q(\bm{\mu}^{(\tau)})$}
		\nl Calculate subgradients of $q(\bm{\mu}^{\tau})$: $g_{v,u,n}^{(\tau)}=f_{v,u,n}^{(\tau)}-y_{v,n}^{(\tau)}-z_{v,u,n}^{(\tau)},\forall\,u,v,n$.  \\%
		\nl Update the dual variables: $\mu_{v,u,n}^{(\tau+1)}=[\mu_{v,u,n}^{(\tau)} + \sigma^{(\tau)}g_{v,u,n}^{(\tau)}]^+,\,\,\forall\,u,v,n$ \\%
		\nl  $\tau\leftarrow \tau+1$ }
	\caption{Iterative Routing - Caching Algorithm}\label{algo:dual-algorithm}
\end{algorithm}\DecMargin{1em}
\normalsize

Further details about the subgradient methods, the selection of the step size and methods for setting the upper bound ($UB$) can be found in \cite{Ber98book}. Moreover, there are techniques that allow the convergence of such iterative algorithms even in the presence of errors in the dual variables, e.g., due to perturbations in the network parameters, by leveraging the properties of stochastic subgradient methods; see \cite{victorcdc16} for a comprehensive discussion on this topic.

It is important to stress that, albeit elegant, these Lagrangian methods have well-known drawbacks. First, (CR4) is an NP-hard problem (reduction from (C1)) and this algorithm does not guarantee an exact solution. Moreover, and perhaps more importantly, the obtained result might violate the relaxed constraints. Namely, given that strong duality does not hold and that often we only obtain an approximate solution for the dual problem, the result might violate the primal constraints that were dualized.  For example, the routing policy might be incompatible with the caching decisions. Hence, there is a need to adjust the solution with problem-specific heuristics, known as \emph{Lagrangian heuristics}, and recover a primal-feasible solution.

Despite these issues, this approach has been proven very effective in practice \cite{geoffrion}. Furthermore, in some cases we can quantify the optimality loss by characterizing the duality gap \cite{bertsekasdualgap}, which is reduced for large problem instances; see \cite{athanasTon17, wangTon17, wangSiam} for applications in base station association, multipath routing, and network economics, respectively. Lagrangian methods have been used in several cases for devising caching policies, e.g., \cite{DaiInfocom12}, \cite{poularakisTNSM16}. Furthermore, in case the caching decisions are continuous (as in the case of coded caching), this decomposition approach can provide an exact solution under some mild assumptions on the objective function and the constraint set \cite{Ber98book}. 

\subsection{Multihop Routing}

We explain how the above models can be generalized to arbitrary networks with multihop and/or multipath routing. The solution approaches in this case can be broadly classified based on whether: (i) they propose deterministic or randomized policies; (ii) they employ static or dynamic policies which make decisions on a slot-by-slot basis; (iii) the routing policy is based on the requester (path selection) or realized in a hop-by-hop fashion based on the requested file. We discuss next some illustrative cases. 


Consider a general caching network $\mathcal{G}=(\mathcal{V}, \mathcal{E}, \bm{\lambda}, \bm{d} )$. For each file $n$ there is a set $\mathcal{V}(n)$ of dedicated root servers in $\mathcal{G}$ that store it permanently. The user requests are routed towards these servers but can be also satisfied by in-network caches along the routing path. Each request $(u,n)$ is specified by the user $u\in \mathcal{V}$ and the requested item $n\in \mathcal{N}$. Let $f_{p,n}\in\{0,1\}$ denote the decision of selecting path $p\in\mathcal{P}_{v,u}$ that has length $p=|\mathcal{P}_{v,u}|$ nodes; and denote with $d_{p_kp_{k+1}}$ the routing cost (delay or bandwidth) for transferring one file over link $(p_k, p_{k+1})$ in this path. The eligible routing policies belong to set:
\begin{align}
	\mathcal{F}_{CRs}\!=\!\{ \bm{f}\in[0,1]^{P\times N}\,\vert\, \sum_{v\in \mathcal{V}(n)} \sum_{p\in\mathcal{P}_{v,u}}f_{p,n}=1,  u\!\in\!\mathcal{V}, n\!\in\!\mathcal{N}\}, \nonumber
\end{align}
while the set of possible caching policies is $\mathcal{Y}_{CRs}=\mathcal{Y}_{CR3}$ for $\mathcal{V}_{c}^0=\mathcal{V}$.

In this \emph{Source-Routing} (SR) model, we can express the cost of serving the requests as follows:
\begin{align}
	J_{SR}(\bm{f},\bm{y})=\sum_{(u,n)}\sum_{v\in\mathcal{V}(n)}\sum_{p\in \mathcal{P}_{v,u}} f_{p,n}\sum_{k=1}^{|p|-1}d_{p_{k+1}p_k} \prod_{k^{\prime}=1}^{k}\big(1-y_{p_k,n}\big)\,, \nonumber
\end{align}
where the last term ensures that whenever item $n$ is not cached at a node $p_k\in\mathcal{P}_{v,u}$, i.e., $y_{p_k,n}=0$, the request is further forwarded to node $p_{k+1}$ and the network incurs cost $d_{p_kp_{k+1}}$. In this simple way, the routing and caching decisions are jointly shaping the objective function. Then, the joint caching and routing optimization problem can be written:
\begin{opt}{Caching and Source-Routing (CRs)}
	\vspace{-0.2in}
	\begin{align}
		\min _{\bm{y}, \bm{f} } \qquad  &  J_{SR}(\bm{f}, \bm{y}) \nonumber \\
		\text{s.t.}\qquad 	&  \bm{y} \in\mathcal{Y}_{CRs}, \bm{f}\in\mathcal{F}_{CRs}
	\end{align}
	\vspace{-0.25in}
\end{opt}

Note that this source-routing approach presumes the computation of all available paths for each request in advance, similarly to the previous models. An alternative approach would be to employ a hop-by-hop routing policy that makes decisions based only on the requested item and not the requester. In this case, each intermediate node $c$ needs to make a next-hop routing decision to node $v$ for every requested content $n$, unless the item is cached at $c$. To avoid inefficient routing, it helps defining the subgraph $G^{(u,n)}$ for each request $(u,n)$ in a way to exclude loops and nodes that are not reachable by node $s$. In this case, the servicing cost can be written 	$J_{HH}(\bm{f},\bm{y})=$
\begin{equation}
\sum_{(u,n)} \sum_{(c,v)\in G^{(u,n)}} d_{c,v} f_{c,v,n}(1-y_{c,n}) \sum_{v\in\mathcal{V}(n)} \sum_{p\in \mathcal{P}_{v,u}^{c}}\prod_{k^{\prime}=1}^{|p|-1}r_{p_{k^{\prime}}p_{k^{\prime}+1}}\big(1-y_{p_{k^{\prime}},n}\big)\,,\nonumber
\end{equation}
where $\mathcal{P}_{(u,n)}^{c}$ is the set of paths connecting node $u$ (origin of request) with node $c$, and is conditioned on the requested item, as this determines the servers that are included in $G^{(u,n)}$.  

For the above models one can devise the respective randomized policies where the caching and routing decisions are probabilistic. Note that, even under this relaxation this problem is challenging as the objective might not be --- in the general case --- a convex function. The authors in \cite{ioannidis_icn17} provide an interesting solution for this problem that relies on a convex approximation, leveraging and extending the methodology in \cite{pipage}.


\subsection{Discussion of Related Work}

Unlike the Lagrange relaxation approach presented above, a technique that also decomposes the problem but finds an exact solution is the Benders' decomposition method \cite{benders62}. This approach does not relax any constraint but rather calculates iteratively cutting planes that are added to the problem in the form of constraints, and essentially simplifying the search space by removing some candidate solutions known to be suboptimal. This process has the advantage that it reduces the size of the initial problem and hence expedites its solution, even if the subproblems have the same order of complexity, but does not offer guarantees on the convergence time. We refer the reader to \cite{benders-overreview} for a recent survey with examples, and to \cite{tasos-jsac18} for its application to caching.


Moreover, \cite{Bektas} discusses a general suite of solution methodologies, from Lagrange relaxations to the Benders' decomposition method. These methods are iterative and trade-off precision with complexity, from the exact solution of Benders' method to the near-optimal with no guarantees Lagrange iteration. The same work also discusses tools for linearization of the objective function which can be very useful especially when the caching and routing decisions interact in the objective. When there are no hard capacity constraints and the only routing criterion is the delay or cost, which is load-independent per hop, then we can precompute all possible paths connecting each pair of user and cache. Without hard capacities, and if the routing costs are constant, the routing policy can always select for each request the shortest path towards the cache. Therefore, the caching policy takes as input the path delays, and the routing decisions are essentially determined, indirectly, by the caching policy. See, for example, \cite{GeorgiadisTPDS06} for a detailed analysis of this methodology in the context of service deployment in networks.

In another example, \cite{applegateConext10} considers the problem of video caching in a Video on Demand network. The model considers strict bounds for the node storage and link capacities, different routing costs, and files with possibly different sizes. The routing decisions are fractional, and hence the problem is formulated as a min-cost mixed integer program. Similarly, \cite{xie2012} formulates the joint optimization problem of collaborative caching and traffic engineering, and assumes a fractional solution (splittable requests) that minimize link utilization. Such joint mechanisms require the coordination of ISPs and CDNs and have attracted significant interest from industry, e.g., see \cite{smaragdBook13}.

Concluding, for all the discussed problems, the solution methods can be roughly classified as follows: (i) algorithms based on facility location theory; (ii) algorithms that rely on submodular optimization theory, often involving matroid constraints;  (iii) techniques that employ some type of relaxation, such as Lagrange or linear relaxations; and (iv) algorithms that find exact solutions, such as the cutting planes Benders' decomposition. Based on the problem formulation, but also on whether we are interested in solution speed or optimality, a different algorithm should be selected. After all, designing and optimizing caching networks is an art, and cannot be fully specified by predetermined rules.


\chapter{Online Bipartite Caching}\label{ch:5}

In this chapter we revisit the problem of bipartite caching in the setting where content popularity is unknown and possibly non-stationary. Upon receiving each request, the network must decide how to route the requested content and how to update the contents in each cache in order to maximize the accumulated utility over an horizon $T$. For this challenging setting we propose the \emph{Bipartite Supergradient Caching Algorithm} (BSCA) and prove that it achieves no regret ($\texttt{Reg}_T/T \to 0$). That is, as the time horizon $T$ increases, BSCA learns to achieve the same performance with the optimal bipartite CN configuration in hindsight, i.e.,  the one we would have chosen knowing all future requests. The learning rate of the algorithm is characterized by its regret expression, found to be $\texttt{Reg}_T = O(\sqrt{JT})$. We focus here on the basic bipartite caching model, but BSCA can be extended to other caching networks studied in Chapter \ref{ch:4}.

\section{Introduction to Online Bipartite Caching}

The  bipartite caching analysis of Chapter \ref{ch:4} introduced a network architecture where caches are attached to small-cell base stations (SBSs), and proposed an efficient algorithm for proactively caching content files at them. An essential weakness of these proactive caching techniques, however, is that they assume static and known file popularity. Practical experience has shown quite the opposite: \emph{file popularities changes rapidly, and learning them is challenging}. In what follows,  we study the performance of femtocaching networks from a new angle. \emph{We seek to find an online caching and routing policy that optimizes, on-the-fly, the network's performance under any file popularity (or, request) model}. This is very important, as it tackles the femtocaching problem in its most general form, and reveals a novel connection between content caching and routing mechanisms and the theory of Online Convex Optimization (OCO) \cite{zinkevich2003online,Shalev12,Mert18}.

\subsection{System model}

\textbf{Network Connectivity}. Following the notation of femtocaching, we consider a set ${\cal V}_R$ of cache-endowed SBSs and an MBS indexed with 0. There is  a set of user locations ${\cal V}_c$, where file requests are created. The existence of a link between cache $v$ and user location $u$ is denoted with $e_{v,u}=1$, and we set $e_{v,u}=0$ otherwise. The MBS is in range with all users.

\textbf{File Requests}. The system evolves in  slots, $t=1,2,\dots,T$, and users submit requests for obtaining files from the catalog $\mathcal{N}$. We denote with $R_{t}^{n,u}\!\in\!\{0,1\}$ the event that a request for file $n\in\mathcal{N}$ has been submitted by a user in location $u\in\mathcal{V}_c$, during slot $t$. At each slot we  assume that there is exactly one request. From a different perspective, this means that the policy is applied after every request, exactly as it happens with the standard LFU/LRU-type of single-cache policies, or the respective multi-cache reactive policies, see \cite{giovanidis-mLRU, leonardi-implicit} and references therein.\footnote{Alternatively, we may  consider batches of requests. If the batch contains one request from each location, the request pattern is biased to equal request rate at each location. An unbiased batch should contain an arbitrary number of requests from each location. The presented guarantees hold even for unbiased batches of arbitrary (but finite) length, multiplied by the  batch size.}
Hence, the request process can be described by a sequence of vectors $\{R_t\}_{t=1,\dots,T}$ drawn from set:
\begin{equation}
	\mathcal{R}=\left\{R\in \{0,1\}^{N\times V_r} ~\bigg |~ \sum_{n\in\mathcal{N}}\sum_{u\in\mathcal{V}_r} R^{n,u}=1\right\}.
\end{equation}

The instantaneous file popularity is expressed by the probability distribution $P(R_t)$ (with support $\mathcal{R}$), which is allowed to be unknown and arbitrary. The same holds for the joint distribution $P(R_1,\dots,R_T)$ that describes the file popularity evolution within the time interval $T$. This generic model captures all studied request sequences in the literature, including stationary (i.i.d. or otherwise), non-stationary, and adversarial models. The latter are the most general models one can employ, as they include request sequences selected by an \emph{adversary} aiming to disrupt the system performance, e.g., consider Denial-of-Service attacks.

\textbf{Caching}. Each SBS $v\!\in\!\mathcal{V}_c$ can cache only $B_v$ of the catalog files, i.e., $B_v\!<\!N, \forall v$, and the MBS can store the entire catalog, i.e., $B_0=N$. Different from the analysis in Chapter 4, here we follow the fractional caching model \cite{femtocachingTransIT13} and use the \emph{Maximum Distance Separable} (MDS) codes, where files are split into a fixed number of $K$ chunks, and each stored chunk is a pseudo-random linear combination of the original $K$ chunks. Using the properties of MDS codes, a user will be able to decode the file (with high probability) if it receives any $K$ coded chunks.

The above model results in the following: the caching decision vector $y_t$ has $N\!\times\!V_c$ elements, and each element $y^t_{v,n}\in [0,1]$ denotes the amount of random coded chunks of file $n$ stored at cache $v$.\footnote{The fractional model is justified by the observation that large files are composed of thousands of chunks, stored independently, see \emph{partial caching}~\cite{maggi2018adapting}.  Hence, by rounding the fractional decisions  to the closest integer in this finer granularity, induces only a small application-specific error. In some prior caching models,  fractional variables represent probabilities of caching \cite{Blaszczyszyn2014Geographic}.}
Based on this, we introduce the set of eligible caching vectors: 
\[
\mathcal{Y}=\left\{y\in [0,1]^{N\times V_c} ~\Bigg|~ \sum_{n\in\mathcal{N}}y_{v,n}\leq B_v, ~v\in \mathcal{V}_c\right\},
\]
which is convex. We can now define the online caching policy as follows:
\begin{dfn}[theorem style=plain]{Content catalog}{}
A caching policy $\sigma$ is a (possibly randomized) rule
\begin{equation}
	\sigma: (R_1, R_2, \ldots, R_{t-1}, y_1, y_2, \ldots, y_{t-1})\longrightarrow y_{t}\in \mathcal{Y}\,. \nonumber
\end{equation}
which at each slot $t=1,\dots,T$ maps past observations $R_1,\dots,R_{t-1}$ and configurations $y_1,\dots,y_{t-1}$ to a caching vector $y_t\in \mathcal{Y}$ for slot $t$. 
\end{dfn}
\noindent Note that unlike previous works that study only \emph{proactive caching} policies, here we assume that content can be fetched and cached dynamically as the requests are submitted by the users.

\textbf{Routing}. Since each location $u\!\in\!\mathcal{V}_r$ is possibly connected to multiple caches, we introduce  \emph{routing variables} to determine the cache from which  the requested content will be fetched. Namely, let $f_{v,u,n}^t\in[0,1]$ denote the portion of request $R_{t}^{u,n}$ that is fetched from cache $v$, and we define the respective routing vector $f_t$ at $t$. There are two important remarks here. Due to the coded caching model, the requests can be simultaneously routed from multiple caches and then combined to deliver the content file. Moreover, the caching and routing decisions are coupled and constrained since: \emph{(i)} a request cannot be routed from an unreachable cache; \emph{(ii)} a cache cannot deliver more data chunks than it has; and \emph{(iii)} each request must be fully routed.

Therefore, we define the set of eligible routing decisions conditioned on caching policy $y_t$ and request $R_t$: 
\[
{\cal F}(y_t)=\left\{z\in [0,1]^{N\!\times\! V_c\!\times \!V_r} \Bigg|
\begin{array}{c}
\sum_{j\in {\cal V}_c}f_{v,u,n}^t=R^{u,n}_t,~\forall n,i \vspace{1.5mm}\\

f_{u,v,n}^t\leq e_{u,v}y_{v,n}, ~\forall n, u, v
\end{array}
\right\}
\]
where the first constraint ensures that the entire request is routed, and the second constraint captures both connectivity and caching limitations. Here we note that routing from MBS (variable $f_{0,i,n}^t$) does not appear in the second constraint, because the MBS stores the entire catalog and is in range with all users. This {last-resort} routing option ensures that the set ${\cal F}(y_t)$ is non-empty for all $y_t\in \mathcal{Y}$. 

\section{Online Bipartite Caching: Statement \& Formulation}\label{sec:formulation}
 
We begin this section by defining the caching objective and then proving that the online femtocaching operation can be modeled as a regret minimization problem.

\subsection{Cache Utility}
 
We consider a utility-cache model and use the utility weights $d^{v,u,n}$ to denote the obtained benefit by retrieving a unit fraction of a content file (i.e., a coded chunk) from cache $j$ instead of the MBS, and we set $d^{0,i,n}=0$. This content-dependent utility can be used to model bandwidth economization from cache hits \cite{maggi2018adapting}, QoS improvement from using caches in proximity \cite{femtocachingTransIT13}, or any other content/cache-related benefit. Our model generalizes the respective model discussed in Chapter 4. We can then define the total network utility accrued in slot $t$ as: 
\begin{equation}\label{eq:biput}
	J_t(y_t)=\max_{f\in {\cal F}(y_t)} \sum_{n\in\mathcal{N}}\sum_{u\in\mathcal{V}_r}\sum_{v\in\mathcal{V}_c} d_{v,u,n} R_t^{n,i}f_{v,u,n}^t,
\end{equation}
where index $t$ is used to remind us that utility is affected by the request arrival vector $R_t$, which is adversarial. It is easy to see that $J_t(\cdot)$ states that benefit $d^{v,u,n}$ is realized when a unit of request is successfully routed to cache $v$ where the content $n$ is available. Note also that we have written $J_t(\cdot)$ only as function of caching, as for each $y_t$ we have already included in the utility definition the selection of the optimal routing $f_t$. As we will see next, this formulation facilitates the solution by simplifying the projection step in our algorithm.

\subsection{Problem Formulation}

Let us now formulate femtocaching as an OCO problem. This is non-trivial and requires certain conceptual innovations. For the description below please refer to Fig.~\ref{fig:model-oco}. First, in order to model that the request sequence can follow any arbitrary, and a priori unknown, probability distribution we use the notion of an \emph{adversary} that selects $R_t$ in each slot $t$. In the worst case, this entity generates requests aiming to degrade the performance of the caching system. This approach ensures that our caching and routing policy has robust performance for any possible request model. Going a step further, we model the adversary as selecting the utility function, instead of the request. Namely, at each slot $t$, the adversary picks a utility function $J_t(y)$ from the family of  functions $\{f(R_t,y); R_t\}$ by deciding the vector $R_t$. We emphasize that these functions are piece-wise linear. In the next subsection we will show that they are concave with respect to $y_t$, but not always differentiable. 

Finally, we consider here the practical online setting where $y_t$ \emph{is decided after the request has arrived} and the caching utility has been calculated. This timing reflects naturally the operation of caches and reactive policies, where first a generated request yields some utility (based on whether there was a cache hit or miss), and then the system reacts by updating the cached files. In other words, caching decisions are taken without knowing the future requests. All the above steps allow us to reformulate the caching problem and place it squarely on the OCO framework \cite{hazan-book}.

\begin{figure}[!t]
	\centering
	\includegraphics[width=5.5in]{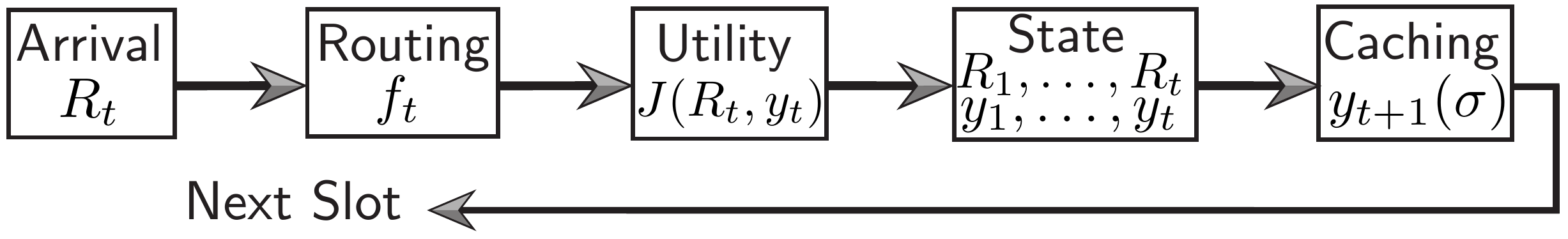} 
	\caption{\textbf{Online femtocaching model}. When a request $R_t$ arrives, the content is routed optimally based on the current cache configuration i.e., $f_t=\mathcal{F}(y_t)$. We accrue utility $J(R_t, y_t)$ and the caching decisions are updated using the state that includes the observed requests and caching decisions.}
	\label{fig:model-oco}
\end{figure} 

Given the adversarial nature of our model, the ability to extract useful conclusions depends  crucially on the choice of the performance metric. Differently from the competitive ratio approach of \cite{Sleator85}, we introduce a new metric for designing our policies. Namely, we will compare how our policy fares against the \emph{best static action in hindsight}. The latter is an ideal, hypothetical, cache configuration which optimizes the average of the  entire request sequence. This metric is commonly used in machine learning \cite{hazan-book, Mert18} and is known as the worst-case static regret. In particular, we define the \emph{regret of policy $\sigma$} as:
\begin{align}\label{eq:regret}
	\texttt{Reg}_T(\sigma)&=\max_{P(R_1,\dots,R_T)}
	\mean{\sum_{t=1}^T J_t\big(y^*\big)-\sum_{t=1}^T J_t\big(y_t(\sigma)\big)},
\end{align}
where $T$ is the time horizon. Note that the maximization is over all possible adversary distributions, and the expectation is taken with respect to the possibly randomized $R_t$ and $\{y_t(\sigma); f_t(y_t)\}$. The hindsight policy is determined by solving $y^*\in\arg\max_{y\in \Yc}\sum_{t=1}^T J_t(y)$.

Our goal is to study how the regret scales with the horizon $T$. A policy with sublinear regret $o(T)$ produces average loss $\lim_{T\rightarrow\infty} \texttt{Reg}_T(\sigma)/T=0$, w.r.t. the hindsight policy. Therefore, $\sigma$ learns which file chunks to store and how to route requests, without knowing the file popularity. We can now formally define the problem at hand as follows:
\begin{opt}{Online Femtocaching Problem (OFP)}
Find a policy $\sigma$ that satisfies:
\[
\texttt{Reg}_T(\sigma) = o(T),
\]
where $\texttt{Reg}_T(\sigma)$ is defined in \eqref{eq:regret}.
\end{opt}

Finally, we stress that while regret minimization typically focuses on time dimension, the number of caches and the catalog size in (OFP) are large enough to induce high regret (independently of $T$). Hence, it is crucial to study how the regret depends on $N$ or $B_v, v\in\mathcal{V}_c$.

\subsection{Problem Properties}

We prove that (OFP) is an OCO problem by establishing the concavity of $J_t(y)$ with respect to $y_t$. Note that we propose here a different formulation from the typical femtocaching model \cite{femtocachingTransIT13}, by including routing variables and the request arrival events. This re-formulation is imperative in order to fit our problem to the OCO framework, but also because otherwise (e.g., if were using \cite{femtocachingTransIT13}) we would need to make in each slot a computationally-challenging projection operation. This will become clear in the sequel.


First, note that we can simplify the expression of $J_t(y)$ by exploiting the fact that there is only one request at each slot. Let $\hat{n}$ be the content and $\hat{u}$  the location of the request in slot $t$. Clearly, $\hat{n},\hat{u}$ depend on $t$ but the notation is omitted for brevity. 
Then $J_t(y_t)$ is zero except for $R_{t}^{\hat{n}, \hat{u}}$. Simplifying the notation by setting $d_{v, \hat{u}, \hat{n}}=d_v$ and $f_{v,\hat{u}, \hat{n}}=f_v$, $\forall v$, \eqref{eq:biput} reduces to: 
\begin{align}
	J( y )\triangleq\,\,\,\,\,\,\max_{f\geq 0 }  \,\,\,&\sum_{v\in {\cal V}_c({\hat{u}}) } d_vf_v \label{eq:opt_zb1}\\
	\text{s.t. }& \sum_{ v\in {\cal V}_c({\hat{u}})  } f_v\leq1 \label{eq:opt_zs}\\
	&  f_v\leq \left\{\begin{array}{ll}
		y_v, &  v\in { {\cal V}_c({\hat{u}})   }, \\
		0, & v\notin {\cal V}_c({\hat{u}}).
	\end{array}\right. \label{eq:opt_zc}
\end{align}

It suffices to prove the following lemma.
\begin{lem}[theorem style=plain]{}{convx}
	The function $J(y)$ is concave in its domain $\Yc$.
\end{lem}
\begin{proof}
Consider two feasible caching vectors $y_1,y_2\in\Yc$. We will show that:
\[
J\big(\lambda y_1+(1-\lambda)y_2\big)\geq \lambda J(y_1)+(1-\lambda) J(y_2),~\forall \lambda \in [0,1].
\]
We begin by denoting with $f_1$ and $f_2$ the routing vectors that maximize \eqref{eq:opt_zb1} for $y_1$, $y_2$ respectively. Immediately, it is $J(y_i)=\sum_v d_vf_{v,u}$, for $v=1,2$. Next, consider a candidate vector
$y_3=\lambda y_1+(1-\lambda)y_2$, for some $\lambda\in [0,1]$. We first show that routing $f_3=\lambda f_1+(1-\lambda) f_2$ is feasible for $y_3$. By the feasibility of $f_1$, $f_2$, we have: 
\begin{align}
	\sum_j f_3^j=\sum_j (\lambda f_1^j+(1-\lambda) f_2^j)= \lambda \sum_j f_1^j+(1-\lambda)  \sum_j f_2^j\leq\!1, \nonumber
\end{align}
which proves that $z_3$ satisfies \eqref{eq:opt_zs}. Also, for all $j$ it is:
\begin{align*}
	f_3^j&=\lambda f_1^j+(1-\lambda) f_2^j \leq \lambda y_1^j+(1-\lambda) y_2^j=y_3^j,
\end{align*}
which proves that $f_3$ also satisfies \eqref{eq:opt_zc}; hence $f_3\in \mathcal{F}(y_3)$. It follows that:
\begin{equation}
	J(y_3)\triangleq  \max_{f\in \mathcal{F}(y_3)}\sum_j d^j f^j\geq \sum_j d^j f^j_3. \nonumber
\end{equation}
Combining the above, we obtain:
\begin{align}
	&J\big(\lambda y_1+(1-\lambda)y_2\big)=J(y_3)\geq \sum_j d^j f^j_3=\nonumber \\
	&\lambda \sum_j d^j f^j_1+(1- \lambda) \sum_j d^j f^j_2\! =\!\lambda J(y_1)+(1-\lambda) J(y_2) \nonumber
\end{align}
which establishes the concavity of $J(y)$.	
\end{proof}
Observe now that the term $-\sum_{t=1}^TJ_t(y_t)$ that appears in the regret definition is convex, and the $\max$ operator applied for all possible request arrivals preserves this convexity. This places (OFP) on the framework of OCO. For the benefit of future work we  mention that the online caching problem remains OCO when we consider (i)  more general graphs, (ii) other families of convex functions $J_t(y_t)$, and (iii) extra convex constraints.

Finally, we can show with a simple example that $J_t(\cdot)$ does not belong to class $\mathbf{C}^1$, i.e., it is not differentiable everywhere. Consider a network with a single file $N=1$, and two caches with $B_1=B_2=1$, that serve one user with utility $d_1=d_2=1$. Assume that $y^t_{1,1}=y^t_{1,2}=0.5-\epsilon$ for some very small $\epsilon$. Notice that the partial derivatives $\partial J/\partial y^t_{1,1}=\partial J/\partial y^t_{1,2}=1$ (equal to $d$). But if we suppose a slight increase in caching variables such that $\epsilon$ term is removed, then the partial derivatives become zero. Since $y^t_{1,1}=y^t_{1,2}=0.5$, the two caches combine for the entire content file and yield maximal utility, hence extra caching of this content will not improve further the obtained utility. The same holds in many scenarios which make it impossible to guess when the objective changes in a non-smooth manner (having points of non-differentiability). Hence our algorithm will be based on supergradients.

\section{Bipartite Supergradient Caching Algorithm}\label{sec:bipartite}

Our solution employs an efficient lightweight gradient-based algorithm for the caching decisions, which also incorporates the optimal routing.

\subsection{Optimal Routing}
 
We first examine the routing decisions. Recall that the routing of a file is naturally decided after a request is submitted, at which time the caching $y_t$ has been determined, Fig.~\ref{fig:model-oco}. Thus, in order to decide $f_t$ we will assume the request $R_t$ and the cache configuration $y_t$ are given. The goal of routing is to determine which chunks of the requested file are fetched from each cache. 

Specifically, let us fix a request for file $\hat{n}$ submitted to location $\hat{u}$ as before, 
we may recover an optimal routing vector as one that maximizes the utility: 
\begin{align}
	\hat{f}\in \,\,\,\,\,\,{\arg\max}_{f\geq 0, \eqref{eq:opt_zs},\eqref{eq:opt_zc} }  \,\,\,&\sum_{v\in\mathcal{V}_c(\hat{u}) } d_vf_v. \label{eq:opt_routing}
\end{align} 
Ultimately, the routing at $t$ is set to be:
\[
f^t_{v,u,n} = \left\{\begin{array}{ll}
\hat{f}_{v} & \text{if } n=\hat{n},~ u=\hat{i}, \\
0 & \text{otherwise}.
\end{array}\right.
\]
Problem \eqref{eq:opt_routing} is a Linear Program (LP) of a dimension at most $\text{deg}$. Computationally, solving such a problem is very cheap and can be done by means of the interior point method, or the simplex method \cite{bertsimas1997introduction}. Interestingly, however, due to the special problem structure, a solution can be found by inspection as follows. First, order the reachable caches in decreasing utility, by letting $\phi(.)$ be a permutation such that $d_{\phi(1)}\geq d_{\phi(2)} \geq \dots \geq d_{\phi(|\mathcal{V}_{c}(\hat{u})|)}$. Start with the first element and set $f_{\phi(1)}=\min\{1,y_{\phi(1),\hat{u}, \hat{n}}\}$. Continue in the following manner, as long as the total routing after $k$ steps is $\sum_{v=1}^k f_{\phi(k)}<1$ proceed to the $k+1$ cache and set: 
\begin{equation}
f_{\phi(k)}=\min\{1,\sum_{v=1}^k f_{\phi(k)}+y_{\phi(k), \hat{u},\hat{n}}\}-\sum_{v=1}^k f_{\phi(k)}. \nonumber
\end{equation}
At every step, we visit the cache with the highest utility and route the available chunks until completing the entire content file. The iterative process stops when either all the reachable caches are visited, or $\sum_{v=1}^k f_{\phi(k)}=1$ for some $k$, where in the latter case the rest caches have $f_v=0$. Both approaches may be helpful in practice. By explicitly solving the LP we also obtain the value of the dual variables, which, as we shall see, will help us to compute the supergradient.

\subsection{Optimal Caching - BSCA Algorithm}

Since $J(y)$ is not necessarily differentiable everywhere the gradient $\nabla J(y)$ might not exist for some $y_t$. Therefore we need to find at each slot a supergradient direction.\footnote{A \emph{Supergradient} $g$ is the equivalent of subgradient for concave functions, i.e., $J(y)\geq J(y')-g^T(y'-y),~\forall y'\in \Yc$. } We describe next how this can be achieved. Consider the partial Lagrangian of \eqref{eq:opt_zb1}: 
\begin{equation}
L(y,f,\alpha,\beta)\!=\!\sum_{v\in\mathcal{V}_c(u)}d_vf_v+\alpha\big(1-\sum_{v\in\mathcal{V}_c(u) } f_v\big)+\sum_{v\in\mathcal{V}_c(u)}\beta_v(y_v-f_v) \label{eq:lagrangian}
\end{equation}
where $d_v\triangleq w^{v,\hat{u}, \hat{n}}$, and define the auxiliary function:
\begin{equation}
\Lambda (y,\beta)=L(y,f^*,\alpha^*,\beta)\triangleq \min_{\alpha\geq 0}\max_{f\geq 0}	L(y,f,\alpha,\beta).  \label{eq:lambda}
\end{equation}
From the strong duality property of linear programming, we may exchange $\min$ and $\max$ in the Lagrangian, and  obtain:
\begin{equation}\label{eq:strdual}
J(y)=\min_{\beta \geq 0}\Lambda (y,\beta).
\end{equation}

\begin{lem}[theorem style=plain]{Supergradient}{subgradient}
Let $\beta^*(y)\triangleq \arg\min_{\beta\geq 0}\Lambda (y,\beta)$ be the vector of optimal multipliers of  \eqref{eq:opt_zc}. Define:
\begin{align}\label{eq:supergrad}
g^{n,u,v} = \left\{\begin{array}{ll}
\beta^{v,*}(y) & \text{if } n=\hat{n},~ u=\hat{u}, ~v\in {\cal V}_c(u) \\
0 & \text{otherwise}.
\end{array}\right.
\end{align}
The vector $g\!\in\!\mathbb R^{N\times V_r \times V_c}$ is a supergradient of $J$ at $y$, i.e., it holds $J(y)\!\geq\! J(y')\!-g^T(y'-y),~\forall y'\!\in\! \Yc$. 
\end{lem}
\begin{proof}
First note that we can write: 
	\begin{align*}
		&J(y)\stackrel{\eqref{eq:strdual}}{=}\min_{\beta \geq 0}\Lambda (y,\beta)\triangleq \Lambda \big(y,\beta^*(y)\big) \stackrel{(a)}{=}\Lambda \big(y',\beta^*(y)\big)-\beta^*(y)^\top(y'-y).
	\end{align*}
	Where $(a)$ holds since it is:
	\begin{align}  
		L(y',f,\alpha,\beta)=\sum_{v\in\mathcal{V}_c(\hat{u}) }\!d_vf_v+\alpha\big(1-\sum_{v\in\mathcal{V}_c(\hat{u})} f_v\big)+\sum_{v\in\mathcal{V}_c(\hat{u})}\beta^v(y'^{\,v}\!-\!f_v)\nonumber
	\end{align}
	and by applying \eqref{eq:lambda}, where the optimization is independent of variables $y$ (or $y'$), we obtain $\Lambda(y',\beta)=L(y', f^*, \alpha^*, \beta)$, with $\alpha^*$ and $f^*$ being the same as those appearing in $\Lambda(y,\beta)=L(y, f^*, \alpha^*, \beta)$ (since their calculation is independent of $y$). Hence, we can subtract the two expressions (observe the linear structure of \eqref{eq:lagrangian}), plug in a certain vector $\beta$ and obtain:
	\begin{equation}
		\Lambda\big(y,\beta^*(y)\big)-\Lambda\big(y', \beta^*(y)\big)=-\beta^*(y)^\top(y'-y).
	\end{equation}
	where $\beta^*(y)\!=\!\arg\min \Lambda(y,\beta)$. Finally, note that it holds $\Lambda (y',\beta^*(y))>\Lambda (y',\beta^*(y'))$ by definition of $\beta^*$, hence:
	\begin{align}
		J(y)&=\Lambda \big(y',\beta^*(y)\big)-\beta^*(y)^\top(y'-y)\nonumber \\
		&\geq \Lambda\big(y', \beta^*(y')\big)-\beta^*(y)^\top(y'-y) \nonumber \\
		&=J(y')-\beta^*(y)^\top(y'-y), \nonumber
	\end{align}
which concludes the proof.
%
\end{proof}


\begin{algorithm}[t]
	\nl \textbf{Input}: $\{e_{v,u}\}_{(v,u)}$; $\{B_v\}_v$; $\mathcal{N}$; $\{d^{v,u,n}\}_{(v,u,n)}$; $\eta_t=\Delta_{\Yc}/K\sqrt{T}$.\\%
	\nl \textbf{Output}: $y_t$, $\forall t$.\\%
	\nl \textbf{Initialize}: $\hat{n}, \hat{u}$, $y_1$ arbitrarily.\\%
	\nl \For{ $t=1,2,\ldots $  }{
		\nl Observe  request $R_{t}$ and set ${\hat{n}, \hat{u}}$ for which  $R_{t}^{\hat{n},\hat{u}}\!=\!1$\,;\\
		\nl Find the  routing $f_t$ solving \eqref{eq:opt_zb1}-\eqref{eq:opt_zc}; $\%$ \emph{decides routing}		\\
		\nl Calculate the accrued utility $J_t(R_t,y_t)$\,;\\
		\nl Calculate the supergradient $g_t$ for $\hat{n}, \hat{u}$ using (\ref{eq:supergrad}); \\%
		\nl Update the vector $q_{t+1}=y_{t}+\eta_tg_t$\,; \\%
		\nl Project: $y_{t+1}=\Pi_{\mathcal{Y}}\left(q_{t+1}\right)$; \,\,$\%$ \emph{decides caching}\\%
	}
	\caption{{\small \!Bipartite Supergradient Caching Algorithm }}	\label{alg1}
\end{algorithm}

Intuitively, the dual variable $\beta^{j,*}(y)$ (element of vector $\beta^*(y)$) is positive only if the respective constraint \eqref{eq:opt_zc} is tight (and some other conditions are met) which ensures that increasing the allocation $y^{\hat{n},\hat{u},v}$ will induce a benefit in case of a request with $R^{\hat{n},\hat{u}}\!=\!1$ occurs in future. The actual value of $\beta^{v,*}(y)$ is proportional to this benefit. The reason the algorithm emphasizes this request, is that in the online gradient-type of algorithms the last function (in this case a linear function with parameters the last request) serves as a corrective step in the ``prediction'' of future. Having this method for calculating a supergradient direction, we can extend the seminal online gradient ascent algorithm \cite{zinkevich2003online}, to design an online caching policy for (OFP). In detail: 

\begin{box_example}[detach title,colback=blue!5!white, before upper={\tcbtitle\quad}]{Bipartite Subgradient Caching Algorithm (BSCA).}
	\small
Upon a request $R_t$, find the optimal routing and then adjust the caching decisions with a supergradient: 
\[
y_{t+1}=\Pi_{\mathcal{Y}}\left(y_t+\eta_t g_t\right),
\]
where $\eta_t$ is the stepsize, $g_t$ can be taken as  in Lemma~\ref{lem:subgradient}, and
\begin{equation}
\Pi_{\mathcal{Y}}\left(q\right)\triangleq \argmin_{y\in\Yc}\|q-y \|, \nonumber
\end{equation}
is the Euclidean projection of the argument vector $q$ onto $\mathcal{Y}$, which is performed with the projection algorithm of Section \ref{sec:projection} in Chapter~\ref{ch:3}.\end{box_example}

Algorithm \ref{alg1} explains how BSCA can be incorporated into the network operation for devising the caching and routing decisions in an online fashion. Note that the algorithm requires as input only the network parameters $\ell_{ij}, B_j, \mathcal{N}, w^{n,i,j}$. The stepsize $\eta_t$ is computed using the set diameter $\Delta_{\mathcal{Y}}$, the upper bound on the supergradient  $G$, and the horizon $T$. The former two depend on the network parameters as well. Specifically, define first the diameter $\Delta_{\mathcal{S} }$ of  set $\mathcal{S}$ to be the largest Euclidean distance between any two elements of this set. In order to calculate this quantity for $\mathcal{Y}$, we select two vectors $y_1, y_2\in \Yc$ which cache $B_j$ different files at each cache $j\!\in\! \mathcal{J}$, and hence we obtain:
\begin{equation}\label{eq:diam}
\Delta_{\mathcal{Y}}=\sqrt{\sum_{n,v}(y^1_{v,u,n}-y^2_{v,u,n})^2}=\sqrt{\sum_{v\in \mathcal{V}_c}2B_v}\leq \sqrt{2BV_{c}},
\end{equation}
where $B\!=\!\max_{v}B_v$. Also, we denote with $G$ the upper bound on the norm of the supergradient vector. By construction this vector is non-zero only at the reachable caches, and only for the specific content. Further, its smallest value is zero by the non-negativity of Lagrangian multipliers, and its  largest is no more than the \textbf{maximum utility}, denoted with $d_{(1)}$. Therefore, we can bound the supergradient norm as follows: 
\begin{equation}\label{eq:lipschitz}
\|g\|= \sqrt{\sum_{ v\in\mathcal{V}_c(\hat{u})} (d_{(1)})^2} \leq d_{(1)}\sqrt{\text{deg}}\triangleq K,
\end{equation}
where $\text{deg}\!=\!\max_u |\mathcal{V}_c(u)|$ is the maximum number of reachable caches from any location $i\in\mathcal{I}$. 

The algorithm proceeds as follows. At each slot, the system observes the submitted request (or, \emph{reacts} to a request) (line 4). Then, it calculates the supergradient (lines 5) using low-complexity computations, e.g. by solving an LP with at most $\text{deg}$ variables and finding the dual variables. This information is used to update the caching policy $\hat{y}$ (line 6), which is then projected onto the feasible caching region (line 7). Then the request arrives (line 8) and the optimal routing decision is calculated based on the current cache configuration (line 8).\footnote{Please note that one can shift the algorithm's steps (preserving the relative order) and start from the request submission.} Since the supergradient computation in  line 5 of BSCA, and the optimal routing  explained in the previous subsection, both require the solution of the same LP, it is possible to combine these as follows. When the optimal routing is found, the dual variables can be stored and used for the direct computation of the supergradient in the next iteration of BSCA. Finally, we mention that the algorithm state is only the vector $y_t$, and therefore the memory requirements of the algorithm are very small.

\subsection{Performance of BSCA}
Following the rationale of the online gradient descent analysis in \cite{zinkevich2003online}, we show that our policy achieves no regret and we analyze how the different network and content file parameters affect the regret expression. The next theorem holds.
\begin{thm}[theorem style=plain]{Regret of BSCA}{BSCA}
\[
\texttt{Reg}_T(\textup{BSCA})\leq 
w^{(1)}\sqrt{2\textup{deg} V_c BT}, \,\,\,\,\,\, B=\max_v B_v
\]
\end{thm}

\begin{proof}
Using now the non-expansiveness property of the Euclidean projection \cite{Ber99book}, we can bound the distance of each new value $y_{t+1}$ from the best static policy in hindsight $y^*$, as follows:
\begin{align}
&\|\Pi_{\mathcal{Y}}\left(y_t+\eta_t  g_t\right)-y^*\|^2\leq \|y_{t}+\eta_tg_t-y^*\|^2=\nonumber \\
&\|y_t-y^*\|^2+2\eta_t{g_t}^\top(y_t-y^*)+\eta_t^2\|g_t\|^2,\label{eq:ineqsub1}
\end{align}
where we expanded the norm. If we fix the step size $\eta^t=\eta$ and sum telescopically over all slots until $T$, we obtain:
\begin{equation*}
\|y_{T}-y^*\|^2\!\!\leq\! \|y_1-y^*\|^2+2\eta\sum_{t=1}^T{ g_t}^\top(y_t-y^*)+\eta^2\sum_{t=1}^T\!\|g_t\|^2.
\end{equation*}
Since $\|y_{T}-y^*\|^2\geq 0$, rearranging the terms and using $\|y_1-y^*\|\leq \Delta_{\Yc}$ and $\|g_t\|\leq K$, we obtain: 
\begin{equation}\label{eq:teleonl}
\sum_{t=1}^T{g_t}^\top(y^*-y_t)\leq \frac{\Delta_{\Yc}^2}{2\eta} +\frac{\eta TK^2}2.
\end{equation}

Any concave function $J_t$ satisfies the inequality $J_t(y_t)\leq J_t(y)+ {g_t}^T(y_t-y)$, $\forall y\in\mathcal{Y}$, and hence the same holds for the $J_t$ that maximizes regret. Plugging in $\texttt{Reg}_T(\cdot)$, we get:
\begin{align*} 
\texttt{Reg}_T(BSCA)\!=\!\sum_{t=1}^T\big(J_t(y^*)-J_t(y_t)\big)\! \leq\! \sum_{t=1}^T{g_t}^{\!\top}\!(y_t\!-\!y^*)\stackrel{\eqref{eq:teleonl}}{\leq}\! \frac{\Delta_{\Yc}^2}{2\eta} +\frac{\eta TK^2}2.
\end{align*}
We can minimize the regret bound by selecting the optimal step size. That is, using the first-order condition w.r.t. $\eta$, for the RHS above we obtain $\eta^*=\Delta_{\Yc}/K\sqrt{T}$, which  yields:
\begin{align}
	\texttt{Reg}_T(BSCA)\leq  \Delta_{\mathcal{Y}}K\sqrt{T} . \label{eq:conv_reg}
\end{align}
The theorem follows from \eqref{eq:diam}-\eqref{eq:lipschitz}.
\end{proof}

Theorem \ref{thm:BSCA} shows that the regret of BSCA scales as $O(\sqrt{T})$ and therefore BSCA solves (OFP). The learning rate $O(\sqrt{T})$ is known to be the fastest possible learning rate achievable when functions $f_t$ are general concave functions \cite{Abernethy08}--it can be improved only if functions $J_t$ have additional properties such as being strongly convex, which does not hold in our problem. This establishes that BSCA offers the fastest possible way to learn to cache and route when (i) the distribution of popularity is unknown and highly time-varying, and (ii) we disregard  constants that do not scale with $T$. Any improvement on regret we hope to achieve over BSCA, can only be in relation to these constants. 

Furthermore, the regret expression is indicative of how fast the algorithm \emph{learns} the right caching decision, and therefore the detailed constants we obtain in the theorem are of great value. For example, we see that the regret bound is independent of the size of the catalog $N$. This is particularly important for caching problems where the catalog is typically the physical quantity that drives the problem's dimension. Another interesting observation from the constants is that the learning rate of the algorithm might become slow (i.e. resembling regret behavior of $\sim\!\!O(T)$) when $B$ is very large, or comparable to $T$. This justifies empirical observations suggesting that in order to extract safe conclusions about caching performance under fluctuating popularity, one should simulate datasets whose size $T$ is larger than the cache size $B$. Note that in practice it is $B\!<<\!N$, since otherwise the caching problem becomes trivial, i.e., any policy performs well. 


Despite the superior learning rate performance, projected gradient algorithms like BSCA may become problematic for large problem instances because they involve an Euclidean projection step (see line 9 of BSCA). Projection operations are oftentimes computationally expensive \cite{duchi-projection, schmidt-projection, hao-projection}, and it is not uncommon to constitute the bottleneck step in, otherwise-fast, algorithms. However, our BSCA algorithm is engineered to have a simple projection step, which can be addressed using the algorithm of Section \ref{sec:projection} in at most $O(N\log N)$ computations. This is an important feature of BSCA, which allows its implementation in caches that can fit a large number of contents.

\subsection{Choosing the Step Size}
We proved that by using the specific constant step size we are able to obtain sublinear regret. However, calculating  $\eta^*=\Delta_{\Yc}/K\sqrt{T}$ requires knowledge of the time horizon $T$, which in some cases might not be convenient. To relax this requirement, one can use the standard \emph{doubling trick} where essentially we select a certain time horizon $T$ and then ``restart'' the algorithm, calculating the optimal step size rule for $2T$, and so on. This is known to  add only a small constant factor to the regret bound \cite[Sec. 2.3]{Shalev12}. An alternative approach is to use a time varying step. In detail, if we start from \eqref{eq:ineqsub1} and sum telescopically for the first $T$ slots and different step sizes, we obtain the  regret performance under varying step size $\texttt{Reg}_{T}^v(BSCA)$:  
\begin{equation}\label{eq:2nd-ubound}
	\texttt{Reg}_{T}^v(BSCA)\leq  \frac{\Delta_{\mathcal{Y}}^2}{\eta_{T}}+ \frac{ K^2 \sum_{t=1}^T\eta_t }{2}.
\end{equation}
Now, it is easy to see that if we set $\eta_t=1/\sqrt{t}$, then the two terms in \eqref{eq:2nd-ubound} yield factors of order $O(\sqrt{T})$, and we obtain:
\begin{equation}\label{eq:22nd-ubound}
\texttt{Reg}_{T}^v(BSCA)\leq  \frac{\Delta_{\mathcal{Y}}^2 \sqrt{T}}{2}+ \Big(\sqrt{T}-\frac{1}{2} \Big)K^2. 
\end{equation}

Comparing the two expressions for the regret of BSCA, \eqref{eq:conv_reg} and \eqref{eq:22nd-ubound}, we see that they are both sublinear with the same order of magnitude learning rate $\sqrt{T}$. Their exact relationship depends on the relative values of parameters $K$ and $\Delta_{\mathcal{Y}}$.

\section{Model Extensions and Performance}\label{sec:extensions}

The model and algorithms we have introduced in this paper can be used in much more general settings than (OFP). First of all, as it was explain in Sec. \ref{chapter4:bipartite} the bipartite  model can be used for caching networks that do not have, at-a-first-glance, the bipartite structure, as long as they do not have hard link capacity constraints or load-dependent routing costs. Figure \ref{Fig:bipartite} showcases two representative cases. Namely, a CDN that can be transformed to a bipartite graph by mapping each uncapacitated path to a link in the bipartite graph, and a system of connected memories that jointly serve requests submited at a central I/O point. Such multi-memory paging systems arise in data centers or disaggregated server systems  \cite{dredbox}, \cite{disag}, and BSCA can effectively determine the routing of requests and the cached contents at each memory.




\begin{figure}[t!]
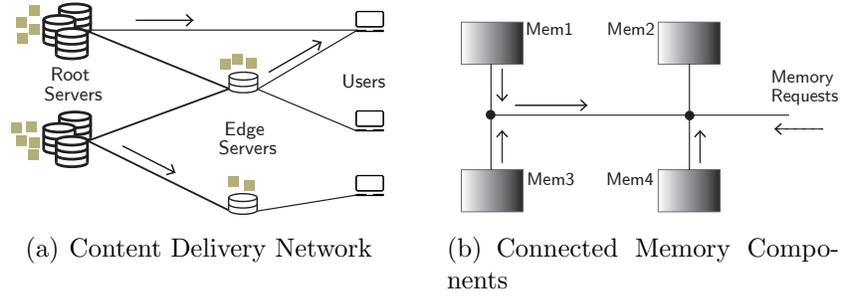

	\centering
	\subfigure[Content Delivery Network]{
		\centering
		\includegraphics[width=2.0in]{intro-model2}\label{Fig:cdn}
	}\,\,\,\,\,\,\,
		\subfigure[Connected Memory Components]{
		\centering
	\includegraphics[width=2.0in]{disag-example}\label{Fig:disag}}
	\caption{Bipartite model for general networks. (a) A content delivery network with root and edge servers, and uncapacitated links. Each path can be modeled as a super-link in a bipartite graph of (root or edge) servers and end-users. (b) A network of connected memory elements in a disaggregated server system. }
	\label{Fig:bipartite}
\end{figure}



Another interesting model extension is when the utility parameters change dynamically. This means that either due to the value of caching or the cost of routing, parameters $d^{t}_{v,u,n}$ vary over time. For example, the connectivity graph $\{e_{v,u}^t\}_{(v,u)}$ might change from slot to slot, because the link channel gains vary. BSCA can handle this effect by simply updating the step where in each slot now we observe both the submitted request $R_{t}^{\hat{n}, \hat{u}}$ and the current utility vector $d_{t}$. Interestingly, this generalization does not affect the regret bound which already encompasses the maximum possible distance between the utility parameters.

Finally, an important case arises when there is cost for prefetching the content over the backhaul SBS links. First, note that BSCA might select to reconfigure the caches in a slot $t$ ($y_t\neq y_{t-1}$) even if the requested files were already available, if this update improves the aggregate utility.\footnote{For instance, in a 3-cache network with 1 file, BSCA might decide to change the $t-1$ configuration $[0.33, 0.33, 0.33]$ to $[0.66, 0.44, 0.33]$ if this increases the content delivery utility, e.g., if cache 1 is  better.} Nevertheless, such changes induce cost due to the bandwidth consumed for fetching the new file chunks, and therefore the question ``\emph{when should a file be prefetched?}'' is of very high importance in caching. This question arises of course in other policies as well, such as the LRU-ALL \cite{giovanidis-mLRU}, that make proactive updates in order to improve their cache hit ratio. Unlike these policies, however, BSCA can be modified in order to make reconfiguration decisions by balancing the expected utility benefits and induced costs.

Namely, if we denote with $c_{n,j}$ the cost for transferring one unit of file from the origin servers to cache $v\!\in\!\mathcal{V}_c$, then we can define the utility-cost function:
\begin{equation}
	H_t(y_t, y_{t-1})=J_{t}(y_t)  - \sum_{n\in\mathcal{N}}\sum_{v\in\mathcal{V}_c}c_{n,j}\max\{y^{t}_{v,n}- y^{t-1}_{v,n},0\}
\end{equation}
where $f_t(y_t)$ is given in \eqref{eq:biput}, and the convex $\max$ operator ensures that we pay cost whenever we increase the cached chunks of a file $n$ at a cache $j$ (but not when we evict data). Function $H_t$ is concave and hence BSCA can be easily updated to accommodate this change. Namely, it suffices to use the supergradient $q_t$ for $H_t(\cdot)$ instead of the supergradient $g_t$ of $J_t(\cdot)$. Using basic subgradient algebra we can write $h_t=g_t + q_t$, where
\begin{align}\label{eq:supergrad2}
	q_{t}^{\hat{n},v}=\left\{\begin{array}{ll}
		-c_{ \hat{n},v }, & \text{if  }\,y^{t}_{v,\hat{n}} - y^{t-1}_{v,\hat{n}}>0 \\
		0 & \text{otherwise}
	\end{array}\right.
\end{align}
which is calculated for each request $R_{t}^{ \hat{n}, j }$ and every cache $j$. The policy's learning rate is not affected by this change, and we only need to redefine parameter $G$ in the regret, by adding the maximum value of $q$. This makes our policy suitable for many placement problems beyond caching, e.g., service deployment \cite{hong-hou-jsac}, where reconfigurations might induce non-negligible costs.




\begin{figure*}[!t]
	\centering
	\subfigure[Example femtocaching network.]
	{
		\includegraphics[width=2.0151in]{bipartite-example2}
		\label{Fig:bipartite-example}
	}\,\,\,\, 
	\subfigure[Utility of BSCA \& Competitors]
	{
		\includegraphics[width=2.0151in, height=1.85in]{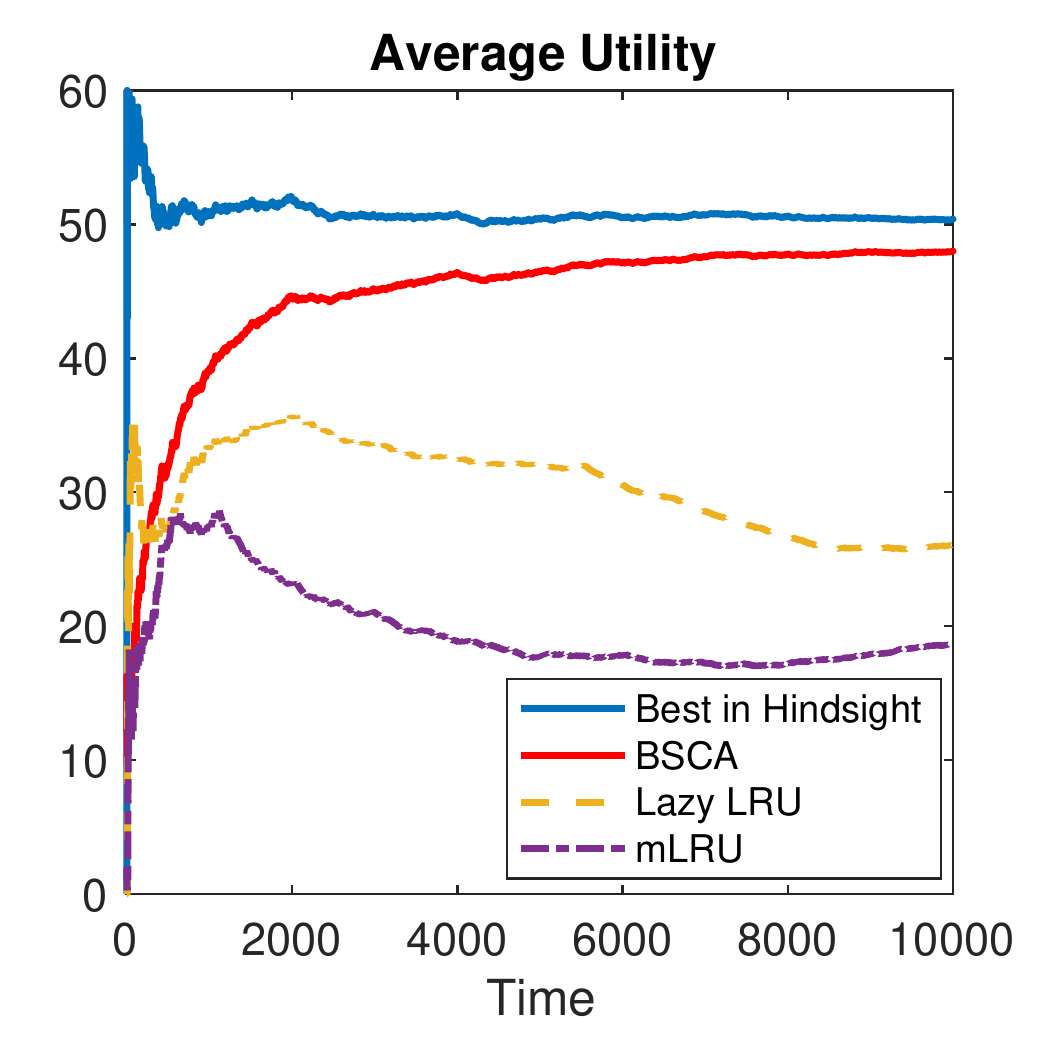}
		\label{Fig:bipartite-utility}
	}
	\caption{(a) User $3$ is connected to caches $v=2$ and $v=1$; a request for $n$ arrives at user $u=2$, $R_{t}^{2,n}=1$, and fetching it from cache $v=1$ yields $d_{1,2,n}$. (b) The average utility of BSCA and competitor policies for this network (from \cite{paschos-infocom19}). 
	} 
	\label{Fig:bipartite-sims}	
\end{figure*}

The generality and robustness of BSCA can be easily observed numerically as well. We consider a bipartite graph with 3 caches of size $B=10$, and 4 user locations, see  Fig.~\ref{Fig:bipartite-example}. The utility vector is $d_{n}=(1,2,100), \forall n$, hence an efficient policy needs to place popular files on cache $3$. The network is fed with stationary Zipf requests from a catalog of $N=100$ files, and each request arrives at a user location uniformly at random. We compare BSCA to the best static policy in hindsight, and state-of-the-art reactive policies for caching networks, namely \emph{(i)} the multi-LRU policy proposed in \cite{giovanidis-mLRU} where a request is routed to a given cache (e.g., the closest) which is updated based on the LRU rule; and \emph{(ii)} and the $q$-LRU policy with the ``lazy'' rule \cite{leonardi-implicit} for $q=1$, which works as the multi-LRU but updates the cache only if the file is not in any other reachable cache.

Fig. \ref{Fig:bipartite-example} depicts the considered network, and Fig. \ref{Fig:bipartite-utility} the comparison results. We see that BSCA converges to the best static policy in hindsight, which verifies that it is a universal no regret policy. That is, BSCA learns gradually what files are popular and increases their cache allocation at the high utility caches, while adjusting accordingly the less popular files at the other caches. The second best policy is lazy-LRU which is outperformed by BSCA by $45.8\%$. On the other hand, both lazy-LRU and mLRU have comparable performance and, interestingly, BSCA has lower performance for the first 500 slots (from lazy-LRU) but quickly adapts to the requests and outperforms its competitors.

\section{Discussion of Related Work}\label{sec:related}

The first OCO-based caching policy was proposed in \cite{paschos-infocom19} which reformulated the caching problem and embedded a learning mechanism into it. This opened the road for studying the broader network caching problem. In CNs one needs to jointly decide which cache will satisfy a request (routing) and which files will be evicted (caching), and these decisions are perplexed when each user is connected to multiple caches over different paths.\footnote{Note that for independent caches, the CN problem is essentially reduced to the single-cache version, as, e.g., in \cite{gunduz-reinforcement, giannakis-q-learning, mihaela-video-caching}.} Thus, it is not surprising that online policies for CNs (including femtocaching) are under-explored.

Placing more emphasis on the network, \cite{ioannidis_icn17, ioannidis-jsac18} introduced a joint routing and caching  algorithm for general graphs, assuming that file popularity is stationary. On the other hand, proposals for reactive CN policies include: \cite{Blaszczyszyn2014Geographic} which designed a randomized caching policy for small-cell networks; \cite{spyropoulos-icc} which studies a wireless network coordinating caching and transmission decisions across base stations; \cite{avrachenkov-acm17} that proposed distributed cooperative caching algorithms; and \cite{C_Dehghan_16} that introduced a TTL-based utility-cache model. Albeit important, all these solutions \emph{presume that the popularity model is fixed and known}. Clearly, in practice, and especially in femtocaching networks, this assumption does not hold.

More recently, \cite{giovanidis-mLRU} proposed the multi-LRU (mLRU) heuristic strategy, and \cite{leonardi-implicit} the ``lazy rule'' extending $q$-LRU to provide local optimality guarantees under stationary requests. These works pioneered the extension of the seminal LFU/LRU-type policies to the case of multiple connected caches and designed efficient caching algorithms with minimal overheads. Nevertheless, dropping the stationarity assumption, the problem of online routing and caching remains open. The approach presented in this chapter is fundamentally different from \cite{giovanidis-mLRU, leonardi-implicit} as \emph{we embed a learning mechanism into the system operation that adapts the caching and routing decisions to any request pattern (even one that is created by an attacker) and also to changes in network links or utilities}.

BSCA exploits the powerful OCO framework  \cite{hazan-book}, which has been used with great success in a variety of large-scale problems \cite{Mert18}. OCO is particularly attractive for the design of network protocols (as our caching solution here) for two reasons. Firstly, it offers general performance bounds that hold with minimum assumptions. Secondly, the algorithms are - potentially - simple, elegant, and can be implemented as protocols. In fact, we use here (a variation of) the online counterpart of the gradient algorithm which has played a key role in the design of Internet protocols under stationary assumptions.

On the other hand, employing OCO in networks raises new challenges. It requires careful problem formulations, handling constraints that typically do not appear in OCO, and ensuring that the bounds do not depend on parameters that can admit large values (as the content catalog size). Furthermore, a key challenge in OCO algorithms is that the obtained solution at each iteration has to be projected onto the feasible set. For the problem of caching, this projection can be addressed by the efficient algorithm given in Chapter~\ref{ch:3}. Another solution is to employ projection-free algorithms, which however typically have poor learning rates (slow convergence)  \cite{martin-jaggi}.


\chapter{Asymptotic laws for caching networks }\label{ch:6}

In chapter~\ref{ch:4} we saw that the different problems of content caching in networks are prohibitively complex to solve to optimality, especially for large networks. In this chapter, we examine a  caching network arranged in the form of a square grid with a number of nodes that grows to infinity, and we show that in such a regime the analysis can be simplified. 
Following the steps of \cite{infocom,savvas_ToIT,paderborn,J_Gitzenis_14} we show that  a relaxed macroscopic convex program can be used to produce  solutions that scale at optimal  rate, i.e., the gap from the optimal is of smaller order (with respect to size parameters) than the value of the optimization.
Therefore, the macroscopic analysis can be used to  extract safe conclusions about the sustainability of large caching networks, with respect to their capability of withstanding a growing video demand.

\subsubsection{Network sustainability}
The sustainability of wireless multihop networks can be  characterized by means of studying a network that grows in size and examining its capacity scaling laws. In their seminal work \cite{Kumar}, Gupta and Kumar studied the asymptotic behaviour of multihop wireless networks when communications take place between $K$ random pairs. They showed that  the maximum data rate is $O\mleft(\sfrac{1}{\sqrt{K}}\mright)$, hence as the network grows ($K\to\infty$) the per node data rate vanishes to zero. 
This finding  argues against the sustainability of multihop  communications. The limit $O\mleft(\sfrac{1}{\sqrt{K}}\mright)$  arises from the fact that as the network grows, the number of wireless hops between a random pair of nodes also increases, and each wireless transmission is limited by fundamental physical laws \cite{franceschetti}. Hence, the rate $O\mleft(\sfrac{1}{\sqrt{K}}\mright)$ is thought impossible to breach under the classical random communicating pair model. 

However, if each node in the network is interested in a content instead of directly communicating with another node, by replicating the contents  in the network caches the  demands can be served by  nearby caches reducing in this way the number of traversed hops. 
Hence, we are motivated to ask the  question: \emph{Can  a cache-enabled network achieve a sustainable throughput scaling?} 
To answer the question we derive the asymptotic laws for capacity scaling with  content replication and show that caching has a powerful effect on the sustainability of wireless networks since in some  regimes the  $O\mleft(\sfrac{1}{\sqrt{K}}\mright)$ law can be breached.

The fundamental  size parameters are  the number of nodes $K$ and the number of contents $N$, and thus the scaling laws will depend on how fast these increase to infinity. Another key parameter is the cache size $M$ which depicts the number of contents that can be cached at each node. 
Taking $M$ to infinity represents an interesting regime that reflects networks where nodes  are upgraded as storage gets abundant and inexpensive\cite{Roberts}.
We will see that an additional  influencing factor is the content popularity, and in particular the power law exponent $\Zipf$. 
With skewed popularity, some contents are requested multiple times and hence a smaller $M$ has a bigger impact on the system performance. 


\section{Analysis of Large Caching Networks}

\subsubsection{System Model}
\label{sec:replication}

$K$ nodes
are  arranged on a $\sqrt{K}\times\sqrt{K}$ square grid on the plane.  Each node is connected via undirected links to its four neighbours that lie next to it on the same row or column. By keeping the node density fixed and increasing the network size $K$, we obtain a scaling network similar to \cite{Kumar}, with random pair throughput scaling as $O\mleft(\sfrac{1}{\sqrt{K}}\mright)$. Moreover, to avoid boundary effects, we consider a toroidal structure as in \cite{franceschetti}.

\begin{figure}[!htp]
	\centerline{
		\includegraphics[width=0.25\columnwidth]{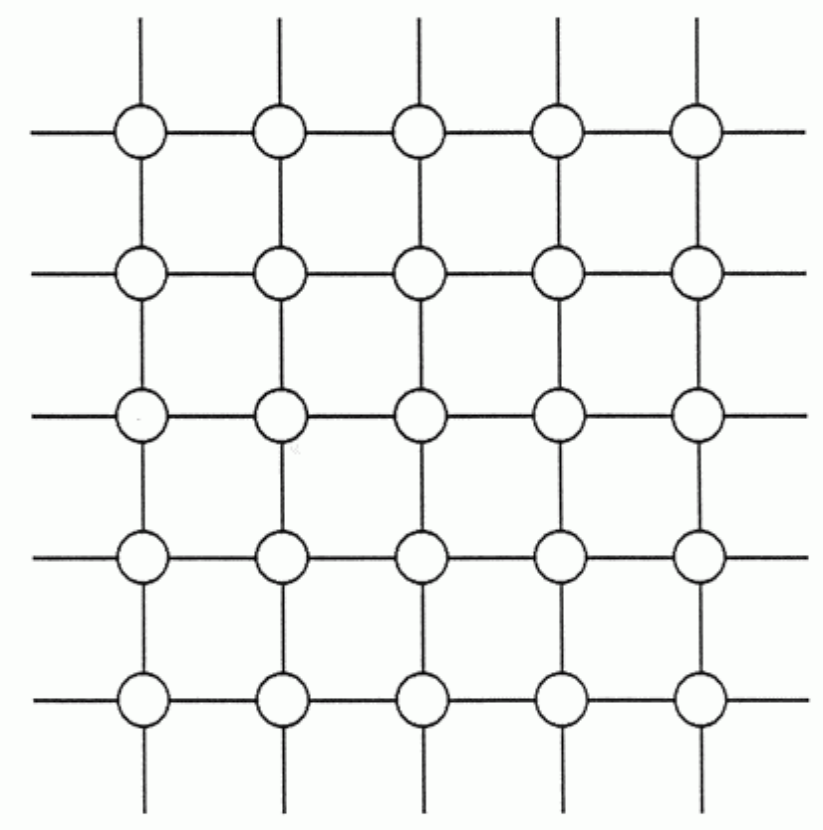}
	}
	\caption{The network topology  is  a toroidal square grid.}
	\label{fig:grid}
\end{figure}


Nodes generate requests to a catalog of contents $\mathcal N \triangleq\{1,2,\dots,N\}$. 
Node $k$  requests  content $n$ with rate $\lambda_n^k$. 
Instead of looking at the capacity region (=set of rate vectors $[\lambda_n^k]$ that are sustainable) as in \cite{niesen}, to obtain a quantitive result, we will consider a symmetric case where all nodes have $\lambda_n^k=\lambda_n,~\forall n$. 
We further consider a Poisson IRM where the requests for content $n$ are due to an independent Poisson process with intensity $\lambda_n=\lambda p_n$, where $[p_n]$ is the  content popularity distribution--here  assumed time-invariant.  
Last, we fix $\lambda=1$ and study the required link capacity to support the communications. 
This is the inverse of the classical approach where we fix the link capacity and study the scaling of maximum supportable throughput. For example, in a network where under fixed link capacities the  throughput scales as $O\mleft(\sfrac{1}{\sqrt{K}}\mright)$, our model will yield required link capacity scaling of $\Theta({\small \sqrt{K}})$ to support constant throughput $\lambda=1$. More generally, the obtained scaling laws for required link capacity can be inverted to reflect the throughput scaling laws. For sustainability we require the link capacity to be as small as possible, ideally to be $O(1)$.

Each node $k$ is equipped with a cache, whose contents are denoted with $\mathcal M_k\subseteq\mathcal N$. If a request at node $k$ regards a content $n\in\mathcal M_k$, then the content can be served locally. Due to the limited storage, $n$  will often be missing from  $\mathcal M_k$, in which case it will be obtained over the network from some other node $w$ such that  $n\in\mathcal M_w$. 
We denote with $[\mathcal{F}_{k,w}]$ the set of routes connecting nearby caches, and we say that a route is feasible if it is a path  that connects a requesting node $k$ to a cache on node $w$ that stores the requested content. The choice  of the sets $[\mathcal{M}_k]$ and $[\mathcal{F}_{k,w}]$ crucially affects the network link loading, and hence to discover the fundamental sustainability laws, we will have to optimize them.

Let $M$ be the cache size measured in the number of contents it can store, which is the same at all nodes, i.e., 
 $|\mathcal M_k| \le M$. 
The important size parameters of the system are $K,N,M$, and the studied regime  satisfies
\begin{align}
	\label{eq:KNM} 
	M & < N\le MK.
\end{align}
The first inequality implies that the cache size $M$ is not enough to fit all contents and hence each node needs to  make a selection of contents to cache. The second inequality  requires that the total  network cache capacity $MK$ (summing up all individual caches) is sufficient to store all contents at least once, and thus each content request can be served.

\subsection{Capacity Optimization}

To yield the correct scaling law of required link capacities, it is necessary to select the best placement and the best routing that minimizes the load at the worst link in the network. Let  $C_{\ell}$ be the  traffic load carried by link $\ell$. We are interested in the following optimization:

\begin{opt}{Worst-cast Link Load Minimization}\vspace{-0.18in}
\begin{align} C^*=~&\text{Minimize}_{[\mathcal{M}_k],[\mathcal F_{k,w}]}\max_{\ell} C_{\ell}\label{opt:wll1} \\
\text{s.t.}~~& |\mathcal{M}_k|\leq M, \quad \forall k \quad\quad\quad\quad\quad\quad ~~ \text{(cache constraint)},\label{opt:wll2} \\
& \sum_k {1}_{\{n\in \mathcal{M}_k\}}\geq 1,\quad \forall n\in\mathcal{N} \quad ~~~\text{(all contents are cached)}, \label{opt:wll3} \\
& [\mathcal F_{k,w}] \quad \text{are feasible routes}.\label{opt:wll4}
\end{align}
\end{opt}

This problem is a variant of the DPR problem for the square grid topology, and as explained in the previous chapter, it is a hard combinatorial problem.
In order to derive the scaling laws we must solve this problem in closed form. To this purpose, we will employ a number of simplifications and arrive at a relaxed problem which is amenable to analysis, while later it will be shown that the solution of the relaxed problem is in fact of the same  order as  the optimal solution of the problem above.
The simplifications are:  (a) the relaxation of the objective to minimizing the average link traffic $\mathop{\mathrm{avg}}_\ell C_\ell$ (instead of the worst-case $\max_{\ell} C_{\ell}$), (b) fixing the routing variables $[\mathcal F_{k,w}]$   to shortest paths, and (c) breaking the coupling between the individual caches $[\mathcal M_k]$, by introducing the notion of \emph{replication density} of content $n$.

\subsubsection{Macroscopic  density reformulation}
For the simplification (c) above, given a placement $[\mathcal M_k]$, consider the frequency of occurrence of each content $n$ in the caches, or {\em replication density} $d_n$ as the fraction of nodes that store content $n$ in the network: \pagebreak[0]
\begin{equation*}
	d_n=\frac{1}K \sum_{k\in\mathcal K} 1_{\{n \in \mathcal{M}_k\}}.
\end{equation*}

Based on this metric, we define a simpler macroscopic problem, where instead of optimizing over content placement, we only optimize the densities of each content (without caring about how the contents are actually cached at each individual node-hence the term macroscopic):
\begin{opt}{Replication Density Optimization}\vspace{-0.18in}
\begin{align} C=~&\text{Minimize}_{[d_n]} \displaystyle \sum_{n\in\mathcal{N}} \left(\frac{1}{\sqrt{d_n}}-1\right) p_n \label{pr:density}\\
\text{s.t.}~~& \sfrac{1}{K}\le d_n \le 1, \quad  \forall n\in\mathcal N\quad  \text{(replication constraint)}\label{pr:density-1}\\
& \sum_{n\in\mathcal N} d_n \le M \quad ~~~ \text{(total network cache constraint)}.\label{pr:density-2}
\end{align}
\end{opt}

In the objective, $\sfrac{1}{\sqrt{d_n}}-1$ approximates (order-wise) the average hop count from a random node to a cache containing $n$. 
The average hopcount (averaged also over the distribution $p_n$) expresses the average link load per request. Additionally, the constraint $\sum_{n\in\mathcal N} d_n \le M$ in \eqref{pr:density-2} reflects another relaxation, whereby the cache size constraint is not satisfied at every node, but only on average across the network. 
Due to this last relaxation, 
a feasible macroscopic solution 
 may yield content densities that force some nodes to cache more than $M$ contents.
It is clear that any feasible solution of problem~\eqref{opt:wll1}-\eqref{opt:wll4}  yields a content density  $[d_n]$ that is feasible for problem~\eqref{pr:density}-\eqref{pr:density-2}, but not vice versa, hence the density-based formulation is  a relaxed version of the original worst-case problem 
 and we readily have $C= O(C^*)$.

Furthermore,  problem~\eqref{pr:density}-\eqref{pr:density-2} is convex and its unique solution  can be found using the Karush-Kuhn-Tucker (KKT) conditions.\footnote{In fact, problem~\eqref{pr:density}-\eqref{pr:density-2} is similar to the Euclidean projection of the gradient step explained in Chapter 3.} Regarding the constraints on $d_n$ about its minimum and maximum value, either one of them can be an equality, or none. This partitions $\mathcal N$ into three subsets, the `up-truncated' $\mathcal N_\uptruncated=\{n:d_n=1\}$ containing contents stored at all nodes, the `down-truncated' $\mathcal N_\downtruncated=\{n:d_n=\sfrac{1}{K}\}$ containing contents stored in just one node, and the complementary `non-truncated' $\mathcal N_\nottruncated=\mathcal N\setminus (\mathcal N_\uptruncated \cup \mathcal N_\downtruncated)$ of contents with $\sfrac{1}K < d_n <1$. Arranging $p_n$ in decreasing order, the partitions become $\mathcal N_\uptruncated  = \{1,2,\dots,l -1\}$, $\mathcal N_\nottruncated  = \{l ,l +1,\dots,r -1\}$, and $\mathcal N_\downtruncated =\{r ,r +2,\dots,N\}$; $l$ and $r$ are integers with $1\le l \le r \le N+1$. 
The solution $d_n$ is equal to\pagebreak[0]
\begin{subnumcases}{\label{eq:density}d_n=}
	1, &   $n\in \mathcal{N}_{\uptruncated}$,\\
	\frac{ M - l + 1 - \frac{N - r + 1}{K}}{
		\sum_{j\in\mathcal N_\nottruncated}p_j^{\frac{2}3}}\  p_n^{\frac{2}3},       &   $n\in \mathcal{N}_{\nottruncated}$,\ \ \ \ \ \ 
	\label{eq:density_nottruncated}
	\\
	\sfrac{1}{K}, & \!\! $n\in \mathcal{N}_{\downtruncated}$.
\end{subnumcases}

Fig.~\ref{fig:density} illustrates such an example solution, depicting the density $d_n$, indices $l$ and $r$, as well as sets $\mathcal  N_\uptruncated$, $\mathcal N_\nottruncated$, $\mathcal N_\downtruncated$ when content popularities follow the Zipf law (see Section~\ref{sec:zipfLaw}).

\begin{figure}[!htp]
	\centerline{
		\begin{overpic}[scale=.277]{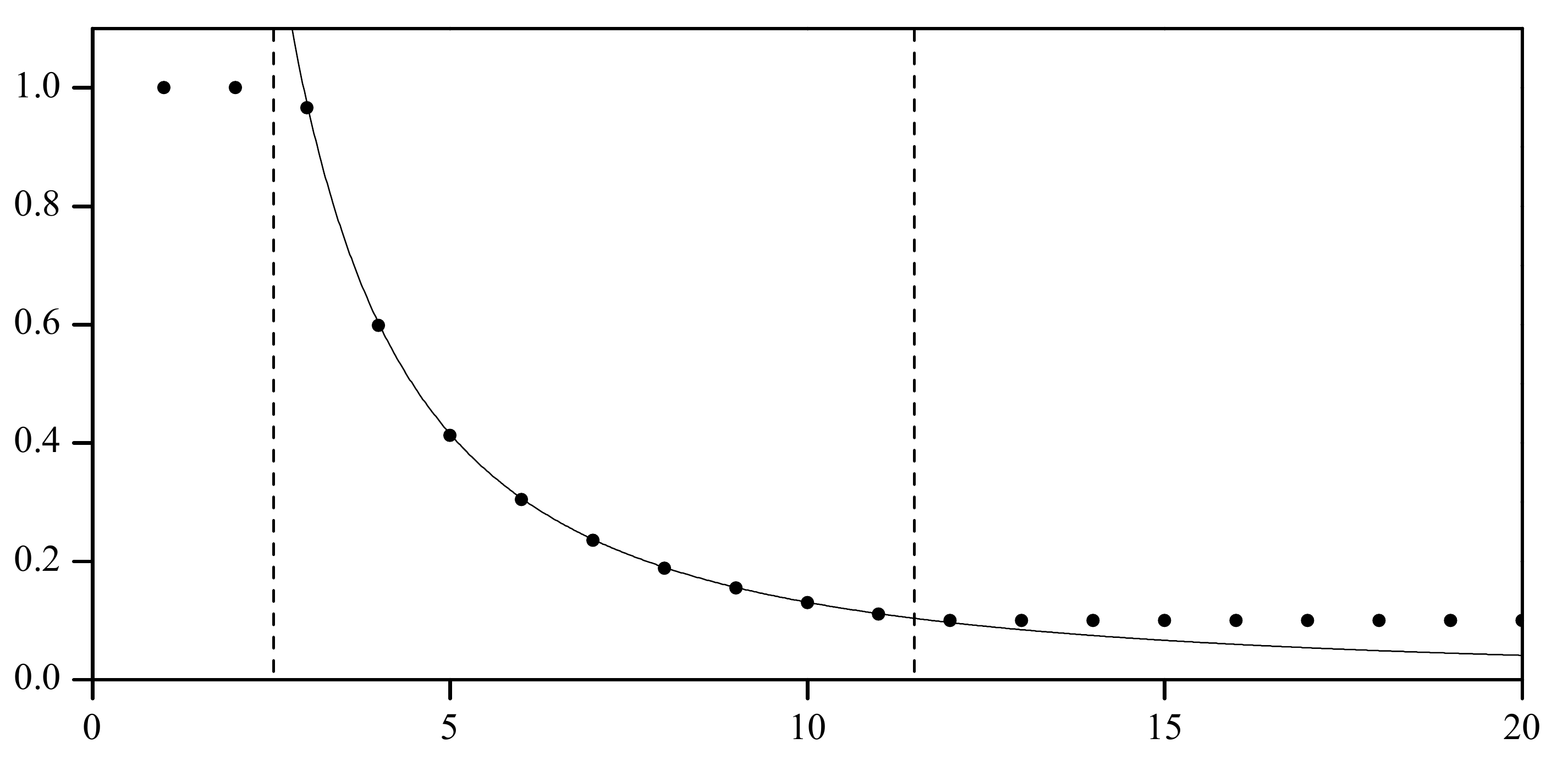}
			\put(-6.5,22.7){\small$d_n$}
			\put(44,-2.5){\small$n$}
			\put(19.1,1){\scriptsize$l$}
			\put(59.1,1){\scriptsize$r$}
			\put(9,22){$\mathcal N_\uptruncated$}
			\put(40,22){$\mathcal N_\nottruncated$}
			\put(77,22){$\mathcal N_\downtruncated$}
	\end{overpic}}
	\vskip 4pt
	\caption[An example case of density and set partitions.]{An example case of density $d_n$ and the $\mathcal N_\uptruncated, \mathcal N_\nottruncated$ and $\mathcal N_\downtruncated$ partitions. Solid line plots the $\sim n^{-\frac{2\tau}3}$ law of $n\in\mathcal N_\nottruncated$, when $p_n$ follows Zipf's law.}
	\label{fig:density}
\end{figure}
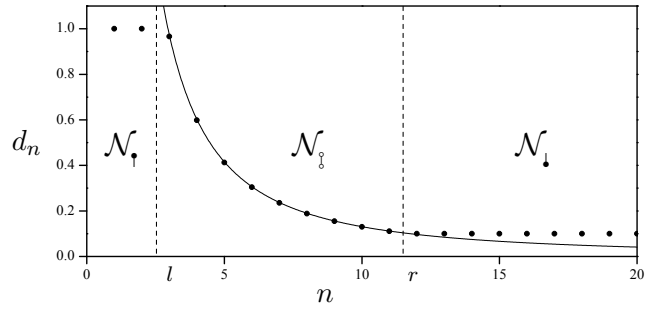

\subsubsection{Order Optimality of Rounded Densities}

The solution \eqref{eq:density} is not directly mapped to a feasible solution for Problem~\eqref{opt:wll1}-\eqref{opt:wll4}, but we may construct one via a two-step process, (i) first rounding $[d_n]$ to $[d_n^{\circ}]$, and (ii) second placing the content symmetrically on the network according to $[d_n^{\circ}]$ so that the constraints $|\mathcal{M}_k|\leq M$ are satisfied. For (i), we simply  define
$d_n^{\circ}\triangleq 4^{-\nu^\circ_n}$ rounded  to the largest power less or equal to $[d_n]$
\begin{equation}
	\label{d_truncated}
	d_n^{\circ}  \triangleq \max \left\{4^{-i}\!:4^{-i}\le d_n,\ i\in\{0,1,\dots,\nu\}\right\}.
\end{equation}

Then for (ii), \cite{savvas_ToIT} gives an algorithm to allocate the contents $\mathcal N$ in the caches $[\mathcal M_k]$ given the replication densities $d_n^{\circ}$. 
The algorithm can be explained by means of Fig.~\ref{fig_grid}. 
Suppose $[d_n^{\circ}]=(1,1/4,1/16,\dots,1/16)$.
We begin with the grey content which has $d_0^{\circ}=1$, this content is simply cached everywhere. Then for the content with $d_1^{\circ}=1/4$, we focus on a $2\times 2$ subgrid of nodes (any such subgrid suffices but we fix the origin to be the top left node in  figure~\ref{fig_grid}). In this subgrid we try to fill the diagonal, which in this case is achieved by placing content 1 at the coordinate $(1,1)$ (top left node in the grid). As a last step for this content, we place replicas by tiling the subgrid everywhere in the network. The result is that content 1 is replicated with density 1/4, as prescribed by the solution. Then for the contents with density 1/16 we enlarge the subgrid to $4\times 4$. In general the subgrid has a size $2^{\nu_m^{\circ}}\times 2^{\nu_m^{\circ}}$  and is aligned with all the considered subgrids. We then  fill the subgrid with the new contents starting  with the diagonal, notably contents 2, and 3 in the example. Then contents 4, 5, 6 are filled in the second diagonal which is below the first, while 7 completes the second diagonal by wrapping up at the coordinate $(1,4)$. During the filling, we only select nodes that have less contents than the maximum. For example, when we are filling object 2 in the subgrid $4\times 4$, we skip the node $(1,1)$ since that node has already two contents (content 0 and content 1), and we place content 2 at node $(2,2)$ which only had one content so far (content 0).
Where would we place content 10 with $d_{10}^{\circ}=1/64$? We would consider the subgrid $8\times 8$ (i.e. the entire grid), the first three diagonals are fully filled, hence we would place it in the first open spot in the fourth diagonal, that is node $(4,1)$.

Using the above canonical placement algorithm it is easy to show that a rounded feasible density solution $(d^{\circ}_n)$ can be mapped to a feasible solution for problem  \eqref{opt:wll1}-\eqref{opt:wll4}. This is because our algorithm places a content to the cache with the less stored contents so far. Therefore, if a cache has $M+1$ contents stored, it follows that all other caches must have at least $M$, and we may conclude that we have cached at least $NM+1$ replicas. This contradicts the fact that combining \eqref{d_truncated} and constraint \eqref{pr:density-2}, the total cached replicas can be at most $NM$. The value of this feasible solution is  $C^{\circ}=\Omega (C)$.
Finally, the following result is established in \cite{savvas_ToIT}.
\begin{figure}[!htp]
	\centerline{
		\begin{overpic}[scale=.377]{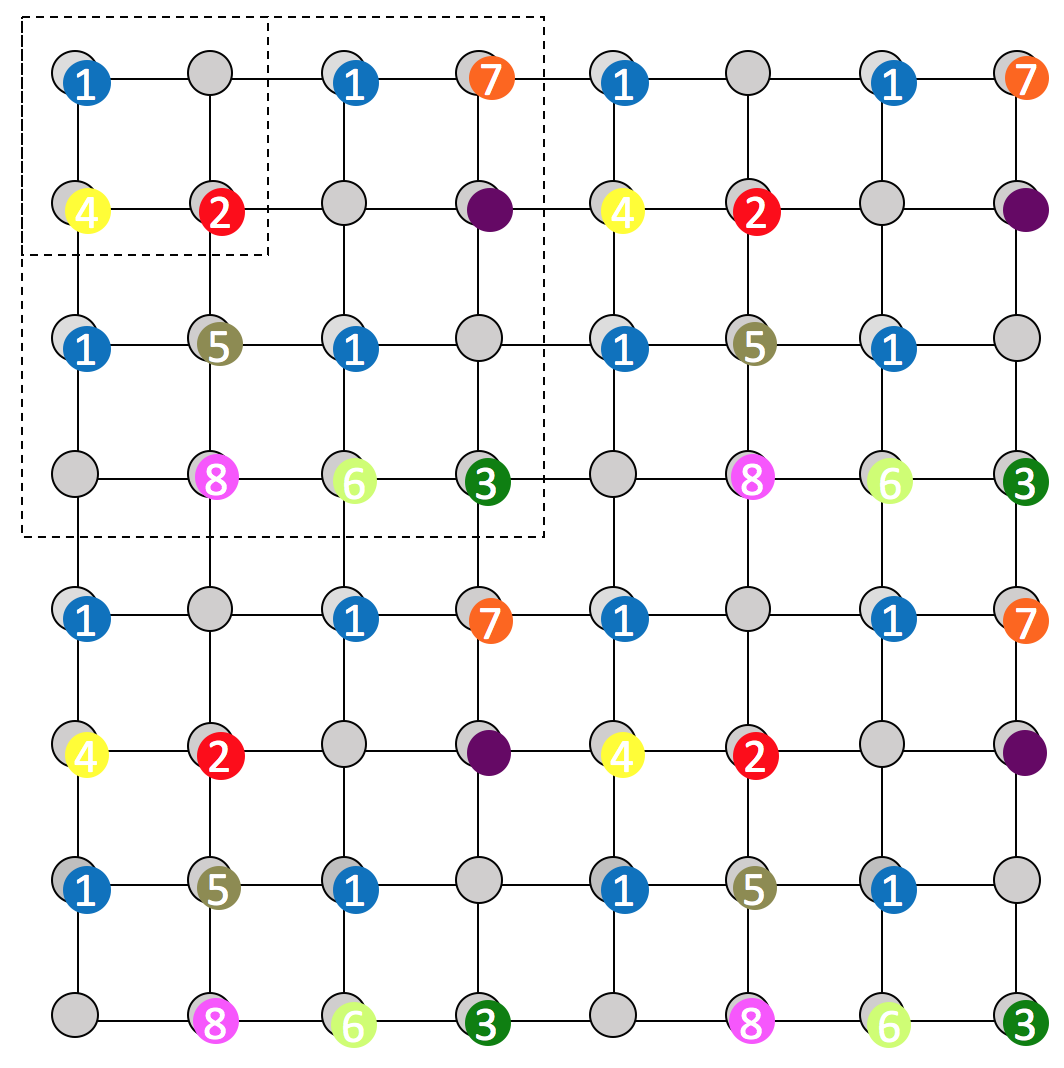} 
		\end{overpic}
	}
	\caption{An example of the canonical rounding used to obtain a feasible cache placement solution. 
		Here, the following items are placed $d_0^{\circ}=1$, $d_1^{\circ}=1/4$, $d_2^{\circ}=\dots=d_8^{\circ}=1/16$.}
	\vskip -3pt
	\label{fig_grid}
\end{figure}

\begin{thm}[theorem style=plain]{Order optimality of rounded densities}{CacheFillingANCNOptimality}
	There exist positive constants $a,b$ that depend on the distribution $(p_m)$, and cache capacity $K$, such that
	$$ 
	C^*\leq C^{\circ} \leq aC^*+b.
	$$
	Furthermore it is $\Theta(C^*)=\Theta(C^{\circ})=\Theta(C)$.
\end{thm}
The theorem is proven in \cite{savvas_ToIT} in two steps. First, it is established that the feasible solution $C^{\circ}$ obtained by rounding the densities and applying the canonical placement algorithm is of the same order as \eqref{opt:wll1}-\eqref{opt:wll4} when the objective is the average load in the network. This fact is established simply by counting the average hopcount of the canonical placement. Second,  a more elaborate proof  shows the order optimality with respect to \eqref{opt:wll1}-\eqref{opt:wll4} itself, by analyzing the local geometry of the canonical solution.
Theorem \ref{thm:CacheFillingANCNOptimality} is very powerful because it allows us to derive the scaling laws of wireless networks with caching $C^*$ using the  macroscopic analysis of~\eqref{pr:density}-\eqref{pr:density-2}. We remark that since $a,b$ above depend on $M$, scaling laws with respect to $M$ need to consider this dependence carefully. This is taken into account in the analysis of \cite{J_Gitzenis_14}.

\subsection{Asymptotic Laws for Zipf Popularity}
\label{sec:asymptotic_constant_cache}
\label{sec:zipfLaw}
To determine the scaling of $C$, we switch from the arbitrary popularity to the Zipf law. 
Substituting the solution \eqref{eq:density} and plugging in the Zipf distribution in the objective of Problem~\ref{pr:density}, it follows:
\begin{align}
	\label{eq:cap}
	C&\triangleq\sum_{m\in\mathcal N}\left(d_m^{-\frac{1}2}-1\right)p_m
	= C_{\nottruncated}+ C_{{\downtruncated}}- \sum_{j=l}^N p_m,
\end{align}
where $\sum_{j=l}^N p_m=O(1)$ (as it lies always in $[0,1]$), and  
\begin{align}
	\label{eq:C_nottruncated}\displaybreak[0]
	C_{\nottruncated} & \triangleq \!\sum_{m\in\mathcal N_\nottruncated}\frac{p_m}{\sqrt{d_m}} \stackrel{\text{\eqref{eq:density},\eqref{eq:powerlaw}}}{=}  \frac{\left[H_{\frac{2\tau}3}(r-1)-H_{\frac{2\tau}3}(l-1)\right]^{\mathrlap{\frac{3}2}}}{\sqrt{M-l+1-\frac{N-r+1}{K}}\ H_{\tau}(N)},\\\displaybreak[0]
	\label{eq:C_downtruncated}
	C_{{\downtruncated}} & \triangleq\! \sum_{m\in\mathcal N_\downtruncated}\frac{p_m}{\sqrt{d_m}} \stackrel{\text{\eqref{eq:density},\eqref{eq:powerlaw}}}{=} \sqrt{K}\ \frac{H_{\tau}(N)-H_{\tau}(r-1)}{H_{\tau}(N)}. \end{align}

To analytically compute the law of $C$, we will approximate $H_{\tau}(x)$ from \eqref{eq:H_approx}, and further analyze $l$ and $r$ using the KKT conditions. Since $N$ scales to infinity, $l,r$ may also scale to infinity, or not, depending on the actual solution \eqref{eq:density}.
Moreover, observing  that the  expressions \eqref{eq:C_nottruncated}-\eqref{eq:C_downtruncated} depend on $H_{\tau},H_{\frac{2\tau}3}$, we expect different cases to appear due to the form of \eqref{eq:H_approx}. 
In fact, previous work \cite{savvas_ToIT}  yields 5 cases depending on the values of $\tau$. These in terms are mapped to scaling laws via \eqref{eq:C_nottruncated}-\eqref{eq:C_downtruncated}. The results are presented in Table~\ref{table:main}.
The cases $\tau=1$ and $\tau=\sfrac{3}2$ are omitted to avoid clutter; they are similar to the  cases $\tau<1$ and $\tau>\sfrac{3}2$ up to a  logarithmic factor.

\begin{table}[t]
	\small
	\caption{Scaling Laws with Constant Cache Size}
	\label{table:main}
	\centering
	\begin{tabular}{ | c|c  || c | c | c | c | }
		\hline
		\multicolumn{2}{|c||}{\multirow{2}{*}{$N$}} & \multicolumn{1}{c|}{\multirow{2}{*}{$N$ finite}} & $\!K\!\rightarrow\!\infty$ then $\!$ &  \multicolumn{2}{c|}{$N\sim MK$, hence $N\!=\!\Theta(K)$}    \\ 
		\cline{5-6} \multicolumn{2}{|c||}{} & & $N\!\rightarrow\!\infty$ & $MK\!-\!N=\omega(1)$ & $\!MK\!-\!N\!=\!O(1)_{\vphantom{A_a}}^{\vphantom{A^a}}\!$   \\ 
		\hline  \multirow{5}{*}[-8pt]{   } & $\tau\!<\!1$ & $\Theta(1)$ & $\Theta\!\left(\sqrt{N}\right)$ $\vphantom{\displaystyle \sum^.}$ & $\Theta\!\left(\sqrt{K}\right)$ & $\Theta\!\left(\sqrt{K}\right)^{\vphantom{A^A}}_{\vphantom{A_A}}$ \\
		\cline{2-6}  $\!C\!$ & $\!1\!<\!\tau\!<\!\frac{3}2\!$ & $\Theta(1)$ & $\Theta\!\left(N^{\frac{3}2-\tau}\right)$ & $\!\Theta\!\left(\!\frac{\sqrt{K}}{(MK\!-\!N)^{\tau -1}}\right)^{\vphantom{A^a}}_{\vphantom{A_a}}$ & $\Theta\!\left(\sqrt{K}\right)$ \\
		\cline{2-6} & $\tau >\frac{3}2$ & $\Theta(1)$ & $\Theta(1)$ &  $\Theta\!\left(\!\frac{\sqrt{K}}{(MK\!-\!N)^{\frac{3(\tau -1)\!}{2\tau}}\!}\!\right)^{\vphantom{A^a}}_{\vphantom{A_a}}$ & $\Theta\!\left(\sqrt{K}\right)$ \\ \hline
	\end{tabular}
\end{table}

\textbf{Scaling laws for $M$ constant.} Table~\ref{table:main} shows how the solution of Problem~\ref{pr:density} scales with the system size parameters $K$ (number of nodes/users), and $N$ (catalog size); $M$ (cache size) in this table is a constant. From Theorem~\ref{thm:CacheFillingANCNOptimality}, the scaling of $C$ also provides the  required link capacity scaling for sustaining a uniform request rate $\lambda=1$, which was chosen to be our sustainability criterion. 
To help the reader, a $C$-scaling $O(1)$ means that the  network can  sustain unit throughput even if the link capacities are sufficiently large constants--the sustainable case. On the other hand,  $\Theta(\sqrt{N})$ means that to sustain unit per-user throughput, the link capacity needs to increase proportionally to the square root of the content catalog. Although we expect  $N$ to grow to infinity, link capacities obey Shannon limits, and this mismatch leads to an  unsustainable network. The admissible throughput scaling law for fixed capacities is given by the inverse;  for example $\Theta(\sqrt{K})$ is inverted to give $\lambda=\Theta\mleft(\sfrac{1}{\sqrt{K}}\mright)$ throughput with fixed capacities, which is the Gupta-Kumar scaling for random communicating pairs in a scaling network \cite{Kumar}. 

Our first observation from the last column of table~\ref{table:main} is that whenever the   replication capacity $MK-N=O(1)$ is low, i.e.~\emph{almost all cache slots are used to store each content once}, we retrieve the Gupta-Kumar regime. A small gain is obtained when we have more replications slots (column with $MK-N=\omega(1)$), but as long as the two scaling factors $K$ and $N$ increase proportionally to infinity, our caching technique has small (or no) effect.

The third column refers to the case where we take the limits in order, first $K\to\infty$ and then $N\to\infty$, i.e. we have $N=o(K)$ and the catalog grows  slower than (sublinearly to)  the network, a regime anticipated to be realistic--e.g. the catalogue of Youtube increases linearly in time, but the Youtube requests increase exponentially. The scaling laws in this case depend on the content catalog size, which yields potentially  a significant improvement over the Gupta-Kumar scaling. Also, the power law exponent of content popularity has a profound effect, where larger values of $\tau$ improve the sustainability (decrease the required link capacity).
Finally, when  $N$ is finite,  the system is always sustainable (required link capacity $\Theta(1)$). 

\textbf{Scaling laws for  $M\to\infty$.}
As memory becomes cheaper and cheaper, we may envisage the scenario where the per-node cache $M$ also scales to infinity. 
The extention of the analysis of \eqref{eq:C_nottruncated}-\eqref{eq:C_downtruncated}  in the case $M\to\infty$ is found in \cite{J_Gitzenis_14}. Here we briefly discuss the resulting scaling laws presented in table~\ref{table:synopsis}.
Considering how $N,M,K$ grow, we can study the system in different regimes of operation, which complicates the exposition of scaling laws. For this reason, we focus here on two specific regimes of interest. 
We compare the total network memory  $MK$ to the content catalog $N$, and split the analysis to two cases (i) $MK=\Omega(N)$ and (ii) $MK=O(N)$ (in fact \cite{J_Gitzenis_14} provides also the asymptotic constant that separates the two regimes).
The first case is called `High' $M$ and the second `Low' $M$.

\begin{table}[t]
	\renewcommand{\arraystretch}{1.9}
	\small
	\caption{Scaling Laws with Scaling Cache Size}
	\label{table:synopsis}
	\centering
	\begin{tabular}{ | c|c  || c | c | }
		\hline
		\multicolumn{2}{|c||}{} & `High' $M$ & `Low' $M$ \\ 
		\hline  \multirow{5}{*}[-8pt]{   } & $\tau<1$ & $O \left( \sqrt{ \frac{N}M}\right)$ & $\Theta\!\left(\sqrt{K}\right)^{\vphantom{A^A}}_{\vphantom{A_A}}$ \\[3pt]
		\cline{2-4}  $\!C\!$ & $\!1\!<\!\tau\!<\!\frac{3}2\!$  & $O\mleft( \frac{N^{\frac{3}2 - \tau}}{\sqrt{M}}\mright)$  & $\Theta\mleft( \frac{\sqrt{K}}{\left(MK -  N\right)^{\tau - 1}} \mright)$ \\[3pt]
		\cline{2-4} & $\tau >\frac{3}2$ &   $O(1)$ & $\Theta\mleft( \sqrt{ \frac{K}{MK - N}}\mright)$ \\[3pt] \hline
	\end{tabular}
\end{table}

\begin{itemize}
\item \textbf{`High' $M$ ($MK=\Omega(N)$).}
The most interesting regime to explore for perfect sustainability is the one of $C=O(1)$. As the formulas show, to keep $C$ bounded, the hardest case is on $\tau<1$:  cache size $M$ should scale as fast as content volume $N$ (e.g. a fixed $\gamma=M/N$ must be achieved). In the intermediate case of $1<\tau<\sfrac{3}2$, cache $M$ has to scale with $N$, but slower, at a sublinear power.
The case of $\tau>\sfrac{3}2$ is quite interesting, as $C=O(1)$ always holds true even for fast scaling catalogues.

\item\textbf{`Low' $M$ ($MK=O(N)$).} This regime  is characterized as unsustainable, because the replication slots $MK-N$ are not sufficient to keep $C$ low;
 $C$ scales as fast as $\sqrt{K}$, the Gupta-Kumar law \cite{Kumar}. The same is true if $\tau<1$. Small gains are observed if $\tau>1$ and $MK-N=\omega(1)$.

\end{itemize}

\section{Discussion of Related Work}

The work of \cite{Azimdoost16} studied the asymptotic laws of caching in the context of Information-Centric networks. Another line of work, extended the asymptotic laws to more realistic models for the wireless physical layer (PHY), and considered the aspect of mobility, showing that mobility does not really improve the scaling laws  \cite{leonardi}. The intricancies of wireless PHY were also captured in caching asymptotic laws in \cite{Liu16}. 

Contrary to the above uncoded considerations,  \cite{MaddahAli2014Fundamental} proposed the technique of ``coded caching'', as the optimal way of caching over a perfect broadcast medium. This idea has stirred a wide interest in the research community, and there exist numerous generalizations and interesting findings, which are not mentioned here; the interested reader is referred to \cite{paschos-jsac} as a starting point.
Continuing in asymptotic laws for caching, \cite{Ji2014Fundamental} studied the fundamental performance laws in device-to-device (D2D) wireless networks, while \cite{Cui16} studied the multicast delivery. A related survey \cite{d2dcaching} suggests that multihop D2D with spatial reuse has the same performance with coded caching.

\bibliographystyle{abbrv}
\bibliography{caching_literature_v2}  
\end{document}